\documentclass[12pt]{article}

\usepackage[margin=1in]{geometry}
\usepackage{amsmath,amssymb,amsthm}
\usepackage{authblk}
\usepackage{bm}
\usepackage{enumitem}
\usepackage{natbib}
\usepackage[colorlinks=true,allcolors=blue]{hyperref}
\usepackage{appendix}
\usepackage{setspace}
\usepackage{booktabs}
\usepackage{multirow}
\usepackage{graphicx}

\newcommand{\E}{\mathbb{E}}

\newcommand{\plim}{\overset{p}{\to}}

\newtheorem{theorem}{Theorem}
\newtheorem{lemma}{Lemma}
\newtheorem{corollary}{Corollary}

\makeatletter
\def\mojiparline#1{
    \newcounter{mpl}
    \setcounter{mpl}{#1}
    \@tempdima=\linewidth
    \advance\@tempdima by-\value{mpl}zw
    \addtocounter{mpl}{-1}
    \divide\@tempdima by \value{mpl}
    \advance\kanjiskip by\@tempdima
    \advance\parindent by\@tempdima
}
\makeatother
\def\linesparpage#1{
    \baselineskip=\textheight
    \divide\baselineskip by #1
}

\title{
\vspace{-10mm}
On the uncertainty from the first-stage estimation of prognostic covariate adjustment in randomized controlled trials}
\author{Nodoka Seya\thanks{Corresponding author: nodoka.seya.bs@gmail.com} }
 \author{Masataka Taguri}
  \affil{Department of Health Data Science, Tokyo Medical University, Tokyo, Japan}
\date{}

\newenvironment{keywords}{\par\noindent\textbf{Keywords:}}

\begin{document}

\onehalfspacing

\maketitle
\begin{abstract}
\noindent Prognostic covariate adjustment (PROCOVA) is a two-sample two-stage estimation method for covariate adjustment in randomized controlled trials.
In the first stage, a prognostic score, defined as the conditional expectation of an outcome given covariates under the control treatment, is estimated using historical data.
In the second stage, analysis of covariance with the estimated prognostic score and treatment assignment as explanatory variables is performed, and the average treatment effect is estimated. 
Although the prognostic score is estimated in this procedure, the variance estimator, which treats the prognostic score as known, has been used. 
Furthermore, the difference in the asymptotic variance between cases where the prognostic score is known versus where it is estimated has not been previously clarified. 
In this study, we derived these two asymptotic variances and showed that they are equal. 
This result also holds when the prognostic score is estimated using machine learning with $L_2$ consistency.
We also constructed two variance estimators: one that treats the prognostic score as known, and another that accounts for its estimation, and compared their performance through simulation studies and data applications.
For PROCOVA, since both variance estimators are asymptotically valid, it is generally recommended to use a variance estimator that treats the prognostic score as known, as it is simpler to derive and implement.
When historical data is small, a variance estimator that explicitly accounts for prognostic score estimation is recommended if conservative inference is preferred.
\end{abstract}


\thispagestyle{empty}
\clearpage
\addtocounter{page}{-1}

\clearpage
\doublespacing
\linesparpage{25}

\section{Introduction}
Randomized controlled trials (RCTs) are designs that randomly assign subjects to treatment or control groups, guaranteeing (baseline) covariate balance between groups at the population level.
In RCTs, the average treatment effect (ATE) can be consistently estimated through simple comparisons between groups without covariate adjustment.
However, even in RCTs, covariates between groups may be imbalanced by chance in a sample.
A chance imbalance in covariates associated with the outcome leads to efficiency loss.

Covariate adjustment in RCTs improves efficiency while ensuring consistency for ATE \citep{tsiatis2008covariate}.
Efficiency gain leads to increased power and a reduced required sample size.
Consequently, covariate adjustment is recommended by guidelines from regulatory agencies such as FDA and EMA \citep{FDA2023covariates, EMA2015baselineCovariates}, and is widely conducted in modern RCTs.
The most popular and widely studied method for adjusting covariates is analysis of covariance (ANCOVA; 
\citealp{leon2003semiparametric, lesaffre2003note}).
ANCOVA regresses the outcome on the treatment and covariates, using the ordinary least squares (OLS) estimator of the treatment coefficient as the estimator of ATE.
If the conditional expectation of the outcome given the treatment and covariates is linear and the assignment design is 1:1 simple randomization, ANCOVA can achieve the minimum asymptotic variance, regardless of whether it includes an interaction term between the treatment and covariates \citep{ye2023toward}.

Recently, \cite{schuler2022increasing} proposed prognostic covariate adjustment (PROCOVA) which uses historical data to create a score that predicts outcomes well and subsequently includes the score in an ANCOVA model.
Historical data is typically previous clinical trial data or real-world data.
PROCOVA is a two-sample two-stage estimation method. 
In the first stage, using historical data, the model for the prognostic score, which is a conditional expectation of the outcome given covariates under the control treatment, is learned. 
In the second stage, prognostic scores are predicted for trial data, and ANCOVA including the predicted prognostic score as an explanatory variable is performed.
EMA has provided a positive opinion on PROCOVA \citep{EMA2022PROCOVA}.
\cite{holzhauer2023super} proposed a covariate adjustment method essentially identical to PROCOVA, referring to the prognostic score as a ``super-covariate.''

Beyond PROCOVA, two-stage estimation is widely used in biostatistics, epidemiology, and econometrics.
Typical examples include inverse probability weighting (IPW; \citealp{horvitz1952generalization}), for addressing missing data \citep{li2013weighting}, for adjusting confounding in observational studies \citep{robins1994estimation}, and for adjusting covariates in RCTs \citep{shen2014inverse}.
Other examples include regression calibration for correcting measurement error \citep{carroll1990approximate} and two-stage least squares (TSLS; \citealp{theil1953repeated}) as an instrumental variable method \citep{imbens2014instrumental}.
In IPW for ATE estimation, using an estimated propensity score model results in smaller asymptotic variance than using the true propensity score model in both observational studies \citep{lunceford2004stratification, henmi2004paradox} and RCTs \citep{shen2014inverse}.
In regression calibration and TSLS, a variance estimator that does not take into account the uncertainty in the first-stage estimation underestimates the true asymptotic variance \citep{carroll1990approximate, murphy2002estimation}. 
Importantly, variance estimators that ignore the first-stage uncertainty generally do not yield valid inferences in two-stage estimation.

In PROCOVA, although the prognostic score model is estimated using historical data, a variance estimator that treats the prognostic score model as known has been used \citep{schuler2022increasing, hojbjerre2025tutorial}. 
To the best of our knowledge, the theoretical justification of this variance estimator has not been discussed, a variance estimator that accounts for the uncertainty in prognostic score model estimation has not been considered, and the difference in variance between cases when the prognostic score model is known and when it is estimated has not been clarified.

To bridge the gap between theory and the current practice of variance estimation for PROCOVA, we derive and compare two types of asymptotic variances for PROCOVA.
One treats the prognostic score model as known, and the other takes into account the uncertainty in prognostic score model estimation.
Our theoretical evaluation reveals that these two asymptotic variances are identical for the PROCOVA-based ATE estimator (see Theorems \ref{thm: asyvar PROCOVA}--\ref{thm: asyvar PROCOVA general}). 

The rest of this article is organized as follows.
In Section 2, we briefly review two-sample two-stage estimation and its special case, PROCOVA.
In Section 3, we derive the two asymptotic variances for two-sample two-stage estimation: one for the case where the first-stage model is known and one for the case where it is estimated, and prove that they are identical for the ATE estimator in PROCOVA and related covariate adjustment methods.
In Section 4, we construct two variance estimators corresponding to the two asymptotic variances derived in Section 3, and discuss their properties and relationships.
In Section 5, we conduct simulation studies to compare the performance of these two variance estimators for PROCOVA. 
In Section 6, we apply PROCOVA with the two variance estimators to actual RCT data.
In Section 7, we summarize this article and provide practical guidance for investigating asymptotic properties, conducting simulation studies, and estimating variance in actual RCTs using PROCOVA.

\section{Two-sample two-stage estimation and PROCOVA}
In Section \ref{subsec: rev two-stage}, we describe a two-sample two-stage estimation framework.
Section \ref{subsec: rev PROCOVA} describes PROCOVA as a special case of this framework. 

\subsection{Two-sample two-stage estimation}
\label{subsec: rev two-stage}
Let $O_1,...,O_n$ be independent and identically distributed copies of the random variable vector $O$ and let $\mathcal{D}:=\{O_i\}_{i=1}^{n}$.
Similarly, let $\tilde O_1,...,\tilde O_{\tilde n}$ be independent and identically distributed copies of the random variable vector $\tilde O$ and let $\mathcal{\tilde D}:=\{\tilde O_i\}_{i=1}^{\tilde n}$.
We assume $\mathcal{D}\perp\mathcal {\tilde D}$.

Then, we describe two-stage estimation using the two samples $\mathcal{D}$ and $\mathcal{\tilde D}$.
Let $\phi(\tilde O; \theta)\in \mathbb{R}^p$ be the first-stage estimating function for a parameter vector $\theta\in \mathbb{R}^p$ and define 
\[
\Phi_{\tilde n}(\theta):=\frac{1}{\tilde n}\sum_{i=1}^{\tilde n}\phi(\tilde O_i;\theta).
\]
Then, we denote the solution of equation $\Phi_{\tilde n}(\theta)=0$ as $\hat{\theta}_{\tilde n}$ and the solution of equation $\mathbb{E}[\phi(\tilde O;\theta)]=0$ as $\theta^*$.
Let $\psi(O; \beta,\theta)\in \mathbb{R}^d$ be the second-stage estimating function for a parameter vector $\beta\in \mathbb{R}^d$ given $\theta\in \mathbb{R}^p$ and define 
\[
\Psi_{n}(\beta,\theta):=\frac{1}{n}\sum_{i=1}^{n}\psi( O_i;\beta,\theta).
\]
Then, we denote the solution of equation $\Psi_{n}(\beta,\theta)=0$ as $\hat{\beta}_n(\theta)$ given $\theta\in \mathbb{R}^p$ and denote the solution of equation $\mathbb{E}[\psi(O;\beta,\theta)]=0$ as $\beta^*(\theta)$.
Especially, we denote $\beta^*:=\beta^*(\theta^*)$.
Then, the two-sample two-stage estimator is written as $\hat{\beta}_n(\hat\theta_{\tilde n})$.

Two-sample two-stage estimation includes several methods of biostatistics and epidemiology.
If $\phi(\tilde O; \theta)$ is an estimating function corresponding to a calibration model using external data $\mathcal{\tilde D}$ and $\psi(O; \beta,\theta)$ is an estimating function corresponding to an analysis model using main data $\mathcal{D}$, this corresponds to regression calibration for correcting measurement error \citep{carroll1990approximate} using external validation data (main/external validation design; \citealp{li2026mainexternal}).
If $\phi(\tilde O; \theta)$ is an estimating function corresponding to a prognostic score model using historical data $\mathcal{\tilde D}$ and $\psi(O; \beta,\theta)$ is an estimating function corresponding to a PROCOVA model using trial data $\mathcal{D}$, this corresponds to PROCOVA for ATE estimation in RCTs \citep{schuler2022increasing}.

\subsection{PROCOVA}
\label{subsec: rev PROCOVA}
Consider $O=(W,A,Y)$, where $W\in \mathbb R^p$ is a baseline covariate vector that includes a constant, $A\in\{0,1\}$ is a treatment assignment and $Y\in \mathbb{R}$ is the outcome.
Then, let trial data be $\mathcal{D}=\{O_i=(W_i,A_i,Y_i)\}_{i=1}^{n}$.
Consider $\tilde O=(\tilde W,\tilde A=0,\tilde Y)$, where $\tilde W\in \mathbb R^p$ is a baseline covariate vector that includes a constant, $\tilde A=0$ is a treatment assignment and $\tilde Y\in\mathbb{R}$ is an outcome. 
Then, let historical data be $\mathcal{\tilde D}=\{\tilde O_i=(\tilde W_i,\tilde A_i=0,\tilde Y_i)\}_{i=1}^{\tilde n}$.

Let $Y^1$ and $Y^0$ be the potential outcomes under treatment and control, respectively.
Then, ATE is defined as $\mathbb{E}[Y^1-Y^0]$. 
We assume $Y=AY^1+(1-A)Y^0$. 
RCTs where the treatment is randomly assigned with a probability $\mathbb{P}[A=1\mid W, Y^1,Y^0]=\mathbb{P}[A=1]=:\pi$ guarantee 
\begin{equation*}
A\perp (W,Y^1,Y^0).    
\end{equation*}
Thus, in RCTs, ATE can be written as $\mathbb{E}[Y\mid A=1]-\mathbb{E}[Y\mid A=0]$.

PROCOVA was proposed by \cite{schuler2022increasing} as a covariate adjustment method for estimating ATE in RCTs and is performed as follows:
\begin{enumerate}
    \item Using historical data $\mathcal{\tilde D}$, learn the model for the prognostic score $\rho(w):=\mathbb{E}[Y\mid A=0, W=w]$, and denote the learned model as $\hat \rho_{\tilde n}(w)$.
    For example, consider the following linear regression model:
    \begin{equation}
    \label{eq: PS model linear}
   \tilde Y_i=\theta^\top \tilde W_i +\tilde\varepsilon_i,
    \end{equation}
    and let $\hat \rho_{\tilde n}(w)=\hat\theta_{\tilde n}^\top w$ be the learned model, where $\hat\theta_{\tilde n}$ is the OLS estimator of $\theta$.
    In this case, the first-stage estimating function can be written as 
    \begin{equation}
    \label{eq: est func prog}
    \phi(\tilde O;\theta)=(\tilde Y-\theta^\top \tilde W)\tilde W.
    \end{equation}
    \item For trial data $\mathcal{D}$, fit the ANCOVA model including the estimated prognostic score $\hat \rho_{\tilde n}(W_i)$. For example, consider the following linear model:
    \begin{equation}
    \label{eq: PROCOVA general}
        Y_i=\beta^\top X_{\hat\rho_{\tilde n},i} + \varepsilon_i,\quad X_{\hat\rho_{\tilde n},i}=(1,A_i,\hat \rho_{\tilde n}(W_i))^\top, \quad \beta=(\beta_0,\beta_A,\beta_1)^\top,
    \end{equation}
    and let the OLS estimator of ${\beta}_{A}$ be the PROCOVA-based ATE estimator.
    If the first-stage estimating function is the equation \eqref{eq: est func prog}, the model \eqref{eq: PROCOVA general} reduces to  
    \begin{equation}
    \label{eq: PROCOVA linear}
        Y_i=\beta^\top X_{\hat\theta_{\tilde n},i} + \varepsilon_i, \quad X_{\hat\theta_{\tilde n},i}=(1,A_i,\hat \theta_{\tilde n}^\top W_i)^\top, \quad \beta=(\beta_0,\beta_A,\beta_1)^\top,
    \end{equation}
    In this case, the second-stage estimating function can be written as
    \begin{equation}
    \label{eq: est func PROCOVA linear}
        \psi(O;\beta,\theta)=(Y-\beta^\top X_{\theta})X_{\theta},\quad X_\theta=(1,A,\theta^\top W)^\top,
    \end{equation}
    and the PROCOVA-based ATE estimator is $e^\top \hat{\beta}_n(\hat{\theta}_{\tilde n})$ with $e=(0,1,0)^\top$.
\end{enumerate}

In RCTs, the PROCOVA-based ATE estimator is consistent and asymptotically normal for ATE, even if the prognostic score and ANCOVA model are misspecified \citep{schuler2022increasing}.
The specification of the prognostic score model affects the asymptotic variance.
If the treatment effect is constant, i.e., $\mathbb{E}[Y^1-Y^0\mid W]= \mathbb{E}[Y^1-Y^0]$ and $\hat \rho_{\tilde n}(W)$ has $L_2$-consistency for $\rho(W)$, then the PROCOVA-based ATE estimator can achieve the minimum asymptotic variance \citep{schuler2022increasing}.

\section{Asymptotic variances}
\label{sec: theory}
In Section \ref{subsec: theory two-stage}, for two-sample two-stage estimation, we derive the difference between the asymptotic variance when the first-stage model is known and that when it is estimated, and provide sufficient conditions for this difference to be zero, i.e., the first-stage uncertainty to be negligible.
The derivations are closely related to the theory of one-sample two-stage estimation developed by \cite{newey1994large} and \cite{lok2024estimating}, and some of the sufficient conditions that we derive are the same as those of the one-sample two-stage estimation case \citep{newey1994large,lok2024estimating}.
In Section \ref{subsec: theory PROCOVA}, we examine whether the condition for the first-stage uncertainty to be negligible holds for PROCOVA and its variant using ANHECOVA or ANCOVA II model \citep{yang2001efficiency, ye2023toward} as the PROCOVA model.
In Section \ref{subsec: theory other}, we also examine whether the condition for the first-stage uncertainty to be negligible holds for one-sample two-stage covariate adjustment methods in RCTs, such as the augmented estimator \citep{tsiatis2008covariate} and the IPW estimator \citep{shen2014inverse}.
In Section \ref{subsec: theory PROCOVA ML}, we extend the result for PROCOVA to cases that do not rely on specific model forms. 
Our theory can be applied when the prognostic score is estimated using machine learning with $L_2$-consistency.

\subsection{Conditions for the first-stage uncertainty to be negligible}
\label{subsec: theory two-stage}
By expanding $\Psi_{n}(\beta,\hat \theta_{\tilde n})$ around $\beta^*$ to solve $\sqrt n(\hat\beta_n(\hat\theta_{\tilde n})-\beta^*)$ and then expanding the result around $\theta^*$, we have
\begin{equation}
\label{eq: expand beta est main}
       \sqrt n(\hat\beta_n(\hat\theta_{\tilde n})-\beta^*)=-Q_0^{-1}\sqrt n\,\Psi_n(\beta^*,\theta^*)+\sqrt{\kappa}Q_0^{-1}Q_1Q_2^{-1}\sqrt{\tilde n}\Phi_{\tilde n}(\theta^*)+o_p(1),
\end{equation}
where $Q_0=\mathbb{E}\!\left[\frac{\partial}{\partial\beta^\top}\psi(O;\beta,\theta^*)\middle|_{\beta=\beta^*}\right]$, $Q_1=\mathbb{E}\!\left[\frac{\partial}{\partial\theta^\top}\psi(O;\beta^*,\theta)\middle|_{\theta=\theta^*}\right]$, $Q_2=\mathbb{E}\!\left[\frac{\partial}{\partial\theta^\top}\phi(\tilde O;\theta)\middle|_{\theta=\theta^*}\right]$ and $\lim_{n,\tilde n\to\infty} {n}/{\tilde n}=:\kappa\in[0,\infty)$. 
Since $\sqrt n(\hat\beta_n(\theta^*)-\beta^*)=-Q_0^{-1}\sqrt n\,\Psi_n(\beta^*,\theta^*)+o_p(1)$, the difference in the asymptotic distributions between the case where $\theta$ is given and the case where $\theta$ is estimated appears in $\sqrt{\kappa}Q_0^{-1}Q_1Q_2^{-1}\sqrt{\tilde n}\Phi_{\tilde n}(\theta^*)$, which is the second term on right-hand side of equation \eqref{eq: expand beta est main}.
For detailed derivations, see Appendix \ref{appendix: thm 1}.
Regularity conditions are given in Appendix \ref{appendix: assumption}.

On the right-hand side of equation \eqref{eq: expand beta est main}, the first term $-Q_0^{-1}\sqrt n\,\Psi_n(\beta^*,\theta^*)$ depends only on $\mathcal{D}$ and the second term $\sqrt{\kappa}Q_0^{-1}Q_1Q_2^{-1}\sqrt{\tilde n}\Phi_{\tilde n}(\theta^*)$ depends only on $\mathcal{\tilde D}$.
Since $\mathcal{D}\perp\mathcal{\tilde D}$, by applying the central limit theorem to equation \eqref{eq: expand beta est main}, $\sqrt n(\hat\beta_n(\hat\theta_{\tilde n})-\beta^*)$ is asymptotically normal with mean zero and variance 
\begin{equation}
\label{eq: asyvar est main}
    V_{\mathrm{est}}=V_{\mathrm{fix}} + \kappa\, Q_0^{-1}Q_1 V_\theta Q_1^\top (Q_0^{-1})^\top,
\end{equation}
where $V_{\mathrm{fix}}=Q_0^{-1}\Omega (Q_0^{-1})^\top$, $\Omega = \E\!\left[\psi(O;\beta^*,\theta^*)\psi(O;\beta^*,\theta^*)^\top\right]$, $V_\theta=Q_2^{-1}Q_3(Q_2^{-1})^{\top}$ and $Q_3=\E\!\left[\phi(\tilde O;\theta^*)\phi(\tilde O;\theta^*)^\top\right]$.
For detailed derivations, see Appendix \ref{appendix: thm 1}.

$V_{\mathrm{fix}}$ represents the asymptotic variance when $\theta$ is fixed at $\theta^*$.
$V_{\mathrm{est}}$ represents the asymptotic variance when $\theta$ is estimated by $\hat\theta_{\tilde n}$.
Equation \eqref{eq: asyvar est main} implies $V_{\mathrm{est}}\succeq V_{\mathrm{fix}}$ because of $\kappa\, Q_0^{-1}Q_1 V_\theta Q_1^\top (Q_0^{-1})^\top\succeq0$.
That is, a variance estimator that ignores uncertainty in the first-stage estimation generally underestimates $V_{\mathrm{est}}$.
If $\kappa=0$ or $Q_0^{-1}Q_1=0$, then $V_{\mathrm{est}}=V_{\mathrm{fix}}$,  i.e., the asymptotic variance of a two-sample two-stage estimator is the same regardless of whether the first-stage parameter is given or estimated. 

The condition $\kappa=0$ means that $\tilde n$ grows faster than $n$.
In practice, $\kappa=0$ would hold when $\tilde n$ is much larger than $n$.
However, this situation may not always be realistic.
The theory of \cite{schuler2022increasing} and simulations of \cite{hojbjerre2025tutorial} do not satisfy $\kappa=0$ because $\tilde n$ grows in tandem with $n$.

Now we consider the condition $Q_0^{-1}Q_1=0$.
This condition is equivalent to $Q_1=0$, which corresponds to Neyman orthogonality \citep{chernozhukov2017double} with $\theta$ as nuisance parameters, because $Q_0$ is nonsingular.
By assuming the exchangeability of differentiation and integration, the implicit function theorem gives 
\begin{equation}
\label{eq: cond Q0Q1}
\begin{split}
    \left.\frac{\partial}{\partial\theta^\top}\beta^*(\theta)\middle |_{\theta=\theta^*}\right.=-Q_0^{-1}Q_1.
\end{split}
\end{equation}
For the detailed derivation, see Appendix \ref{appendix: thm 2}. 
From equation \eqref{eq: cond Q0Q1}, we see that the condition $Q_0^{-1}Q_1=0$ is equivalent to $\left.\frac{\partial}{\partial\theta^\top}\beta^*(\theta)\middle |_{\theta=\theta^*}\right.=0$.
In several cases, $\left.\frac{\partial}{\partial\theta^\top}\beta^*(\theta)\middle |_{\theta=\theta^*}\right.=0$ may be easier to check than $Q_0^{-1}Q_1=0$.
If $\beta^*(\theta)$ is the same regardless of $\theta$, i.e., the probability limit of $\hat\beta_n(\theta)$ is the same regardless of $\theta$, then $\left.\frac{\partial}{\partial\theta^\top}\beta^*(\theta)\middle |_{\theta=\theta^*}\right.=0$ holds.
In terms of causal inference, this means that $\theta$ does not affect the identification of $\beta$.
For the one-sample case, an example of $\theta$ that does not affect the identification of $\beta$ is the parameter in the numerator of stabilized inverse probability weights \citep{lok2024estimating}.

If interest is limited to the specific element of $\beta$, specifically, $e^\top \beta$ with a constant vector $e\in \mathbb{R}^d$, this condition can be relaxed as follows:
\begin{equation}
\label{eq: cond a}
\begin{split}
    \left.\frac{\partial}{\partial\theta^\top}e^\top\beta^*(\theta)\middle |_{\theta=\theta^*}\right.=-e^\top Q_0^{-1}Q_1=0.
\end{split}
\end{equation}
The condition \eqref{eq: cond a} is essentially the same as that for the one-sample case (see equation (6.8) in \cite{newey1994large}).
Sections \ref{subsec: theory PROCOVA} and \ref{subsec: theory other} focus on the condition \eqref{eq: cond a}. 
Specifically, we consider whether the additional variance term due to the first-stage estimation, which corresponds to the second term on the right-hand side in equation \eqref{eq: asyvar est main}, vanishes for the ATE estimator in PROCOVA and related covariate adjustment methods through examining whether the condition \eqref{eq: cond a} holds.

\subsection{Asymptotic variances of PROCOVA and its variant} 
\label{subsec: theory PROCOVA}
Recall that the PROCOVA-based ATE estimator is consistent for ATE regardless of the misspecification of the prognostic score model \citep{schuler2022increasing}.  
This suggests that $e^\top\beta^*(\theta)$ with $e=(0,1,0)^\top$ is the same regardless of $\theta$, i.e., the condition \eqref{eq: cond a} holds.
In fact, the following theorem holds. 
The proof is given in Appendix \ref{appendix: thm 3}.
\begin{theorem}
\label{thm: asyvar PROCOVA}
Consider the two-sample two-stage estimation with the first-stage estimating function \eqref{eq: est func prog} and the second-stage estimating function \eqref{eq: est func PROCOVA linear}.
Assume regularity conditions (A1)--(A4), $\mathbb{E}[Y^2]<\infty$ and $\mathbb{E}[||W||^2]<\infty$.
Additionally, we assume $A\perp W$ and $0<\pi< 1$.
Then, for $e=(0,1,0)^\top$,
\begin{equation*}
   e^\top\beta^*(\theta)=\mathbb{E}[Y\mid A=1]-\mathbb{E}[Y\mid A=0] \quad\text{and}\quad\left.\frac{\partial}{\partial\theta^\top}e^\top\beta^*(\theta)\middle |_{\theta=\theta^*}\right.=-e^\top Q_0^{-1}Q_1 =0.
\end{equation*}
\end{theorem}
Note that RCTs guarantee $A\perp W$ and $0<\pi < 1$.
Then, Theorem \ref{thm: asyvar PROCOVA} implies that the PROCOVA-based ATE estimator, i.e., the estimator of $\beta_A$ in the PROCOVA model \eqref{eq: PROCOVA linear}, has the same asymptotic variance regardless of whether the prognostic score model is known or estimated. 
Therefore, the variance estimator considered in previous studies \citep{hojbjerre2025tutorial,schuler2022increasing} is appropriate.
However, for the intercept $\beta_0$ and coefficient of the prognostic score $\beta_1$ in the PROCOVA model \eqref{eq: PROCOVA linear}, the variance estimator that ignores uncertainty in prognostic score estimation generally leads to underestimation (see Appendix \ref{appendix: thm 3}).

Next, we examine a variant of PROCOVA replacing the model \eqref{eq: PROCOVA linear} with
\begin{equation}
\label{eq: PROCOVAII linear emp}
\begin{split}
     &Y_i=\beta^\top X_{\hat\theta_{\tilde n},i}^{\mathrm{emp}} + \varepsilon_i,\quad \beta=(\beta_0,\beta_A,\beta_1, \beta_2)^\top,\\
     &X_{\hat\theta_{\tilde n},i}^{\mathrm{emp}}=\left(1,A_i,\hat\theta_{\tilde n}^\top W_i-\frac{1}{n}\sum_{j=1}^n\hat\theta_{\tilde n}^\top W_j,  A_i\left(\hat\theta_{\tilde n}^\top W_i-\frac{1}{n}\sum_{j=1}^n\hat\theta_{\tilde n}^\top W_j\right)\right)^\top,
\end{split}
\end{equation}
where the product term of the treatment assignment and the prognostic score is added.
While the model \eqref{eq: PROCOVA linear} is a variant of ANCOVA or ANCOVA I, the model \eqref{eq: PROCOVAII linear emp} is a variant of ANHECOVA or ANCOVA II \citep{yang2001efficiency, ye2023toward}.
PROCOVA using model \eqref{eq: PROCOVAII linear emp} does not directly fit into the framework described in Section \ref{subsec: rev two-stage}.
However, as shown in Appendix \ref{appendix: proof pop emp}, the difference between the estimator when $\theta$ is estimated by $\hat \theta_{\tilde{n}}$ and that when $\theta$ is given as $\theta^*$ in the case of model \eqref{eq: PROCOVAII linear emp} is asymptotically equivalent to their difference in the case of model corresponding to the following estimating function 
\begin{equation}
    \label{eq: est func PROCOVAII linear}
        \psi(O;\beta,\theta)=(Y-\beta^\top X_{\theta})X_{\theta},\quad X_\theta=(1,A,\theta^\top W-\mathbb{E}[\theta^\top W],(\theta^\top W-\mathbb{E}[\theta^\top W])A)^\top.
\end{equation}
Thus, it is enough to examine the asymptotic behavior of PROCOVA with the estimating function \eqref{eq: est func PROCOVAII linear}.
Now, the following theorem holds. The proof is given in Appendix \ref{appendix: thm 4}.
\begin{theorem}
\label{thm: asyvar PROCOVAII}
Consider the two-sample two-stage estimation with the first-stage estimating function \eqref{eq: est func prog} and the second-stage estimating function \eqref{eq: est func PROCOVAII linear}.
Assume regularity conditions (A1)--(A4), $\mathbb{E}[Y^2]<\infty$ and $\mathbb{E}[||W||^2]<\infty$.
Additionally, we assume $A\perp W$ and $0<\pi< 1$.
Then, for $e=(0,1,0,0)^\top$, 
\begin{equation*}
    e^\top\beta^*(\theta)=\mathbb{E}[Y\mid A=1]-\mathbb{E}[Y\mid A=0] \quad\text{and}\quad\left.\frac{\partial}{\partial\theta^\top}e^\top\beta^*(\theta)\middle |_{\theta=\theta^*}\right.=-e^\top Q_0^{-1}Q_1 =0.
\end{equation*}
Furthermore, for $e=(1,a,0,0)^\top,\ a=0,1$,
\begin{equation*}
    e^\top\beta^*(\theta)=\mathbb{E}[Y\mid A=a] \quad\text{and}\quad\left.\frac{\partial}{\partial\theta^\top}e^\top\beta^*(\theta)\middle |_{\theta=\theta^*}\right.=-e^\top Q_0^{-1}Q_1 =0.
\end{equation*}
\end{theorem}
Theorem \ref{thm: asyvar PROCOVAII} shows that even when using \eqref{eq: PROCOVAII linear emp} as the PROCOVA model, the ATE estimator, i.e., the estimator of $\beta_A$, has the same asymptotic variance regardless of whether the prognostic score model is known or estimated.
Theorem \ref{thm: asyvar PROCOVAII} additionally implies that this equivalence holds not only for the coefficient of the treatment $\beta_A$ but also for the intercept $\beta_0$ and the sum $\beta_0 + \beta_A$ in the PROCOVA model \eqref{eq: PROCOVAII linear emp}. 
This additional property results from centering the prognostic score so that its mean is zero.
Similar results are obtained for the PROCOVA model \eqref{eq: PROCOVA linear} when the prognostic score is centered (see Appendix \ref{appendix: PROCOVAI centering}).
In contrast, for the coefficients $\beta_1$ and $\beta_2$ in model \eqref{eq: PROCOVAII linear emp} and $\beta_1$ in model \eqref{eq: PROCOVA linear} with centering, the variance estimator that ignores uncertainty in prognostic score estimation generally leads to underestimation (see Appendix \ref{appendix: thm 4} and \ref{appendix: PROCOVAI centering}).

\subsection{Asymptotic variances of related covariate adjustment methods}
\label{subsec: theory other}
\cite{tsiatis2008covariate} proposed the following augmented estimator:
\begin{equation}
\label{aug tsiatis}
    \frac{\sum_{i=1}^n A_iY_i}{\sum_{i=1}^n A_i} - \frac{\sum_{i=1}^n (1-A_i)Y_i }{\sum_{i=1}^n (1-A_i)}-\sum_{i=1}^n\left\{A_i-\frac{1}{n}\sum_{i=1}^n A_i\right\}\left\{\frac{{\rho}^{(1)}(W_i;\hat\theta^{(1)}_n)}{\sum_{i=1}^nA_i} + \frac{{\rho}^{(0)}(W_i;\hat\theta^{(0)}_n)}{\sum_{i=1}^n(1-A_i)}\right\},
\end{equation}
where $\rho^{(a)}(W;\theta^{(a)})$ is a model for $\mathbb{E}[Y\mid A=a, W]$ and $\hat\theta^{(a)}_n$ is an estimator of $\theta^{(a)}$ using trial data $\mathcal D$, for $a=0,1$.
\cite{tsiatis2008covariate} showed that the uncertainty of the first-stage estimators, $\hat\theta^{(0)}_n$ and $\hat\theta^{(1)}_n$, is asymptotically negligible if $A\perp W$ holds.

\cite{shen2014inverse} proposed the following IPW estimator:
\begin{equation}
\label{ipw}
    \frac{1}{n}\sum_{i=1}^n\frac{A_iY_i}{e(W_i;\hat \theta_n)} - \frac{1}{n}\sum_{i=1}^n\frac{(1-A_i)Y_i}{1-e(W_i;\hat \theta_n)},
\end{equation}
where $e(W_i; \theta)$ is a model for $\mathbb{P}[A=1\mid W]$ and $\hat \theta_n$ is the maximum likelihood estimator of $\theta$ using trial data $\mathcal D$, which is the first-stage estimator.
In RCTs where the treatment is randomly assigned with a probability $\mathbb{P}[A=1\mid W]=\pi$, the true $e(W_i; \theta^*)$ is known as $\pi$.
\cite{shen2014inverse} showed that (deliberately) using $e(W_i; \hat\theta_{n})$ results in a smaller asymptotic variance than using $e(W_i; \theta^*)$, specifically, it becomes the same as ANHECOVA, when $e(W_i; \theta)$ is a logistic model.
This result implies that ignoring the uncertainty of $\hat{\theta}_n$ leads to an overestimation of the variance.

These two apparently contradictory results are easier to understand when considering how consistency for ATE is ensured for each estimator.
The consistency of the estimator \eqref{aug tsiatis} is ensured regardless of the correct specification of $\rho^{(0)}(W;\theta^{(0)})$ and $\rho^{(1)}(W;\theta^{(1)})$.
That is, even if $\theta^{(0)}$ and $\theta^{(1)}$ are fixed at specific values other than $\theta^{(0)*}$ and $\theta^{(1)*}$, the estimator \eqref{aug tsiatis} still converges to ATE.
This suggests that the condition \eqref{eq: cond a} holds.
In contrast, the consistency of the estimator \eqref{ipw} follows from the property that $e(W_i;\hat \theta_n)$ converges to the true value $e(W_i; \theta^*)=\pi$ due to $A\perp W$.
That is, when $\theta$ is fixed at a specific value other than the true value $\theta^*$, the estimator \eqref{ipw} converges to a value different from ATE.
This suggests that the condition \eqref{eq: cond a} does not hold.

\subsection{Extension to nonparametric prognostic score models}
\label{subsec: theory PROCOVA ML}
In this section, we return our focus to PROCOVA.
In the previous sections, we assumed a specific linear model \eqref{eq: PS model linear} for the prognostic score model.
However, prognostic scores can also be estimated using models other than linear regression, such as machine learning.
For example, \cite{schuler2022increasing} and \cite{EMA2022PROCOVA} recommend machine learning such as random forests or deep learning for large historical data, and \cite{hojbjerre2025tutorial} used a super-learner.
As shown by the following theorem, even when machine learning or parametric models beyond the linear regression model are used to estimate prognostic scores, as long as $L_2$-consistency holds, the same statement holds.
The proof is given in Appendix \ref{appendix: for thm3} and \ref{appendix: thm 5}.
\begin{theorem}
\label{thm: asyvar PROCOVA general}
Let $\hat\rho_{\tilde n}$ be a prognostic score function estimated using historical data $\tilde D$.
Additionally, let $\psi(O;\beta,\rho)=(Y-\beta^\top X_\rho)X_\rho$ where $X_\rho=(1, A, \rho(W))^\top$ and $\beta=(\beta_0,\beta_A,\beta_1)^\top$, and let $\hat\beta_n(\rho)$ be the corresponding OLS estimator using the trial data $\mathcal{D}$.
Assume $D\perp \tilde D$, $A\perp W$, $0<\pi<1$, $\mathbb{E}[Y^2]<\infty$, $\mathbb{E}[\{\rho^*(W)\}^2]<\infty$, $\mathbb{V}[\rho^*(W)]>0$ and $\mathbb{E}[\{\hat\rho_{\tilde n}(W)-\rho^*(W) \}^2\mid \mathcal{\tilde D}]=o_p(1)$.
Then, for $e=(0,1,0)^\top$,
\begin{equation*}
    \sqrt{n}\{e^\top\hat\beta_n(\hat\rho_{\tilde n})-e^\top\beta^*\}
=
\sqrt{n}\{e^\top\hat\beta_n(\rho^*)-e^\top\beta^*\}+o_p(1).
\end{equation*}
\end{theorem}
Although theoretical frameworks in the context of observational studies typically assume that nuisance functions such as conditional mean functions are $L_2$-consistent for the target (oracle) functions at a sufficiently fast rate, Theorem \ref{thm: asyvar PROCOVA general} holds as long as $\hat\rho_{\tilde n}$ is $L_2$-consistent for some function $\rho^*$, which does not necessarily need to be the oracle function.
Theorem \ref{thm: asyvar PROCOVA general} shows that the asymptotic distributions of $e^\top\hat{\beta}_n(\hat\rho_{\tilde n})$ using the estimated prognostic score and $e^\top\hat{\beta}_n(\rho^*)$ using the known prognostic score are identical for $e=(0,1,0)^\top$, provided that $L_2$-consistency holds.
Therefore, the statement that the uncertainty in the estimation of the prognostic score model is negligible for ATE inference would hold regardless of whether a parametric or nonparametric model with $L_2$-consistency is used to estimate the prognostic scores.
When estimating prognostic scores using nonparametric models, it is difficult to explicitly account for that uncertainty.
Thus, this result is important in practice. 

\section{Variance estimation}
In Section \ref{subsec: var est general}, we construct two variance estimators corresponding to the two asymptotic variances derived in Section 3 for the two-sample two-stage estimator, and discuss their properties.
In Section \ref{subsec: var est PROCOVA}, we provide the specific forms of these variance estimators for PROCOVA, and discuss their properties.

\subsection{Variance estimation in two-sample two-stage estimation}
\label{subsec: var est general}
The naive variance estimator for $\hat\beta_n(\hat\theta_{\tilde n})$ is $1/n$ times $\hat V_{\mathrm{fix}}:= \hat Q_0^{-1}\hat \Omega (\hat Q_0^{-1})^\top$, where
\begin{equation*}
    \hat Q_0=\frac{1}{n}\sum_{i=1}^n\left\{\frac{\partial}{\partial\beta^\top}\psi\left(O_i;\beta,\hat \theta_{\tilde n}\right)\middle|_{\beta=\hat\beta_n(\hat\theta_{\tilde n})}\right\}
\end{equation*}
and
\begin{equation*}
    \hat \Omega=\frac{1}{n}\sum_{i=1}^n \left\{\psi\left(O_i;\hat\beta_n(\hat\theta_{\tilde n}),\hat\theta_{\tilde n}\right)\psi\left(O_i;\hat\beta_n(\hat\theta_{\tilde n}),\hat\theta_{\tilde n}\right)^\top\right\}.
\end{equation*}
$\hat V_{\mathrm{fix}}$ is an empirical plug-in estimator of $V_{\mathrm{fix}}$ and can be consistent for $V_{\mathrm{fix}}$ under suitable regularity conditions. 
For any constant vector $e\in \mathbb{R}^d$, we can also construct the variance estimator for $e^\top \hat\beta_n(\hat\theta_{\tilde n})$ as $e^\top\hat V_{\mathrm{fix}}e/n$.
However, $\hat V_{\mathrm{fix}}$ is generally not a consistent estimator of $V_{\mathrm{est}}$, which is the asymptotic variance of the two-sample two-stage estimator $\hat\beta_n(\hat\theta_{\tilde n})$, due to ignoring the first-stage uncertainty.

The correct variance estimator for $\hat\beta_n(\hat\theta_{\tilde n})$ is $1/n$ times $\hat V_{\mathrm{est}}:=\hat V_{\mathrm{fix}} + \frac{n}{\tilde n}\hat Q_0^{-1}\hat Q_1\hat V_\theta \hat Q_1^\top(\hat Q_0^{-1})^\top$, where
\begin{equation*}
\hat Q_1=\frac{1}{n}\sum_{i=1}^n \left\{\frac{\partial}{\partial\theta^\top}\psi\left(O_i;\hat\beta_n(\hat\theta_{\tilde n}),\theta\right)\middle|_{\theta=\hat\theta_{\tilde n}}\right\}  
\end{equation*}
and
\begin{equation*}
\hat V_\theta=\hat Q_2^{-1}\hat Q_3(\hat Q_2^{-1})^{\top},\quad\hat Q_2=\frac{1}{\tilde n}\sum_{i=1}^{\tilde n}\left\{\frac{\partial}{\partial\theta^\top}\phi\left(\tilde O_i;\theta\right)\middle|_{\theta=\hat\theta_{\tilde n}}\right\},\quad \hat Q_3=\frac{1}{\tilde n}\sum_{i=1}^{\tilde n}\left\{\phi\left (\tilde O_i;\hat\theta_{\tilde n}\right)\phi\left(\tilde O_i;\hat\theta_{\tilde n}\right)^\top\right\}.   
\end{equation*}
$\hat V_{\mathrm{est}}$ is an empirical plug-in estimator of $V_{\mathrm{est}}$ and can be consistent for $V_{\mathrm{est}}$ under suitable regularity conditions. 
For any constant vector $e\in \mathbb{R}^d$, we can also construct the variance estimator for $e^\top \hat\beta_n(\hat\theta_{\tilde n})$ as $e^\top\hat V_{\mathrm{est}}e/n$.

\subsection{Variance estimation in PROCOVA}
\label{subsec: var est PROCOVA}
We provide specific forms of $\hat V_{\text{fix}}$ and $\hat V_{\text{est}}$ for the case of PROCOVA with the first-stage estimating function \eqref{eq: est func prog} and the second-stage estimating function \eqref{eq: est func PROCOVA linear} in Appendix \ref{appendix: PROCOVA varest}.
For the inference for ATE using PROCOVA, both variance estimators, $e^\top\hat V_{\mathrm{fix}}e/n$ and $e^\top\hat V_{\mathrm{est}}e/n$, are asymptotically valid. 
As estimates, $e^\top\hat V_{\mathrm{est}}e/n$ always takes on values greater than or equal to $e^\top\hat V_{\mathrm{fix}}e/n$ because the additional term $\frac{1}{\tilde n}e^\top\hat Q_0^{-1}\hat Q_1\hat V_\theta \hat Q_1^\top(\hat Q_0^{-1})^\top e$ is non-negative.
If $n/\tilde n$ is large, then the additional term is large.

The heteroskedasticity-consistent covariance matrix estimator corresponding to $e^\top\hat V_{\mathrm{fix}}e/n$ is known to have a downward bias with finite samples in some realistic situations \citep{chesher1987bias}.
If more conservative inference in finite samples is preferred, it may be better to use $e^\top\hat V_{\mathrm{est}}e/n$ rather than $e^\top\hat V_{\mathrm{fix}}e/n$.

However, for the one-sample two-stage estimator in equation \eqref{aug tsiatis} discussed in Section \ref{subsec: theory other}, no general ordering exists between the variance estimator based on equation (A.6) in \cite{tsiatis2008covariate}, which corresponds to ignoring the first-stage uncertainty, and the variance estimator based on equation (A.6)+(A.7) in \cite{tsiatis2008covariate}, which corresponds to taking the first-stage uncertainty into account.
Note that both variance estimators are asymptotically valid.

\section{Simulation studies}
In this section, we conduct simulation studies to compare the performance of two types of variance estimators for PROCOVA.
Section \ref{subsec: sim set} describes our simulation settings.
Section \ref{subsec: sim result} presents the simulation results.

\subsection{Simulation settings}
\label{subsec: sim set}
We generate simulation data following \cite{hojbjerre2025tutorial}.
For each subject of trial data $i = 1, \dots, n$, observed covariates $W_i = (1, W_{1,i}, W_{2,i}, \dots, W_{7,i})^\top$ are generated as $W_{1,i} \sim \mathrm{Unif}(-2,\, 1)$, $W_{2,i} \sim \mathrm{Unif}(-2,\, 1)$, $W_{3,i} \sim {N}(0,\, 3^2)$, $W_{4,i} \sim \mathrm{Exp}(0.8)$, $W_{5,i} \sim \mathrm{Gamma}(5,\, 10)$, $W_{6,i}, W_{7,i} \sim \mathrm{Unif}(1,\, 2)$, an unobserved covariate $U_i$ is generated as $U_i \sim \mathrm{Unif}(0,\, 1)$, a treatment assignment $A_i$ is generated as $ A_i \sim \mathrm{Bernoulli}(0.5)$ and an outcome $Y_i$ is generated as $Y_i\mid A_i, W_i,U_i \sim N(A_im_{1}(W_i, U_i) + (1-A_i)m_{0}(W_i, U_i),\, 1)$.
The model form of $m_{1}(W_i, U_i),\ m_{0}(W_i, U_i)$ differs depending on the scenario.
For each subject of historical data $i = 1, \dots, \tilde n$, observed covariates $\tilde W_i = (1, \tilde W_{1,i}, \tilde W_{2,i}, \dots, \tilde W_{7,i})^\top$ are generated as $\tilde W_{1,i} \sim \mathrm{Unif}(-2+b,\, 1+b)$, $\tilde W_{2,i} \sim \mathrm{Unif}(-2,\, 1)$, $\tilde W_{3,i} \sim {N}(0,\, 3^2)$, $\tilde W_{4,i} \sim \mathrm{Exp}(0.8)$, $\tilde W_{5,i} \sim \mathrm{Gamma}(5,\, 10)$, $\tilde W_{6,i}, \tilde W_{7,i} \sim \mathrm{Unif}(1,\, 2)$, an unobserved covariate $\tilde U_i$ is generated as $\tilde U_i \sim \mathrm{Unif}(c,\, 1+c)$, a treatment assignment $\tilde A_i$ is generated as $\tilde A_i\equiv0$ and an outcome $\tilde Y_i$ is generated as $\tilde Y_i \mid \tilde A_i, \tilde W_i,\tilde U_i \sim N(\tilde A_im_{1}(\tilde W_i, \tilde U_i) + (1-\tilde A_i)m_{0}(\tilde W_i, \tilde U_i),\, 1)$. 
The parameters $b,\ c$ represent the difference in the covariate distribution between trial and historical data and differ depending on the scenario.

We examined 36 scenarios by combining the four patterns of the model form of $m_{1}(W_i, U_i),\ m_{0}(W_i, U_i)$ and nine patterns of the covariate shift parameter $b,\ c$.
Each scenario is denoted as Scenario X-x. X before the hyphen corresponds to the pattern of the model form of $m_{1}(W_i, U_i),\ m_{0}(W_i, U_i)$, and x after the hyphen corresponds to the pattern of the covariate shift parameter $b,\ c$.

In Scenarios A-x and B-x, a linear regression model with the dependent variable $Y$ and independent variables $W$ holds; however, it does not hold in Scenarios C-x and D-x.
Scenarios A-x and C-x have homogeneous treatment effect, whereas Scenarios B-x and D-x have heterogeneous treatment effect.
In Scenario A-x, we set
\begin{equation}
\label{eq: dgp correct m0}
m_0(W_i, U_i) = W_{1,i} + 4.1 W_{2,i} + 1.4 W_{3,i} - 1.5 W_{4,i} + 1.5 W_{5,i} - W_{6,i} + W_{7,i}
\end{equation}
and $m_1(W_i, U_i) = m_0(W_i, U_i) + 0.835$.
In Scenario B-x, we set \eqref{eq: dgp correct m0}, and 
\begin{equation*}
 m_1(W_i, U_i)= -4.184 + 0.1 W_{1,i}^2 + 0.41W_{2,i}^2 + 0.14 W_{3,i}^2 - 0.15 W_{4,i}^2 + 0.15 W_{5,i}^2 - 0.1W_{6,i}^2 + 0.1W_{7,i}^2.
\end{equation*}
In Scenario C-x, we set 
\begin{equation}
\label{eq: dgp miss m0}
\begin{split}
m_0(W_i, U_i) =\;&
4.1 \sin(|W_{2,i}|)
+ 1.4 \mathbb{I}(|W_{3,i}| > 2.5) + 1.5 \mathbb{I}(|W_{4,i}| > 0.25)+ 1.5 \sin(|W_{5,i}|) \\
&- 4.1 \mathbb{I}(W_{1,i} < -4.1) \sin(|W_{2,i}|)
- 4.1 \mathbb{I}(W_{1,i} < -6.1) \sin(|W_{2,i}|) \\
& - 4.1 \sin(|W_{2,i}|)\mathbb{I}(U_i > 1.55) - 4.1 \sin(|W_{2,i}|)\mathbb{I}(U_i > 1.1).
\end{split}
\end{equation}
and $m_1(W_i, U_i) = m_0(W_i, U_i) + 0.835$. 
In Scenario D-x, we set \eqref{eq: dgp miss m0}, and 
\begin{equation*}
\begin{split}
m_1(W_i, U_i) =\;&
4.3 \sin^2(|W_{2,i}|)
+ 1.4 \mathbb{I}(|W_{3,i}| > 2.5)
+ 1.3 \mathbb{I}(|W_{4,i}| > 0.25) \\
&+ 4.1 \mathbb{I}(W_{2,i} > 0)\sin(|W_{5,i}|)
+ 1.6 \sin(|W_{6,i}|) \\
&- 4.1 \sin(|W_{2,i}|)\mathbb{I}(W_{1,i} < -4.1)
- 4.1 \sin(|W_{2,i}|)\mathbb{I}(W_{1,i} < -6.1) \\
&- 4.1 \sin(|W_{2,i}|)\mathbb{I}(U_i > 1.1)
- 4.1 \sin(|W_{2,i}|)\mathbb{I}(U_i > 1.55).  
\end{split}    
\end{equation*}
We set $b=c=0$ in Scenario X-1, $b=0,c=0.5$ in Scenario X-2, $b=0,c=1.5$ in Scenario X-3, $b=-2, c=0$ in Scenario X-4, $b=-2,c=0.5$ in Scenario X-5, $b=-2,c=1.5$ in Scenario X-6, $b=-5, c=0$ in Scenario X-7, $b=-5,c=0.5$ in Scenario X-8, $b=-5,c=1.5$ in Scenario X-9.
In Scenarios A-x and C-x, the true value of ATE is $0.835$.
In Scenarios B-x and D-x, the true value of ATE is calculated as $0.835$ by generating one dataset with the sample size $n=\tilde n=10^7$ for the corresponding scenario and calculating the point estimate using that dataset.

For each scenario, sample sizes are varied as $n=40, 60, 80, 100, 200, 400, 600, 800, 1000$ and $\tilde n = n/4, n/2,n,2n, 4n, 10n$.
For each combination of the scenario and sample sizes $n,\tilde n$, we repeat the simulation 1000 times and evaluate the performance of the two types of variance estimators with the first-stage estimating function \eqref{eq: est func prog} and the second-stage estimating function \eqref{eq: est func PROCOVA linear}.
To calculate the confidence intervals (CIs), we use $t$-distribution approximation with $n-3$ degrees of freedom.

\subsection{Simulation results}
\label{subsec: sim result}
Due to page constraints, we present only the results for Scenario D-5, which may be the most realistic scenario. 
However, the trends in the results described below were generally similar for the other scenarios.
The results for other scenarios are presented in Appendix \ref{appendix: sim}.

Figures \ref{fig: coverage_nboth} and \ref{fig: varianceratio_nboth} show the results when the sample sizes are increased while maintaining the relationship $\tilde n=10n$.
Figure \ref{fig: coverage_nboth} shows the coverage probability which is the proportion of 1000 simulations in which the 95\% CI includes the true value.
For $\beta_0$ and $\beta_1$, the 95\% CI based on $\hat V_{\text{est}}$ achieves the nominal coverage probability as the sample size increases, whereas the 95\% CI based on $\hat V_{\text{fix}}$ remains under coverage.
However, the 95\% CI for $\beta_A$ calculated using either of the variance estimators mostly achieves nominal coverage probability.
Figure \ref{fig: varianceratio_nboth} shows the mean of the ratio of two variance estimators, i.e., $e^\top\hat V_{\text{est}}e/e^\top\hat V_{\text{fix}}e$ with $e=(1,0,0)^\top$ for $\beta_0$, with $e=(0,1,0)^\top$ for $\beta_A$ and $e=(0,0,1)^\top$ for $\beta_1$, over 1000 simulations.
The ratio of the two variance estimators for $\beta_A$ converges to one as the sample size increases, whereas the ratios for $\beta_0$ and $\beta_1$ converge to values greater than one.
These results are consistent with the theory in Section \ref{sec: theory}.
\begin{center}
    [Figures \ref{fig: coverage_nboth} and \ref{fig: varianceratio_nboth} about here]
\end{center}

Note that even for $\beta_A$, because $e^\top\hat V_{\text{est}}e$ always takes a larger value than $e^\top\hat V_{\text{fix}}e$, this difference can affect performance when sample sizes are small.
Figure \ref{fig: coverage_nhist_ntrial100} shows the coverage probability for $\beta_A$ when the sample size of trial data is $n=100$ and that of historical data is varied.
The 95\% CI based on $\hat V_{\text{fix}}$ is under coverage, whereas that based on $\hat V_{\text{est}}$ is closer to the nominal level.
\begin{center}
    [Figure \ref{fig: coverage_nhist_ntrial100} about here]
\end{center}

Figure \ref{fig: coverage_nhist_ntrial1000} shows the coverage probability when the sample size of trial data is $n=1000$ and that of historical data is varied.
Even for $\beta_0$ and $\beta_1$, the 95\% CI based on $\hat V_{\text{fix}}$ achieves the nominal level, as $\tilde n$ increases.
This result corresponds to the condition $\kappa =0$.
\begin{center}
    [Figure \ref{fig: coverage_nhist_ntrial1000} about here]
\end{center}

\section{Data application}
In this section, we apply PROCOVA with two variance estimators to ACTG 175 data from the BART package in R.
ACTG 175 is a randomized, double-blind, placebo-controlled trial that compares monotherapy with zidovudine or didanosine with combination therapy with zidovudine and didanosine or zidovudine and zalcitabine in HIV-1-infected subjects with CD4 T cell counts between 200 and 500 per cubic millimeter \citep{hammer1996trial}. 
We use CD4 T cell count ($\text{cells}/\text{mm}^3$) at 20 weeks as the outcome.
We use the monotherapy with zidovudine arm as the control group and the combination therapy with zidovudine and didanosine arm as the treatment group, and randomly split them into historical data and trial data artificially as follows:
\begin{enumerate}
    \item Randomly sample 100 subjects each from 532 control subjects and 522 treatment subjects, and use these 200 subjects as trial data (i.e., $n=200$).
    \item Randomly sample $\tilde n$ subjects from the remaining control group data, and use them as historical data ($\tilde n= 100,200,400$).  
\end{enumerate}

Then, we apply PROCOVA to the data constructed through the above procedures.
First, we learn the prognostic score model \eqref{eq: PS model linear} including CD4 T cell count at baseline, age, Karnofsky score and past treatment stratification as covariates.
Second, we fit the PROCOVA model \eqref{eq: PROCOVA linear} using two variance estimators: $\hat V_{\mathrm{fix}}$ and $\hat V_{\mathrm{est}}$.

Table \ref{tb: app} and Figure \ref{fig: app} show application results. 
For $\beta_A$, the standard errors and 95\% CIs using $t$-distribution approximation with $n-3$ degrees of freedom are almost identical for the two variance estimators across the three historical sample sizes.
For $\beta_0$, while variance estimates treating the prognostic score model as known are almost identical across three historical sample sizes, variance estimates accounting for prognostic score model estimation are smaller when the historical sample size is larger.
Additionally, the ratio of the two variance estimates is close to one for $\beta_A$ for each historical sample size.
For $\beta_0$ and $\beta_1$, the ratio deviates from one, and this degree is smaller if historical sample size is larger.
\begin{center}
    [Table  \ref{tb: app} and Figure \ref{fig: app} about here]
\end{center}

\section{Conclusion}
In this study, we derived and compared two asymptotic variances for PROCOVA: one for the case where the first-stage model is known and one for the case where it is estimated.
As discussed in Section \ref{subsec: theory two-stage}, in the two-sample two-stage estimation framework, the variance estimation that ignores the uncertainty in the first-stage estimation generally leads to underestimation.
However, as shown in Theorems \ref{thm: asyvar PROCOVA}--\ref{thm: asyvar PROCOVA general}, in the special case of estimating ATE using PROCOVA in RCTs, a variance estimator that ignores the uncertainty in the prognostic score estimation is asymptotically valid.
Furthermore, we constructed two variance estimators for PROCOVA, each corresponding to cases for which the uncertainty in prognostic score estimation is accounted for or not, and compared their performances via simulations and data application.
For the coefficient of the treatment, which represents ATE, the two variance estimators performed comparably, except when both sample sizes were small.
For the intercept and coefficient of the prognostic score, there was a non-negligible difference between the two variance estimators.

According to our theory, the asymptotic variance of the PROCOVA-based ATE estimator can be treated as if prognostic scores are known.
This simplifies the asymptotic theory, which is beneficial for theoretical researchers.
Our theory also confirmed that the current practice of variance estimation in PROCOVA \citep{schuler2022increasing, hojbjerre2025tutorial} is appropriate.
When estimating ATE using PROCOVA in actual RCTs, it is asymptotically valid to use a variance estimator that ignores the uncertainty in the prognostic score estimation.
In particular, when machine learning is used to estimate prognostic scores, the fact that there is no need to explicitly account for that uncertainty is advantageous for data analysts.
When the sample size of historical data is small, it is preferable to use a linear regression model to estimate prognostic scores and a variance estimator based on $\hat{V}_{\mathrm{est}}$, which yields larger variance estimates if one prefers more conservative inference.

\bibliographystyle{apalike}
\bibliography{PROCOVAvar}

\clearpage

\begin{figure}[p]
    \centering
    \includegraphics[width=\linewidth]{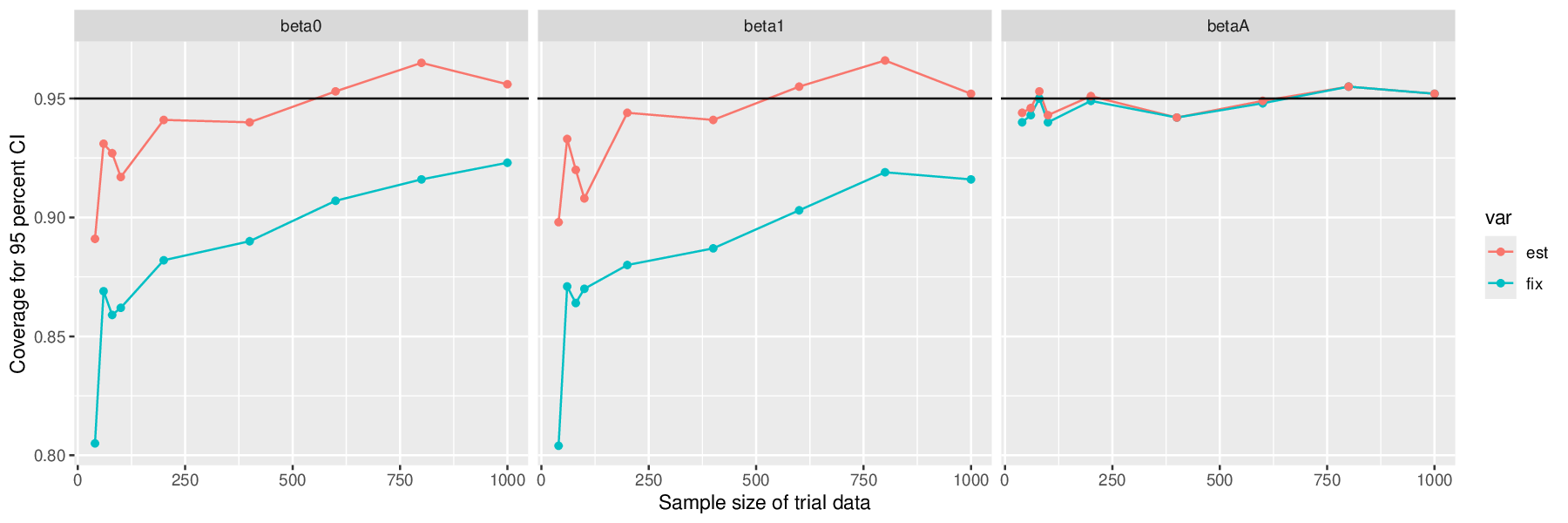}
    \caption{Plots of the coverage probability of 95\% CI over 1000 simulations for Scenario D-5.
    ``beta0'', ``betaA'', ``beta1'' represent the intercept $\beta_0$, the coefficient for the treatment assignment $\beta_A$, the coefficient for the prognostic score $\beta_1$ in the PROCOVA model \eqref{eq: PROCOVA linear}, respectively.
     ``fix'' represents $e^\top\hat V_{\text{fix}}e$ with $e=(1,0,0)^\top$ for $\beta_0$, with $e=(0,1,0)^\top$ for $\beta_A$ and $e=(0,0,1)^\top$ for $\beta_1$.
    ``est'' represents $e^\top\hat V_{\text{est}}e$ with $e=(1,0,0)^\top$ for $\beta_0$, with $e=(0,1,0)^\top$ for $\beta_A$ and $e=(0,0,1)^\top$ for $\beta_1$.
    The x-axis represents the sample size of trial data $n$.
    The sample size of historical data is $\tilde n = 10n$. 
    The y-axis represents the coverage probability which is the proportion of 1000 simulations in which the 95\% CI using each variance estimator includes the true value.
    }
    \label{fig: coverage_nboth}
\end{figure}

\clearpage
\begin{figure}[p]
    \centering
    \includegraphics[width=0.4\linewidth]{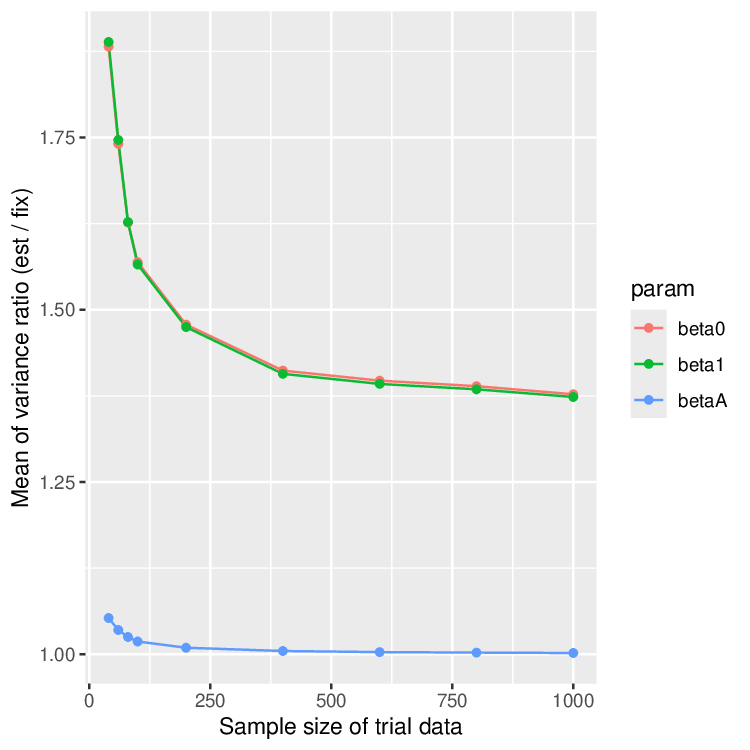}
    \caption{Plots of the mean of the ratio of two variance estimators over 1000 simulations for Scenario D-5.
    ``beta0'', ``betaA'', ``beta1'' represent the intercept $\beta_0$, the coefficient for the treatment assignment $\beta_A$, the coefficient for the prognostic score $\beta_1$ in the PROCOVA model \eqref{eq: PROCOVA linear}, respectively.
    The x-axis represents the sample size of trial data $n$.
    The sample size of historical data is $\tilde n = 10n$. 
    The y-axis represents the mean of the ratio of two variance estimators, i.e., $e^\top\hat V_{\text{est}}e/e^\top\hat V_{\text{fix}}e$ with $e=(1,0,0)^\top$ for $\beta_0$, with $e=(0,1,0)^\top$ for $\beta_A$ and $e=(0,0,1)^\top$ for $\beta_1$, over 1000 simulations.}
    \label{fig: varianceratio_nboth}
\end{figure}
\clearpage
\begin{figure}[p]
    \centering
    \includegraphics[width=0.4\linewidth]{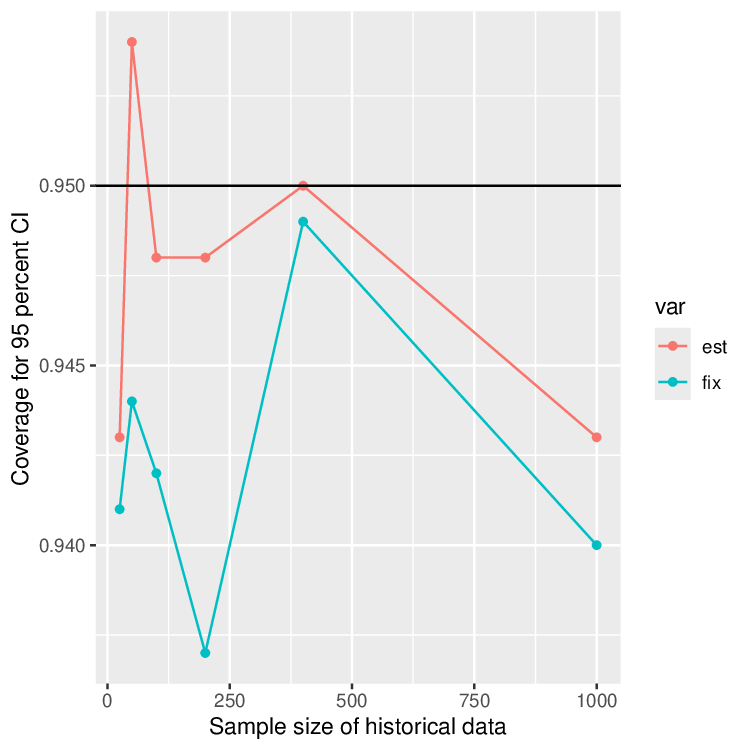}
    \caption{Plots of the coverage probability of the 95\% CI for $\beta_A$ in the PROCOVA model \eqref{eq: PROCOVA linear} over 1000 simulations for Scenario D-5.
     ``fix'' represents $e^\top\hat V_{\text{fix}}e$ with $e=(1,0,0)^\top$ for $\beta_0$, with $e=(0,1,0)^\top$ for $\beta_A$ and $e=(0,0,1)^\top$ for $\beta_1$.
    ``est'' represents $e^\top\hat V_{\text{est}}e$ with $e=(1,0,0)^\top$ for $\beta_0$, with $e=(0,1,0)^\top$ for $\beta_A$ and $e=(0,0,1)^\top$ for $\beta_1$.
    The sample size of trial data is $n=100$. 
    The x-axis represents the sample size of historical data $\tilde n$.
    The y-axis represents the coverage probability which is the proportion of 1000 simulations in which the 95\% CI using each variance estimator includes the true value.}
    \label{fig: coverage_nhist_ntrial100}
\end{figure}
\clearpage
\begin{figure}[p]
    \centering
    \includegraphics[width=\linewidth]{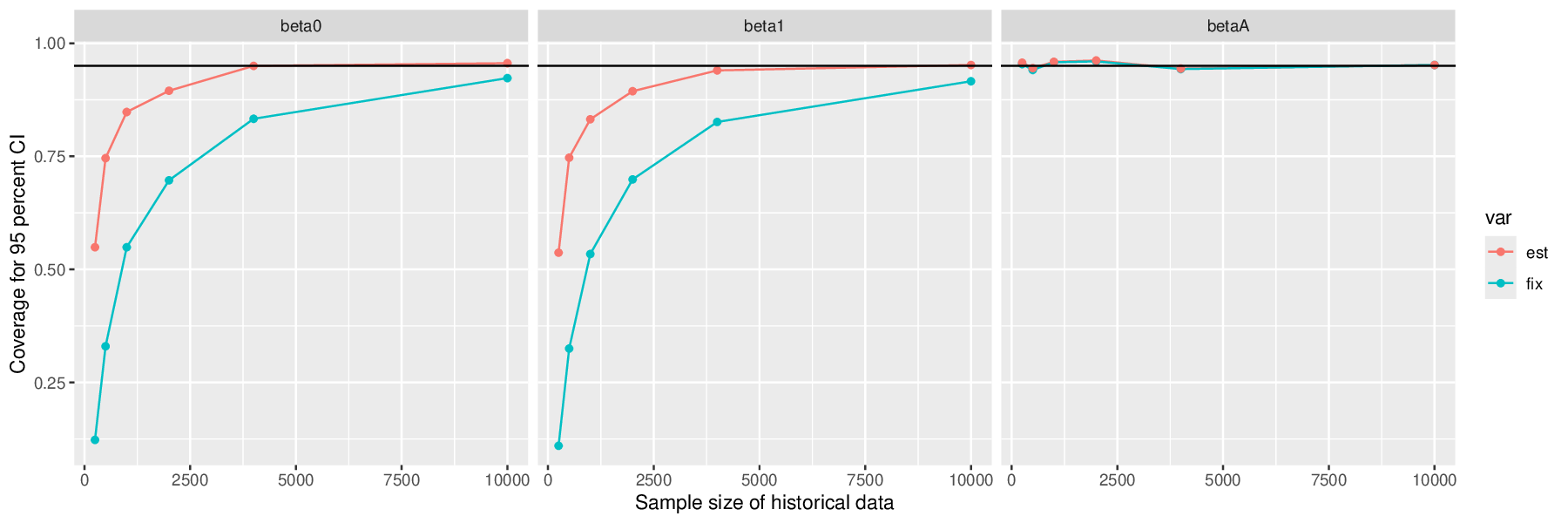}
    \caption{Plots of the coverage probability of 95\% CI over 1000 simulations for Scenario D-5.
    ``beta0'', ``betaA'', ``beta1'' represent the intercept $\beta_0$, the coefficient for the treatment assignment $\beta_A$, the coefficient for the prognostic score $\beta_1$ in the PROCOVA model \eqref{eq: PROCOVA linear}, respectively.
     ``fix'' represents $e^\top\hat V_{\text{fix}}e$ with $e=(1,0,0)^\top$ for $\beta_0$, with $e=(0,1,0)^\top$ for $\beta_A$ and $e=(0,0,1)^\top$ for $\beta_1$.
    ``est'' represents $e^\top\hat V_{\text{est}}e$ with $e=(1,0,0)^\top$ for $\beta_0$, with $e=(0,1,0)^\top$ for $\beta_A$ and $e=(0,0,1)^\top$ for $\beta_1$.
    The sample size of trial data is $n=1000$. 
    The x-axis represents the sample size of historical data $\tilde n$.
    The y-axis represents the coverage probability which is the proportion of 1000 simulations in which the 95\% CI using each variance estimator includes the true value.}
    \label{fig: coverage_nhist_ntrial1000}
\end{figure}
\clearpage
\begin{table}[p]
\centering
\caption{Results of applying PROCOVA with two variance estimators to ACTG 175 data. 
$\beta_0$, $\beta_A$, $\beta_1$ represent the intercept, the coefficient for the treatment assignment, the coefficient for the prognostic score in the PROCOVA model \eqref{eq: PROCOVA linear}, respectively. 
The sample size of trial data is $n=200$. 
The 2nd column represents the sample size of historical data $\tilde n$.
The 3rd column represents the point estimate.
The 4th column represents the standard error calculated by $\hat V_{\text{fix}}$ and the 5th column represents the standard error calculated by $\hat V_{\text{est}}$. 
The 6th column represents 95\% CI calculated by $\hat V_{\text{fix}}$. 
The 7th column represents 95\% CI calculated by $\hat V_{\text{est}}$.
95\% CIs use $t$-distribution approximation with $n-3$ degrees of freedom.
\\}
\label{tb: app}
\begin{tabular}{llcccccc}
\toprule
 &  &  &  & \multicolumn{2}{c}{Standard error} & \multicolumn{2}{c}{95\% CI} \\
\cmidrule(lr){5-6} \cmidrule(lr){7-8}
 & $\tilde{n}$ & Estimate &  & Fix & Est & Fix & Est \\
\midrule
\multirow{3}{*}{$\beta_0$}
 & 400 & 60.05 &  & 35.68 & 39.88 & [-10.30, 130.41] & [-18.47, 138.58] \\
 & 200 & 58.73 &  & 36.02 & 42.38 & [-12.31, 129.77] & [-24.84, 142.31] \\
 & 100 & 47.52 &  & 36.76 & 46.00 & [-24.98, 120.02] & [-43.17, 138.22] \\
\midrule
\multirow{3}{*}{$\beta_A$}
 & 400 & 71.92 &  & 17.34 & 17.35 & [37.73, 106.11] & [37.71, 106.13] \\
 & 200 & 70.87 &  & 17.21 & 17.23 & [36.92, 104.81] & [36.88, 104.85] \\
 & 100 & 71.11 &  & 17.17 & 17.21 & [37.25, 104.96] & [37.17, 105.04] \\
\midrule
\multirow{3}{*}{$\beta_1$}
 & 400 & 0.79 &  & 0.11 & 0.12 & [0.58, 1.00] & [0.56, 1.02] \\
 & 200 & 0.81 &  & 0.11 & 0.13 & [0.60, 1.03] & [0.57, 1.06] \\
 & 100 & 0.83 &  & 0.11 & 0.13 & [0.62, 1.04] & [0.57, 1.09] \\
\bottomrule
\end{tabular}
\end{table}
\clearpage
\begin{figure}[p]
    \centering
    \includegraphics[width=0.5\linewidth]{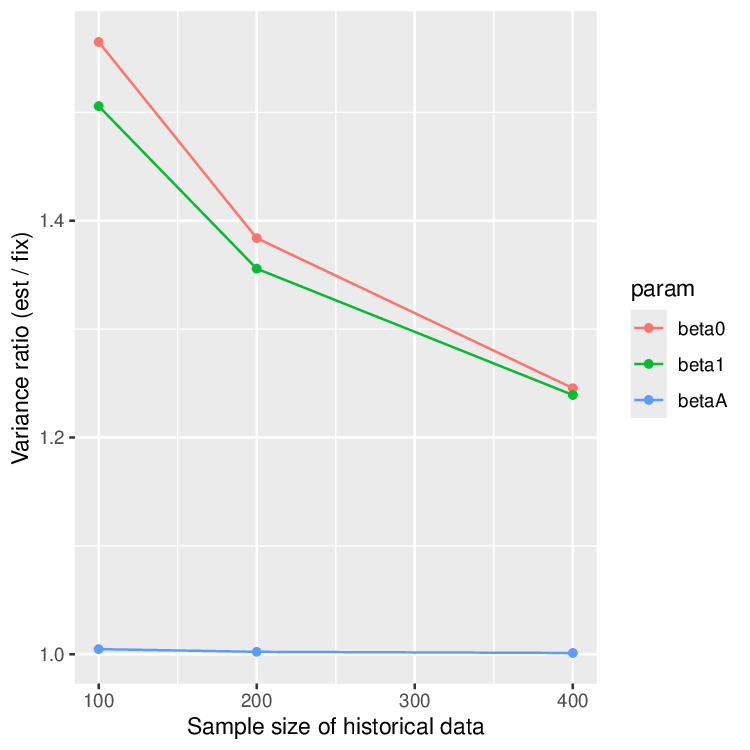}
    \caption{Plots of the ratio of two variance estimators when PROCOVA is applied to the ACTG 175 data with varying the sample size of historical data.
    ``beta0'', ``betaA'', ``beta1'' represent the intercept $\beta_0$, the coefficient for the treatment assignment $\beta_A$, the coefficient for the prognostic score $\beta_1$ in the PROCOVA model \eqref{eq: PROCOVA linear}, respectively.
    The sample size of trial data is $n=200$.
    The x-axis represents the sample size of historical data $\tilde n$.
    The y-axis represents the ratio of two variance estimators, i.e., $e^\top\hat V_{\text{est}}e/e^\top\hat V_{\text{fix}}e$ with $e=(1,0,0)^\top$ for $\beta_0$, with $e=(0,1,0)^\top$ for $\beta_A$ and $e=(0,0,1)^\top$ for $\beta_1$.}
    \label{fig: app}
\end{figure}

\clearpage
\begin{appendix}

\section{Appendix: Proofs}

\subsection{Regularity conditions}
\label{appendix: assumption}
Let $\Theta \subset \mathbb{R}^p$ and $\mathcal{B} \subset \mathbb{R}^d$ be open parameter spaces.
Let $N_{\theta}$ be a neighborhood of $\theta^*$. Let $N_{\beta}$ be a neighborhood of $\beta^*$.
\begin{enumerate}
\item[(A1)] There exists a unique parameter $\theta^* \in \Theta$ such that $\mathbb E[\phi(\tilde O;\theta^*)] = 0$.
For every $\theta\in N_\theta$, there exists a unique $\beta^*(\theta) \in \mathcal{B}$ satisfying $\mathbb E[\psi(O;\beta^*(\theta),\theta)]=0$.
\item[(A2)] The function $\phi(\tilde O;\theta)$ is continuously differentiable in $\theta\in N_\theta$.
The function $\psi(O;\beta,\theta)$ is continuously differentiable in $(\beta,\theta)\in N_\beta\times N_\theta$.

\item[(A3)] There exist integrable random variables $M_\beta(O)$, $M_\theta(O)$ and
$\tilde M(\tilde O)$ such that
\[
\sup_{(\beta,\theta)\in N_\beta\times N_\theta}
\left\|
\frac{\partial}{\partial \beta^\top}\psi(O;\beta,\theta)
\right\|
\le M_\beta(O),
\]
\[
\sup_{(\beta,\theta)\in N_\beta\times N_\theta}
\left\|
\frac{\partial}{\partial \theta^\top}\psi(O;\beta,\theta)
\right\|
\le M_\theta(O),
\]
\[
\sup_{\theta\in N_\theta}
\left\|
\frac{\partial}{\partial \theta^\top}\phi(\tilde O;\theta)
\right\|
\le \tilde M(\tilde O).
\]
Additionally,
\[
\mathbb{E}[\|\psi(O;\beta^*,\theta^*)\|^2]<\infty,
\quad
\mathbb{E}[\|\phi(\tilde O;\theta^*)\|^2]<\infty.
\]
\item[(A4)] $Q_0$ and $Q_2$ are nonsingular.
\item[(A5)] $\hat\theta_{\tilde n} \plim \theta^*$. Additionally, $\sup_{\theta\in N_\theta}
\|\hat\beta_n(\theta)-\beta^*(\theta)\| \plim 0$.
\end{enumerate}

\subsection{Asymptotic variances in two-sample two-stage estimation}
\label{appendix: thm 1}

\begin{lemma}
\label{thm: asyvar mat}
Assume the regularity conditions (A1)--(A5), $\mathcal{D}\perp\mathcal{\tilde D}$ and $\lim_{n,\tilde n\to\infty} {n}/{\tilde n}=:\kappa\in[0,\infty)$. 
Then, the following hold:
\begin{equation*}
    \sqrt n\{\hat\beta_n(\theta^*)-\beta^*\}
=-Q_0^{-1}\sqrt n\,\Psi_n(\beta^*,\theta^*)+o_p(1),
\end{equation*}
\begin{equation*}
         \sqrt n(\hat\beta_n(\hat\theta_{\tilde n})-\beta^*)=-Q_0^{-1}\sqrt n\,\Psi_n(\beta^*,\theta^*)+\sqrt{\kappa}Q_0^{-1}Q_1Q_2^{-1}\sqrt{\tilde n}\Phi_{\tilde n}(\theta^*)+o_p(1),
\end{equation*}
and
\begin{equation*}
    \sqrt n\{\hat\beta_n(\theta^*)-\beta^*\}\xrightarrow[]{d} N\!\left(0,\ V_{\mathrm{fix}}\right), \text{ with }V_{\mathrm{fix}}=Q_0^{-1}\Omega (Q_0^{-1})^\top,
\end{equation*}
\begin{equation*}
    \sqrt n\bigl(\hat\beta_n(\hat\theta_{\tilde n})-\beta^*\bigr)\xrightarrow[]{d}N\!\left(0,\ V_{\mathrm{est}}\right),\text{ with }V_{\mathrm{est}}=V_{\mathrm{fix}} + \kappa\, Q_0^{-1}Q_1 V_\theta Q_1^\top (Q_0^{-1})^\top,
\end{equation*}
where, $\Omega = \E\!\left[\psi(O;\beta^*,\theta^*)\psi(O;\beta^*,\theta^*)^\top\right]$, $Q_0=\mathbb{E}\!\left[\frac{\partial}{\partial\beta^\top}\psi(O;\beta,\theta^*)\middle|_{\beta=\beta^*}\right]$, $Q_1=\mathbb{E}\!\left[\frac{\partial}{\partial\theta^\top}\psi(O;\beta^*,\theta)\middle|_{\theta=\theta^*}\right]$, $Q_2=\mathbb{E}\!\left[\frac{\partial}{\partial\theta^\top}\phi(\tilde O;\theta)\middle|_{\theta=\theta^*}\right]$, $Q_3=\E\!\left[\phi(\tilde O;\theta^*)\phi(\tilde O;\theta^*)^\top\right]$ and $V_\theta=Q_2^{-1}Q_3(Q_2^{-1})^{\top}$.
\end{lemma}

\begin{proof}
By the mean value theorem to $\beta$, there exists $\bar\beta_n$ between $\hat\beta_n(\theta)$ and $\beta^*$ such that
\begin{equation*}
\Psi_n(\hat\beta_n(\theta^*), \theta^*) = \Psi_n(\beta^*,\theta^*) + \left\{\frac{\partial}{\partial\beta^\top}\Psi_n(\beta,\theta^*)\ \middle|_{\beta=\bar\beta_n}\right\}\bigl(\hat\beta_n(\theta^*)-\beta^*\bigr).
\end{equation*}
Since $\Psi_n(\hat\beta_n(\theta^*),\theta^*)=0$, the following holds:
\begin{equation*}
\hat\beta_n(\theta^*)-\beta^* = -\left\{\frac{\partial}{\partial\beta^\top}\Psi_n(\beta,\theta^*)\ \middle|_{\beta=\bar\beta_n}\right\}^{-1}\Psi_n(\beta^*,\theta^*).
\end{equation*}
Thus, by (A2) and (A3),
\begin{equation*}
    \left\{\frac{\partial}{\partial\beta^\top}\Psi_n(\beta,\theta^*)\ \middle|_{\beta=\bar\beta_n}\right\}\plim Q_0,
\end{equation*}
and then
\begin{equation*}
\label{eq: conv Qinv}
    \left\{\frac{\partial}{\partial\beta^\top}\Psi_n(\beta,\theta^*)\ \middle|_{\beta=\bar\beta_n}\right\}^{-1}
\plim Q_0^{-1}.
\end{equation*}
Therefore, the following holds:
\begin{equation}
\label{eq: expand beta fix}
\begin{split}
    \sqrt n\bigl\{\hat\beta_n(\theta^*)-\beta^*\bigr\}&=-\left\{\frac{\partial}{\partial\beta^\top}\Psi_n(\beta,\theta^*)\ \middle|_{\beta=\bar\beta_n}\right\}^{-1}
\sqrt n\,\Psi_n(\beta^*,\theta^*)\\
&=-Q_0^{-1}\sqrt n\,\Psi_n(\beta^*,\theta^*)+o_p(1).
\end{split}
\end{equation}
Additionally, by (A2) and (A3), the following holds:
\begin{equation}
\label{eq: psi CLT}
    \sqrt n\,\Psi_n(\beta^*,\theta^*)\xrightarrow[]{d} N(0,\Omega).
\end{equation}
By equations \eqref{eq: expand beta fix} and \eqref{eq: psi CLT}, the following holds:
\begin{equation}
\label{eq: asydist beta fix}
\sqrt n\{\hat\beta_n(\theta^*)-\beta^*\}\xrightarrow[]{d} N\!\left(0,\ Q_0^{-1}\Omega (Q_0^{-1})^\top\right).    
\end{equation}

Next, we decompose
\begin{equation}
\label{eq: decom}
    \hat\beta_n(\hat\theta_{\tilde n})-\beta^*
=\bigl\{\hat\beta_n(\theta^*)-\beta^*\bigr\} + \bigl\{\hat\beta_n(\hat\theta_{\tilde n})-\hat\beta_n(\theta^*)\bigr\}.
\end{equation}
The asymptotic distribution of the first term has already been derived as \eqref{eq: asydist beta fix}. 
We now derive the asymptotic distribution of the second term.
By the implicit function theorem, the following holds:
\begin{equation}
\label{eq: implicit thm}
\frac{\partial}{\partial\theta^\top}\hat\beta_n(\theta)=-\left\{\frac{\partial}{\partial\beta^\top}\Psi_n(\beta,\theta)\ \middle |_{\beta=\hat\beta_n(\theta)}\right\}^{-1}\left\{\frac{\partial}{\partial\theta^\top}\Psi_n(\hat\beta_n(\theta),\theta)\right\}.
\end{equation}
By the mean value theorem to $\hat\beta_n(\theta)$, there exists $\bar\theta_{\tilde n}$ between $\hat\theta_{\tilde n}$ and $\theta^*$ such that
\begin{equation}
\label{eq: mean value thm}
\hat\beta_n(\hat\theta_{\tilde n})-\hat\beta_n(\theta^*)
=
\left\{\frac{\partial}{\partial\theta^\top}\hat\beta_n(\theta)\ \middle|_{\theta=\bar\theta_{\tilde n}}\right\}
(\hat\theta_{\tilde n}-\theta^*).
\end{equation}
Substituting \eqref{eq: implicit thm} into \eqref{eq: mean value thm}:
\begin{equation*}
    \hat\beta_n(\hat\theta_{\tilde n})-\hat\beta_n(\theta^*)=-\left\{\frac{\partial}{\partial\beta^\top}\Psi_n(\beta,\bar\theta_{\tilde n})\ \middle |_{\beta=\hat\beta(\bar\theta_{\tilde n})}\right\}^{-1}\left\{\frac{\partial}{\partial\theta^\top}\Psi_n(\hat\beta_n(\theta),\theta)\ \middle |_{\theta=\bar\theta_{\tilde n}}\right\}
(\hat\theta_{\tilde n}-\theta^*).
\end{equation*}
By (A2) and (A3), 
\begin{equation*}
    \left\{\frac{\partial}{\partial\beta^\top}\Psi_n(\beta,\bar\theta_{\tilde n})\ \middle |_{\beta=\hat\beta(\bar\theta_{\tilde n})}\right\}^{-1}\plim Q_0^{-1},
\end{equation*}
and 
\begin{equation*}
    \left\{\frac{\partial}{\partial\theta^\top}\Psi_n(\hat\beta_n(\theta),\theta)\ \middle |_{\theta=\bar\theta_{\tilde n}}\right\}\plim Q_1.
\end{equation*}
Thus, the following holds:
\begin{equation*}
\hat\beta_n(\hat\theta_{\tilde n})-\hat\beta_n(\theta^*)
=
- Q_0^{-1}Q_1(\hat\theta_{\tilde n}-\theta^*) + o_p(\|\hat\theta_{\tilde n}-\theta^*\|).    
\end{equation*}
Multiplying by $\sqrt n$, 
\begin{equation*}
    \sqrt n\{\hat\beta_n(\hat\theta_{\tilde n})-\hat\beta_n(\theta^*)\}=- Q_0^{-1}Q_1\,\sqrt n(\hat\theta_{\tilde n}-\theta^*) + o_p\!\left(\sqrt n\,\|\hat\theta_{\tilde n}-\theta^*\|\right).
\end{equation*}
Since $\hat\theta_{\tilde n}-\theta^*=O_p(1/\sqrt {\tilde n})$, we have $\sqrt n\,\|\hat\theta_{\tilde n}-\theta^*\|=O_p(\sqrt{n/\tilde n})$.
Since $\lim_{n,\tilde n\to\infty} {n}/{\tilde n}=:\kappa\in[0,\infty)$, we have $o_p\!\left(\sqrt n\,\|\hat\theta_{\tilde n}-\theta^*\|\right)=o_p(1)$.
Thus, the following holds:
\begin{equation}
\label{eq: expand second term}
    \sqrt n\{\hat\beta_n(\hat\theta_{\tilde n})-\hat\beta_n(\theta^*)\}=- Q_0^{-1}Q_1\,\sqrt n(\hat\theta_{\tilde n}-\theta^*) + o_p(1).
\end{equation}
By \eqref{eq: expand beta fix}, \eqref{eq: decom} and \eqref{eq: expand second term}, the following holds:
\begin{equation*}
       \sqrt n(\hat\beta_n(\hat\theta_{\tilde n})-\beta^*)=-Q_0^{-1}\sqrt n\,\Psi_n(\beta^*,\theta^*)-Q_0^{-1}Q_1\,\sqrt n(\hat\theta_{\tilde n}-\theta^*)+o_p(1).
\end{equation*}
Since $\sqrt{\tilde n}(\hat\theta_{\tilde n}-\theta^*)=-Q_2^{-1}\sqrt{\tilde n}\Phi_{\tilde n}(\theta^*)+o_p(1)$ and $\lim_{n,\tilde n\to\infty} {n}/{\tilde n}=:\kappa\in[0,\infty)$, we have
\begin{equation}
\label{eq: expand beta est}
       \sqrt n(\hat\beta_n(\hat\theta_{\tilde n})-\beta^*)=-Q_0^{-1}\sqrt n\,\Psi_n(\beta^*,\theta^*)+\sqrt{\kappa}Q_0^{-1}Q_1Q_2^{-1}\sqrt{\tilde n}\Phi_{\tilde n}(\theta^*)+o_p(1).
\end{equation}
On the right-hand side of \eqref{eq: expand beta est}, the first term $-Q_0^{-1}\sqrt n\,\Psi_n(\beta^*,\theta^*)$ depends only on dataset $\mathcal{D}$ and the second term $\sqrt{\kappa}Q_0^{-1}Q_1Q_2^{-1}\sqrt{\tilde n}\Phi_{\tilde n}(\theta^*)$ depends only on dataset $\mathcal{\tilde D}$.
Since $\mathcal{D}\perp\mathcal{\tilde D}$, the following holds:
\begin{equation*}
    \sqrt n\bigl(\hat\beta_n(\hat\theta_{\tilde n})-\beta^*\bigr)\xrightarrow[]{d}N\!\left(0,\ Q_0^{-1}\Omega (Q_0^{-1})^\top + \kappa\, Q_0^{-1}Q_1 V_\theta Q_1^\top (Q_0^{-1})^\top\right).
\end{equation*}
\end{proof}

\begin{corollary}
\label{cor: asyvar scl}
Under the same assumptions as in Lemma \ref{thm: asyvar mat}, for any constant vector $e\in \mathbb{R}^d$, the following statements hold:
\begin{equation*}
    \sqrt{n}\,(e^\top\hat\beta_n(\theta^*)-e^\top\beta^*)\ \xrightarrow[]{d}\ N\!\left(0,\ e^\top V_{\mathrm{fix}} e\right),
\end{equation*}
\begin{equation*}
    \sqrt{n}\,(e^\top\hat\beta_n(\hat \theta_{\tilde n})-e^\top\beta^*)\ \xrightarrow[]{d}\ N\!\left(0,\ e^\top V_{\mathrm{est}} e\right),
\end{equation*}
and 
\begin{equation*}
    e^\top V_{\mathrm{est}} e-e^\top V_{\mathrm{fix}} e=\kappa\, b^\top V_\theta b \geq 0,
\end{equation*}
where $b^\top=e^\top Q_0^{-1}Q_1$.
Thus, $e^\top V_{\mathrm{est}} e=e^\top V_{\mathrm{fix}} e$ holds if $V_\theta=0$, $\kappa=0$ or $b=0$.
\end{corollary}

\subsection{Derivation of equation \eqref{eq: cond Q0Q1}}
\label{appendix: thm 2}
\begin{lemma}
\label{thm: param and a}
Under regularity conditions (A1)--(A4), the following holds:
\begin{equation*}
\begin{split}
    \left.\frac{\partial}{\partial\theta^\top}\beta^*(\theta)\middle |_{\theta=\theta^*}\right.=-Q_0^{-1}Q_1.
\end{split}
\end{equation*}
\end{lemma}
\begin{proof}
By the definition of $\beta^*(\theta)$,
\begin{equation*}
    \mathbb{E}[\psi(O;\beta^*(\theta),\theta)]=0.
\end{equation*}
Thus, since the exchangeability of differentiation and integration by (A2) and (A3),
\begin{equation*}
    \mathbb{E}\left[\frac{\partial}{\partial\theta^\top}\psi(O;\beta^*(\theta),\theta)\right]=0.
\end{equation*}
By the chain rule, 
\begin{equation*}
\begin{split}
    \frac{\partial}{\partial\theta^\top}\psi(O;\beta^*(\theta),\theta)=\left\{\frac{\partial}{\partial\beta^\top}\psi(O;\beta,\theta)\middle |_{\beta= \beta^*(\theta)}\right\}\frac{\partial}{\partial\theta^\top}\beta^*(\theta) + \left\{\frac{\partial}{\partial\theta^\top}\psi(O;\beta,\theta)\middle |_{\beta= \beta^*(\theta)}\right\}.
\end{split}
\end{equation*}
Thus, by taking expectation, 
\begin{equation*}
\begin{split}
    0=\mathbb{E}\left[\frac{\partial}{\partial\beta^\top}\psi(O;\beta,\theta)\middle |_{\beta= \beta^*(\theta)}\right]\frac{\partial}{\partial\theta^\top}\beta^*(\theta) + \mathbb{E}\left[\frac{\partial}{\partial\theta^\top}\psi(O;\beta,\theta)\middle |_{\beta= \beta^*(\theta)}\right],
\end{split}
\end{equation*}
that is, 
\begin{equation*}
\begin{split}
    \frac{\partial}{\partial\theta^\top}\beta^*(\theta)=-\mathbb{E}\left[\frac{\partial}{\partial\beta^\top}\psi(O;\beta,\theta)\middle |_{\beta= \beta^*(\theta)}\right]^{-1} \mathbb{E}\left[\frac{\partial}{\partial\theta^\top}\psi(O;\beta,\theta)\middle |_{\beta= \beta^*(\theta)}\right].
\end{split}
\end{equation*}
Then, by evaluating at $\theta=\theta^*$,
\begin{equation*}
\begin{split}
    \left.\frac{\partial}{\partial\theta^\top}\beta^*(\theta)\middle |_{\theta=\theta^*}\right.=-Q_0^{-1}Q_1.
\end{split}
\end{equation*}
\end{proof}

\subsection{Proof of Theorem \ref{thm: asyvar PROCOVA}}
\label{appendix: thm 3}
\begin{proof}
First, we prove $e_{A}^\top Q_0^{-1}Q_1=0$ for $e_{A}=(0,1,0)^\top$.
The second-stage estimating function \eqref{eq: est func PROCOVA linear} can be expressed as follows:
\begin{equation*}
\begin{split}
    \psi(O;\beta,\theta)&=(Y-\beta^\top X_{\theta})X_\theta=\begin{pmatrix}
        Y-\beta^\top X_{\theta}\\
       (Y-\beta^\top X_{\theta})A\\
        (Y-\beta^\top X_{\theta})\theta^\top W
    \end{pmatrix}=\begin{pmatrix}
        Y-\beta_0-\beta_AA-\beta_1\theta^\top W\\
        (Y-\beta_0-\beta_AA-\beta_1\theta^\top W)A\\
        (Y-\beta_0-\beta_AA-\beta_1\theta^\top W)\theta^\top W
    \end{pmatrix}.
\end{split}
\end{equation*}
Thus, $Q_0$ can be expressed as follows:
\begin{equation*}
\begin{split}
        Q_0&= \mathbb{E}\!\left[\frac{\partial}{\partial\beta^\top}\psi(O;\beta,\theta^*)\Big|_{\beta=\beta^*}\right]=\mathbb{E}\!\left[\frac{\partial}{\partial\beta^\top}(Y-\beta^\top X_{\theta^*})X_{\theta^*}\Big|_{\beta=\beta^*}\right]=-\mathbb{E}[X_{\theta^*}X_{\theta^*}^\top]\\
    &=-\mathbb{E}\left[\begin{pmatrix}
   1 & A & \theta^{*\top} W \\
   A & A & \theta^{*\top} WA \\
   \theta^{*\top} W & \theta^{*\top}WA & (\theta^{*\top} W)^2
\end{pmatrix}\right]=-\begin{pmatrix}
   1 & \pi & \mathbb{E}[\theta^{*\top} W] \\
   \pi & \pi & \mathbb{E}[\theta^{*\top} W]\pi \\
   \mathbb{E}[\theta^{*\top} W] & \mathbb{E}[\theta^{*\top} W]\pi & \mathbb{E}[(\theta^{*\top} W)^2]
\end{pmatrix}.
\end{split}
\end{equation*}
The last equation holds because of $A\perp W$.
By denoting $(i,j)$ component of $Q_0^{-1}$ as $\check q_{i,j}$, following holds:
\begin{equation*}
\begin{split}
    \check q_{1,1}&=(-1)^{1+1}\begin{vmatrix}
   \pi  & \mathbb{E}[\theta^{*\top} W]\pi \\
   \mathbb{E}[\theta^{*\top} W]\pi  & \mathbb{E}[(\theta^{*\top} W)^2]
\end{vmatrix}/|Q_0|=\frac{\pi\mathbb{V}[\theta^{*\top} W]+\pi(1-\pi)\mathbb{E}[\theta^{*\top} W]^2}{|Q_0|},\\
\check  q_{1,2}&=\check q_{2,1}=(-1)^{1+2}\begin{vmatrix}
   \pi  & \mathbb{E}[\theta^{*\top} W]\pi \\
   \mathbb{E}[\theta^{*\top} W] & \mathbb{E}[(\theta^{*\top} W)^2]
\end{vmatrix}/|Q_0|=-\frac{\pi\mathbb{V}[\theta^{*\top} W]}{|Q_0|},\\
 \check q_{1,3}&=\check q_{3,1}=(-1)^{1+3}\begin{vmatrix}
   \pi  & \pi \\ 
   \mathbb{E}[\theta^{*\top} W]  & \mathbb{E}[\theta^{*\top} W]\pi
\end{vmatrix}/|Q_0|=-\frac{\pi(1-\pi)\mathbb{E}[\theta^{*\top} W]}{|Q_0|},\\
 \check q_{2,2}&=(-1)^{2+2}\begin{vmatrix}
   1  & \mathbb{E}[\theta^{*\top} W] \\
   \mathbb{E}[\theta^{*\top} W] & \mathbb{E}[(\theta^{*\top} W)^2]
\end{vmatrix}/|Q_0|=\frac{\mathbb{V}[\theta^{*\top} W]}{|Q_0|},\\
 \check q_{2,3}&=\check q_{3,2}=(-1)^{2+3}\begin{vmatrix}
   1  &  \pi\\
   \mathbb{E}[\theta^{*\top} W] & \mathbb{E}[\theta^{*\top} W]\pi
\end{vmatrix}/|Q_0|=0,\\
\check q_{3,3}&=(-1)^{3+3}\begin{vmatrix}
   1  & \pi\\
   \pi & \pi
\end{vmatrix}/|Q_0|=\frac{\pi(1-\pi)}{|Q_0|}.
\end{split}
\end{equation*}
Thus, the following holds:
\begin{equation*}
    \check q_{1,2}=\check q_{2,1}=-\pi\check q_{2,2},\quad \check q_{1,3}=\check q_{3,1}=-\check q_{3,3}\mathbb{E}[\theta^{*\top} W],\quad \text{and} \quad\check q_{1,1}= \pi \check q_{2,2} + \check q_{3,3}\mathbb{E}[\theta^{*\top} W]^2.
\end{equation*}
Then, $Q_0^{-1}$ can be expressed as follows:
\begin{equation*}
    Q_0^{-1}=\begin{pmatrix}
  \pi \check q_{2,2} + \check q_{3,3}\mathbb{E}[\theta^{*\top} W]^2 & -\pi \check q_{2,2} & -\check q_{3,3}\mathbb{E}[\theta^{*\top} W] \\
    -\pi \check q_{2,2} & \check q_{2,2} & 0 \\
  -\check q_{3,3}\mathbb{E}[\theta^{*\top} W] & 0 & \check q_{3,3}
\end{pmatrix}.
\end{equation*}
Additionally, $Q_1$ can be expressed as follows:
\begin{equation*}
    Q_1=\mathbb{E}\!\left[\frac{\partial}{\partial\theta^\top}\psi(O;\beta^*,\theta)\Big|_{\theta=\theta^*}\right]=\mathbb{E}\!\left[\frac{\partial}{\partial\theta^\top}\begin{pmatrix}
        Y-\beta_0^*-\beta_A^*A-\beta_1^*\theta^\top W\\
        (Y-\beta_0^*-\beta_A^*A-\beta_1^*\theta^\top W)A\\
        (Y-\beta_0^*-\beta_A^*A-\beta_1^*\theta^\top W)\theta^\top W
    \end{pmatrix}\middle |_{\theta=\theta^*}\right]=\begin{pmatrix}
        q_1\\
        q_2\\
        q_3
    \end{pmatrix},
\end{equation*}
where 
\begin{equation*}
\begin{split}
    \beta_0^*&=e_0^\top \beta^*,\quad \text{with}\quad e_0=(1,0,0)^\top,\\
    \beta_A^*&=e_A^\top \beta^*,\quad \text{with}\quad e_A=(0,1,0)^\top,\\
    \beta_1^*&=e_1^\top \beta^*,\quad \text{with}\quad e_1=(0,0,1)^\top,
\end{split}
\end{equation*}
and 
\begin{equation*}
\begin{split}
    q_1&=-\beta_{1}^* \mathbb{E}[W^\top],\\
    q_2&=-\beta_{1}^* \mathbb{E}[AW^\top],\\
    q_3&=\mathbb{E}[YW^\top]-\beta_0^*\mathbb{E}[W^\top]-\beta_A^*\mathbb{E}[AW^\top]-2\beta_1^*\mathbb{E}[\theta^{*\top}WW^\top].
\end{split}
\end{equation*}
Since $A\perp W$, the following hold:
\begin{equation*}
    q_2=-\beta_{1}^*\pi\mathbb{E}[W^\top]=\pi q_1\quad\text{and}\quad q_3 =\mathbb{E}[YW^\top]-\beta_0^*\mathbb{E}[W^\top]-\beta_A^*\pi\mathbb{E}[W^\top]-2\beta_1^*\mathbb{E}[\theta^{*\top}WW^\top].
\end{equation*}
Since 
\begin{equation*}
    |Q_0|=-\pi(1-\pi)\mathbb{V}[\theta^{*\top}W]
\end{equation*}
and
\begin{equation*}
    \beta^*=\mathbb{E}[X_{\theta^*}X_{\theta^*}^\top]^{-1}\mathbb{E}[X_{\theta^*}Y]=-\begin{pmatrix}
  \pi \check q_{2,2} + \check q_{3,3}\mathbb{E}[\theta^{*\top} W]^2 & -\pi \check q_{2,2} & -\check q_{3,3}\mathbb{E}[\theta^{*\top} W] \\
    -\pi \check q_{2,2} & \check q_{2,2} & 0 \\
  -\check q_{3,3}\mathbb{E}[\theta^{*\top} W] & 0 & \check q_{3,3}
\end{pmatrix}\begin{pmatrix}
        \mathbb{E}[Y]\\
        \mathbb{E}[AY]\\
        \mathbb{E}[\theta^{*\top} W Y]
    \end{pmatrix},
\end{equation*}
the components of $\beta^*=(\beta^*_0,\beta^*_A,\beta_1^*)^\top$ can be expressed as follows:
\begin{equation*}
\begin{split}
    \beta_A^*&=\check q_{2,2}(\pi\mathbb{E}[Y]-\mathbb{E}[AY])=-\frac{\mathbb{V}[\theta^{*\top}W]}{\pi(1-\pi)\mathbb{V}[\theta^{*\top}W]}(-1)\pi(1-\pi)\{\mathbb{E}[Y\mid A=1]-\mathbb{E}[Y\mid A=0]\}\\&=\mathbb{E}[Y\mid A=1]-\mathbb{E}[Y\mid A=0],\\
    \beta_1^*&=\check q_{3,3}(\mathbb{E}[\theta^{*\top}W]\mathbb{E}[Y]-\mathbb{E}[\theta^{*\top}WY])=-\frac{\pi(1-\pi)}{\pi(1-\pi)\mathbb{V}[\theta^{*\top}W]}(-1)\mathbb{COV}[\theta^{*\top}W, Y]\\
    &=\frac{\mathbb{COV}[\theta^{*\top}W, Y]}{\mathbb{V}[\theta^{*\top}W]},\\
    \beta_0^*&=-\pi\check q_{2,2}(\mathbb{E}[Y]-\mathbb{E}[AY])-\check q_{3,3}\mathbb{E}[\theta^{*\top}W](\mathbb{E}[\theta^{*\top}W]\mathbb{E}[Y]-\mathbb{E}[\theta^{*\top}WY])\\
    &=\pi\frac{\mathbb{V}[\theta^{*\top}W]}{\pi(1-\pi)\mathbb{V}[\theta^{*\top}W]}(1-\pi)\mathbb{E}[Y\mid A=0]+\frac{\pi(1-\pi)}{\pi(1-\pi)\mathbb{V}[\theta^{*\top}W]}\mathbb{E}[\theta^{*\top}W](-1)\mathbb{COV}[\theta^{*\top}W, Y]\\
    &=\mathbb{E}[Y\mid A=0]-\frac{\mathbb{E}[\theta^{*\top}W]}{\mathbb{V}[\theta^{*\top}W]}\mathbb{COV}[\theta^{*\top}W, Y].
\end{split}
\end{equation*}
Then, $q_1$ and $q_3$ can be expressed as follows:
\begin{equation*}
\begin{split}
    q_1&=-\frac{\mathbb{COV}[\theta^{*\top}W, Y]}{\mathbb{V}[\theta^{*\top}W]} \mathbb{E}[W^\top],\\
    q_3 &=\mathbb{E}[YW^\top]-\mathbb{E}[Y]\mathbb{E}[W^\top]-\frac{\mathbb{COV}[\theta^{*\top}W,Y]}{\mathbb{V}[\theta^{*\top}W]}(2\mathbb{E}[\theta^{*\top}WW^\top]-\mathbb{E}[\theta^{*\top}W]\mathbb{E}[W^\top]).
\end{split}   
\end{equation*}
Then, for $e_A=(0,1,0)^\top$, 
\begin{equation*}
     e_{A}^\top Q_0^{-1}Q_1=-\pi\check q_{2,2}q_1+\pi\check q_{2,2}q_1 +0=0.
\end{equation*}
For $e_0=(1,0,0)^\top$ and $e_1=(0,0,1)^\top$, 
\begin{equation*}
\begin{split}
  e_{1}^\top Q_0^{-1}Q_1&= -\check q_{3,3}\mathbb{E}[\theta^{*\top} W]q_1 + \check q_{3,3}q_3\\
    &=-\frac{1}{\mathbb{V}[\theta^{*\top}W]}\left\{\mathbb{COV}[Y,W^\top]-2\frac{\mathbb{COV}[\theta^{*\top}W, Y]}{\mathbb{V}[\theta^{*\top}W]}\mathbb{COV}[\theta^{*\top}W,W^\top]\right\},\\
     e_0^\top Q_0^{-1}Q_1&=(\pi \check q_{2,2} + \check q_{3,3}\mathbb{E}[\theta^{*\top} W]^2)q_1 -\pi^2 \check q_{2,2}q_1 - \check q_{3,3}\mathbb{E}[\theta^{*\top} W]q_3\\
    &=-\mathbb{E}[\theta^{*\top} W]e_{1}^\top Q_0^{-1}Q_1+\pi(1-\pi) \check q_{2,2}q_1\\
    &=\frac{1}{\mathbb{V}[\theta^{*\top}W]}\left\{\mathbb{E}[\theta^{*\top} W]\left(\mathbb{COV}[Y,W^\top]-2\frac{\mathbb{COV}[\theta^{*\top}W, Y]}{\mathbb{V}[\theta^{*\top}W]}\mathbb{COV}[\theta^{*\top}W,W^\top]\right)\right.\\
    &\left.\qquad\qquad\qquad+\mathbb{COV}[\theta^{*\top}W, Y]\mathbb{E}[W^\top]\right\},
\end{split}
\end{equation*}
where 
\begin{equation*}
    \mathbb{COV}[Y,W^\top]=\mathbb{E}[YW^\top]-\mathbb{E}[Y]\mathbb{E}[W^\top]
\end{equation*}
and
\begin{equation*}
    \mathbb{COV}[\theta^{*\top}W,W^\top]=\mathbb{E}[\theta^{*\top}WW^\top]-\mathbb{E}[\theta^{*\top}W]\mathbb{E}[W^\top].
\end{equation*}

Next, we prove $\left.\frac{\partial}{\partial\theta^\top}e_A^\top\beta^*(\theta)\middle |_{\theta=\theta^*}\right.=0$ for $e_A=(0,1,0)^\top$.
As with the derivation of the components of $\beta^*$, the components of $\beta^*(\theta)=(\beta^*_0(\theta),\beta^*_A(\theta),\beta_1^*(\theta))^\top$ can be expressed as follows:
\begin{equation*}
\begin{split}
    \beta_A^*(\theta)&=-\frac{\mathbb{V}[\theta^{\top}W]}{\pi(1-\pi)\mathbb{V}[\theta^{\top}W]}(-1)\pi(1-\pi)\{\mathbb{E}[Y\mid A=1]-\mathbb{E}[Y\mid A=0]\}=\mathbb{E}[Y\mid A=1]-\mathbb{E}[Y\mid A=0],\\
    \beta_1^*(\theta)&=-\frac{\pi(1-\pi)}{\pi(1-\pi)\mathbb{V}[\theta^{\top}W]}(-1)\mathbb{COV}[\theta^{\top}W, Y]=\frac{\mathbb{COV}[\theta^{\top}W, Y]}{\mathbb{V}[\theta^{\top}W]},\\
    \beta_0^*(\theta)&=\pi\frac{\mathbb{V}[\theta^{\top}W]}{\pi(1-\pi)\mathbb{V}[\theta^{\top}W]}(1-\pi)\mathbb{E}[Y\mid A=0]+\frac{\pi(1-\pi)}{\pi(1-\pi)\mathbb{V}[\theta^{\top}W]}\mathbb{E}[\theta^{\top}W](-1)\mathbb{COV}[\theta^{\top}W, Y]\\
    &=\mathbb{E}[Y\mid A=0]-\frac{\mathbb{E}[\theta^{\top}W]}{\mathbb{V}[\theta^{\top}W]}\mathbb{COV}[\theta^{\top}W, Y].
\end{split}
\end{equation*}
That is, 
\begin{equation*}
    \left.\frac{\partial}{\partial\theta^\top}e_A^\top\beta^*(\theta)\middle |_{\theta=\theta^*}\right.=\left.\frac{\partial}{\partial\theta^\top}\beta^*_A(\theta)\middle |_{\theta=\theta^*}\right.=0, \quad \text{with}\quad e_A=(0,1,0)^\top.
\end{equation*}
\end{proof}

\subsection{Connection between model \eqref{eq: PROCOVAII linear emp} and estimating function \eqref{eq: est func PROCOVAII linear}}

\begin{lemma}
\label{thm: pop emp}
Let $\hat\beta_n^{\mathrm{emp}}(\hat \theta_{\tilde n})$ be the OLS estimator of the model \eqref{eq: PROCOVAII linear emp}. Let $\hat\beta_n(\hat \theta_{\tilde n})$ be the OLS estimator based on the estimating function \eqref{eq: est func PROCOVAII linear}.
Assume the same assumptions as in Lemma \ref{thm: asyvar mat}.
Additionally, assume $\mathbb{E}[Y^2]<\infty$ and $\mathbb{E}[||W||^2]<\infty$.
Then, the following holds:
\[
\sqrt n\,
\Bigl\{
\hat\beta_n^{\mathrm{emp}}(\hat\theta_{\tilde n})
-
\hat\beta_n^{\mathrm{emp}}(\theta^*)
\Bigr\}
=
\sqrt n\,
\Bigl\{
\hat\beta_n(\hat\theta_{\tilde n})
-
\hat\beta_n(\theta^*)
\Bigr\}
+
o_p(1).
\]
\end{lemma}

\label{appendix: proof pop emp}

\begin{proof}
By denoting 
\[
 \bar D_\theta:= \frac{1}{n}\sum_{j=1}^n\theta^\top W_j- \mathbb{E}[\theta^\top W],
\]
the following holds:
\[
\theta^\top W_i- \frac{1}{n}\sum_{j=1}^n\theta^\top W_j
=
\theta^\top W_i-\mathbb{E}[\theta^\top W]-\bar D_\theta.
\]
Then, the following holds:
\[
X_{\theta,i}^{\mathrm{emp}}
= B_n(\theta) X_{\theta,i},
\]
where
\[
B_n(\theta)
=
\begin{pmatrix}
1 & 0 & 0 & 0\\
0 & 1 & 0 & 0\\
-\bar D_\theta & 0 & 1 & 0\\
0 & -\bar D_\theta & 0 & 1
\end{pmatrix}.
\]
Let ${\mathbf X}_\theta^{\mathrm{emp}}$ and ${\mathbf X}_\theta$  be the $n\times 4$ matrix whose $i$-th rows are $X_{\theta,i}^{\mathrm{emp}}$ and $X_{\theta,i}$, respectively. 
Then,
\[
{\mathbf X}_\theta^{\mathrm{emp}}
=
{\mathbf X}_\theta B_n(\theta)^\top.
\]
Therefore,
\begin{align*}
\hat\beta_{n}^{\mathrm{emp}}(\theta)
&=
\{({\mathbf X}_\theta^{\mathrm{emp}})^{\top}{\mathbf X}_\theta^{\mathrm{emp}}\}^{-1}
({\mathbf X}_\theta^{\mathrm{emp}})^{\top} Y \\
&=
\{B_n(\theta)({\mathbf X}_\theta)^{\top}{\mathbf X}_\theta B_n(\theta)^\top\}^{-1}
B_n(\theta)({\mathbf X}_\theta)^{\top} Y \\
&=
B_n(\theta)^{-\top}
\{({\mathbf X}_\theta)^{\top}{\mathbf X}_\theta\}^{-1}
({\mathbf X}_\theta)^{\top} Y \\
&=
B_n(\theta)^{-\top}\hat\beta_n(\theta).
\end{align*}
Since
\[
B_n(\theta)^{-1}=\begin{pmatrix}
1 & 0 & 0 & 0\\
0 & 1 & 0 & 0\\
\bar D_\theta & 0 & 1 & 0\\
0 & \bar D_\theta & 0 & 1
\end{pmatrix},
\]
the components of $\hat\beta_{n}^{\mathrm{emp}}(\theta)=( \hat\beta_{0,n}^{\mathrm{emp}}(\theta), \hat\beta_{A,n}^{\mathrm{emp}}(\theta), \hat\beta_{1,n}^{\mathrm{emp}}(\theta), \hat\beta_{2,n}^{\mathrm{emp}}(\theta))^\top$ can be expressed as 
\begin{equation*}
\label{eq: emp pop diff 0}
    \hat\beta_{0,n}^{\mathrm{emp}}(\theta)
=
\hat\beta_{0,n}(\theta)
+
\bar D_{\theta}\hat\beta_{1,n}(\theta),
\end{equation*}
\begin{equation}
\label{eq: emp pop diff A}
    \hat\beta_{A,n}^{\mathrm{emp}}(\theta)
=
\hat\beta_{A,n}(\theta)
+
\bar D_{\theta}\hat\beta_{2,n}(\theta),
\end{equation}
and
\begin{equation}
\label{eq: pop emp equal 1 2}
    \hat\beta_{1,n}^{\mathrm{emp}}(\theta)
=
\hat\beta_{1,n}(\theta),\quad \hat\beta_{2,n}^{\mathrm{emp}}(\theta)
=
\hat\beta_{2,n}(\theta).
\end{equation}

By $\mathbb{E}[\|W\|^2]<\infty$,
\[
\bar W - \mathbb{E}[W] = O_p(n^{-1/2}).
\]
Then, 
\[
\bar D_{\theta^*}
=
\theta^{*\top}\{\bar W-\mathbb{E}[W]\}
=
O_p(n^{-1/2}),
\]
and
\begin{equation}
\label{eq: D diff}
    \bar D_{\hat \theta_{\tilde n}}-\bar D_{\theta^*}
=
(\hat\theta_{\tilde n}-\theta^\ast)^\top\{\bar W-\mathbb{E}[W]\}
=
O_p(n^{-1/2}\tilde n^{-1/2}).
\end{equation}
Additionally, by the same argument as in the proof of Theorem \ref{thm: asyvar mat},
\[
\hat\beta_n(\hat\theta_{\tilde n})
-
\hat\beta_n(\theta^\ast)
=
O_p(\|\hat\theta_{\tilde n}-\theta^\ast\|)
=
O_p(\tilde n^{-1/2}).
\]
Thus,
\begin{equation*}
\label{eq: pop diff 1}
    \hat\beta_{1,n}(\hat\theta_{\tilde n})
-
\hat\beta_{1,n}(\theta^\ast)
=
O_p(\tilde n^{-1/2}),\quad \hat\beta_{1,n}(\hat\theta_{\tilde n})=O_p(1),
\end{equation*}
and
\begin{equation}
\label{eq: pop diff 2}
\hat\beta_{2,n}(\hat\theta_{\tilde n})
-
\hat\beta_{2,n}(\theta^\ast)
=
O_p(\tilde n^{-1/2}),\quad \hat\beta_{2,n}(\hat\theta_{\tilde n})=O_p(1).
\end{equation}

For $e_A=(0,1,0,0)^\top$, by \eqref{eq: emp pop diff A}, \eqref{eq: D diff} and \eqref{eq: pop diff 2}, 
\begin{align*}
&\Bigl\{
\hat\beta_{A,n}^{\mathrm{emp}}(\hat\theta_{\tilde n})
-
\hat\beta_{A,n}^{\mathrm{emp}}(\theta^\ast)
\Bigr\}
-
\Bigl\{
\hat\beta_{A,n}(\hat\theta_{\tilde n})
-
\hat\beta_{A,n}(\theta^\ast)
\Bigr\} \\
& =
\bar D_{\hat \theta_{\tilde n}}\hat\beta_{2,n}(\hat\theta_{\tilde n})
-
\bar D_{\theta^*}\hat\beta_{2,n}(\theta^\ast) \\
& =
\bigl\{
\bar D_{\hat \theta_{\tilde n}}-\bar D_{\theta^*}
\bigr\}
\hat\beta_{2,n}(\hat\theta_{\tilde n})
+
\bar D_{\theta^*}
\bigl\{
\hat\beta_{2,n}(\hat\theta_{\tilde n})
-
\hat\beta_{2,n}(\theta^\ast)
\bigr\}\\
&= O_p(n^{-1/2}\tilde n^{-1/2}),
\end{align*}
and then
\[
\sqrt n\,
\Bigl\{
\hat\beta_{A,n}^{\mathrm{emp}}(\hat\theta_{\tilde n})
-
\hat\beta_{A,n}^{\mathrm{emp}}(\theta^\ast)
\Bigr\}
=
\sqrt n\,
\Bigl\{
\hat\beta_{A,n}(\hat\theta_{\tilde n})
-
\hat\beta_{A,n}(\theta^\ast)
\Bigr\}
+
o_p(1).
\]
The same argument holds for $e_0=(1,0,0,0)^\top$:
\[
\sqrt n\,
\Bigl\{
\hat\beta_{0,n}^{\mathrm{emp}}(\hat\theta_{\tilde n})
-
\hat\beta_{0,n}^{\mathrm{emp}}(\theta^\ast)
\Bigr\}
=
\sqrt n\,
\Bigl\{
\hat\beta_{0,n}(\hat\theta_{\tilde n})
-
\hat\beta_{0,n}(\theta^\ast)
\Bigr\}
+
o_p(1).
\]
By combining \eqref{eq: pop emp equal 1 2}, the following holds:
\[
\sqrt n\,
\Bigl\{
\hat\beta_n^{\mathrm{emp}}(\hat\theta_{\tilde n})
-
\hat\beta_n^{\mathrm{emp}}(\theta^*)
\Bigr\}
=
\sqrt n\,
\Bigl\{
\hat\beta_n(\hat\theta_{\tilde n})
-
\hat\beta_n(\theta^*)
\Bigr\}
+
o_p(1).
\]
\end{proof}

\subsection{Proof of Theorem \ref{thm: asyvar PROCOVAII}}
\label{appendix: thm 4}
\begin{proof}
First, we prove $e^\top Q_0^{-1}Q_1=0$ for $e=(0,1,0,0)^\top$, $e=(1,0,0,0)^\top$ and $e=(1,1,0,0)^\top$.
The second-stage estimating function \eqref{eq: est func PROCOVAII linear} can be expressed as follows:
\begin{equation*}
\begin{split}
    \psi(O;\beta,\theta)&=(Y-\beta^\top X_{\theta})X_\theta\\
    &=\begin{pmatrix}
        Y-\beta^\top X_{\theta}\\
       (Y-\beta^\top X_{\theta})A\\
        (Y-\beta^\top X_{\theta})(\theta^\top W- \mathbb{E}[\theta^\top W])\\
        (Y-\beta^\top X_{\theta})(\theta^\top W- \mathbb{E}[\theta^\top W])A
    \end{pmatrix}\\
    &=\begin{pmatrix}
        Y-\beta_0-\beta_AA-\beta_1(\theta^\top W- \mathbb{E}[\theta^\top W]) -\beta_2(\theta^\top W- \mathbb{E}[\theta^\top W])A\\
        \{Y-\beta_0-\beta_AA-\beta_1(\theta^\top W- \mathbb{E}[\theta^\top W]) -\beta_2(\theta^\top W- \mathbb{E}[\theta^\top W])A\}A\\
         \{Y-\beta_0-\beta_AA-\beta_1(\theta^\top W- \mathbb{E}[\theta^\top W]) -\beta_2(\theta^\top W- \mathbb{E}[\theta^\top W])A\}(\theta^\top W- \mathbb{E}[\theta^\top W])\\
         \{Y-\beta_0-\beta_AA-\beta_1(\theta^\top W- \mathbb{E}[\theta^\top W]) -\beta_2(\theta^\top W- \mathbb{E}[\theta^\top W])A\}(\theta^\top W- \mathbb{E}[\theta^\top W])A
    \end{pmatrix}.
\end{split}
\end{equation*}
Thus, $Q_0$ can be expressed as follows:
\begin{equation*}
\begin{split}
        Q_0&= \mathbb{E}\!\left[\frac{\partial}{\partial\beta^\top}\psi(O;\beta,\theta^*)\Big|_{\beta=\beta^*}\right]=\mathbb{E}\!\left[\frac{\partial}{\partial\beta^\top}(Y-\beta^\top X_{\theta^*})X_{\theta^*}\Big|_{\beta=\beta^*}\right]=-\mathbb{E}[X_{\theta^*}X_{\theta^*}^\top]\\
    &=-\mathbb{E}\left[\begin{pmatrix}
   1 & A & \theta^\top W- \mathbb{E}[\theta^\top W] & (\theta^\top W- \mathbb{E}[\theta^\top W])A \\
   A & A & (\theta^\top W- \mathbb{E}[\theta^\top W])A &  (\theta^\top W- \mathbb{E}[\theta^\top W])A \\
   \theta^\top W- \mathbb{E}[\theta^\top W]  & (\theta^\top W- \mathbb{E}[\theta^\top W])A & (\theta^\top W- \mathbb{E}[\theta^\top W])^2 & (\theta^\top W- \mathbb{E}[\theta^\top W])^2 A\\
     (\theta^\top W- \mathbb{E}[\theta^\top W]) A & (\theta^\top W- \mathbb{E}[\theta^\top W]) A & (\theta^\top W- \mathbb{E}[\theta^\top W])^2 A & (\theta^\top W- \mathbb{E}[\theta^\top W])^2 A
\end{pmatrix}\right]\\
&=-\begin{pmatrix}
   1 & \pi & 0 &  0 \\
   \pi & \pi & 0 & 0 \\
  0 & 0 & \mathbb{V}[\theta^{*\top} W] & \mathbb{V}[\theta^{*\top} W]\pi \\
    0 & 0  & \mathbb{V}[\theta^{*\top} W]\pi  & \mathbb{V}[\theta^{*\top} W]\pi \\
\end{pmatrix}.
\end{split}
\end{equation*}
The last equation holds by $A\perp W$.
Thus, $Q_0^{-1}$ can be expressed as follows:
\begin{equation*}
    Q_0^{-1}=-\begin{pmatrix}
   \frac{1}{1-\pi} & -\frac{1}{1-\pi} & 0 &  0 \\
    -\frac{1}{1-\pi} & \frac{1}{\pi(1-\pi)} & 0 & 0 \\
  0 & 0 & \frac{1}{(1-\pi)\mathbb{V}[\theta^{*\top} W]} & -\frac{1}{(1-\pi)\mathbb{V}[\theta^{*\top} W]} \\
    0 & 0  & -\frac{1}{(1-\pi)\mathbb{V}[\theta^{*\top} W]}  & \frac{1}{\pi(1-\pi)\mathbb{V}[\theta^{*\top} W]} \\
\end{pmatrix}
\end{equation*}
Additionally, $Q_1$ can be expressed as follows:
\begin{equation*}
\begin{split}
     &Q_1\\
     &=\mathbb{E}\!\left[\frac{\partial}{\partial\theta^\top}\psi(O;\beta^*,\theta)\Big|_{\theta=\theta^*}\right]\\
     &=\mathbb{E}\!\left[\frac{\partial}{\partial\theta^\top}\begin{pmatrix}
        Y-\beta_0^*-\beta_A^*A-\beta_1^*(\theta^\top W- \mathbb{E}[\theta^\top W])-\beta_2^*(\theta^\top W- \mathbb{E}[\theta^\top W])A\\
        \{Y-\beta_0^*-\beta_A^*A-\beta_1^*(\theta^\top W- \mathbb{E}[\theta^\top W])-\beta_2^*(\theta^\top W- \mathbb{E}[\theta^\top W])A\}A\\
        \{Y-\beta_0^*-\beta_A^*A-\beta_1^*(\theta^\top W- \mathbb{E}[\theta^\top W])-\beta_2^*(\theta^\top W- \mathbb{E}[\theta^\top W])A\}(\theta^\top W- \mathbb{E}[\theta^\top W])\\
         \{Y-\beta_0^*-\beta_A^*A-\beta_1^*(\theta^\top W- \mathbb{E}[\theta^\top W])-\beta_2^*(\theta^\top W- \mathbb{E}[\theta^\top W])A\}(\theta^\top W- \mathbb{E}[\theta^\top W])A
    \end{pmatrix}\middle |_{\theta=\theta^*}\right],
\end{split}
\end{equation*} 
where
\begin{equation*}
\begin{split}
    \beta_0^*&=e_0^\top \beta^*,\quad \text{with}\quad e_0=(1,0,0,0)^\top,\\
    \beta_A^*&=e_A^\top \beta^*,\quad \text{with}\quad e_A=(0,1,0,0)^\top,\\
    \beta_1^*&=e_1^\top \beta^*,\quad \text{with}\quad e_1=(0,0,1,0)^\top,\\
    \beta_2^*&=e_2^\top \beta^*,\quad \text{with}\quad e_2=(0,0,0,1)^\top.
\end{split}
\end{equation*}
By denoting $j$ th row component of $Q_1$ as $q_j$, 
\begin{equation*}
    \begin{split}
        q_1&=-\beta_1^*(\mathbb{E}[W^\top]-\mathbb{E}[W^\top]) -\beta_2^*(\mathbb{E}[A
        W^\top]-\pi\mathbb{E}[W^\top])=-\beta_2^*(\mathbb{E}[A
        W^\top]-\pi\mathbb{E}[W^\top]),\\
        q_2 &=-\beta_1^*(\mathbb{E}[AW^\top]-\pi\mathbb{E}[W^\top]) -\beta_2^*(\mathbb{E}[AW^\top]-\pi\mathbb{E}[W^\top])=-(\beta_1^*+\beta_2^*)(\mathbb{E}[AW^\top]-\pi\mathbb{E}[W^\top]),\\
        q_3&=\mathbb{E}[YW^\top]-\mathbb{E}[Y]\mathbb{E}[W^\top]-\beta_0^*(\mathbb{E}[W^\top]-\mathbb{E}[W^\top])-\beta_A^*(\mathbb{E}[AW^\top]-\pi\mathbb{E}[W^\top])\\
        &\quad - 2\beta_1^*\theta ^\top \mathbb{V}[W] - 2\beta_2^*\theta^\top\pi\mathbb{V}[W\mid A=1]\\
        &=\mathbb{E}[YW^\top]-\mathbb{E}[Y]\mathbb{E}[W^\top]-\beta_A^*(\mathbb{E}[AW^\top]-\pi\mathbb{E}[W^\top])- 2\beta_1^*\theta ^\top \mathbb{V}[W] - 2\beta_2^*\theta^\top\pi\mathbb{V}[W\mid A=1],\\
        q_4&=\mathbb{E}[AYW^\top]-\mathbb{E}[AY]\mathbb{E}[W^\top]-\beta_0^*(\mathbb{E}[AW^\top]-\pi\mathbb{E}[W^\top])-\beta_A^*(\mathbb{E}[AW^\top]-\pi\mathbb{E}[W^\top])\\
        &\quad - 2\beta_1^*\theta ^\top \pi\mathbb{V}[W\mid A=1] - 2\beta_2^*\theta^\top\pi\mathbb{V}[W\mid A=1]\\
        &=\mathbb{E}[AYW^\top]-\mathbb{E}[AY]\mathbb{E}[W^\top]-(\beta_0^*+\beta_A^*)(\mathbb{E}[AW^\top]-\pi\mathbb{E}[W^\top])- 2(\beta_1^*+\beta_2^*)\theta^\top\pi\mathbb{V}[W\mid A=1].
    \end{split}
\end{equation*}
Since $A\perp W$, 
\begin{equation*}
    \begin{split}
        q_1&=q_2=0,\\
        q_3&=\mathbb{E}[YW^\top]-\mathbb{E}[Y]\mathbb{E}[W^\top]- 2(\beta_1^*+\beta_2^*\pi)\theta ^\top \mathbb{V}[W],\\
        q_4
        &=\mathbb{E}[AYW^\top]-\mathbb{E}[AY]\mathbb{E}[W^\top]- 2(\beta_1^*+\beta_2^*)\pi\theta^\top\mathbb{V}[W],
    \end{split}
\end{equation*}
that is, $Q_1$ can be expressed as follows:
\begin{equation*}
    Q_1=\begin{pmatrix}
        0\\
        0\\
        \mathbb{E}[YW^\top]-\mathbb{E}[Y]\mathbb{E}[W^\top]- 2(\beta_1^*+\beta_2^*\pi)\theta ^\top \mathbb{V}[W]\\
        \mathbb{E}[AYW^\top]-\mathbb{E}[AY]\mathbb{E}[W^\top]- 2(\beta_1^*+\beta_2^*)\pi\theta^\top\mathbb{V}[W]
    \end{pmatrix}.
\end{equation*}
Then, 
\begin{equation*}
\begin{split}
 e_A^\top Q_0^{-1}Q_1&=\frac{1}{1-\pi}\cdot0 - \frac{1}{\pi(1-\pi)}\cdot0 + 0\cdot q_3 + 0\cdot q_4 =0 \quad \text{with}\quad e_A=(0,1,0,0)^\top,\\
  e_{0}^\top Q_0^{-1}Q_1&=-\frac{1}{1-\pi}\cdot0 + \frac{1}{1-\pi}\cdot0 + 0\cdot q_3 + 0\cdot q_4 =0 \quad \text{with}\quad e_0=(1,0,0,0)^\top,\\
   e_{0A}^\top Q_0^{-1}Q_1&=0\cdot0 - \frac{1}{\pi}\cdot0 + 0\cdot q_3 + 0\cdot q_4 =0 \quad \text{with}\quad e_{0A}=(1,1,0,0)^\top.
\end{split}
\end{equation*}

Next, we prove $\left.\frac{\partial}{\partial\theta^\top}e^\top\beta^*(\theta)\middle |_{\theta=\theta^*}\right.=0$ for $e=(0,1,0,0)^\top$, $e=(1,0,0,0)^\top$ and $e=(1,1,0,0)^\top$.
As with the derivation of $Q_0^{-1}$, under $A\perp W$,
\begin{equation*}
 \mathbb{E}[X_{\theta}X_\theta^\top]^{-1}=\begin{pmatrix}
   \frac{1}{1-\pi} & -\frac{1}{1-\pi} & 0 &  0 \\
    -\frac{1}{1-\pi} & \frac{1}{\pi(1-\pi)} & 0 & 0 \\
  0 & 0 & \frac{1}{(1-\pi)\mathbb{V}[\theta^{\top} W]} & -\frac{1}{(1-\pi)\mathbb{V}[\theta^{\top} W]} \\
    0 & 0  & -\frac{1}{(1-\pi)\mathbb{V}[\theta^{\top} W]}  & \frac{1}{\pi(1-\pi)\mathbb{V}[\theta^{\top} W]} \\
\end{pmatrix}.
\end{equation*}
Additionally,
\begin{equation*}
\begin{split}
    \mathbb{E}[X_\theta Y]&=
    \begin{pmatrix}
        \mathbb{E}[Y]\\
        \mathbb{E}[AY]\\
        \mathbb{E}[\theta^\top WY]-\mathbb{E}[\theta^\top W]\mathbb{E}[Y]\\
        \mathbb{E}[\theta^\top W AY] - \mathbb{E}[\theta^\top W]\mathbb{E}[AY]
    \end{pmatrix}
    =\begin{pmatrix}
         \pi\mathbb{E}[Y\mid A=1] +(1-\pi)\mathbb{E}[Y\mid A=0]\\
        \pi\mathbb{E}[Y\mid A=1]\\
       \pi g_\theta(1) + (1-\pi)g_\theta(0)\\
       \pi g_\theta(1)
    \end{pmatrix},
\end{split}
\end{equation*}
where $g_\theta(a)=\mathbb{E}[\theta^\top W Y\mid A=a] - \mathbb{E}[\theta^\top W]\mathbb{E}[Y\mid A=a]$ for $a=0,1$.
Then, 
\begin{equation*}
\begin{split}
    \beta^*(\theta)=
\begin{pmatrix}
\frac{\pi}{1-\pi}\mathbb{E}[Y\mid A=1] + \mathbb{E}[Y\mid A=0]  -     \frac{\pi}{1-\pi}\mathbb{E}[Y\mid A=1]\\
-\frac{\pi}{1-\pi}\mathbb{E}[Y\mid A=1] - \mathbb{E}[Y\mid A=0]  +     \frac{1}{1-\pi}\mathbb{E}[Y\mid A=1]\\
\frac{\pi g_\theta(1)}{(1-\pi)\mathbb{V}[\theta^{\top} W]} + \frac{g_\theta(0)}{\mathbb{V}[\theta^{\top} W]} -\frac{\pi g_\theta(1)}{(1-\pi)\mathbb{V}[\theta^{\top} W]}\\
-\frac{\pi g_\theta(1)}{(1-\pi)\mathbb{V}[\theta^{\top} W]} - \frac{g_\theta(0)}{\mathbb{V}[\theta^{\top} W]} +\frac{ g_\theta(1)}{(1-\pi)\mathbb{V}[\theta^{\top} W]}
\end{pmatrix}=
\begin{pmatrix}
  \mathbb{E}[Y\mid A=0]\\
  \mathbb{E}[Y\mid A=1]-\mathbb{E}[Y\mid A=0]\\
  \frac{g_\theta(0)}{\mathbb{V}[\theta^{\top} W]}\\
  \frac{g_\theta(1)}{\mathbb{V}[\theta^{\top} W]}-\frac{g_\theta(0)}{\mathbb{V}[\theta^{\top} W]}
\end{pmatrix}.
\end{split}
\end{equation*}
That is, 
\begin{equation*}
\begin{split}
     \left.\frac{\partial}{\partial\theta^\top}e_A^\top\beta^*(\theta)\middle |_{\theta=\theta^*}\right.&=\left.\frac{\partial}{\partial\theta^\top}\beta^*_A(\theta)\middle |_{\theta=\theta^*}\right.=0, \quad \text{with}\quad e_A=(0,1,0,0)^\top,\\
      \left.\frac{\partial}{\partial\theta^\top}e_0^\top\beta^*(\theta)\middle |_{\theta=\theta^*}\right.&=\left.\frac{\partial}{\partial\theta^\top}\beta^*_0(\theta)\middle |_{\theta=\theta^*}\right.=0, \quad \text{with}\quad e_0=(1,0,0,0)^\top,\\
      \left.\frac{\partial}{\partial\theta^\top}e_{0A}^\top\beta^*(\theta)\middle |_{\theta=\theta^*}\right.&=\left.\frac{\partial}{\partial\theta^\top}\beta^*_0(\theta)\middle |_{\theta=\theta^*}\right.+\left.\frac{\partial}{\partial\theta^\top}\beta^*_A(\theta)\middle |_{\theta=\theta^*}\right.=0, \quad \text{with}\quad e_{0A}=(1,1,0,0)^\top.
\end{split}
\end{equation*}
\end{proof}

\subsection{PROCOVA model \eqref{eq: PROCOVA linear} with centering}
\label{appendix: PROCOVAI centering}

\begin{theorem}
\label{thm: asyvar PROCOVAII centering}
Consider the two-sample two-stage estimation with the first-stage estimating function \eqref{eq: est func prog} and the second-stage estimating function
\begin{equation}
    \label{eq: est func PROCOVAII linear centering}
        \psi(O;\beta,\theta)=(Y-\beta^\top X_{\theta})X_{\theta},\quad X_\theta=(1,A,\theta^\top W-\mathbb{E}[\theta^\top W])^\top.
\end{equation}
Assume regularity conditions (A1)--(A4), $\mathbb{E}[Y^2]<\infty$ and $\mathbb{E}[||W||^2]<\infty$.
Additionally, assume $A\perp W$ and $0<\pi< 1$.
Then, for $e=(0,1,0)^\top$, 
\begin{equation*}
    e^\top\beta^*(\theta)=\mathbb{E}[Y\mid A=1]-\mathbb{E}[Y\mid A=0] \quad\text{and}\quad\left.\frac{\partial}{\partial\theta^\top}e^\top\beta^*(\theta)\middle |_{\theta=\theta^*}\right.=-e^\top Q_0^{-1}Q_1 =0.
\end{equation*}
Additionally, for $e=(1,0,0)^\top$,
\begin{equation*}
    e^\top\beta^*(\theta)=\mathbb{E}[Y\mid A=0] \quad\text{and}\quad\left.\frac{\partial}{\partial\theta^\top}e^\top\beta^*(\theta)\middle |_{\theta=\theta^*}\right.=-e^\top Q_0^{-1}Q_1 =0.
\end{equation*}
For $e=(1,1,0)^\top$,
\begin{equation*}
    e^\top\beta^*(\theta)=\mathbb{E}[Y\mid A=1]\quad\text{and}\quad\left.\frac{\partial}{\partial\theta^\top}e^\top\beta^*(\theta)\middle |_{\theta=\theta^*}\right.=-e^\top Q_0^{-1}Q_1 =0.
\end{equation*}
\end{theorem}

\begin{proof}

First, we prove $e^\top Q_0^{-1}Q_1=0$ for $e=(0,1,0)^\top$, $e=(1,0,0)^\top$ and $e=(1,1,0)^\top$.
The second-stage estimating function \eqref{eq: est func PROCOVAII linear centering} can be expressed as follows:
\begin{equation*}
\begin{split}
    \psi(O;\beta,\theta)&=(Y-\beta^\top X_{\theta})X_\theta=\begin{pmatrix}
        Y-\beta^\top X_{\theta}\\
       (Y-\beta^\top X_{\theta})A\\
        (Y-\beta^\top X_{\theta})(\theta^\top W-\mathbb{E}[\theta^\top W])
    \end{pmatrix}\\
    &=\begin{pmatrix}
        Y-\beta_0-\beta_AA-\beta_1(\theta^\top W-\mathbb{E}[\theta^\top W])\\
        \{Y-\beta_0-\beta_AA-\beta_1(\theta^\top W-\mathbb{E}[\theta^\top W])\}A\\
        \{Y-\beta_0-\beta_AA-\beta_1(\theta^\top W-\mathbb{E}[\theta^\top W])\}(\theta^\top W-\mathbb{E}[\theta^\top W])
    \end{pmatrix}.
\end{split}
\end{equation*}
Thus, $Q_0$ can be expressed as follows:
\begin{equation*}
\begin{split}
        Q_0&= \mathbb{E}\!\left[\frac{\partial}{\partial\beta^\top}\psi(O;\beta,\theta^*)\Big|_{\beta=\beta^*}\right]=\mathbb{E}\!\left[\frac{\partial}{\partial\beta^\top}(Y-\beta^\top X_{\theta^*})X_{\theta^*}\Big|_{\beta=\beta^*}\right]=-\mathbb{E}[X_{\theta^*}X_{\theta^*}^\top]\\
    &=-\mathbb{E}\left[\begin{pmatrix}
   1 & A & \theta^{*\top} W-\mathbb{E}[\theta^{*\top} W] \\
   A & A & (\theta^{*\top} W-\mathbb{E}[\theta^{*\top} W])A \\
   \theta^{*\top}W-\mathbb{E}[\theta^{*\top} W] & (\theta^{*\top} W-\mathbb{E}[\theta^{*\top} W])A & (\theta^{*\top} W-\mathbb{E}[\theta^{*\top} W])^2
\end{pmatrix}\right]\\
&=-\begin{pmatrix}
   1 & \pi & 0 \\
   \pi & \pi & 0 \\
   0 & 0 & \mathbb{V}[\theta^{*\top} W]
\end{pmatrix}.
\end{split}
\end{equation*}
The last equation holds because of $A\perp W$.
By denoting $(i,j)$ component of $Q_0^{-1}$ as $\check q_{i,j}$, following holds:
\begin{equation*}
\begin{split}
    \check q_{1,1}&=(-1)^{1+1}\begin{vmatrix}
   \pi  & 0 \\
   0  & \mathbb{V}[\theta^{*\top} W]
\end{vmatrix}/|Q_0|=\frac{\pi\mathbb{V}[\theta^{*\top} W]}{|Q_0|},\\
\check  q_{1,2}&=\check q_{2,1}=(-1)^{1+2}\begin{vmatrix}
   \pi  & 0 \\
   0 & \mathbb{V}[\theta^{*\top} W]
\end{vmatrix}/|Q_0|=-\frac{\pi\mathbb{V}[\theta^{*\top} W]}{|Q_0|},\\
 \check q_{1,3}&=\check q_{3,1}=(-1)^{1+3}\begin{vmatrix}
   \pi  & \pi \\ 
   0  & 0
\end{vmatrix}/|Q_0|=0,\\
 \check q_{2,2}&=(-1)^{2+2}\begin{vmatrix}
   1  & 0 \\
   0 & \mathbb{V}[\theta^{*\top} W]
\end{vmatrix}/|Q_0|=\frac{\mathbb{V}[\theta^{*\top} W]}{|Q_0|},\\
 \check q_{2,3}&=\check q_{3,2}=(-1)^{2+3}\begin{vmatrix}
   1  &  \pi\\
   0 & 0
\end{vmatrix}/|Q_0|=0,\\
\check q_{3,3}&=(-1)^{3+3}\begin{vmatrix}
   1  & \pi\\
   \pi & \pi
\end{vmatrix}/|Q_0|=\frac{\pi(1-\pi)}{|Q_0|}.
\end{split}
\end{equation*}
Thus, $\check q_{1,1}=\pi\check q_{2,2}$ and $\check q_{1,2}=\check q_{2,1}=-\pi\check q_{2,2}$.
Then, $Q_0^{-1}$ can be expressed as follows:
\begin{equation*}
    Q_0^{-1}=\begin{pmatrix}
  \pi \check q_{2,2} & -\pi \check q_{2,2} &0 \\
    -\pi \check q_{2,2} & \check q_{2,2} & 0 \\
  0 & 0 & \check q_{3,3}
\end{pmatrix}.
\end{equation*}
Additionally, $Q_1$ can be expressed as follows:
\begin{equation*}
\begin{split}
      Q_1&=\mathbb{E}\!\left[\frac{\partial}{\partial\theta^\top}\psi(O;\beta^*,\theta)\Big|_{\theta=\theta^*}\right]\\
      &=\mathbb{E}\!\left[\frac{\partial}{\partial\theta^\top}\begin{pmatrix}
        Y-\beta_0^*-\beta_A^*A-\beta_1^*(\theta^\top W-\mathbb{E}[\theta^\top W])\\
        \{Y-\beta_0^*-\beta_A^*A-\beta_1^*(\theta^\top W-\mathbb{E}[\theta^\top W])\}A\\
        \{Y-\beta_0^*-\beta_A^*A-\beta_1^*(\theta^\top W-\mathbb{E}[\theta^\top W])\}(\theta^\top W-\mathbb{E}[\theta^\top W])
    \end{pmatrix}\middle |_{\theta=\theta^*}\right]=\begin{pmatrix}
        q_1\\
        q_2\\
        q_3
    \end{pmatrix},
\end{split}
\end{equation*}
where 
\begin{equation*}
\begin{split}
    \beta_0^*&=e_0^\top \beta^*,\quad \text{with}\quad e_0=(1,0,0)^\top,\\
    \beta_A^*&=e_A^\top \beta^*,\quad \text{with}\quad e_A=(0,1,0)^\top,\\
    \beta_1^*&=e_1^\top \beta^*,\quad \text{with}\quad e_1=(0,0,1)^\top,
\end{split}
\end{equation*}
and $q_1=q_2=0$ and $q_3=\mathbb{COV}[Y,W^\top ]-2\beta_1^*\theta^{*\top}\mathbb{V}[W]$.
Since 
\begin{equation*}
    |Q_0|=-\pi(1-\pi)\mathbb{V}[\theta^{*\top}W]
\end{equation*}
and
\begin{equation*}
    \beta^*=\mathbb{E}[X_{\theta^*}X_{\theta^*}^\top]^{-1}\mathbb{E}[X_{\theta^*}Y]=-\begin{pmatrix}
  \pi \check q_{2,2} & -\pi \check q_{2,2} &0 \\
    -\pi \check q_{2,2} & \check q_{2,2} & 0 \\
  0 & 0 & \check q_{3,3}
\end{pmatrix}\begin{pmatrix}
        \mathbb{E}[Y]\\
        \mathbb{E}[AY]\\
        \mathbb{E}[(\theta^{*\top} W -\mathbb{E}[\theta^{*\top} W ])Y]
    \end{pmatrix},
\end{equation*}
the components of $\beta^*=(\beta^*_0,\beta^*_A,\beta_1^*)^\top$ can be expressed as follows:
\begin{equation*}
\begin{split}
    \beta_A^*&=\check q_{2,2}(\pi\mathbb{E}[Y]-\mathbb{E}[AY])=-\frac{\mathbb{V}[\theta^{*\top}W]}{\pi(1-\pi)\mathbb{V}[\theta^{*\top}W]}(-1)\pi(1-\pi)\{\mathbb{E}[Y\mid A=1]-\mathbb{E}[Y\mid A=0]\}\\&=\mathbb{E}[Y\mid A=1]-\mathbb{E}[Y\mid A=0],\\
    \beta_1^*&=\check q_{3,3}(\mathbb{E}[\theta^{*\top}W]\mathbb{E}[Y]-\mathbb{E}[\theta^{*\top}WY])=-\frac{\pi(1-\pi)}{\pi(1-\pi)\mathbb{V}[\theta^{*\top}W]}(-1)\mathbb{COV}[\theta^{*\top}W, Y]\\
    &=\frac{\mathbb{COV}[\theta^{*\top}W, Y]}{\mathbb{V}[\theta^{*\top}W]},\\
    \beta_0^*&=-\pi\check q_{2,2}(\mathbb{E}[Y]-\mathbb{E}[AY])=\pi\frac{\mathbb{V}[\theta^{*\top}W]}{\pi(1-\pi)\mathbb{V}[\theta^{*\top}W]}(1-\pi)\mathbb{E}[Y\mid A=0]\\
    &=\mathbb{E}[Y\mid A=0].
\end{split}
\end{equation*}
Then, $q_3$ can be expressed as follows:
\begin{equation*}
\begin{split}
    q_3 &=\mathbb{COV}[Y,W^\top]-2\frac{\mathbb{COV}[\theta^{*\top}W,Y]}{\mathbb{V}[\theta^{*\top}W]}\theta^{*\top} \mathbb{V}[W].
\end{split}   
\end{equation*}
Then,  
\begin{equation*}
\begin{split}
    e_{A}^\top Q_0^{-1}Q_1&=-\pi\check q_{2,2}\cdot0+\check q_{2,2}\cdot0 +0\cdot q_3=0, \quad \text{with} \quad e_A=(0,1,0)^\top,\\
     e_0^\top Q_0^{-1}Q_1&=\pi \check q_{2,2}\cdot 0 -\pi \check q_{2,2}\cdot 0 + 0\cdot q_3 =0, \quad \text{with} \quad e_0=(1,0,0)^\top,\\
     e_{0A}^\top Q_0^{-1}Q_1&=0\cdot 0 +(1-\pi) \check q_{2,2}\cdot 0 + 0\cdot q_3 =0, \quad \text{with} \quad e_{0A}=(1,1,0)^\top,\\
      e_{1}^\top Q_0^{-1}Q_1&= 0\cdot 0+ 0\cdot 0+\check q_{3,3}q_3=-\frac{\mathbb{COV}[Y,W^\top ]}{\mathbb{V}[\theta^{*\top}W]}+2\frac{\mathbb{COV}[\theta^{*\top}W,Y]}{\mathbb{V}[\theta^{*\top}W]^2}\theta^{*\top} \mathbb{V}[W], \quad \text{with} \quad e_1=(0,0,1)^\top.
\end{split}
\end{equation*}

Next, we prove $\left.\frac{\partial}{\partial\theta^\top}e^\top\beta^*(\theta)\middle |_{\theta=\theta^*}\right.=0$ for $e=(0,1,0)^\top$, $e=(1,0,0)^\top$ and $e=(1,1,0)^\top$.
As with the derivation of the components of $\beta^*$, the components of $\beta^*(\theta)=(\beta^*_0(\theta),\beta^*_A(\theta),\beta_1^*(\theta))^\top$ can be expressed as follows:
\begin{equation*}
\begin{split}
    \beta_A^*(\theta)&=-\frac{\mathbb{V}[\theta^{\top}W]}{\pi(1-\pi)\mathbb{V}[\theta^{\top}W]}(-1)\pi(1-\pi)\{\mathbb{E}[Y\mid A=1]-\mathbb{E}[Y\mid A=0]\}=\mathbb{E}[Y\mid A=1]-\mathbb{E}[Y\mid A=0],\\
    \beta_1^*(\theta)&=-\frac{\pi(1-\pi)}{\pi(1-\pi)\mathbb{V}[\theta^{\top}W]}(-1)\mathbb{COV}[\theta^{\top}W, Y]=\frac{\mathbb{COV}[\theta^{\top}W, Y]}{\mathbb{V}[\theta^{\top}W]},\\
    \beta_0^*(\theta)&=\pi\frac{\mathbb{V}[\theta^{\top}W]}{\pi(1-\pi)\mathbb{V}[\theta^{\top}W]}(1-\pi)\mathbb{E}[Y\mid A=0]=\mathbb{E}[Y\mid A=0].
\end{split}
\end{equation*}
That is, 
\begin{equation*}
\begin{split}
      \left.\frac{\partial}{\partial\theta^\top}e_A^\top\beta^*(\theta)\middle |_{\theta=\theta^*}\right.&=\left.\frac{\partial}{\partial\theta^\top}\beta^*_A(\theta)\middle |_{\theta=\theta^*}\right.=0, \quad \text{with}\quad e_A=(0,1,0)^\top.\\
       \left.\frac{\partial}{\partial\theta^\top}e_0^\top\beta^*(\theta)\middle |_{\theta=\theta^*}\right.&=\left.\frac{\partial}{\partial\theta^\top}\beta^*_0(\theta)\middle |_{\theta=\theta^*}\right.=0, \quad \text{with}\quad e_0=(1,0,0)^\top,\\
        \left.\frac{\partial}{\partial\theta^\top}e_{0A}^\top\beta^*(\theta)\middle |_{\theta=\theta^*}\right.&=\left.\frac{\partial}{\partial\theta^\top}\beta^*_0(\theta)\middle |_{\theta=\theta^*}\right.+\left.\frac{\partial}{\partial\theta^\top}\beta^*_A(\theta)\middle |_{\theta=\theta^*}\right.=0, \quad \text{with}\quad e_{0A}=(1,1,0)^\top.
\end{split}
\end{equation*}
\end{proof}

\newpage
\subsection{Preparation for proof of Theorem \ref{thm: asyvar PROCOVA general}}
\label{appendix: for thm3}
\begin{lemma}
\label{lemma: for thm3}
For any scalar random variables \(Z\) and \(V\), denote the sample
covariance by
\[
\hat c(Z,V)
=
\frac{1}{n}
\sum_{i=1}^n
(Z_i-\overline Z)(V_i-\overline V),
\]
where 
\[\overline Z=\frac{1}{n}\sum_{i=1}^n Z_i\quad\text{ and }\quad\overline V=\frac{1}{n}\sum_{i=1}^n V_i.\]
The population OLS coefficient given $\mathcal{\tilde D}$, $\beta^*(\hat \rho_{\tilde n})=(\beta^*_0(\hat \rho_{\tilde n}), \beta_A^*(\hat \rho_{\tilde n}), \beta^*_1(\hat \rho_{\tilde n}))^\top$, is defined by
\[
\beta^*(\hat \rho_{\tilde n})
=
\operatorname*{argmin}_{\beta\in\mathbb R^3}
\mathbb{E}\left[
\left(Y-\beta^\top X_{\hat \rho_{\tilde n}}\right)^2
\mid\tilde{\mathcal D}
\right],
\]
where $X_{\hat \rho_{\tilde n}}=(1,A,\hat \rho_{\tilde n}(W))^\top$.
Furthermore, define $\varepsilon_{\hat \rho_{\tilde n}}=Y-\beta^*(\hat \rho_{\tilde n})^\top X_{\hat \rho_{\tilde n}}$.
Then, under the same assumptions as in Theorem \ref{thm: asyvar PROCOVA general}, the following (i)--(v) hold:
\[
\begin{aligned}
\text{(i)}\quad&
\hat{c}(A,A)
=
\pi(1-\pi)+o_p(1)
\\
\text{(ii)}\quad&
\hat{c}
\bigl(
\hat{\rho}_{\tilde n}(W),
\hat{\rho}_{\tilde n}(W)
\bigr)
=
\mathbb{V}[
\hat{\rho}_{\tilde n}(W)
\mid
\mathcal{\tilde D}
]
+o_p(1)
=
\mathbb{V}[\rho^*(W)]+o_p(1), \quad\mathbb{V}[
\hat{\rho}_{\tilde n}(W)
\mid
\mathcal{\tilde D}
]=O_p(1)
\\
\text{(iii)}\quad&
\hat{c}
\bigl(
A,\hat{\rho}_{\tilde n}(W)
\bigr)
=
O_p(n^{-1/2})
\\
\text{(iv)}\quad&
\sqrt{n}\,
\hat{c}
\bigl(
A,\varepsilon_{\hat{\rho}_{\tilde n}}
\bigr)
=
\frac{1}{\sqrt{n}}
\sum_{i=1}^n
(A_i-\pi)
\varepsilon_{\hat{\rho}_{\tilde n},i}
+o_p(1),\quad
\frac{1}{\sqrt{n}}
\sum_{i=1}^n
(A_i-\pi)
\varepsilon_{\hat{\rho}_{\tilde n},i}
=
O_p(1)
\\
\text{(v)}\quad&
\hat{c}
\bigl(
\hat{\rho}_{\tilde n}(W),
\varepsilon_{\hat{\rho}_{\tilde n}}
\bigr)
=
o_p(1).
\end{aligned}
\] 
\end{lemma}

\begin{proof}[Proof of (i)]
Since \(A_i\in\{0,1\}\),
\[
\hat{c}(A,A)
=
\frac{1}{n}\sum_{i=1}^n(A_i-\overline{A})^2
=
\overline{A}(1-\overline{A}),
\]
where 
\[
\overline A
=
\frac{1}{n}\sum_{i=1}^n A_i.
\]
Since \(\overline{A}=\pi+o_p(1)\), it follows that
\begin{equation*}
\hat{c}(A,A)
=
\pi(1-\pi)+o_p(1).
\end{equation*}
\end{proof}

\begin{proof}[Proof of (ii)]
Define
\[
e_{\tilde n}(w)
=
\hat{\rho}_{\tilde n}(w)-\rho^*(w)
\]
and
\[
R_{\tilde n}
=
\mathbb{E}[
e_{\tilde n}(W)^2
\mid
\mathcal{\tilde D}
].
\]
By assumption $\mathbb{E}[\{\hat\rho_{\tilde n}(W)-\rho^*(W) \}^2\mid \mathcal{\tilde D}]=o_p(1)$,
\[
R_{\tilde n}=o_p(1).
\]
Let
\[
T_{n,\tilde n}
=
\frac{1}{n}
\sum_{i=1}^n
e_{\tilde n}(W_i)^2.
\]
Then,
\[
\mathbb{E}[
T_{n,\tilde n}
\mid
\mathcal{\tilde D}
]
=
\frac{1}{n}
\sum_{i=1}^n
\mathbb{E}[
e_{\tilde n}(W_i)^2
\mid
\mathcal{\tilde D}
]
=
R_{\tilde n}.
\]
For any constants \(\eta>0\) and \(\delta>0\),
\begin{align*}
\mathbb{P}(
T_{n,\tilde n}>\eta
)
&=
\mathbb{E}\left[
\mathbb{P}(
T_{n,\tilde n}>\eta
\mid
\mathcal{\tilde D}
)
\right]
\\
&\leq
\mathbb{P}(
R_{\tilde n}>\delta
)
+
\mathbb{E}\left[
\mathbb{I}(R_{\tilde n}\leq\delta)
\mathbb{P}(
T_{n,\tilde n}>\eta
\mid
\mathcal{\tilde D}
)
\right]
\\
&\leq
\mathbb{P}(
R_{\tilde n}>\delta
)
+
\mathbb{E}\left[
\mathbb{I}(R_{\tilde n}\leq\delta)
\frac{
\mathbb{E}[T_{n,\tilde n}\mathcal{\mid\tilde D}]
}{\eta}
\right]
\\
&=
\mathbb{P}(
R_{\tilde n}>\delta
)
+
\mathbb{E}\left[
\mathbb{I}(R_{\tilde n}\leq\delta)
\frac{R_{\tilde n}}{\eta}
\right]
\\
&\leq
\mathbb{P}(
R_{\tilde n}>\delta
)
+
\frac{\delta}{\eta}.
\end{align*}
where the second inequality follows from Markov's inequality.
Since \(R_{\tilde n}=o_p(1)\), for any $\delta>0$,
\[
\lim_{n,\tilde n\rightarrow \infty}\sup \mathbb{P}(
T_{n,\tilde n}>\eta
)\le\frac{\delta}{\eta}.
\]
Since $\delta>0$ is arbitrary, 
\[
\lim_{n,\tilde n\rightarrow \infty}\sup \mathbb{P}(
T_{n,\tilde n}>\eta
)\le \inf_{\delta>0}\frac{\delta}{\eta}=0.
\]
Thus, 
\begin{equation*}
T_{n,\tilde n}
=
o_p(1).
\end{equation*}
Let
\[
\overline e_{\tilde n}
=
\frac{1}{n}
\sum_{i=1}^n
e_{\tilde n}(W_i).
\]
Since 
\[
\left|
\overline e_{\tilde n}
\right|^2
\leq
T_{n,\tilde n}
=
o_p(1),
\]
the following holds:
\[
\overline e_{\tilde n}=o_p(1).
\]

We now consider the sample variance of
\(\hat{\rho}_{\tilde n}(W)\). 
\begin{align}
&
\hat{c}
\bigl(
\hat{\rho}_{\tilde n}(W),
\hat{\rho}_{\tilde n}(W)
\bigr)
\nonumber\\
&=
\hat{c}
\bigl(
\rho^*(W)+e_{\tilde n}(W),
\rho^*(W)+e_{\tilde n}(W)
\bigr)
\nonumber\\
&=
\hat{c}
\bigl(
\rho^*(W),\rho^*(W)
\bigr)
+
2\hat{c}
\bigl(
\rho^*(W),e_{\tilde n}(W)
\bigr)
+
\hat{c}
\bigl(
e_{\tilde n}(W),e_{\tilde n}(W)
\bigr).
\label{eq:thm3-sample-var-decomposition}
\end{align}
Since \(\mathbb{E}[\{\rho^*(W)\}^2]<\infty\), the law of large numbers gives
\[
\frac{1}{n}
\sum_{i=1}^n
\{\rho^*(W_i)\}^2
=
\mathbb{E}[\{\rho^*(W)\}^2]+o_p(1)
\]
and
\[
\frac{1}{n}
\sum_{i=1}^n
\rho^*(W_i)
=
\mathbb{E}[\rho^*(W)]+o_p(1).
\]
Therefore,
\begin{align}
\hat{c}
\bigl(
\rho^*(W),\rho^*(W)
\bigr)
&=
\frac{1}{n}
\sum_{i=1}^n
\{\rho^*(W_i)\}^2
-
\left\{
\frac{1}{n}
\sum_{i=1}^n
\rho^*(W_i)
\right\}^2
\nonumber\\
&=
\mathbb{E}[\{\rho^*(W)\}^2]
-
\mathbb{E}[\rho^*(W)]^2
+
o_p(1)
\nonumber\\
&=
\mathbb{V}[\rho^*(W)]+o_p(1).
\label{eq:thm3-rhostar-sample-var}
\end{align}
Furthermore,
\begin{equation}
\hat{c}
\bigl(
e_{\tilde n}(W),
e_{\tilde n}(W)
\bigr)\leq
T_{n,\tilde n}=
o_p(1).
\label{eq:thm3-e-sample-var}
\end{equation}
By Cauchy--Schwarz inequality,
\begin{align}
&
\left|
\hat{c}
\bigl(
\rho^*(W),e_{\tilde n}(W)
\bigr)
\right|
\nonumber\\
&\leq
\left\{
\hat{c}
\bigl(
\rho^*(W),\rho^*(W)
\bigr)
\right\}^{1/2}
\left\{
\hat{c}
\bigl(
e_{\tilde n}(W),e_{\tilde n}(W)
\bigr)
\right\}^{1/2}
\nonumber\\
&=
O_p(1)o_p(1)
=
o_p(1).
\label{eq:thm3-rho-e-sample-var}
\end{align}
Substituting \eqref{eq:thm3-rhostar-sample-var}--\eqref{eq:thm3-rho-e-sample-var} into
\eqref{eq:thm3-sample-var-decomposition} yields
\begin{equation}
\hat{c}
\bigl(
\hat{\rho}_{\tilde n}(W),
\hat{\rho}_{\tilde n}(W)
\bigr)
=
\mathbb{V}[\rho^*(W)]+o_p(1).
\label{eq:thm3-rhohat-sample-var}
\end{equation}

Next, we examine the conditional population variance of ${\hat \rho}_{\tilde n}(W)$.
Since \(\mathcal{D}\perp\mathcal{\tilde D}\) and \(\rho^*(W)\) does not depend on
\(\mathcal{\tilde D}\),
\[
\mathbb{E}[
\rho^*(W)
\mid
\mathcal{\tilde D}
]
=
\mathbb{E}[\rho^*(W)]
\]
and
\[
\mathbb{V}[
\rho^*(W)
\mid
\tilde D
]
=
\mathbb{V}[\rho^*(W)].
\]
By Cauchy--Schwarz inequality,
\[
\left|
\mathbb{E}[
e_{\tilde n}(W)
\mid
\tilde D
]
\right|^2
\leq
R_{\tilde n}
=
o_p(1).
\]
Furthermore, 
\begin{align*}
\mathbb{V}[
e_{\tilde n}(W)
\mid
\tilde D
]
&\leq
R_{\widetilde n}
=
o_p(1).
\end{align*}
By Cauchy--Schwarz inequality,
\begin{align*}
&
\left|
\mathbb{COV}
\bigl(
\rho^*(W),e_{\tilde n}(W)
\mid
\mathcal{\tilde D}
\bigr)
\right|
\\
&\leq
\left\{
\mathbb{V}[
\rho^*(W)
\mid
\mathcal{\tilde D}
]
\right\}^{1/2}
\left\{
\mathbb{V}[
e_{\tilde n}(W)
\mid
\mathcal{\tilde D}
]
\right\}^{1/2}
\\
&=
\left\{
\mathbb{V}[\rho^*(W)]
\right\}^{1/2}
o_p(1)
\\
&=
o_p(1).
\end{align*}
Consequently,
\begin{align}
\mathbb{V}[
\hat{\rho}_{\tilde n}(W)
\mid
\mathcal{\tilde D}
]
&=
\mathbb{V}[
\rho^*(W)+e_{\tilde n}(W)
\mid
\mathcal{\tilde D}
]
\nonumber\\
&=
\mathbb{V}[\rho^*(W)]
+
2\mathbb{COV}
\bigl(
\rho^*(W),e_{\tilde n}(W)
\mid
\mathcal{\tilde D}
\bigr)+
\mathbb{V}[
e_{\tilde n}(W)
\mid
\mathcal{\tilde D}
]
\nonumber\\
&=
\mathbb{V}[\rho^*(W)]+o_p(1).
\label{eq:thm3-rhohat-conditional-var}
\end{align}

Combining
\eqref{eq:thm3-rhohat-sample-var} and
\eqref{eq:thm3-rhohat-conditional-var}, we obtain
\begin{equation*}
\hat{c}
\bigl(
\hat{\rho}_{\tilde n}(W),
\hat{\rho}_{\tilde n}(W)
\bigr)
=
\mathbb{V}[
\hat{\rho}_{\tilde n}(W)
\mid
\mathcal{\tilde D}
]
+o_p(1)
=
V[\rho^*(W)]+o_p(1).
\end{equation*}
Thus, 
\begin{equation*}
    \mathbb{V}[
\hat{\rho}_{\tilde n}(W)
\mid
\mathcal{\tilde D}
]=O_p(1).
\end{equation*}
\end{proof}

\begin{proof}[Proof of (iii)]
Define
\[
m_{\tilde n}
=
\mathbb{E}[\hat{\rho}_{\tilde n}(W)\mid\mathcal{\tilde D}],
\]
and
\[
\overline{\hat{\rho}}_{\tilde n}
=
\frac{1}{n}\sum_{i=1}^n
\hat{\rho}_{\tilde n}(W_i).
\]
We have
\begin{align}
\sqrt{n}\,
\hat{c}
\bigl(
A,\hat{\rho}_{\tilde n}(W)
\bigr)
={}&
\frac{1}{\sqrt{n}}
\sum_{i=1}^n
(A_i-\pi)
\{
\hat{\rho}_{\tilde n}(W_i)-m_{\tilde n}
\}
-
\sqrt{n}
(\overline{A}-\pi)
(
\overline{\widehat{\rho}}_{\widetilde n}-m_{\widetilde n}
).
\label{eq:thm3-cArho-decomposition}
\end{align}
Since $A\perp W$ and ${\mathcal{D}}\perp\tilde{\mathcal{D}}$, conditional on \(\tilde{\mathcal D}\), the first term on the right-hand side in equation \eqref{eq:thm3-cArho-decomposition} has mean
\begin{equation*}
 \mathbb{E}[(A-\pi)(\hat \rho_{\tilde n}(W)-m_{\tilde n})
  \mid\tilde{\mathcal D}]=0 
\end{equation*}
and variance
\begin{equation*}
\begin{split}
    \mathbb{V}[(A-\pi)(\hat \rho_{\tilde n}(W)-m_{\tilde n})
  \mid\tilde{\mathcal D}]
  &=\mathbb{E}[(A-\pi)^2(\hat \rho_{\tilde n}(W)-m_{\tilde n})^2
  \mid\tilde{\mathcal D}]\\
  &=\mathbb{E}[(A-\pi)^2]\mathbb{E}[(\hat \rho_{\tilde n}(W)-m_{\tilde n})^2
  \mid\tilde{\mathcal D}]\\
  &=\pi(1-\pi)\mathbb{V}[\hat \rho_{\tilde n}(W)\mid \tilde{\mathcal{D}}]\\
  &=O_p(1),
\end{split}
\end{equation*}
where the last equality follows from (ii).
By Chebyshev's inequality,
\[
\frac{1}{\sqrt{n}}
\sum_{i=1}^n
(A_i-\pi)
\{
\hat{\rho}_{\tilde n}(W_i)-m_{\tilde n}
\}
=
O_p(1).
\]

Furthermore, since $\overline A-\pi=O_p(n^{-1/2})$ and $\overline {\hat \rho_{\tilde n}}-m_{\tilde n}=O_p(n^{-1/2})$, the second term on the right-hand side in equation \eqref{eq:thm3-cArho-decomposition} is 
\begin{equation*}
\sqrt n(\overline A-\pi)(\overline {\hat \rho_{\tilde n}}-m_{\tilde n})
=
o_p(1).
\end{equation*}
Thus,
\begin{equation*}
\sqrt n\hat  c({A,\hat \rho_{\tilde n}(W)})
=
O_p(1),
\end{equation*}
equivalently,
\begin{equation*}
\hat  c({A,\hat \rho_{\tilde n}(W)})
=
O_p(n^{-1/2}).
\end{equation*}
\end{proof}

\begin{proof}[Proof of (iv)]
Recall that the population OLS coefficient given $\mathcal{\tilde D}$, $\beta^*(\hat \rho_{\tilde n})=(\beta^*_0(\hat \rho_{\tilde n}), \beta_A^*(\hat \rho_{\tilde n}), \beta^*_1(\hat \rho_{\tilde n}))^\top$, is defined by
\[
\beta^*(\hat \rho_{\tilde n})
=
\operatorname*{argmin}_{\beta\in\mathbb R^3}
\mathbb{E}\left[
\left(Y-\beta^\top X_{\hat \rho_{\tilde n}}\right)^2
\mid\tilde{\mathcal D}
\right],
\]
where $X_{\hat \rho_{\tilde n}}=(1,A,\hat \rho_{\tilde n}(W))^\top$.
Differentiating the objective function with respect to \(\beta\)
and setting the derivative equal to zero yields
\[
\mathbb{E}[X_{\hat \rho_{\tilde n}}\varepsilon_{\hat \rho_{\tilde n}}
\mid\tilde{\mathcal D}
]=0,
\]
where $\varepsilon_{\hat \rho_{\tilde n}}=Y-\beta^*(\hat \rho_{\tilde n})^\top X_{\hat \rho_{\tilde n}}$.
That is, 
\begin{equation}
\label{eq:OLS residual}
\mathbb{E}[\varepsilon_{\hat \rho_{\tilde n}}
\mid\tilde{\mathcal D}]=0,
\qquad
\mathbb{E}[A\varepsilon_{\hat \rho_{\tilde n}}
\mid\tilde{\mathcal D}]=0,
\qquad
E[\hat \rho_{\tilde n}(W)\varepsilon_{\hat \rho_{\tilde n}}\mid\tilde{\mathcal D}]=0.
\end{equation}

Since
\(\sum_{i=1}^n(A_i-\overline A)=0\),
\begin{equation*}
\sqrt n\hat c({A,\varepsilon_{\hat \rho_{\tilde n}}})
=
\frac{1}{\sqrt n}
\sum_{i=1}^n
(A_i-\overline A)\varepsilon_{\hat \rho_{\tilde n},i}
=
\frac{1}{\sqrt n}
\sum_{i=1}^n
(A_i-\pi)\varepsilon_{\hat \rho_{\tilde n},i}
-
\sqrt n(\overline A-\pi)\overline{\varepsilon_{\hat \rho_{\tilde n}}},
\end{equation*}
where
\[
\overline{\varepsilon_{\hat \rho_{\tilde n}}}
=
\frac{1}{n}\sum_{i=1}^n\varepsilon_{\hat \rho_{\tilde n},i}.
\]
By equation \eqref{eq:OLS residual}, $\mathbb{E}[\varepsilon_{\hat \rho_{\tilde n}}\mid\tilde{\mathcal D}]=0$.
Furthermore, $\mathbb{E}[\varepsilon_{\hat \rho_{\tilde n}}^2\mid\tilde{\mathcal D}]=O_p(1)$ since
\[
\mathbb{E}[\varepsilon_{\hat \rho_{\tilde n}}^2\mid\tilde{\mathcal D}]
=
\min_{\beta\in\mathbb R^3}
\mathbb{E}[
\left(Y-\beta^\top X_{\hat \rho_{\tilde n}}\right)^2
\mid\tilde{\mathcal D}
]
\leq
\mathbb{E}[Y^2\mid\tilde{\mathcal D}]
=
\mathbb{E}[Y^2]
<\infty.
\]
By Chebyshev's inequality, $\overline{\varepsilon_{\hat \rho_{\tilde n}}}=O_p(n^{-1/2})$.
Since \(\overline A-\pi=O_p(n^{-1/2})\),
\[
\sqrt n
(\overline A-\pi)\overline{\varepsilon_{\hat \rho_{\tilde n}}}
=
o_p(1).
\]
By equation \eqref{eq:OLS residual}, $\mathbb{E}[(A-\pi)\varepsilon_{\hat \rho_{\tilde n}}\mid\tilde{\mathcal D}]=0$.
Furthermore, 
\[
\mathbb{V}[(A-\pi)\varepsilon_{\hat \rho_{\tilde n}}
  \mid\tilde{\mathcal D}]=\mathbb{E}[(A-\pi)^2\varepsilon_{\hat \rho_{\tilde n}}^2
  \mid\tilde{\mathcal D}]
\leq
\mathbb{E}[\varepsilon_{\hat \rho_{\tilde n}}^2
  \mid\tilde{\mathcal D}]<\infty.
\]
Thus, 
\[
\frac{1}{\sqrt n}
\sum_{i=1}^n
(A_i-\pi)\varepsilon_{\hat \rho_{\tilde n},i}=O_p(1).
\]
Consequently,
\begin{equation*}
\sqrt n\hat c({A,\varepsilon_{\hat \rho_{\tilde n}}})
=
\frac{1}{\sqrt n}
\sum_{i=1}^n
(A_i-\pi)\varepsilon_{\hat \rho_{\tilde n},i}
+
o_p(1).
\end{equation*}
\end{proof}

\begin{proof}[Proof of (v)]
Let
\[
X_{\rho^*}
=
(1,A,\rho^*(W))^\top,
\qquad
\varepsilon_{\rho^*}
=
Y-\beta^*(\rho^*)^\top X_{\rho^*}.
\]
Define
\[
M_{\tilde n}
=
\mathbb{E}[
X_{\hat{\rho}_{\tilde n}}
X_{\hat{\rho}_{\tilde n}}^\top
\mid\tilde D
],
\qquad
r_{\tilde n}
=
\mathbb{E}[
X_{\hat{\rho}_{\tilde n}}Y
\mid\tilde D
],
\]
and
\[
M_*
=
E[X_{\rho^*}X_{\rho^*}^\top],
\qquad
r_*
=
E[X_{\rho^*}Y].
\]
By assumptions $\mathbb{E}[Y^2]<\infty$ and $\mathbb{E}[\{\hat\rho_{\tilde n}(W)-\rho^*(W) \}^2\mid \mathcal{\tilde D}]=o_p(1)$ and Cauchy--Schwarz inequality,
\[
M_{\widetilde n}
=
M_*+o_p(1),
\qquad
r_{\widetilde n}
=
r_*+o_p(1).
\]
Since \(A\perp W\),
\[
|M_*|
=
\pi(1-\pi)\mathbb{V}[\rho^*(W)]
>
0.
\]
Thus, 
\[
\beta^*(\hat{\rho}_{\tilde n})
=
M_{\tilde n}^{-1}r_{\tilde n}
=
\beta^*(\rho^*)+o_p(1).
\]
Consequently,
\[
\mathbb{E}[
\{
\varepsilon_{\hat{\rho}_{\tilde n}}
-
\varepsilon_{\rho^*}
\}^2
\mid\mathcal{\tilde D}
]
=
o_p(1).
\]
Then, as with the proof of (ii), Markov's inequality gives
\[
\frac{1}{n}\sum_{i=1}^n
\{
\varepsilon_{\hat{\rho}_{\tilde n},i}
-
\varepsilon_{\rho^*,i}
\}^2
=
o_p(1).
\]
Since \(\mathbb{E}[\varepsilon_{\hat{\rho}_{\tilde n}}^2\mid \mathcal{\tilde D}]\leq\mathbb{E}[Y^2]<\infty\),
\[
\frac{1}{n}\sum_{i=1}^n
\varepsilon_{\hat{\rho}_{\tilde n},i}^2
=
O_p(1).
\]

Let
\[
\overline{\rho^*}
=
\frac{1}{n}\sum_{i=1}^n\rho^*(W_i).
\]
Then,
\begin{align*}
&\hat{c}
\bigl(
\hat{\rho}_{\tilde n}(W),
\varepsilon_{\hat{\rho}_{\tilde n}}
\bigr)\\
&=
\hat{c}
\bigl(
\rho^*(W),\varepsilon_{\rho^*}
\bigr)
+
\frac{1}{n}\sum_{i=1}^n
\{
e_{\tilde n}(W_i)-\overline e_{\tilde n}
\}
\varepsilon_{\hat{\rho}_{\tilde n},i}
+
\frac{1}{n}\sum_{i=1}^n
\{
\rho^*(W_i)-\overline{\rho^*}
\}
\{
\varepsilon_{\hat{\rho}_{\tilde n},i}
-
\varepsilon_{\rho^*,i}
\}.
\end{align*}
The second and third terms on the right-hand side are \(o_p(1)\)
by the Cauchy--Schwarz inequality.
By the law of large numbers and the population
normal equations for \(\beta^*(\rho^*)\),
\[
\hat{c}
\bigl(
\rho^*(W),\varepsilon_{\rho^*}
\bigr)
\xrightarrow{p}
\mathbb{COV}
\bigl(
\rho^*(W),\varepsilon_{\rho^*}
\bigr)
=
0.
\]
Therefore,
\begin{equation*}
\hat{c}
\bigl(
\hat{\rho}_{\tilde n}(W),
\varepsilon_{\hat{\rho}_{\tilde n}}
\bigr)
=
o_p(1).
\end{equation*}    
\end{proof}

\subsection{Proof of Theorem \ref{thm: asyvar PROCOVA general}}
\label{appendix: thm 5}
\begin{proof}
The estimating function for PROCOVA can be expressed as follows:
\begin{equation*}
\begin{split}
    \psi(O;\beta,\rho)&=(Y-\beta^\top X_{\rho})X_\rho=\begin{pmatrix}
        Y-\beta^\top X_{\rho}\\
       (Y-\beta^\top X_{\rho})A\\
        (Y-\beta^\top X_{\rho})\rho(W)
    \end{pmatrix}=\begin{pmatrix}
        Y-\beta_0-\beta_AA-\beta_1\rho(W)\\
        \{Y-\beta_0-\beta_AA-\beta_1\rho(W)\}A\\
        \{Y-\beta_0-\beta_AA-\beta_1\rho(W)\}\rho(W)
    \end{pmatrix}.
\end{split}
\end{equation*}
The OLS estimator, $\hat\beta_n(\hat \rho_{\tilde n})=(\hat\beta_{0,n}(\hat \rho_{\tilde n}),\hat\beta_{A,n}(\hat \rho_{\tilde n}), \hat\beta_{1,n}(\hat \rho_{\tilde n}))^\top$, is defined by
\[
\hat\beta_n(\hat \rho_{\tilde n})
=
\operatorname*{argmin}_{\beta\in\mathbb R^3}
\frac{1}{n}
\sum_{i=1}^n
\left(Y_i-\beta^\top X_{\hat \rho_{\tilde n},i}\right)^2,
\]
where $X_{\hat \rho_{\tilde n},i}=(1, A_i, \hat \rho_{\tilde n}(W_i))^\top$.
Differentiating the objective function with respect to \(\beta\) and setting the derivative equal to zero yields
\begin{equation}
\label{eq:sample-normal-equations}
\frac{1}{n}
\sum_{i=1}^n
X_{\hat \rho_{\tilde n},i}
\left\{
Y_i-\hat\beta_n(\hat \rho_{\tilde n})^\top X_{\hat \rho_{\tilde n},i}
\right\}
=0.
\end{equation}
We denote $\hat\beta_n(\hat \rho_{\tilde n})-\beta^*(\hat \rho_{\tilde n})=(\delta_{0,n}(\hat \rho_{\tilde n}), \delta_{A,n}(\hat \rho_{\tilde n}), \delta_{1,n}(\hat \rho_{\tilde n}))^\top$.
Since $\varepsilon_{\hat \rho_{\tilde n},i}=Y_i-\beta^*(\hat \rho_{\tilde n})^\top X_{\hat \rho_{\tilde n},i}$, the OLS residual can be written as
\begin{equation*}
\begin{split}
    Y_i-\hat\beta_n(\hat \rho_{\tilde n})^\top X_{\hat \rho_{\tilde n},i}
&=\{Y_i-\beta^*(\hat \rho_{\tilde n})^\top X_{\hat \rho_{\tilde n},i}\}-\{\hat\beta_n(\hat \rho_{\tilde n})-\beta^*(\hat \rho_{\tilde n})\}^\top X_{\hat \rho_{\tilde n},i} \\
&=
\varepsilon_{\hat \rho_{\tilde n},i}
-\delta_{0,n}(\hat \rho_{\tilde n})
-\delta_{A,n}(\hat \rho_{\tilde n})A_i
-\delta_{1,n}(\hat \rho_{\tilde n})\hat \rho_{\tilde n}(W_i).
\end{split}    
\end{equation*}
Thus, equation \eqref{eq:sample-normal-equations} can be written as 
\[
\frac{1}{n}
\sum_{i=1}^n
X_{\hat \rho_{\tilde n},i}
\left\{
\varepsilon_{\hat \rho_{\tilde n},i}
-\delta_{0,n}(\hat \rho_{\tilde n})
-\delta_{A,n}(\hat \rho_{\tilde n})A_i
-\delta_{1,n}(\hat \rho_{\tilde n})\hat \rho_{\tilde n}(W_i)
\right\}
=0.
\]
Equivalently,
\begin{equation}
\label{eq:intercept-normal-equation0}
    \frac{1}{n}
\sum_{i=1}^n
\left\{
\varepsilon_{\hat \rho_{\tilde n},i}
-\delta_{0,n}(\hat \rho_{\tilde n})
-\delta_{A,n}(\hat \rho_{\tilde n})A_i
-\delta_{1,n}(\hat \rho_{\tilde n})\hat \rho_{\tilde n}(W_i)
\right\}
=0,    
\end{equation}
\begin{equation}
\label{eq:treatment-normal-equation0}
\frac{1}{n}
\sum_{i=1}^n
A_{i}
\left\{
\varepsilon_{\hat \rho_{\tilde n},i}
-\delta_{0,n}(\hat \rho_{\tilde n})
-\delta_{A,n}(\hat \rho_{\tilde n})A_i
-\delta_{1,n}(\hat \rho_{\tilde n})\hat \rho_{\tilde n}(W_i)
\right\}
=0,
\end{equation}
\begin{equation}
\label{eq:prog-normal-equation0}
\frac{1}{n}
\sum_{i=1}^n
\hat \rho_{\tilde n}(W_i)
\left\{
\varepsilon_{\hat \rho_{\tilde n},i}
-\delta_{0,n}(\hat \rho_{\tilde n})
-\delta_{A,n}(\hat \rho_{\tilde n})A_i
-\delta_{1,n}(\hat \rho_{\tilde n})\hat \rho_{\tilde n}(W_i)
\right\}
=0.
\end{equation}
By equation \eqref{eq:intercept-normal-equation0},
\begin{equation}
\label{eq:intercept-normal-equation}
\delta_{0,n}(\hat \rho_{\tilde n})
=
\overline{\varepsilon_{\hat \rho_{\tilde n}}}
-\delta_{A,n}(\hat \rho_{\tilde n})\overline A
-\delta_{1,n}(\hat \rho_{\tilde n})\overline {\hat \rho_{\tilde n}}.
\end{equation}
Substituting \eqref{eq:intercept-normal-equation} into equations \eqref{eq:treatment-normal-equation0} and \eqref{eq:prog-normal-equation0} gives
\[
\hat c(A,A)\delta_{A,n}(\hat \rho_{\tilde n})
+
\hat c({A,\hat \rho_{\tilde n}(W)})\delta_{1,n}(\hat \rho_{\tilde n})
=
\hat c({A,\varepsilon_{\hat \rho_{\tilde n}}}),
\]
\[
\hat c(A,\hat \rho_{\tilde n}(W))\delta_{A,n}(\hat \rho_{\tilde n})
+
\hat c({\hat \rho_{\tilde n}(W),\hat \rho_{\tilde n}(W)})\delta_{1,n}(\hat \rho_{\tilde n})
=
\hat c({\hat \rho_{\tilde n}(W),\varepsilon_{\hat \rho_{\tilde n}}}).
\]
If the sample design matrix is full column rank, 
\[
\hat c(A,A)\hat c({\hat \rho_{\tilde n}(W),\hat \rho_{\tilde n}(W)})
-\hat  c({A,\hat \rho_{\tilde n}(W)})^2
\neq 0,
\]
and then, solving the above equations for
\(\delta_{A,n}(\hat \rho_{\tilde n})\) yields 
\begin{equation}
\label{eq:exact-treatment-error}
\begin{split}
    &\sqrt n\{\hat\beta_{A,n}(\hat \rho_{\tilde n})-\beta_A^*(\hat \rho_{\tilde n})\}\\
&=
\frac{
\hat c({\hat \rho_{\tilde n}(W),\hat \rho_{\tilde n}(W)})\sqrt n\hat c({A,\varepsilon_{\hat \rho_{\tilde n}}})
-
\sqrt n\hat  c({A,\hat \rho_{\tilde n}(W)})\hat c({\hat \rho_{\tilde n}(W),\varepsilon_{\hat \rho_{\tilde n}}})
}{
\hat c(A,A)\hat c({\hat \rho_{\tilde n}(W),\hat \rho_{\tilde n}(W)})
-\hat  c({A,\hat \rho_{\tilde n}(W)})^2
}.
\end{split}
\end{equation}
By Lemma \ref{lemma: for thm3} (i)--(v), equation \eqref{eq:exact-treatment-error} reduces to
\begin{equation}
\label{eq: thm3 step1}
    \sqrt n\{\hat\beta_{A,n}(\hat \rho_{\tilde n})-\beta_A^*(\hat \rho_{\tilde n})\}=\frac{1}{\pi(1-\pi)}\frac{1}{\sqrt n}
\sum_{i=1}^n
(A_i-\pi)\varepsilon_{\hat \rho_{\tilde n},i}+o_p(1).
\end{equation}

Since 
\begin{equation}
\label{eq: beta star est}
\begin{split}
      \beta^*(\hat{\rho}_{\tilde n})&=\begin{pmatrix}
  \frac{1}{1-\pi} + \frac{\mathbb{E}[\hat{\rho}_{\tilde n}(W)\mid\mathcal{\tilde D}]^2}{\mathbb{V}[\hat{\rho}_{\tilde n}(W)\mid\mathcal{\tilde D}]} & -\frac{1}{1-\pi} & -\frac{\mathbb{E}[\hat{\rho}_{\tilde n}(W)\mid\mathcal{\tilde D}]}{\mathbb{V}[\hat{\rho}_{\tilde n}(W)\mid\mathcal{\tilde D}]} \\
   -\frac{1}{1-\pi} &  \frac{1}{\pi(1-\pi)} & 0 \\
  -\frac{\mathbb{E}[\hat{\rho}_{\tilde n}(W)\mid\mathcal{\tilde D}]}{\mathbb{V}[\hat{\rho}_{\tilde n}(W)\mid\mathcal{\tilde D}]} & 0 & \frac{1}{\mathbb{V}[\hat{\rho}_{\tilde n}(W)\mid\mathcal{\tilde D}]} 
\end{pmatrix}\begin{pmatrix}
        \mathbb{E}[Y]\\
        \mathbb{E}[AY]\\
        \mathbb{E}[\hat{\rho}_{\tilde n}(W)Y\mid\mathcal{\tilde D}]
    \end{pmatrix}\\
    &=\begin{pmatrix}
        \mathbb{E}[Y\mid A=0]-\frac{\mathbb{E}[\hat{\rho}_{\tilde n}(W)\mid\mathcal{\tilde D}]\mathbb{COV}[\hat{\rho}_{\tilde n}(W),Y\mid\mathcal{\tilde D}]}{\mathbb{V}[\hat{\rho}_{\tilde n}(W)\mid\mathcal{\tilde D}]}\\
        \mathbb{E}[Y\mid A=1]-\mathbb{E}[Y\mid A=0]\\
        \frac{\mathbb{COV}[\hat{\rho}_{\tilde n}(W),Y\mid\mathcal{\tilde D}]}{\mathbb{V}[\hat{\rho}_{\tilde n}(W)\mid\mathcal{\tilde D}]}
    \end{pmatrix},
\end{split} 
\end{equation}
the following equation holds:
\begin{equation}
\label{eq: thm3 step2}
\begin{split}
    \frac{1}{\pi(1-\pi)}\frac{1}{\sqrt n}
\sum_{i=1}^n
(A_i-\pi)\varepsilon_{\hat \rho_{\tilde n},i}=T_n-b(\hat\rho_{\tilde n})S_n(\hat\rho_{\tilde n}),
\end{split}
\end{equation}
where 
\begin{equation*}
    T_n=\frac{1}{\pi(1-\pi)}\frac{1}{\sqrt n}\sum_{i=1}^n  (A_i-\pi)\left\{Y_i - A_i\mathbb{E}[Y\mid A=1] - (1-A_i)\mathbb{E}[Y\mid A=0]\right\},
\end{equation*}
\begin{equation*}
    b(\hat\rho_{\tilde n})=\frac{\mathbb{COV}[\hat{\rho}_{\tilde n}(W),Y\mid\mathcal{\tilde D}]}{\mathbb{V}[\hat{\rho}_{\tilde n}(W)\mid\mathcal{\tilde D}]},
\end{equation*}
and
\begin{equation*}
    S_n(\hat\rho_{\tilde n})=\frac{1}{\pi(1-\pi)}\frac{1}{\sqrt n}\sum_{i=1}^n  (A_i-\pi)\{\hat{\rho}_{\tilde n}(W_i)-\mathbb{E}[\hat{\rho}_{\tilde n}(W)\mid\mathcal{\tilde D}]\}.
\end{equation*}
Additionally, by the equation \eqref{eq: beta star est}, 
\begin{equation}
\label{eq: beta star est equal}
    \beta_A^*(\hat{\rho}_{\tilde n})=\beta_A^*,
\end{equation}
where $\beta_A^*=\mathbb{E}[Y\mid A=1]-\mathbb{E}[Y\mid A=0]$.
By equations \eqref{eq: thm3 step1}, \eqref{eq: thm3 step2} and \eqref{eq: beta star est equal},
\begin{equation}
\label{eq: thm3 target2}
\begin{split}
    \sqrt n\{\hat\beta_{A,n}(\hat \rho_{\tilde n})-\beta_A^*\}=T_n-b(\hat\rho_{\tilde n})S_n(\hat\rho_{\tilde n})+ o_p(1).
\end{split}
\end{equation}
Note that 
\begin{equation}
\label{eq: thm3 asydist star}
\begin{split}
    \sqrt n\{\hat\beta_{A,n}(\rho^*)-\beta_A^*\}=T_n-b(\rho^*)S_n(\rho^*)+ o_p(1),
\end{split}
\end{equation}
where 
\begin{equation*}
    b({\rho}^*)=\frac{\mathbb{COV}[{\rho}^*(W),Y]}{\mathbb{V}[{\rho}^*(W)]},
\end{equation*}
and
\begin{equation*}
    S_n({\rho}^*)=\frac{1}{\pi(1-\pi)}\frac{1}{\sqrt n}\sum_{i=1}^n  (A_i-\pi)\{{\rho}^*(W_i)-\mathbb{E}[{\rho}^*(W)]\}.
\end{equation*}

We now investigate the asymptotic behavior of $S_n(\hat{\rho}_{\tilde n})$ and 
\begin{equation*}
    S_n(\hat{\rho}_{\tilde n})-S_n({\rho}^*)=\frac{1}{\pi(1-\pi)\sqrt n}\sum_{i=1}^n (A_i-\pi)
\Bigl\{e_{\tilde n}(W_i)-\mathbb{E}[e_{\tilde n}(W)\mid \tilde D]\Bigr\}.
\end{equation*}
Since $A\perp W$ and $\mathcal D \perp  \mathcal{\tilde D}$,
\begin{equation*}
    \mathbb{E}[S_n(\hat{\rho}_{\tilde n})-S_n({\rho}^*)\mid \mathcal{\tilde D}]=\frac{\sqrt{n}\mathbb{E}[A-\pi](\mathbb{E}[e_{\tilde n}(W)\mid \tilde D]-\mathbb{E}[e_{\tilde n}(W)\mid \tilde D])}{\pi(1-\pi)}=0,
\end{equation*}
and 
\begin{align*}
\mathbb{V}[ S_n(\hat{\rho}_{\tilde n})-S_n({\rho}^*)\mid\mathcal{\tilde D}]
&=
\frac{1}{\pi^2(1-\pi)^2}
\mathbb{V}\!\left[(A-\pi)(e_{\tilde n}(W)-\mathbb{E}[e_{\tilde n}(W)\mid \mathcal{\tilde D}])\mid \tilde D\right]\\
&=
\frac{1}{\pi^2(1-\pi)^2}
\mathbb{E}[(A-\pi)^2]\,
\mathbb{V}[e_{\tilde n}(W)\mid \mathcal{\tilde D}]\\
&=
\frac{1}{\pi(1-\pi)}\mathbb{V}[e_{\tilde n}(W)\mid \mathcal{\tilde D}]\\
&\;\le\;
\frac{1}{\pi(1-\pi)}\mathbb{E}[\{e_{\tilde n}(W)\}^2\mid \mathcal{\tilde D}].
\end{align*}
The assumption $\mathbb{E}[\{\hat\rho_{\tilde n}(W)-\rho^*(W) \}^2\mid \mathcal{\tilde D}]=o_p(1)$ implies $\mathbb{E}[\{e_{\tilde n}(W)\}^2\mid \mathcal{\tilde D}]=o_p(1)$.
Therefore, by Chebyshev's inequality,
\begin{equation}
\label{eq: Sdiff op(1)}
    S_n(\hat{\rho}_{\tilde n})-S_n({\rho}^*)= o_p(1).
\end{equation}
Moreover, by $S_n({\rho}^*)=O_p(1)$,
\begin{equation}
\label{eq: Shat Op(1)}
    S_n(\hat{\rho}_{\tilde n})=O_p(1).
\end{equation}

We then investigate the asymptotic behavior of $b(\hat{\rho}_{\tilde n})-b(\rho^*)$.
We decompose the numerator of $b(\hat{\rho}_{\tilde n})$:
\[
\mathbb{COV}[\hat\rho_{\tilde n}(W),Y\mid \mathcal{\tilde D}]
=
\mathbb{COV}[\rho^*(W),Y]
+
\mathbb{COV}[e_{\tilde n}(W),Y\mid \mathcal{\tilde D}].
\]
By Cauchy--Schwarz inequality, $\mathbb{E}[Y^2]<\infty$ and $\mathbb{E}[\{\hat\rho_{\tilde n}(W)-\rho^*(W) \}^2\mid \mathcal{\tilde D}]=o_p(1)$,
\[
\bigl|
\mathbb{COV}[e_{\tilde n}(W),Y\mid \mathcal{\tilde D}]\bigr|
\le
\sqrt{\mathbb{E}[\{e_{\tilde n}(W)\}^2\mid \tilde D]}\,\sqrt{\mathbb E[Y^2]}
=
o_p(1).
\]
Hence,
\begin{equation}
\label{eq: b numdiff op(1)}
\mathbb{COV}[\hat\rho_{\tilde n}(W),Y\mid \mathcal{\tilde D}]
=
\mathbb{COV}[\rho^*(W),Y]
+
o_p(1).    
\end{equation}
By equation \eqref{eq: b numdiff op(1)} and Lemma \ref{lemma: for thm3} (ii),
\begin{equation}
\label{eq: bdiff op(1)}
    b(\hat\rho_{\tilde n})
    =
    \frac{\mathbb{COV}[\rho^*(W),Y]+o_p(1)}{\mathbb{V}[\rho^*(W)]+o_p(1)}
    =
    b(\rho^*)+o_p(1).
\end{equation}

Finally, we decompose
\[
b(\hat\rho_{\tilde n})S_n(\hat\rho_{\tilde n})-b(\rho^*)S_n(\rho^*)
=
\{b(\hat\rho_{\tilde n})-b(\rho^*)\}S_n(\hat\rho_{\tilde n}) + b(\rho^*)\{S_n(\hat\rho_{\tilde n})-S_n(\rho^*)\}.
\]
By equations \eqref{eq: Sdiff op(1)}, \eqref{eq: Shat Op(1)} and \eqref{eq: bdiff op(1)},
\begin{equation*}
    b(\hat\rho_{\tilde n})S_n(\hat\rho_{\tilde n})-b(\rho^*)S_n(\rho^*)=o_p(1).
\end{equation*}
Thus, by equations \eqref{eq: thm3 target2} and \eqref{eq: thm3 asydist star},
\begin{equation*}
\begin{split}
    &\sqrt n\{\hat\beta_{A,n}(\hat \rho_{\tilde n})-\beta_A^*\}-\sqrt n\{\hat\beta_{A,n}( \rho^*)-\beta_A^*\}\\
    &=-\{b(\hat\rho_{\tilde n})S_n(\hat\rho_{\tilde n})-b(\rho^*)S_n(\rho^*)\}+ o_p(1)=o_p(1).
\end{split}
\end{equation*}
\end{proof}

\subsection{Variance estimators for PROCOVA with estimating functions \eqref{eq: est func prog} and \eqref{eq: est func PROCOVA linear}}
\label{appendix: PROCOVA varest}
The variance estimator corresponding to $ V_{\mathrm{fix}}$ is $1/n$ times 
\begin{equation*}
    \hat V_{\mathrm{fix}}= \hat Q_0^{-1}\hat \Omega (\hat Q_0^{-1})^\top,
\end{equation*}
where
\begin{equation*}
    \hat Q_0=  -\frac{1}{n}\sum_{i=1}^n X_{\hat \theta_{\tilde n},i}X_{\hat \theta_{\tilde n},i}^\top =-\frac{1}{n}\sum_{i=1}^n \begin{pmatrix}
   1 & A_i &  \hat\theta_{\tilde n}^{\top} W_i \\
   A_i & A_i &  \hat\theta_{\tilde n}^{\top} W_iA_i \\
   \hat\theta_{\tilde n}^{\top} W_i &   \hat\theta_{\tilde n}^{\top} W_iA_i & ( \hat\theta_{\tilde n}^{\top} W_i)^2
\end{pmatrix}
\end{equation*}
and
\begin{equation*}
   \hat \Omega = \frac{1}{n}\sum_{i=1}^n e_i^2  X_{\hat \theta_{\tilde n},i}X_{\hat \theta_{\tilde n},i}^\top=\frac{1}{n}\sum_{i=1}^n
   \begin{pmatrix}
        e_i^2 & e_i^2A_i &  e_i^2 \hat \theta_{\tilde n}^\top W_i\\
       e_i^2 A_i & e_i^2 A_i & e_i^2 \hat \theta_{\tilde n}^\top W_i A_i\\
        e_i^2 \hat \theta_{\tilde n}^\top W_i & e_i^2\hat \theta_{\tilde n}^\top W_i A_i & e_i^2 (\hat \theta_{\tilde n}^\top W_i )^2
    \end{pmatrix}
\end{equation*}
with $X_{\hat \theta_{\tilde n},i}=(1, A_i, \hat \theta_{\tilde n}^\top W_i)^\top$ and $\hat e_i= Y_i-\hat \beta_n(\hat \theta_{\tilde n})^\top X_{\hat\theta_{\tilde n},i}$.

The variance estimator corresponding to $ V_{\mathrm{est
}}$ is $1/n$ times 
\begin{equation*}
    \hat V_{\mathrm{est}}=\hat V_{\mathrm{fix}} + \frac{n}{\tilde n}\hat Q_0^{-1}\hat Q_1\hat V_\theta \hat Q_1^\top(\hat Q_0^{-1})^\top,
\end{equation*}
where
\begin{equation*}
\begin{split}
    \hat Q_1&= \frac{1}{n}\sum_{i=1}^n\begin{pmatrix}
      -\hat \beta_1 W_i^\top\\
      -\hat \beta_1 A_iW_i^\top\\
     Y_iW_i^\top  - \hat \beta_n(\hat \theta_{\tilde n})^\top 
      \begin{pmatrix}
          W_i^\top & A_iW_i^\top & 2\hat \theta_{\tilde n}^\top W_iW_i^\top 
      \end{pmatrix}^\top
    \end{pmatrix}\\
     &=\frac{1}{n}\sum_{i=1}^n\begin{pmatrix}
      -\hat \beta_1 W_i^\top\\
      -\hat \beta_1 A_iW_i^\top\\
     Y_iW_i^\top  - \hat \beta_0 W_i^\top -\hat \beta_A A_iW_i^\top -2\hat \beta_1 \theta_{\tilde n}^\top W_iW_i^\top
    \end{pmatrix}
\end{split}
\end{equation*}
with $\hat \beta_0 =(1,0,0) \hat \beta_n(\hat \theta_{\tilde n})$, $\hat \beta_A =(0,1,0) \hat \beta_n(\hat \theta_{\tilde n})$ and $\hat \beta_1 =(0,0,1) \hat \beta_n(\hat \theta_{\tilde n})$,
and
\begin{equation*}
\hat V_\theta=\hat Q_2^{-1}\hat Q_3(\hat Q_2^{-1})^{\top},\quad\hat Q_2=- \frac{1}{\tilde n}\sum_{i=1}^{\tilde n}\tilde W_i\tilde W_i^\top,\quad   \hat Q_3= \frac{1}{\tilde n}\sum_{i=1}^{\tilde n}\tilde e_i ^2 \tilde W_i \tilde W_i^\top
\end{equation*}
with $\tilde e_i =\tilde  Y_i-\hat \theta_{\tilde n}^\top \tilde  W_i$.

\clearpage
\section{Appendix: Additional simulation results}
\label{appendix: sim}

\subsection{Scenario A-1}

\begin{figure}[h]
    \centering
    \includegraphics[width=\linewidth]{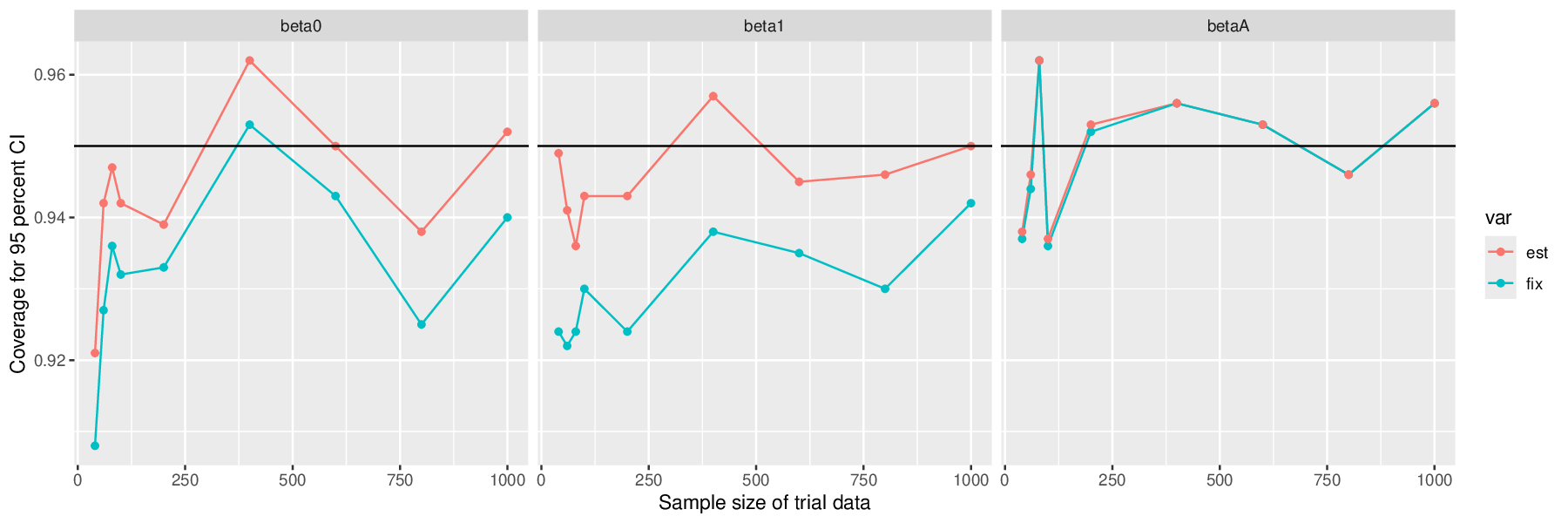}
    \caption{Plots of the coverage probability of 95\% CI over 1000 simulations for Scenario A-1.
    ``beta0'', ``betaA'', ``beta1'' represent the intercept $\beta_0$, the coefficient for the treatment assignment $\beta_A$, the coefficient for the prognostic score $\beta_1$ in the PROCOVA model \eqref{eq: PROCOVA linear}, respectively.
     ``fix'' represents $e^\top\hat V_{\text{fix}}e$ with $e=(1,0,0)^\top$ for $\beta_0$, with $e=(0,1,0)^\top$ for $\beta_A$ and $e=(0,0,1)^\top$ for $\beta_1$.
    ``est'' represents $e^\top\hat V_{\text{est}}e$ with $e=(1,0,0)^\top$ for $\beta_0$, with $e=(0,1,0)^\top$ for $\beta_A$ and $e=(0,0,1)^\top$ for $\beta_1$.
    The x-axis represents the sample size of trial data $n$.
    The sample size of historical data is $\tilde n = 10n$. 
    The y-axis represents the coverage probability which is the proportion of 1000 simulations in which the 95\% CI using each variance estimator includes the true value.
    }
\end{figure}

\clearpage
\begin{figure}[h]
    \centering
    \includegraphics[width=0.4\linewidth]{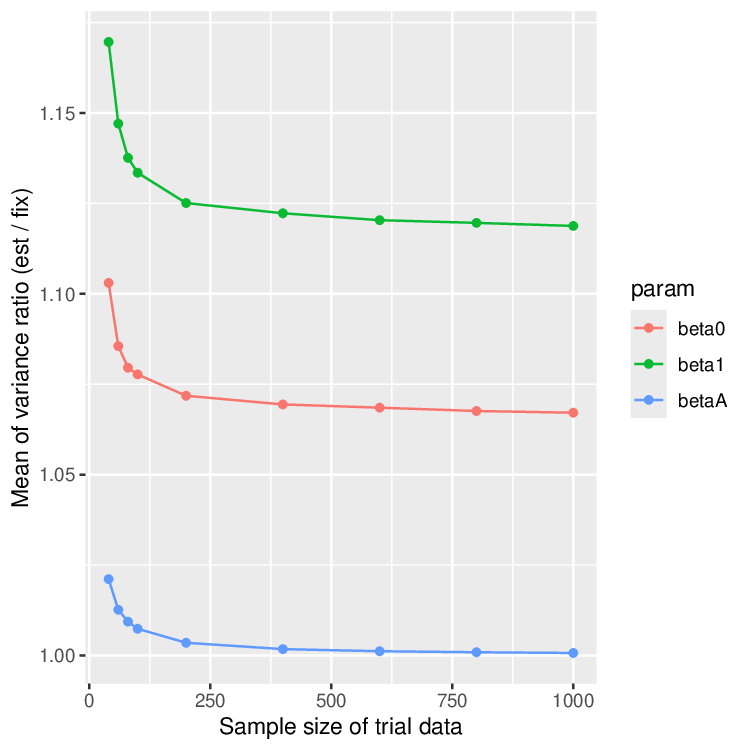}
    \caption{Plots of the mean of the ratio of two variance estimators over 1000 simulations for Scenario A-1.
    ``beta0'', ``betaA'', ``beta1'' represent the intercept $\beta_0$, the coefficient for the treatment assignment $\beta_A$, the coefficient for the prognostic score $\beta_1$ in the PROCOVA model \eqref{eq: PROCOVA linear}, respectively.
    The x-axis represents the sample size of trial data $n$.
    The sample size of historical data is $\tilde n = 10n$. 
    The y-axis represents the mean of the ratio of two variance estimators, i.e., $e^\top\hat V_{\text{est}}e/e^\top\hat V_{\text{fix}}e$ with $e=(1,0,0)^\top$ for $\beta_0$, with $e=(0,1,0)^\top$ for $\beta_A$ and $e=(0,0,1)^\top$ for $\beta_1$, over 1000 simulations.}
\end{figure}
\clearpage
\begin{figure}[h]
    \centering
    \includegraphics[width=0.4\linewidth]{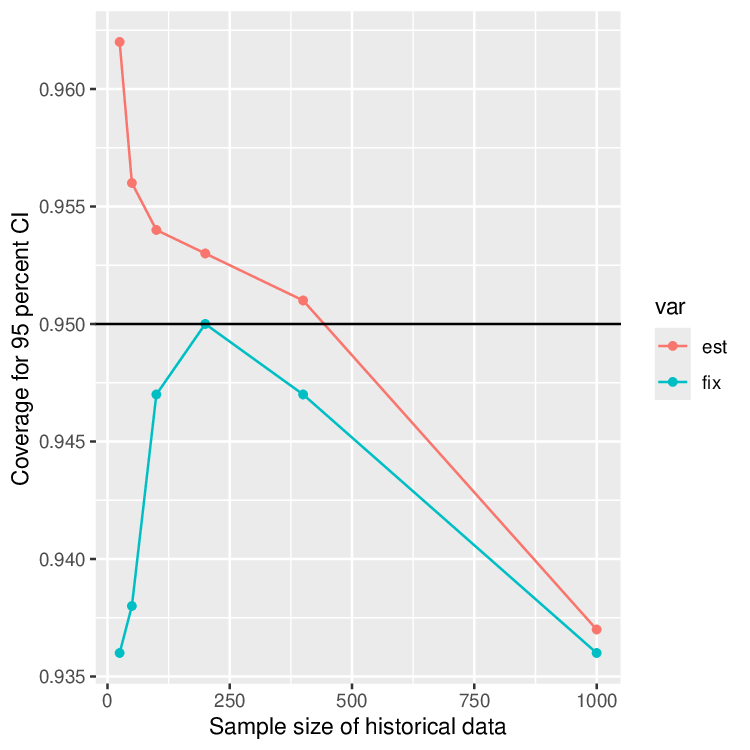}
    \caption{Plots of the coverage probability of the 95\% CI for $\beta_A$ in the PROCOVA model \eqref{eq: PROCOVA linear} over 1000 simulations for Scenario A-1.
     ``fix'' represents $e^\top\hat V_{\text{fix}}e$ with $e=(1,0,0)^\top$ for $\beta_0$, with $e=(0,1,0)^\top$ for $\beta_A$ and $e=(0,0,1)^\top$ for $\beta_1$.
    ``est'' represents $e^\top\hat V_{\text{est}}e$ with $e=(1,0,0)^\top$ for $\beta_0$, with $e=(0,1,0)^\top$ for $\beta_A$ and $e=(0,0,1)^\top$ for $\beta_1$.
    The sample size of trial data is $n=100$. 
    The x-axis represents the sample size of historical data $\tilde n$.
    The y-axis represents the coverage probability which is the proportion of 1000 simulations in which the 95\% CI using each variance estimator includes the true value.}
\end{figure}
\clearpage
\begin{figure}[h]
    \centering
    \includegraphics[width=\linewidth]{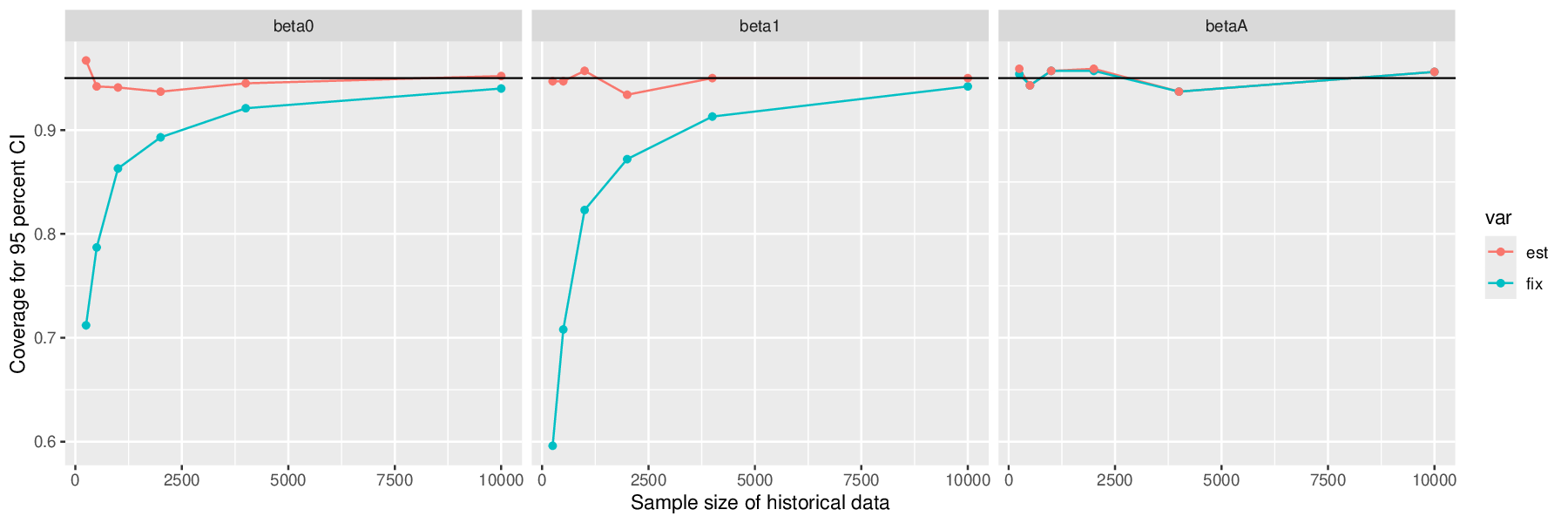}
    \caption{Plots of the coverage probability of 95\% CI over 1000 simulations for Scenario A-1.
    ``beta0'', ``betaA'', ``beta1'' represent the intercept $\beta_0$, the coefficient for the treatment assignment $\beta_A$, the coefficient for the prognostic score $\beta_1$ in the PROCOVA model \eqref{eq: PROCOVA linear}, respectively.
     ``fix'' represents $e^\top\hat V_{\text{fix}}e$ with $e=(1,0,0)^\top$ for $\beta_0$, with $e=(0,1,0)^\top$ for $\beta_A$ and $e=(0,0,1)^\top$ for $\beta_1$.
    ``est'' represents $e^\top\hat V_{\text{est}}e$ with $e=(1,0,0)^\top$ for $\beta_0$, with $e=(0,1,0)^\top$ for $\beta_A$ and $e=(0,0,1)^\top$ for $\beta_1$.
    The sample size of trial data is $n=1000$. 
    The x-axis represents the sample size of historical data $\tilde n$.
    The y-axis represents the coverage probability which is the proportion of 1000 simulations in which the 95\% CI using each variance estimator includes the true value.}
\end{figure}

\clearpage
\subsection{Scenario A-2}

\begin{figure}[h]
    \centering
    \includegraphics[width=\linewidth]{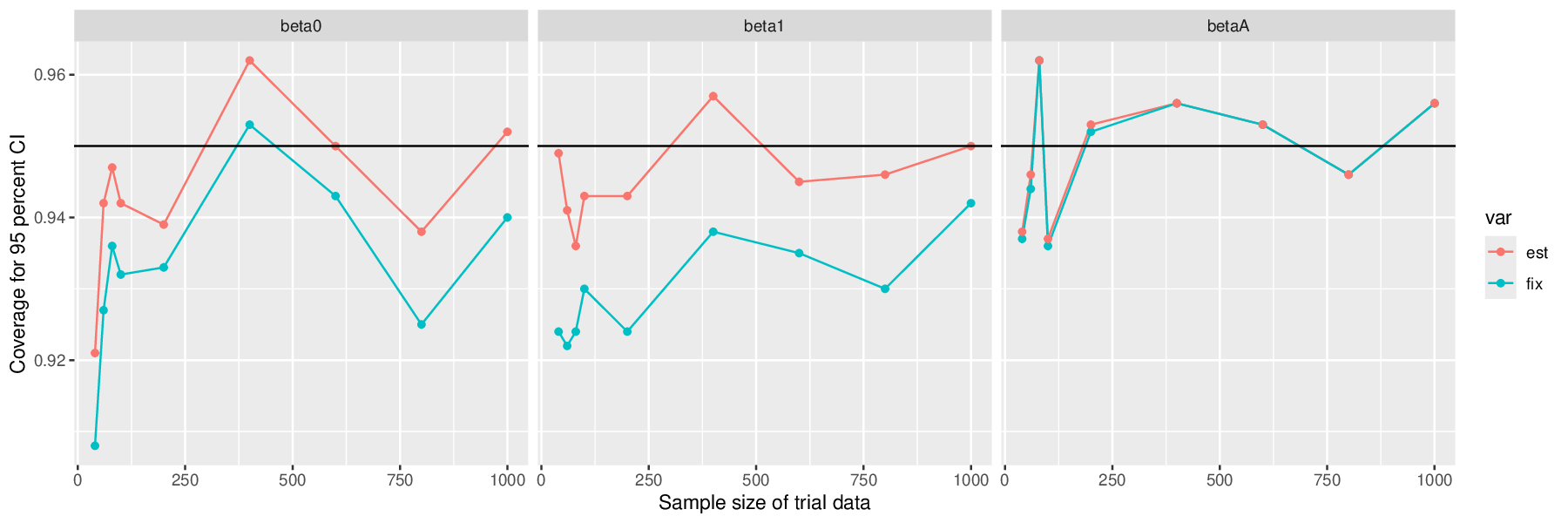}
    \caption{Plots of the coverage probability of 95\% CI over 1000 simulations for Scenario A-2.
    ``beta0'', ``betaA'', ``beta1'' represent the intercept $\beta_0$, the coefficient for the treatment assignment $\beta_A$, the coefficient for the prognostic score $\beta_1$ in the PROCOVA model \eqref{eq: PROCOVA linear}, respectively.
     ``fix'' represents $e^\top\hat V_{\text{fix}}e$ with $e=(1,0,0)^\top$ for $\beta_0$, with $e=(0,1,0)^\top$ for $\beta_A$ and $e=(0,0,1)^\top$ for $\beta_1$.
    ``est'' represents $e^\top\hat V_{\text{est}}e$ with $e=(1,0,0)^\top$ for $\beta_0$, with $e=(0,1,0)^\top$ for $\beta_A$ and $e=(0,0,1)^\top$ for $\beta_1$.
    The x-axis represents the sample size of trial data $n$.
    The sample size of historical data is $\tilde n = 10n$. 
    The y-axis represents the coverage probability which is the proportion of 1000 simulations in which the 95\% CI using each variance estimator includes the true value.
    }
\end{figure}

\clearpage
\begin{figure}[h]
    \centering
    \includegraphics[width=0.4\linewidth]{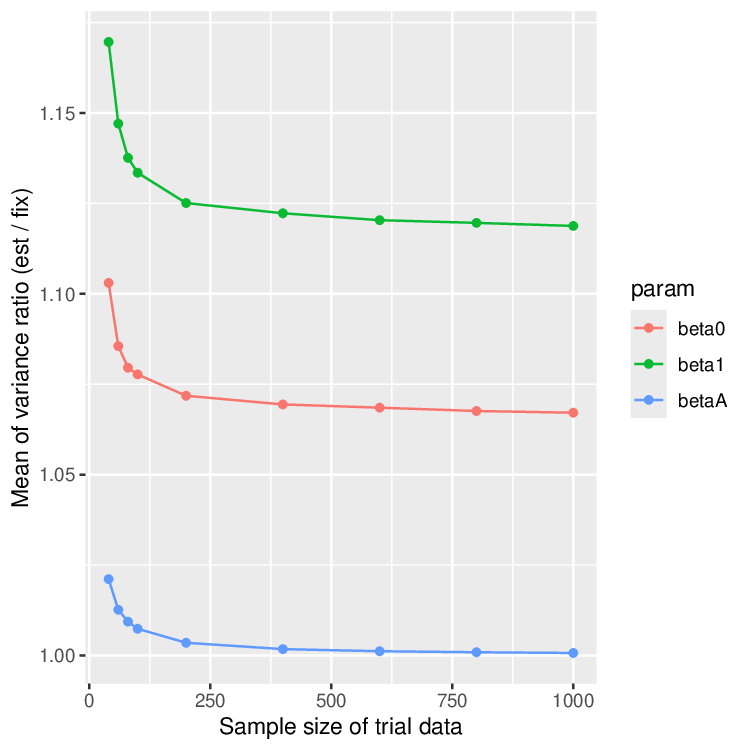}
    \caption{Plots of the mean of the ratio of two variance estimators over 1000 simulations for Scenario A-2.
    ``beta0'', ``betaA'', ``beta1'' represent the intercept $\beta_0$, the coefficient for the treatment assignment $\beta_A$, the coefficient for the prognostic score $\beta_1$ in the PROCOVA model \eqref{eq: PROCOVA linear}, respectively.
    The x-axis represents the sample size of trial data $n$.
    The sample size of historical data is $\tilde n = 10n$. 
    The y-axis represents the mean of the ratio of two variance estimators, i.e., $e^\top\hat V_{\text{est}}e/e^\top\hat V_{\text{fix}}e$ with $e=(1,0,0)^\top$ for $\beta_0$, with $e=(0,1,0)^\top$ for $\beta_A$ and $e=(0,0,1)^\top$ for $\beta_1$, over 1000 simulations.}
\end{figure}
\clearpage
\begin{figure}[h]
    \centering
    \includegraphics[width=0.4\linewidth]{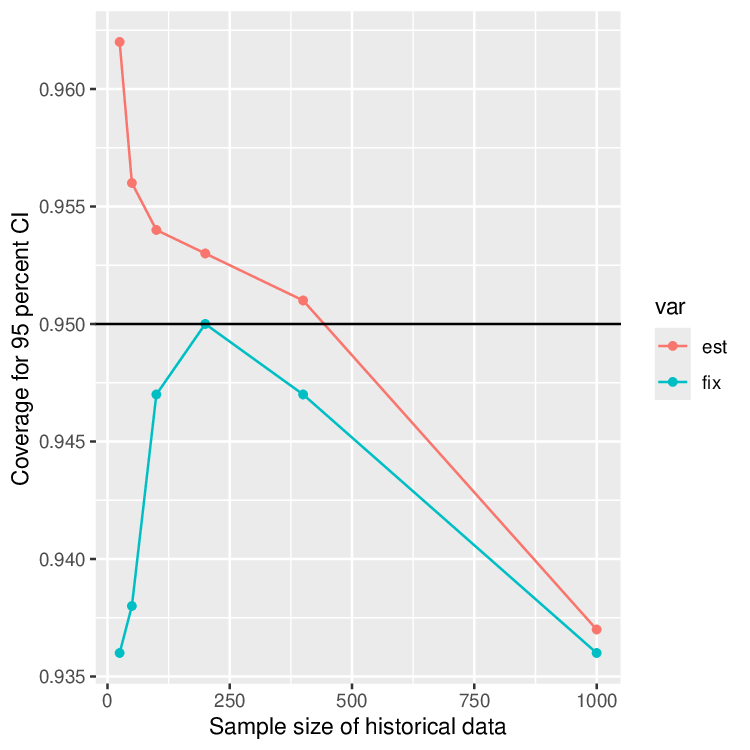}
    \caption{Plots of the coverage probability of the 95\% CI for $\beta_A$ in the PROCOVA model \eqref{eq: PROCOVA linear} over 1000 simulations for Scenario A-2.
     ``fix'' represents $e^\top\hat V_{\text{fix}}e$ with $e=(1,0,0)^\top$ for $\beta_0$, with $e=(0,1,0)^\top$ for $\beta_A$ and $e=(0,0,1)^\top$ for $\beta_1$.
    ``est'' represents $e^\top\hat V_{\text{est}}e$ with $e=(1,0,0)^\top$ for $\beta_0$, with $e=(0,1,0)^\top$ for $\beta_A$ and $e=(0,0,1)^\top$ for $\beta_1$.
    The sample size of trial data is $n=100$. 
    The x-axis represents the sample size of historical data $\tilde n$.
    The y-axis represents the coverage probability which is the proportion of 1000 simulations in which the 95\% CI using each variance estimator includes the true value.}
\end{figure}
\clearpage
\begin{figure}[h]
    \centering
    \includegraphics[width=\linewidth]{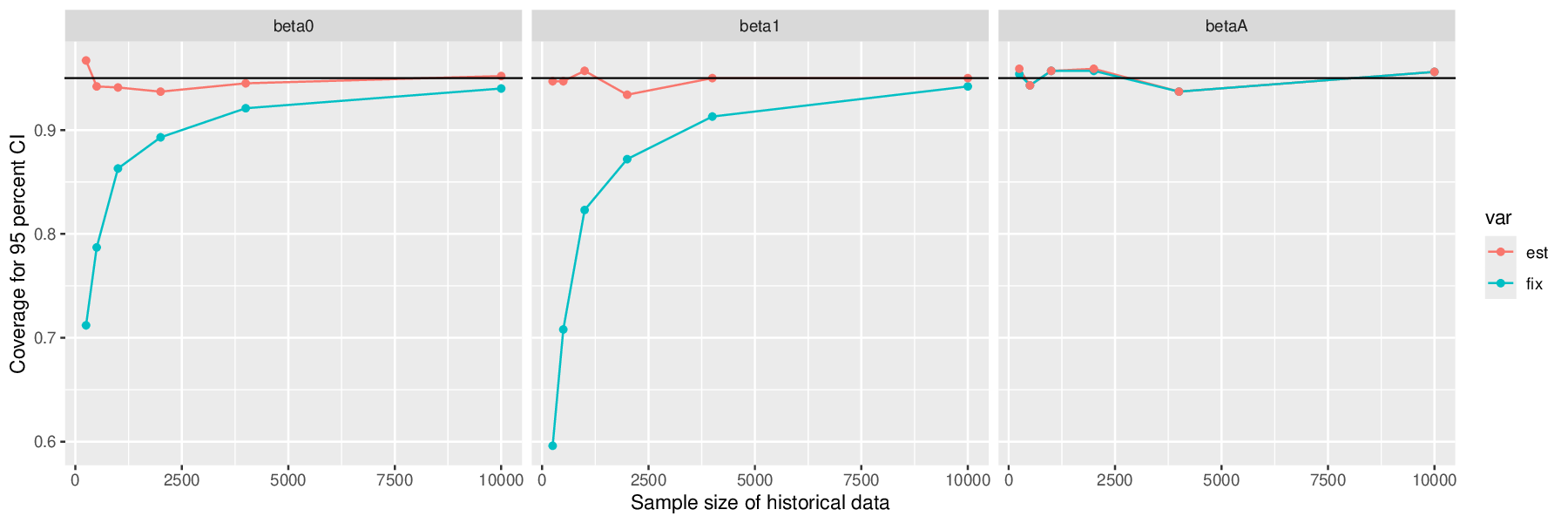}
    \caption{Plots of the coverage probability of 95\% CI over 1000 simulations for Scenario A-2.
    ``beta0'', ``betaA'', ``beta1'' represent the intercept $\beta_0$, the coefficient for the treatment assignment $\beta_A$, the coefficient for the prognostic score $\beta_1$ in the PROCOVA model \eqref{eq: PROCOVA linear}, respectively.
     ``fix'' represents $e^\top\hat V_{\text{fix}}e$ with $e=(1,0,0)^\top$ for $\beta_0$, with $e=(0,1,0)^\top$ for $\beta_A$ and $e=(0,0,1)^\top$ for $\beta_1$.
    ``est'' represents $e^\top\hat V_{\text{est}}e$ with $e=(1,0,0)^\top$ for $\beta_0$, with $e=(0,1,0)^\top$ for $\beta_A$ and $e=(0,0,1)^\top$ for $\beta_1$.
    The sample size of trial data is $n=1000$. 
    The x-axis represents the sample size of historical data $\tilde n$.
    The y-axis represents the coverage probability which is the proportion of 1000 simulations in which the 95\% CI using each variance estimator includes the true value.}
\end{figure}

\clearpage
\subsection{Scenario A-3}

\begin{figure}[h]
    \centering
    \includegraphics[width=\linewidth]{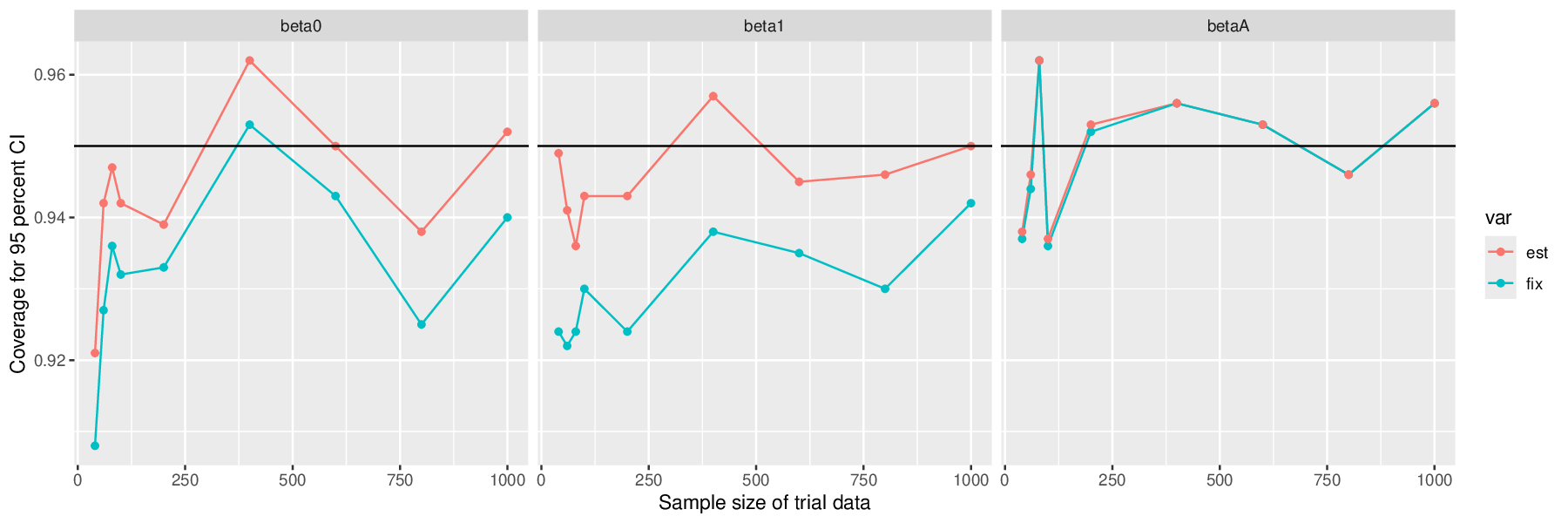}
    \caption{Plots of the coverage probability of 95\% CI over 1000 simulations for Scenario A-3.
    ``beta0'', ``betaA'', ``beta1'' represent the intercept $\beta_0$, the coefficient for the treatment assignment $\beta_A$, the coefficient for the prognostic score $\beta_1$ in the PROCOVA model \eqref{eq: PROCOVA linear}, respectively.
     ``fix'' represents $e^\top\hat V_{\text{fix}}e$ with $e=(1,0,0)^\top$ for $\beta_0$, with $e=(0,1,0)^\top$ for $\beta_A$ and $e=(0,0,1)^\top$ for $\beta_1$.
    ``est'' represents $e^\top\hat V_{\text{est}}e$ with $e=(1,0,0)^\top$ for $\beta_0$, with $e=(0,1,0)^\top$ for $\beta_A$ and $e=(0,0,1)^\top$ for $\beta_1$.
    The x-axis represents the sample size of trial data $n$.
    The sample size of historical data is $\tilde n = 10n$. 
    The y-axis represents the coverage probability which is the proportion of 1000 simulations in which the 95\% CI using each variance estimator includes the true value.
    }
\end{figure}

\clearpage
\begin{figure}[h]
    \centering
    \includegraphics[width=0.4\linewidth]{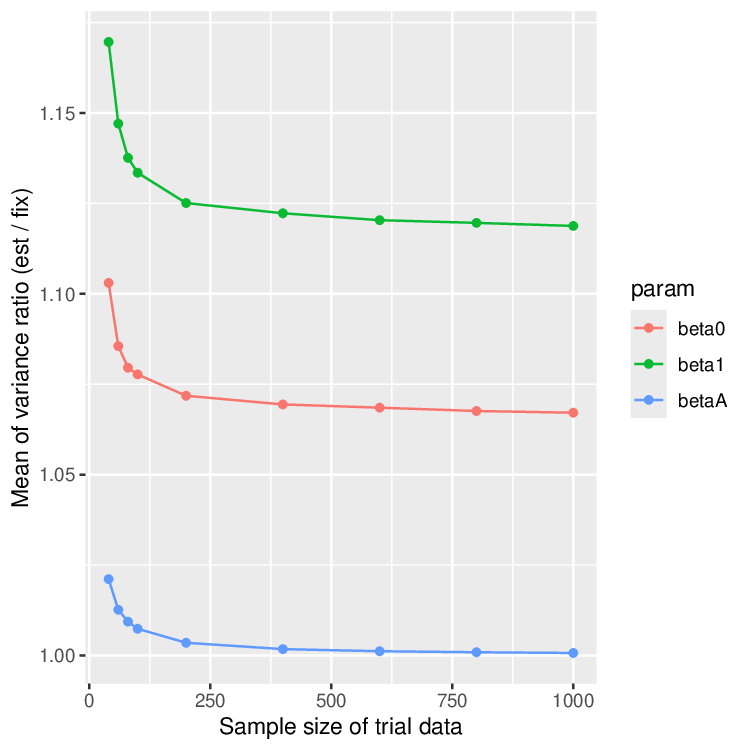}
    \caption{Plots of the mean of the ratio of two variance estimators over 1000 simulations for Scenario A-3.
    ``beta0'', ``betaA'', ``beta1'' represent the intercept $\beta_0$, the coefficient for the treatment assignment $\beta_A$, the coefficient for the prognostic score $\beta_1$ in the PROCOVA model \eqref{eq: PROCOVA linear}, respectively.
    The x-axis represents the sample size of trial data $n$.
    The sample size of historical data is $\tilde n = 10n$. 
    The y-axis represents the mean of the ratio of two variance estimators, i.e., $e^\top\hat V_{\text{est}}e/e^\top\hat V_{\text{fix}}e$ with $e=(1,0,0)^\top$ for $\beta_0$, with $e=(0,1,0)^\top$ for $\beta_A$ and $e=(0,0,1)^\top$ for $\beta_1$, over 1000 simulations.}
\end{figure}
\clearpage
\begin{figure}[h]
    \centering
    \includegraphics[width=0.4\linewidth]{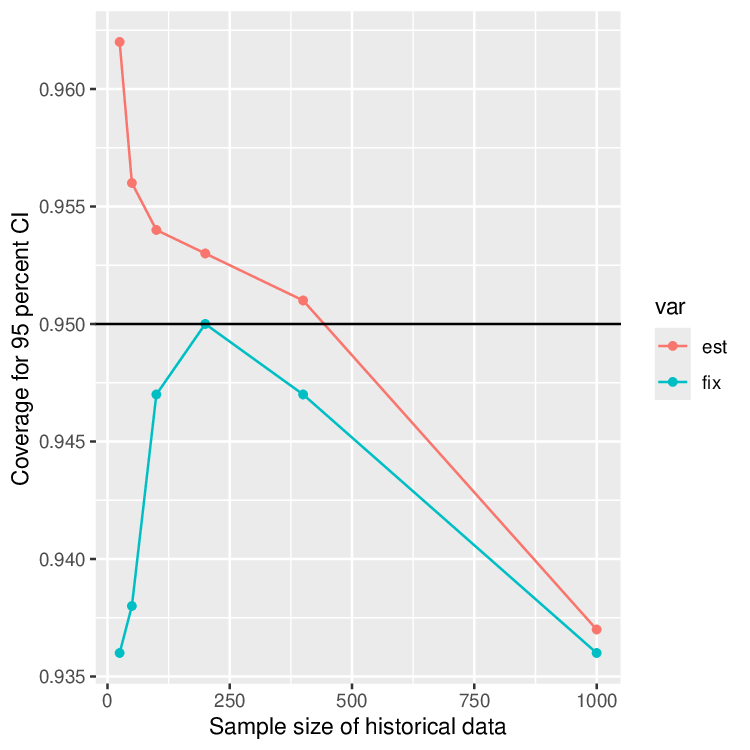}
    \caption{Plots of the coverage probability of the 95\% CI for $\beta_A$ in the PROCOVA model \eqref{eq: PROCOVA linear} over 1000 simulations for Scenario A-3.
     ``fix'' represents $e^\top\hat V_{\text{fix}}e$ with $e=(1,0,0)^\top$ for $\beta_0$, with $e=(0,1,0)^\top$ for $\beta_A$ and $e=(0,0,1)^\top$ for $\beta_1$.
    ``est'' represents $e^\top\hat V_{\text{est}}e$ with $e=(1,0,0)^\top$ for $\beta_0$, with $e=(0,1,0)^\top$ for $\beta_A$ and $e=(0,0,1)^\top$ for $\beta_1$.
    The sample size of trial data is $n=100$. 
    The x-axis represents the sample size of historical data $\tilde n$.
    The y-axis represents the coverage probability which is the proportion of 1000 simulations in which the 95\% CI using each variance estimator includes the true value.}
\end{figure}
\clearpage
\begin{figure}[h]
    \centering
    \includegraphics[width=\linewidth]{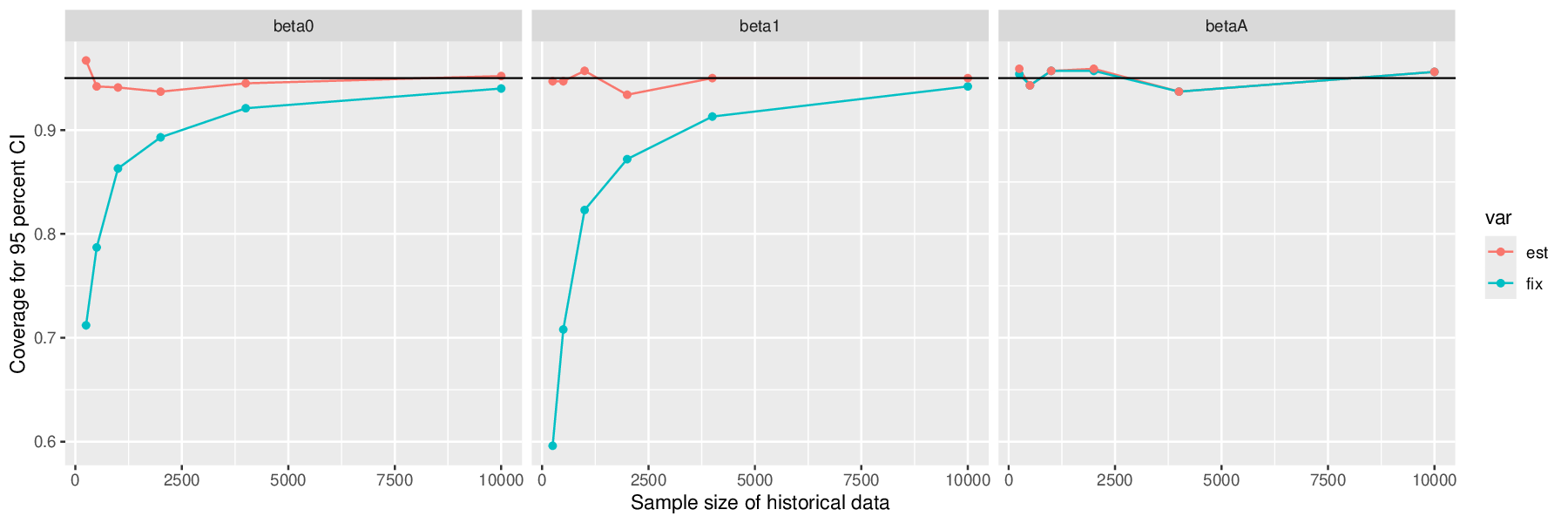}
    \caption{Plots of the coverage probability of 95\% CI over 1000 simulations for Scenario A-3.
    ``beta0'', ``betaA'', ``beta1'' represent the intercept $\beta_0$, the coefficient for the treatment assignment $\beta_A$, the coefficient for the prognostic score $\beta_1$ in the PROCOVA model \eqref{eq: PROCOVA linear}, respectively.
     ``fix'' represents $e^\top\hat V_{\text{fix}}e$ with $e=(1,0,0)^\top$ for $\beta_0$, with $e=(0,1,0)^\top$ for $\beta_A$ and $e=(0,0,1)^\top$ for $\beta_1$.
    ``est'' represents $e^\top\hat V_{\text{est}}e$ with $e=(1,0,0)^\top$ for $\beta_0$, with $e=(0,1,0)^\top$ for $\beta_A$ and $e=(0,0,1)^\top$ for $\beta_1$.
    The sample size of trial data is $n=1000$. 
    The x-axis represents the sample size of historical data $\tilde n$.
    The y-axis represents the coverage probability which is the proportion of 1000 simulations in which the 95\% CI using each variance estimator includes the true value.}
\end{figure}

\clearpage
\subsection{Scenario A-4}

\begin{figure}[h]
    \centering
    \includegraphics[width=\linewidth]{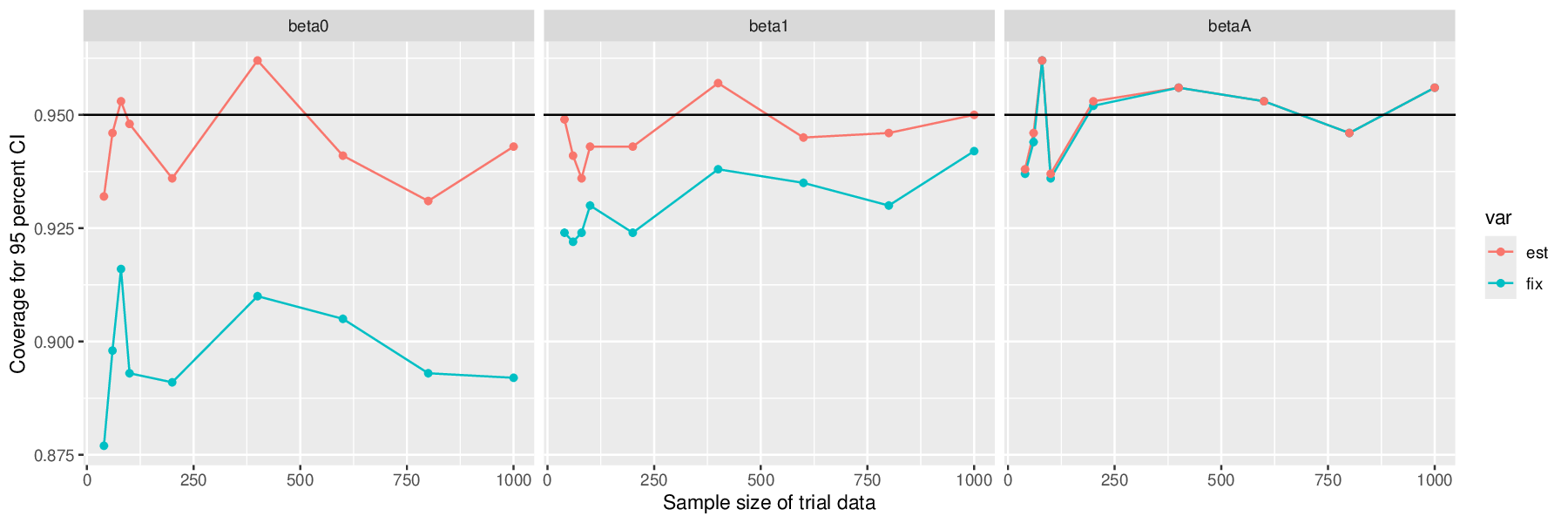}
    \caption{Plots of the coverage probability of 95\% CI over 1000 simulations for Scenario A-4.
    ``beta0'', ``betaA'', ``beta1'' represent the intercept $\beta_0$, the coefficient for the treatment assignment $\beta_A$, the coefficient for the prognostic score $\beta_1$ in the PROCOVA model \eqref{eq: PROCOVA linear}, respectively.
     ``fix'' represents $e^\top\hat V_{\text{fix}}e$ with $e=(1,0,0)^\top$ for $\beta_0$, with $e=(0,1,0)^\top$ for $\beta_A$ and $e=(0,0,1)^\top$ for $\beta_1$.
    ``est'' represents $e^\top\hat V_{\text{est}}e$ with $e=(1,0,0)^\top$ for $\beta_0$, with $e=(0,1,0)^\top$ for $\beta_A$ and $e=(0,0,1)^\top$ for $\beta_1$.
    The x-axis represents the sample size of trial data $n$.
    The sample size of historical data is $\tilde n = 10n$. 
    The y-axis represents the coverage probability which is the proportion of 1000 simulations in which the 95\% CI using each variance estimator includes the true value.
    }
\end{figure}

\clearpage
\begin{figure}[h]
    \centering
    \includegraphics[width=0.4\linewidth]{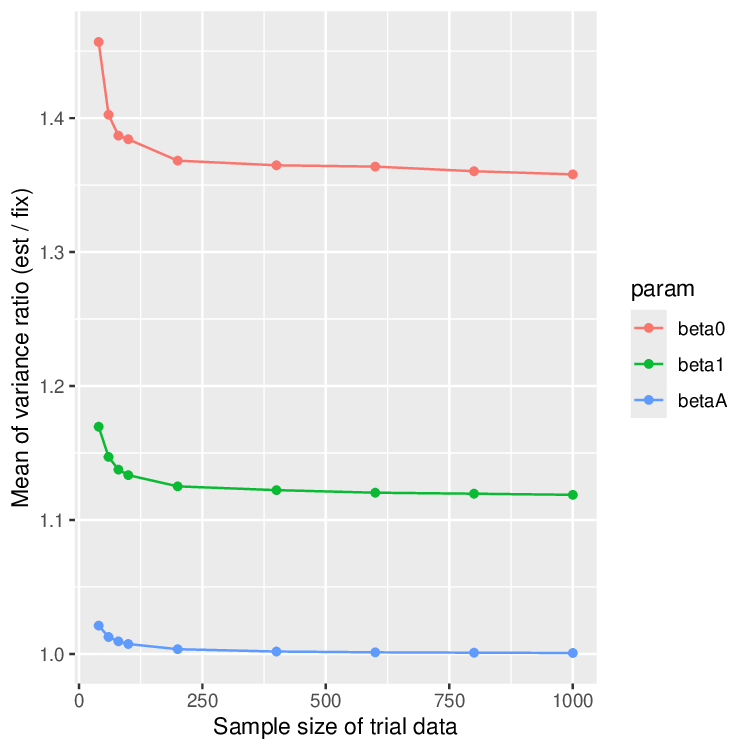}
    \caption{Plots of the mean of the ratio of two variance estimators over 1000 simulations for Scenario A-4.
    ``beta0'', ``betaA'', ``beta1'' represent the intercept $\beta_0$, the coefficient for the treatment assignment $\beta_A$, the coefficient for the prognostic score $\beta_1$ in the PROCOVA model \eqref{eq: PROCOVA linear}, respectively.
    The x-axis represents the sample size of trial data $n$.
    The sample size of historical data is $\tilde n = 10n$. 
    The y-axis represents the mean of the ratio of two variance estimators, i.e., $e^\top\hat V_{\text{est}}e/e^\top\hat V_{\text{fix}}e$ with $e=(1,0,0)^\top$ for $\beta_0$, with $e=(0,1,0)^\top$ for $\beta_A$ and $e=(0,0,1)^\top$ for $\beta_1$, over 1000 simulations.}
\end{figure}

\clearpage
\begin{figure}[h]
    \centering
    \includegraphics[width=0.4\linewidth]{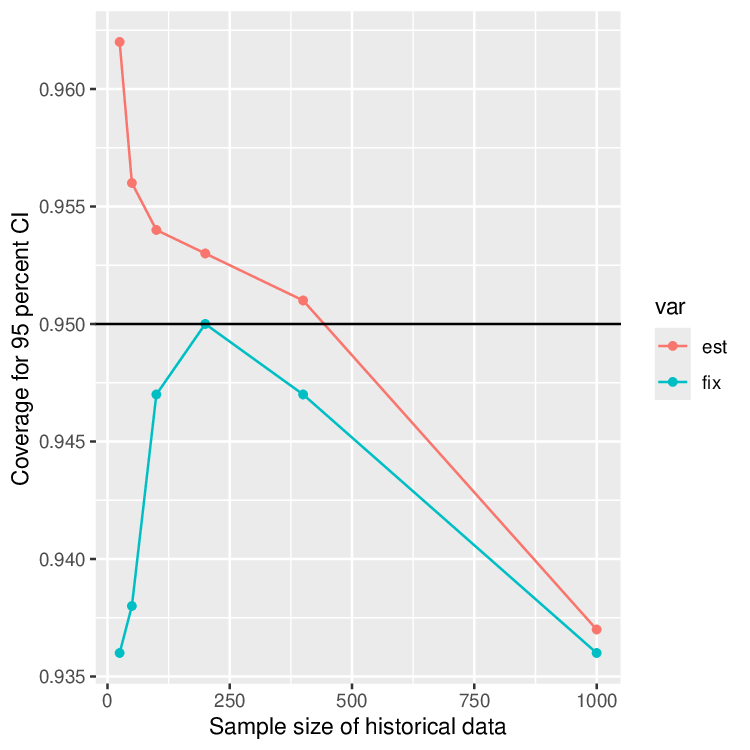}
    \caption{Plots of the coverage probability of the 95\% CI for $\beta_A$ in the PROCOVA model \eqref{eq: PROCOVA linear} over 1000 simulations for Scenario A-4.
     ``fix'' represents $e^\top\hat V_{\text{fix}}e$ with $e=(1,0,0)^\top$ for $\beta_0$, with $e=(0,1,0)^\top$ for $\beta_A$ and $e=(0,0,1)^\top$ for $\beta_1$.
    ``est'' represents $e^\top\hat V_{\text{est}}e$ with $e=(1,0,0)^\top$ for $\beta_0$, with $e=(0,1,0)^\top$ for $\beta_A$ and $e=(0,0,1)^\top$ for $\beta_1$.
    The sample size of trial data is $n=100$. 
    The x-axis represents the sample size of historical data $\tilde n$.
    The y-axis represents the coverage probability which is the proportion of 1000 simulations in which the 95\% CI using each variance estimator includes the true value.}
\end{figure}

\clearpage
\begin{figure}[h]
    \centering
    \includegraphics[width=\linewidth]{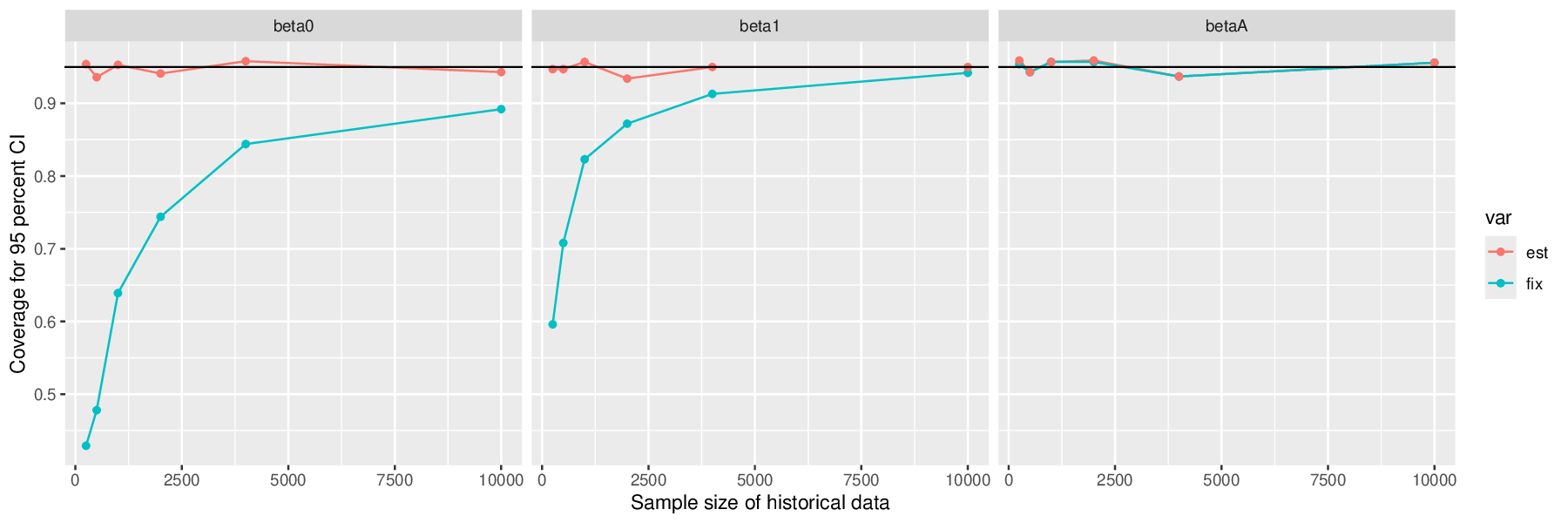}
    \caption{Plots of the coverage probability of 95\% CI over 1000 simulations for Scenario A-4.
    ``beta0'', ``betaA'', ``beta1'' represent the intercept $\beta_0$, the coefficient for the treatment assignment $\beta_A$, the coefficient for the prognostic score $\beta_1$ in the PROCOVA model \eqref{eq: PROCOVA linear}, respectively.
     ``fix'' represents $e^\top\hat V_{\text{fix}}e$ with $e=(1,0,0)^\top$ for $\beta_0$, with $e=(0,1,0)^\top$ for $\beta_A$ and $e=(0,0,1)^\top$ for $\beta_1$.
    ``est'' represents $e^\top\hat V_{\text{est}}e$ with $e=(1,0,0)^\top$ for $\beta_0$, with $e=(0,1,0)^\top$ for $\beta_A$ and $e=(0,0,1)^\top$ for $\beta_1$.
    The sample size of trial data is $n=1000$. 
    The x-axis represents the sample size of historical data $\tilde n$.
    The y-axis represents the coverage probability which is the proportion of 1000 simulations in which the 95\% CI using each variance estimator includes the true value.}
\end{figure}

\clearpage
\subsection{Scenario A-5}

\begin{figure}[h]
    \centering
    \includegraphics[width=\linewidth]{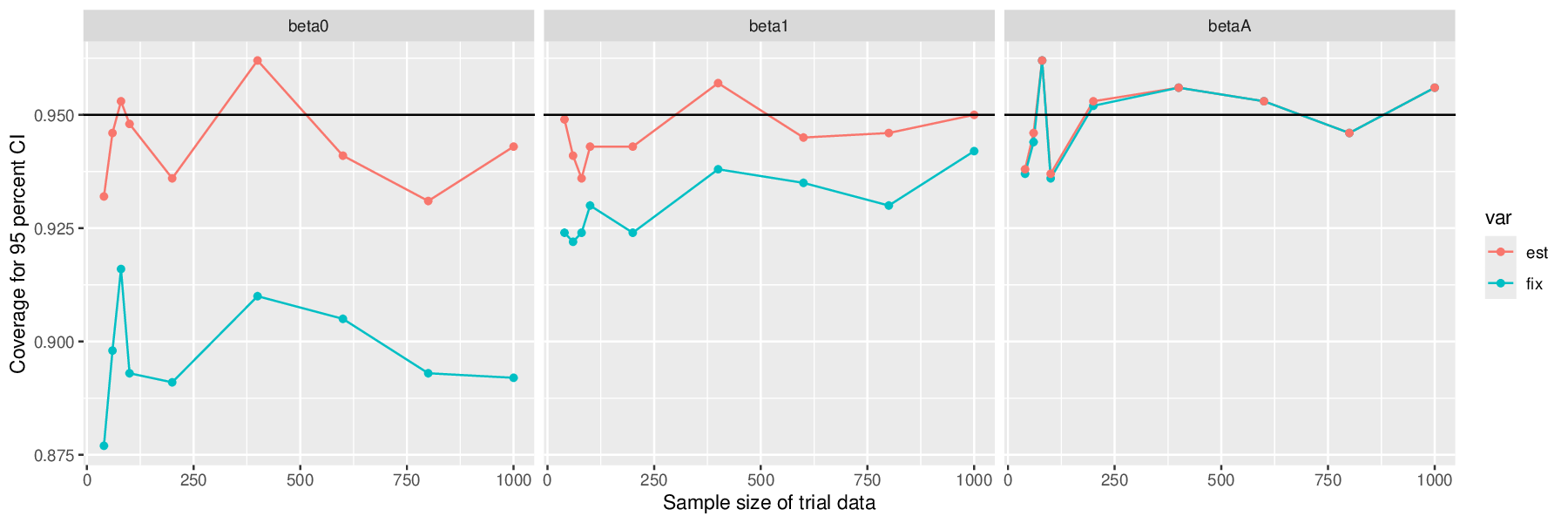}
    \caption{Plots of the coverage probability of 95\% CI over 1000 simulations for Scenario A-5.
    ``beta0'', ``betaA'', ``beta1'' represent the intercept $\beta_0$, the coefficient for the treatment assignment $\beta_A$, the coefficient for the prognostic score $\beta_1$ in the PROCOVA model \eqref{eq: PROCOVA linear}, respectively.
     ``fix'' represents $e^\top\hat V_{\text{fix}}e$ with $e=(1,0,0)^\top$ for $\beta_0$, with $e=(0,1,0)^\top$ for $\beta_A$ and $e=(0,0,1)^\top$ for $\beta_1$.
    ``est'' represents $e^\top\hat V_{\text{est}}e$ with $e=(1,0,0)^\top$ for $\beta_0$, with $e=(0,1,0)^\top$ for $\beta_A$ and $e=(0,0,1)^\top$ for $\beta_1$.
    The x-axis represents the sample size of trial data $n$.
    The sample size of historical data is $\tilde n = 10n$. 
    The y-axis represents the coverage probability which is the proportion of 1000 simulations in which the 95\% CI using each variance estimator includes the true value.
    }
\end{figure}

\clearpage
\begin{figure}[h]
    \centering
    \includegraphics[width=0.4\linewidth]{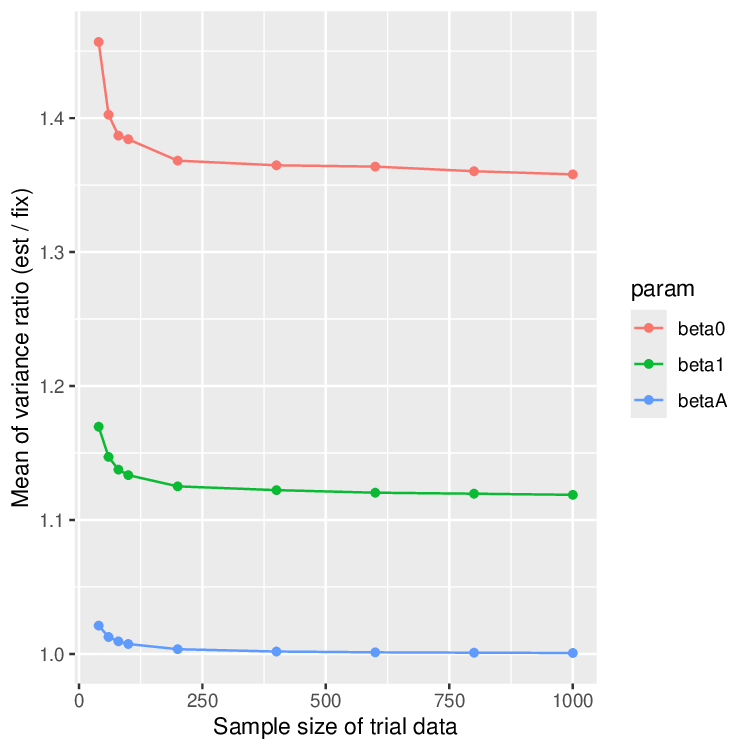}
    \caption{Plots of the mean of the ratio of two variance estimators over 1000 simulations for Scenario A-5.
    ``beta0'', ``betaA'', ``beta1'' represent the intercept $\beta_0$, the coefficient for the treatment assignment $\beta_A$, the coefficient for the prognostic score $\beta_1$ in the PROCOVA model \eqref{eq: PROCOVA linear}, respectively.
    The x-axis represents the sample size of trial data $n$.
    The sample size of historical data is $\tilde n = 10n$. 
    The y-axis represents the mean of the ratio of two variance estimators, i.e., $e^\top\hat V_{\text{est}}e/e^\top\hat V_{\text{fix}}e$ with $e=(1,0,0)^\top$ for $\beta_0$, with $e=(0,1,0)^\top$ for $\beta_A$ and $e=(0,0,1)^\top$ for $\beta_1$, over 1000 simulations.}
\end{figure}

\clearpage
\begin{figure}[h]
    \centering
    \includegraphics[width=0.4\linewidth]{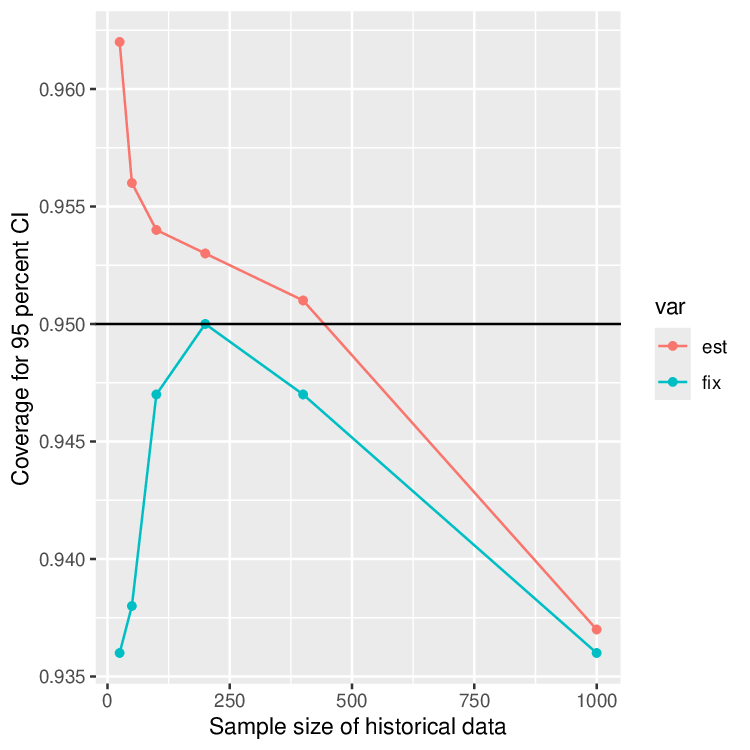}
    \caption{Plots of the coverage probability of the 95\% CI for $\beta_A$ in the PROCOVA model \eqref{eq: PROCOVA linear} over 1000 simulations for Scenario A-5.
     ``fix'' represents $e^\top\hat V_{\text{fix}}e$ with $e=(1,0,0)^\top$ for $\beta_0$, with $e=(0,1,0)^\top$ for $\beta_A$ and $e=(0,0,1)^\top$ for $\beta_1$.
    ``est'' represents $e^\top\hat V_{\text{est}}e$ with $e=(1,0,0)^\top$ for $\beta_0$, with $e=(0,1,0)^\top$ for $\beta_A$ and $e=(0,0,1)^\top$ for $\beta_1$.
    The sample size of trial data is $n=100$. 
    The x-axis represents the sample size of historical data $\tilde n$.
    The y-axis represents the coverage probability which is the proportion of 1000 simulations in which the 95\% CI using each variance estimator includes the true value.}
\end{figure}

\clearpage
\begin{figure}[h]
    \centering
    \includegraphics[width=\linewidth]{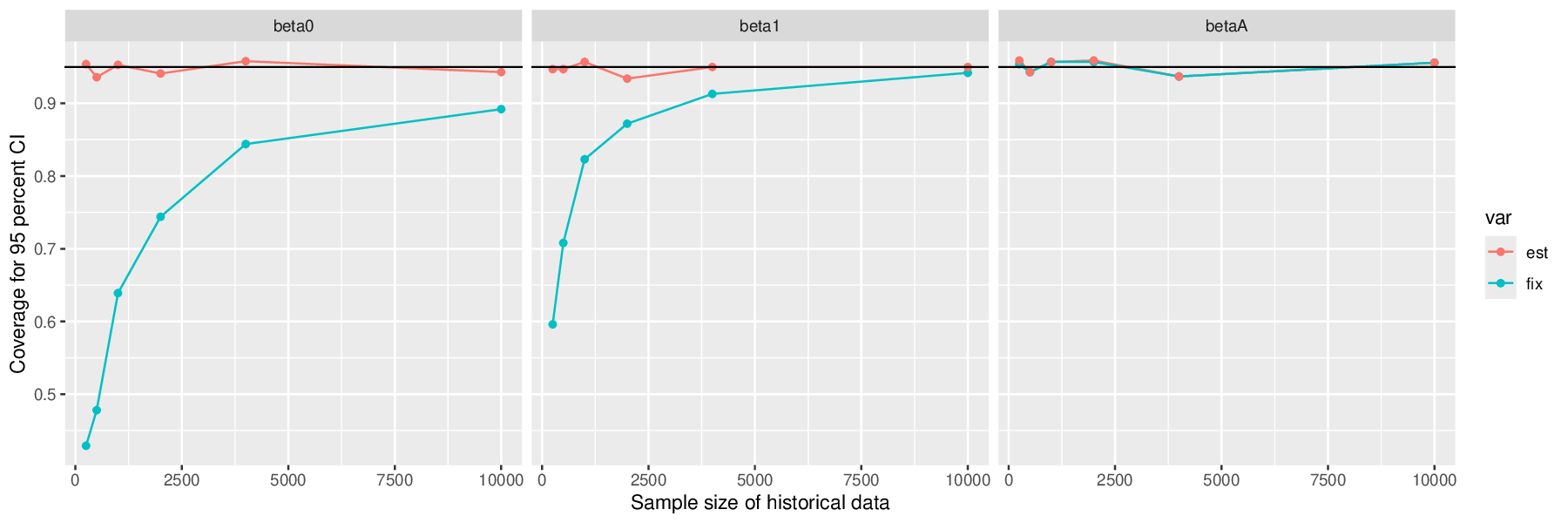}
    \caption{Plots of the coverage probability of 95\% CI over 1000 simulations for Scenario A-5.
    ``beta0'', ``betaA'', ``beta1'' represent the intercept $\beta_0$, the coefficient for the treatment assignment $\beta_A$, the coefficient for the prognostic score $\beta_1$ in the PROCOVA model \eqref{eq: PROCOVA linear}, respectively.
     ``fix'' represents $e^\top\hat V_{\text{fix}}e$ with $e=(1,0,0)^\top$ for $\beta_0$, with $e=(0,1,0)^\top$ for $\beta_A$ and $e=(0,0,1)^\top$ for $\beta_1$.
    ``est'' represents $e^\top\hat V_{\text{est}}e$ with $e=(1,0,0)^\top$ for $\beta_0$, with $e=(0,1,0)^\top$ for $\beta_A$ and $e=(0,0,1)^\top$ for $\beta_1$.
    The sample size of trial data is $n=1000$. 
    The x-axis represents the sample size of historical data $\tilde n$.
    The y-axis represents the coverage probability which is the proportion of 1000 simulations in which the 95\% CI using each variance estimator includes the true value.}
\end{figure}

\clearpage
\subsection{Scenario A-6}

\begin{figure}[h]
    \centering
    \includegraphics[width=\linewidth]{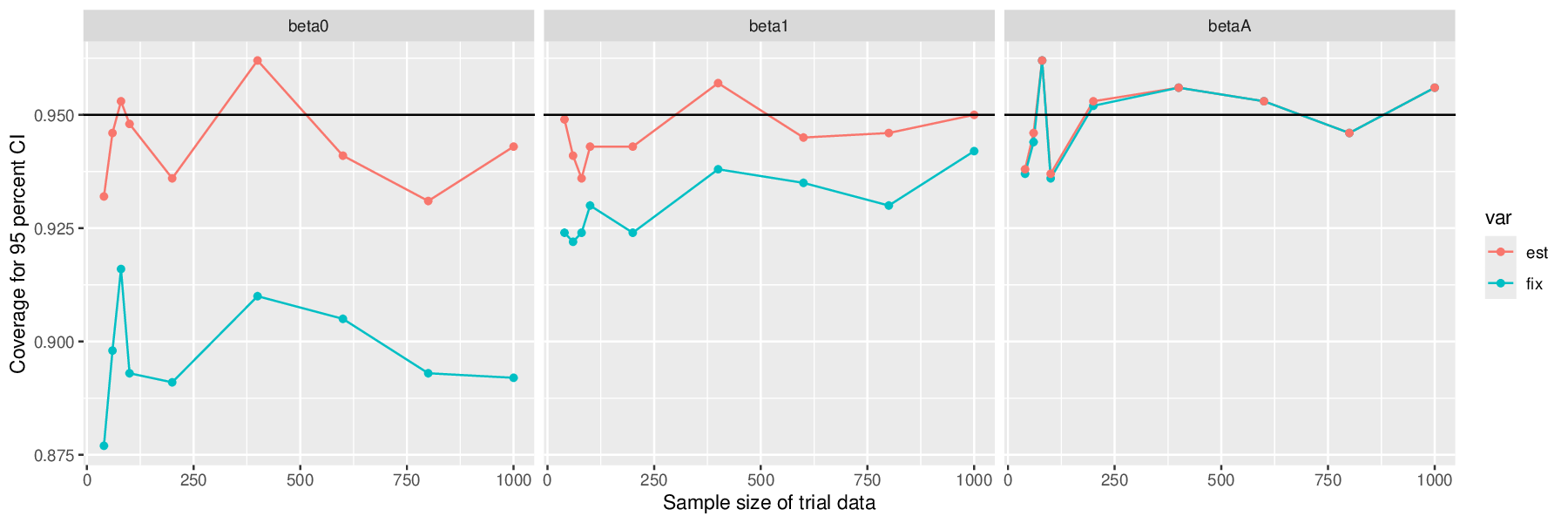}
    \caption{Plots of the coverage probability of 95\% CI over 1000 simulations for Scenario A-6.
    ``beta0'', ``betaA'', ``beta1'' represent the intercept $\beta_0$, the coefficient for the treatment assignment $\beta_A$, the coefficient for the prognostic score $\beta_1$ in the PROCOVA model \eqref{eq: PROCOVA linear}, respectively.
     ``fix'' represents $e^\top\hat V_{\text{fix}}e$ with $e=(1,0,0)^\top$ for $\beta_0$, with $e=(0,1,0)^\top$ for $\beta_A$ and $e=(0,0,1)^\top$ for $\beta_1$.
    ``est'' represents $e^\top\hat V_{\text{est}}e$ with $e=(1,0,0)^\top$ for $\beta_0$, with $e=(0,1,0)^\top$ for $\beta_A$ and $e=(0,0,1)^\top$ for $\beta_1$.
    The x-axis represents the sample size of trial data $n$.
    The sample size of historical data is $\tilde n = 10n$. 
    The y-axis represents the coverage probability which is the proportion of 1000 simulations in which the 95\% CI using each variance estimator includes the true value.
    }
\end{figure}

\clearpage
\begin{figure}[h]
    \centering
    \includegraphics[width=0.4\linewidth]{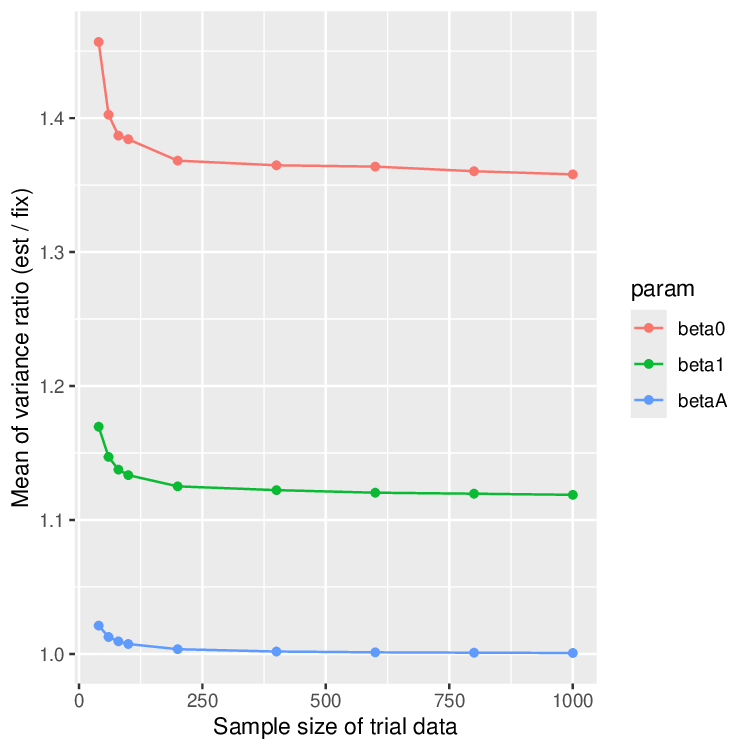}
    \caption{Plots of the mean of the ratio of two variance estimators over 1000 simulations for Scenario A-6.
    ``beta0'', ``betaA'', ``beta1'' represent the intercept $\beta_0$, the coefficient for the treatment assignment $\beta_A$, the coefficient for the prognostic score $\beta_1$ in the PROCOVA model \eqref{eq: PROCOVA linear}, respectively.
    The x-axis represents the sample size of trial data $n$.
    The sample size of historical data is $\tilde n = 10n$. 
    The y-axis represents the mean of the ratio of two variance estimators, i.e., $e^\top\hat V_{\text{est}}e/e^\top\hat V_{\text{fix}}e$ with $e=(1,0,0)^\top$ for $\beta_0$, with $e=(0,1,0)^\top$ for $\beta_A$ and $e=(0,0,1)^\top$ for $\beta_1$, over 1000 simulations.}
\end{figure}

\clearpage
\begin{figure}[h]
    \centering
    \includegraphics[width=0.4\linewidth]{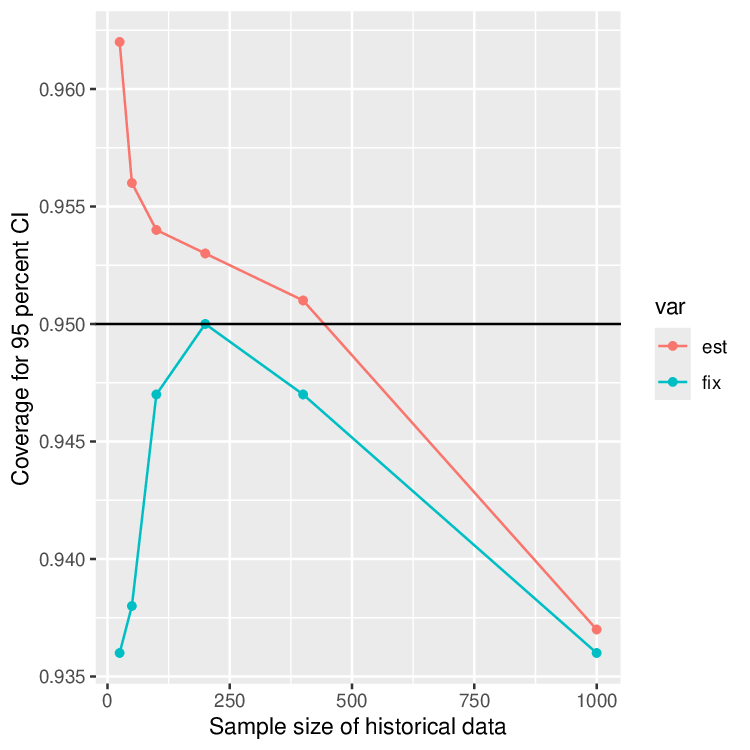}
    \caption{Plots of the coverage probability of the 95\% CI for $\beta_A$ in the PROCOVA model \eqref{eq: PROCOVA linear} over 1000 simulations for Scenario A-6.
     ``fix'' represents $e^\top\hat V_{\text{fix}}e$ with $e=(1,0,0)^\top$ for $\beta_0$, with $e=(0,1,0)^\top$ for $\beta_A$ and $e=(0,0,1)^\top$ for $\beta_1$.
    ``est'' represents $e^\top\hat V_{\text{est}}e$ with $e=(1,0,0)^\top$ for $\beta_0$, with $e=(0,1,0)^\top$ for $\beta_A$ and $e=(0,0,1)^\top$ for $\beta_1$.
    The sample size of trial data is $n=100$. 
    The x-axis represents the sample size of historical data $\tilde n$.
    The y-axis represents the coverage probability which is the proportion of 1000 simulations in which the 95\% CI using each variance estimator includes the true value.}
\end{figure}

\clearpage
\begin{figure}[h]
    \centering
    \includegraphics[width=\linewidth]{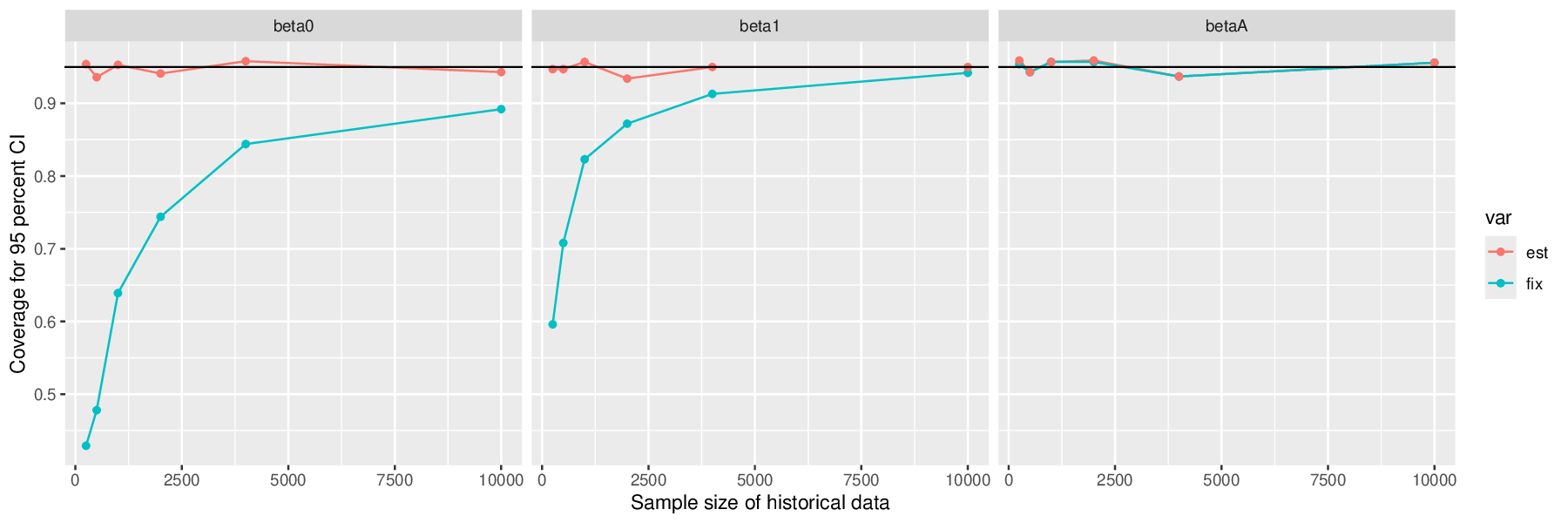}
    \caption{Plots of the coverage probability of 95\% CI over 1000 simulations for Scenario A-6.
    ``beta0'', ``betaA'', ``beta1'' represent the intercept $\beta_0$, the coefficient for the treatment assignment $\beta_A$, the coefficient for the prognostic score $\beta_1$ in the PROCOVA model \eqref{eq: PROCOVA linear}, respectively.
     ``fix'' represents $e^\top\hat V_{\text{fix}}e$ with $e=(1,0,0)^\top$ for $\beta_0$, with $e=(0,1,0)^\top$ for $\beta_A$ and $e=(0,0,1)^\top$ for $\beta_1$.
    ``est'' represents $e^\top\hat V_{\text{est}}e$ with $e=(1,0,0)^\top$ for $\beta_0$, with $e=(0,1,0)^\top$ for $\beta_A$ and $e=(0,0,1)^\top$ for $\beta_1$.
    The sample size of trial data is $n=1000$. 
    The x-axis represents the sample size of historical data $\tilde n$.
    The y-axis represents the coverage probability which is the proportion of 1000 simulations in which the 95\% CI using each variance estimator includes the true value.}
\end{figure}

\clearpage
\subsection{Scenario A-7}

\begin{figure}[h]
    \centering
    \includegraphics[width=\linewidth]{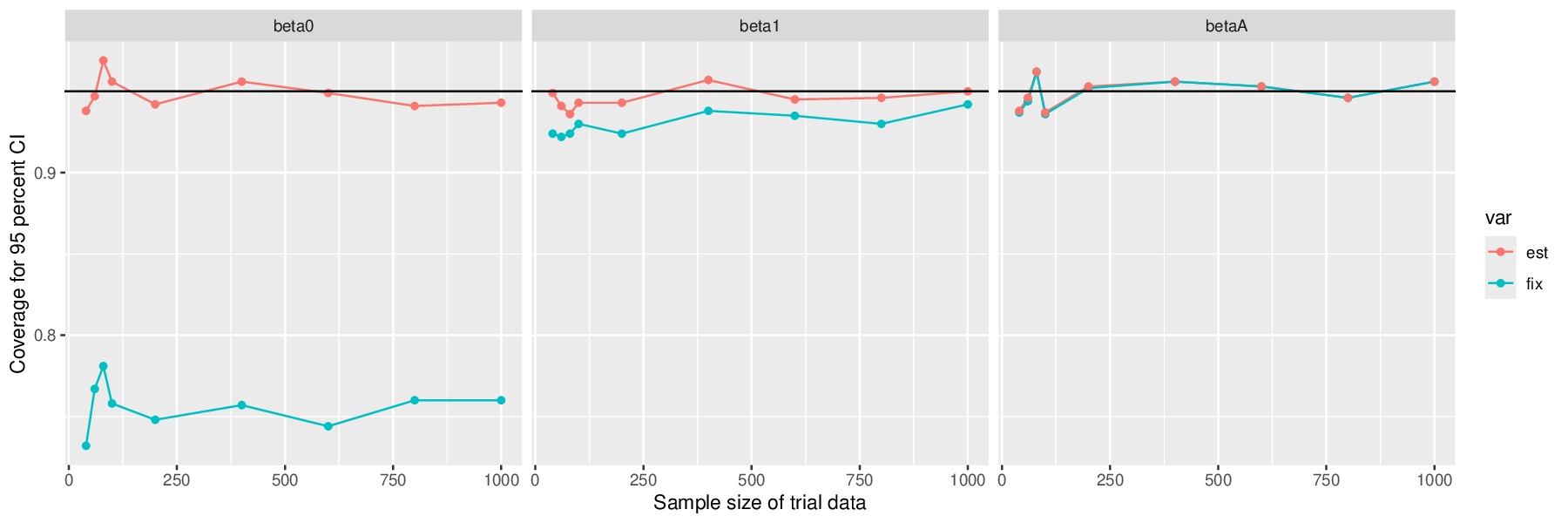}
    \caption{Plots of the coverage probability of 95\% CI over 1000 simulations for Scenario A-7.
    ``beta0'', ``betaA'', ``beta1'' represent the intercept $\beta_0$, the coefficient for the treatment assignment $\beta_A$, the coefficient for the prognostic score $\beta_1$ in the PROCOVA model \eqref{eq: PROCOVA linear}, respectively.
     ``fix'' represents $e^\top\hat V_{\text{fix}}e$ with $e=(1,0,0)^\top$ for $\beta_0$, with $e=(0,1,0)^\top$ for $\beta_A$ and $e=(0,0,1)^\top$ for $\beta_1$.
    ``est'' represents $e^\top\hat V_{\text{est}}e$ with $e=(1,0,0)^\top$ for $\beta_0$, with $e=(0,1,0)^\top$ for $\beta_A$ and $e=(0,0,1)^\top$ for $\beta_1$.
    The x-axis represents the sample size of trial data $n$.
    The sample size of historical data is $\tilde n = 10n$. 
    The y-axis represents the coverage probability which is the proportion of 1000 simulations in which the 95\% CI using each variance estimator includes the true value.
    }
\end{figure}

\clearpage
\begin{figure}[h]
    \centering
    \includegraphics[width=0.4\linewidth]{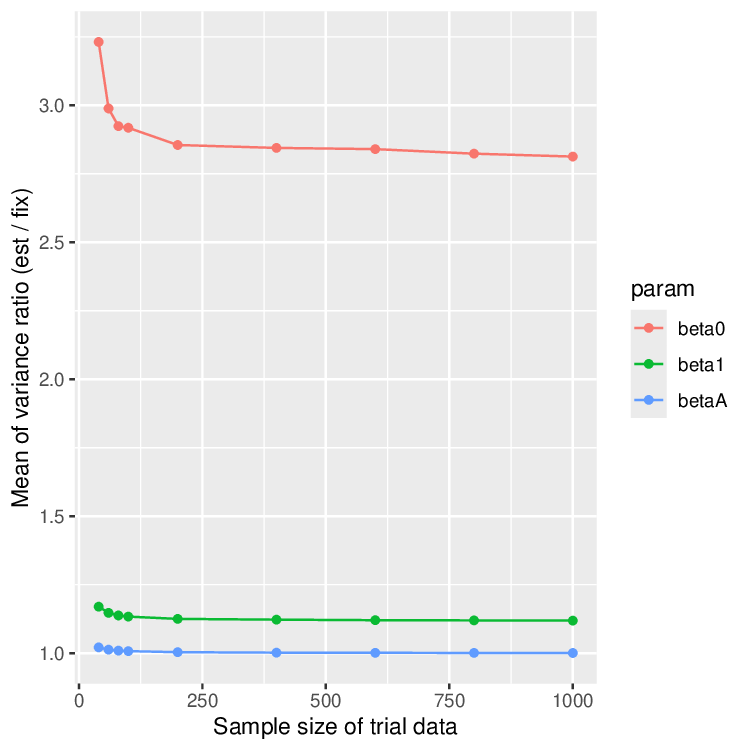}
    \caption{Plots of the mean of the ratio of two variance estimators over 1000 simulations for Scenario A-7.
    ``beta0'', ``betaA'', ``beta1'' represent the intercept $\beta_0$, the coefficient for the treatment assignment $\beta_A$, the coefficient for the prognostic score $\beta_1$ in the PROCOVA model \eqref{eq: PROCOVA linear}, respectively.
    The x-axis represents the sample size of trial data $n$.
    The sample size of historical data is $\tilde n = 10n$. 
    The y-axis represents the mean of the ratio of two variance estimators, i.e., $e^\top\hat V_{\text{est}}e/e^\top\hat V_{\text{fix}}e$ with $e=(1,0,0)^\top$ for $\beta_0$, with $e=(0,1,0)^\top$ for $\beta_A$ and $e=(0,0,1)^\top$ for $\beta_1$, over 1000 simulations.}
\end{figure}
\clearpage
\begin{figure}[h]
    \centering
    \includegraphics[width=0.4\linewidth]{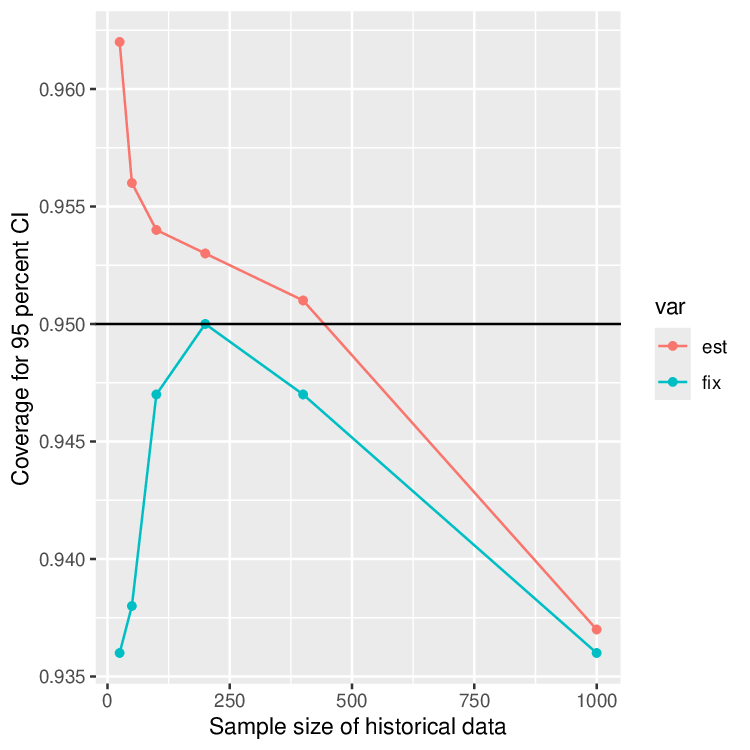}
    \caption{Plots of the coverage probability of the 95\% CI for $\beta_A$ in the PROCOVA model \eqref{eq: PROCOVA linear} over 1000 simulations for Scenario A-7.
     ``fix'' represents $e^\top\hat V_{\text{fix}}e$ with $e=(1,0,0)^\top$ for $\beta_0$, with $e=(0,1,0)^\top$ for $\beta_A$ and $e=(0,0,1)^\top$ for $\beta_1$.
    ``est'' represents $e^\top\hat V_{\text{est}}e$ with $e=(1,0,0)^\top$ for $\beta_0$, with $e=(0,1,0)^\top$ for $\beta_A$ and $e=(0,0,1)^\top$ for $\beta_1$.
    The sample size of trial data is $n=100$. 
    The x-axis represents the sample size of historical data $\tilde n$.
    The y-axis represents the coverage probability which is the proportion of 1000 simulations in which the 95\% CI using each variance estimator includes the true value.}
\end{figure}
\clearpage
\begin{figure}[h]
    \centering
    \includegraphics[width=\linewidth]{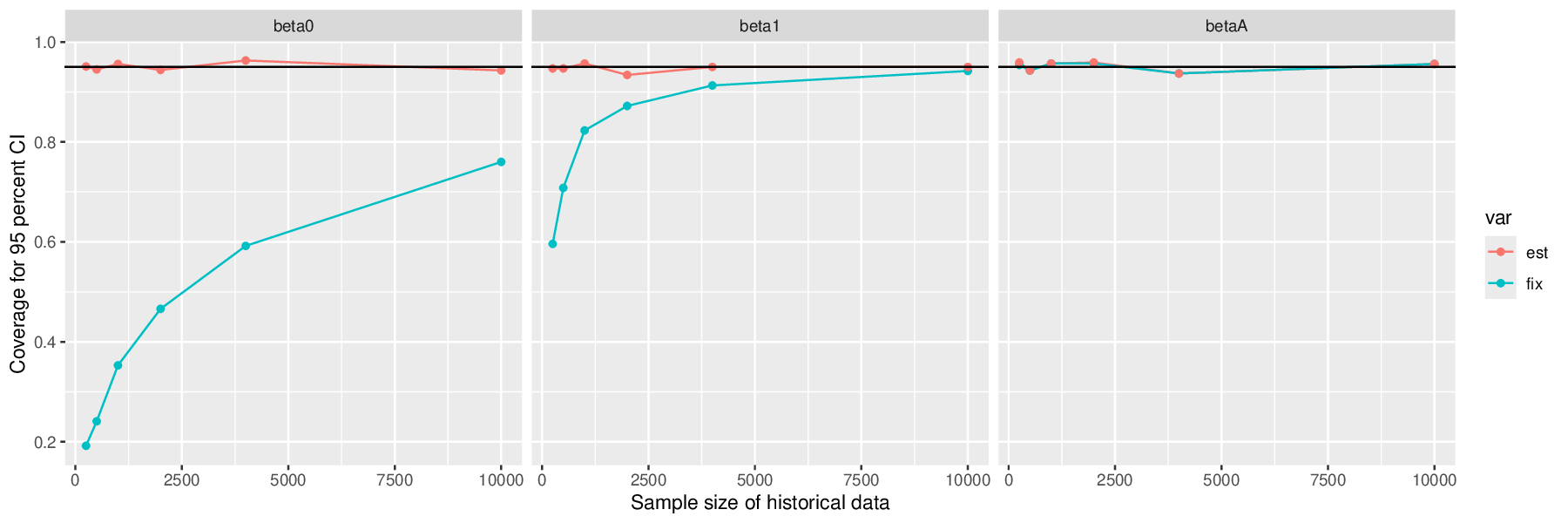}
    \caption{Plots of the coverage probability of 95\% CI over 1000 simulations for Scenario A-7.
    ``beta0'', ``betaA'', ``beta1'' represent the intercept $\beta_0$, the coefficient for the treatment assignment $\beta_A$, the coefficient for the prognostic score $\beta_1$ in the PROCOVA model \eqref{eq: PROCOVA linear}, respectively.
     ``fix'' represents $e^\top\hat V_{\text{fix}}e$ with $e=(1,0,0)^\top$ for $\beta_0$, with $e=(0,1,0)^\top$ for $\beta_A$ and $e=(0,0,1)^\top$ for $\beta_1$.
    ``est'' represents $e^\top\hat V_{\text{est}}e$ with $e=(1,0,0)^\top$ for $\beta_0$, with $e=(0,1,0)^\top$ for $\beta_A$ and $e=(0,0,1)^\top$ for $\beta_1$.
    The sample size of trial data is $n=1000$. 
    The x-axis represents the sample size of historical data $\tilde n$.
    The y-axis represents the coverage probability which is the proportion of 1000 simulations in which the 95\% CI using each variance estimator includes the true value.}
\end{figure}

\clearpage
\subsection{Scenario A-8}

\begin{figure}[h]
    \centering
    \includegraphics[width=\linewidth]{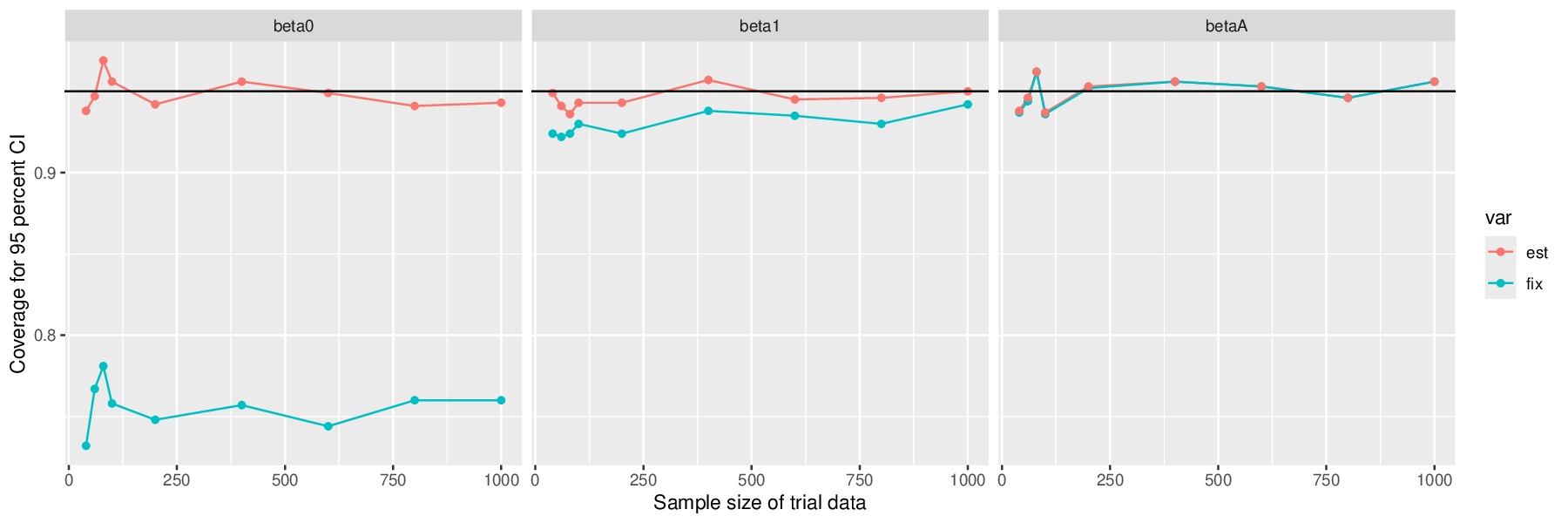}
    \caption{Plots of the coverage probability of 95\% CI over 1000 simulations for Scenario A-8.
    ``beta0'', ``betaA'', ``beta1'' represent the intercept $\beta_0$, the coefficient for the treatment assignment $\beta_A$, the coefficient for the prognostic score $\beta_1$ in the PROCOVA model \eqref{eq: PROCOVA linear}, respectively.
     ``fix'' represents $e^\top\hat V_{\text{fix}}e$ with $e=(1,0,0)^\top$ for $\beta_0$, with $e=(0,1,0)^\top$ for $\beta_A$ and $e=(0,0,1)^\top$ for $\beta_1$.
    ``est'' represents $e^\top\hat V_{\text{est}}e$ with $e=(1,0,0)^\top$ for $\beta_0$, with $e=(0,1,0)^\top$ for $\beta_A$ and $e=(0,0,1)^\top$ for $\beta_1$.
    The x-axis represents the sample size of trial data $n$.
    The sample size of historical data is $\tilde n = 10n$. 
    The y-axis represents the coverage probability which is the proportion of 1000 simulations in which the 95\% CI using each variance estimator includes the true value.
    }
\end{figure}

\clearpage
\begin{figure}[h]
    \centering
    \includegraphics[width=0.4\linewidth]{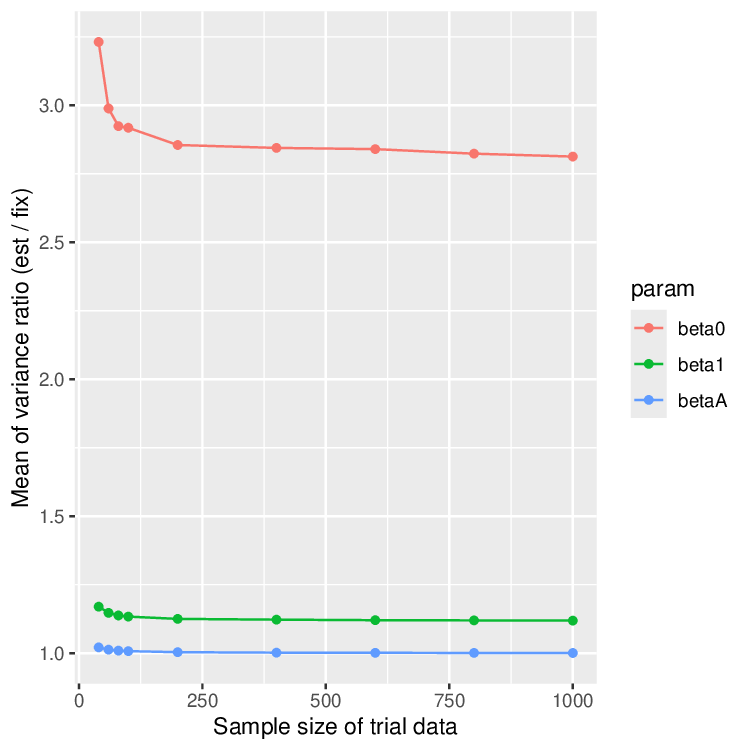}
    \caption{Plots of the mean of the ratio of two variance estimators over 1000 simulations for Scenario A-8.
    ``beta0'', ``betaA'', ``beta1'' represent the intercept $\beta_0$, the coefficient for the treatment assignment $\beta_A$, the coefficient for the prognostic score $\beta_1$ in the PROCOVA model \eqref{eq: PROCOVA linear}, respectively.
    The x-axis represents the sample size of trial data $n$.
    The sample size of historical data is $\tilde n = 10n$. 
    The y-axis represents the mean of the ratio of two variance estimators, i.e., $e^\top\hat V_{\text{est}}e/e^\top\hat V_{\text{fix}}e$ with $e=(1,0,0)^\top$ for $\beta_0$, with $e=(0,1,0)^\top$ for $\beta_A$ and $e=(0,0,1)^\top$ for $\beta_1$, over 1000 simulations.}
\end{figure}
\clearpage
\begin{figure}[h]
    \centering
    \includegraphics[width=0.4\linewidth]{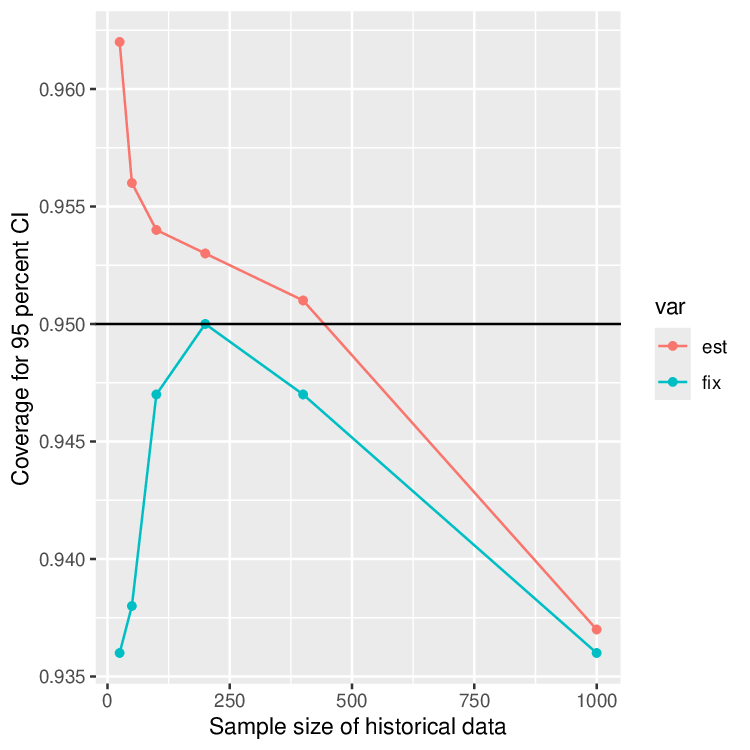}
    \caption{Plots of the coverage probability of the 95\% CI for $\beta_A$ in the PROCOVA model \eqref{eq: PROCOVA linear} over 1000 simulations for Scenario A-8.
     ``fix'' represents $e^\top\hat V_{\text{fix}}e$ with $e=(1,0,0)^\top$ for $\beta_0$, with $e=(0,1,0)^\top$ for $\beta_A$ and $e=(0,0,1)^\top$ for $\beta_1$.
    ``est'' represents $e^\top\hat V_{\text{est}}e$ with $e=(1,0,0)^\top$ for $\beta_0$, with $e=(0,1,0)^\top$ for $\beta_A$ and $e=(0,0,1)^\top$ for $\beta_1$.
    The sample size of trial data is $n=100$. 
    The x-axis represents the sample size of historical data $\tilde n$.
    The y-axis represents the coverage probability which is the proportion of 1000 simulations in which the 95\% CI using each variance estimator includes the true value.}
\end{figure}
\clearpage
\begin{figure}[h]
    \centering
    \includegraphics[width=\linewidth]{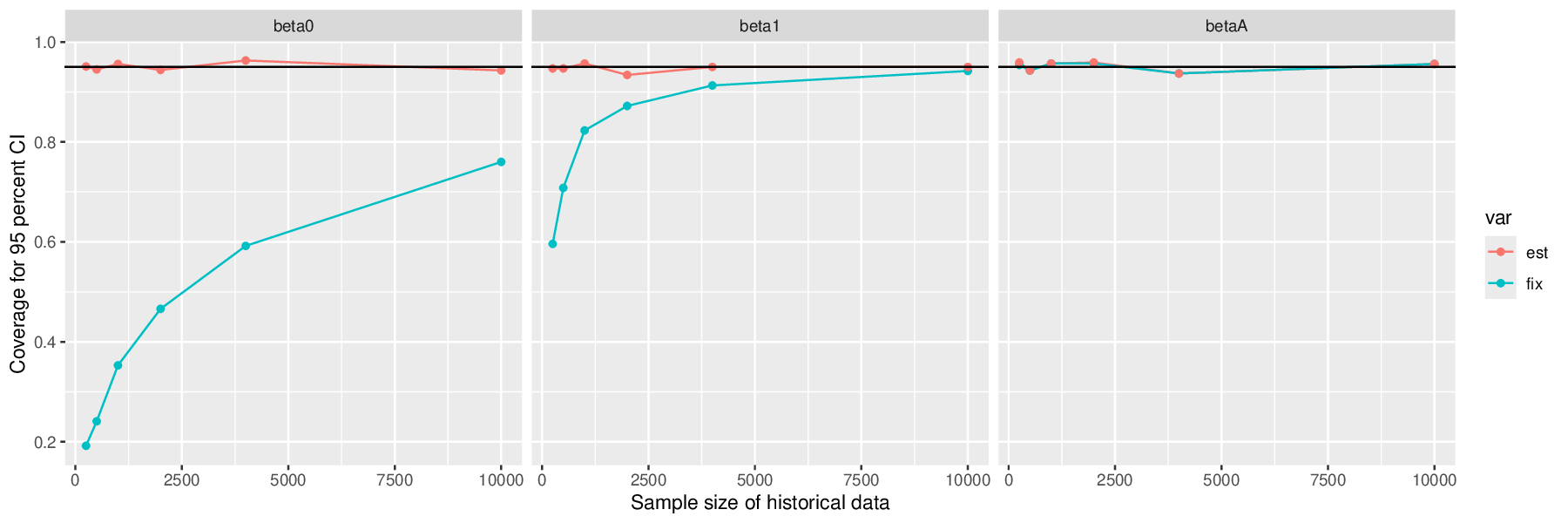}
    \caption{Plots of the coverage probability of 95\% CI over 1000 simulations for Scenario A-8.
    ``beta0'', ``betaA'', ``beta1'' represent the intercept $\beta_0$, the coefficient for the treatment assignment $\beta_A$, the coefficient for the prognostic score $\beta_1$ in the PROCOVA model \eqref{eq: PROCOVA linear}, respectively.
     ``fix'' represents $e^\top\hat V_{\text{fix}}e$ with $e=(1,0,0)^\top$ for $\beta_0$, with $e=(0,1,0)^\top$ for $\beta_A$ and $e=(0,0,1)^\top$ for $\beta_1$.
    ``est'' represents $e^\top\hat V_{\text{est}}e$ with $e=(1,0,0)^\top$ for $\beta_0$, with $e=(0,1,0)^\top$ for $\beta_A$ and $e=(0,0,1)^\top$ for $\beta_1$.
    The sample size of trial data is $n=1000$. 
    The x-axis represents the sample size of historical data $\tilde n$.
    The y-axis represents the coverage probability which is the proportion of 1000 simulations in which the 95\% CI using each variance estimator includes the true value.}
\end{figure}

\clearpage
\subsection{Scenario A-9}

\begin{figure}[h]
    \centering
    \includegraphics[width=\linewidth]{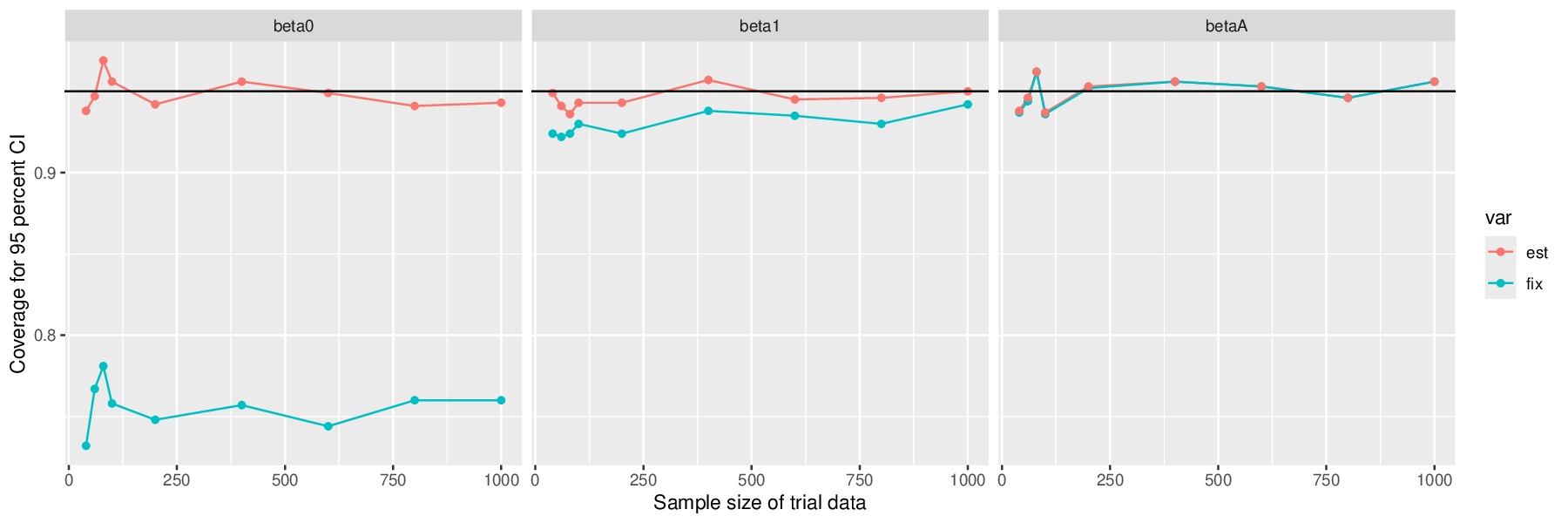}
    \caption{Plots of the coverage probability of 95\% CI over 1000 simulations for Scenario A-9.
    ``beta0'', ``betaA'', ``beta1'' represent the intercept $\beta_0$, the coefficient for the treatment assignment $\beta_A$, the coefficient for the prognostic score $\beta_1$ in the PROCOVA model \eqref{eq: PROCOVA linear}, respectively.
     ``fix'' represents $e^\top\hat V_{\text{fix}}e$ with $e=(1,0,0)^\top$ for $\beta_0$, with $e=(0,1,0)^\top$ for $\beta_A$ and $e=(0,0,1)^\top$ for $\beta_1$.
    ``est'' represents $e^\top\hat V_{\text{est}}e$ with $e=(1,0,0)^\top$ for $\beta_0$, with $e=(0,1,0)^\top$ for $\beta_A$ and $e=(0,0,1)^\top$ for $\beta_1$.
    The x-axis represents the sample size of trial data $n$.
    The sample size of historical data is $\tilde n = 10n$. 
    The y-axis represents the coverage probability which is the proportion of 1000 simulations in which the 95\% CI using each variance estimator includes the true value.
    }
\end{figure}

\clearpage
\begin{figure}[h]
    \centering
    \includegraphics[width=0.4\linewidth]{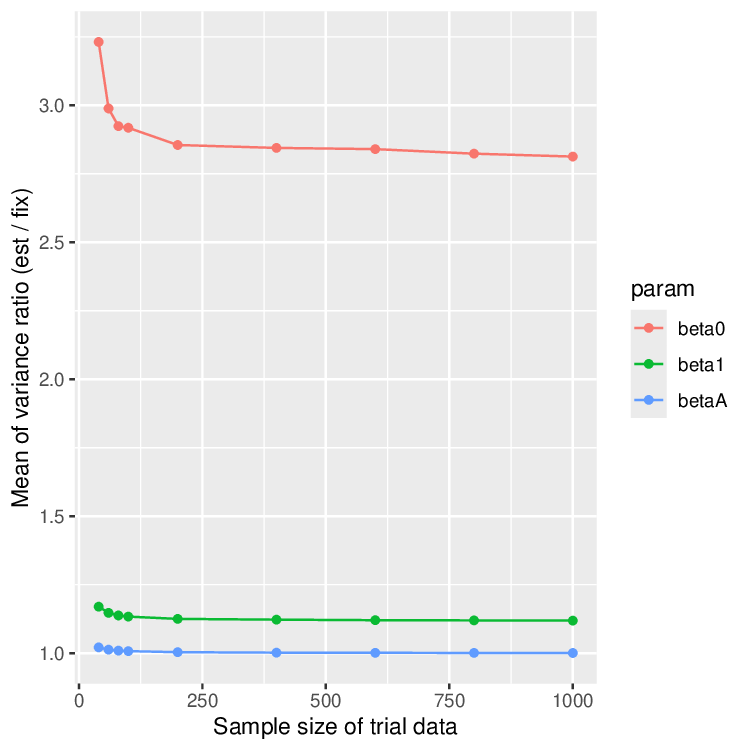}
    \caption{Plots of the mean of the ratio of two variance estimators over 1000 simulations for Scenario A-9.
    ``beta0'', ``betaA'', ``beta1'' represent the intercept $\beta_0$, the coefficient for the treatment assignment $\beta_A$, the coefficient for the prognostic score $\beta_1$ in the PROCOVA model \eqref{eq: PROCOVA linear}, respectively.
    The x-axis represents the sample size of trial data $n$.
    The sample size of historical data is $\tilde n = 10n$. 
    The y-axis represents the mean of the ratio of two variance estimators, i.e., $e^\top\hat V_{\text{est}}e/e^\top\hat V_{\text{fix}}e$ with $e=(1,0,0)^\top$ for $\beta_0$, with $e=(0,1,0)^\top$ for $\beta_A$ and $e=(0,0,1)^\top$ for $\beta_1$, over 1000 simulations.}
\end{figure}
\clearpage
\begin{figure}[h]
    \centering
    \includegraphics[width=0.4\linewidth]{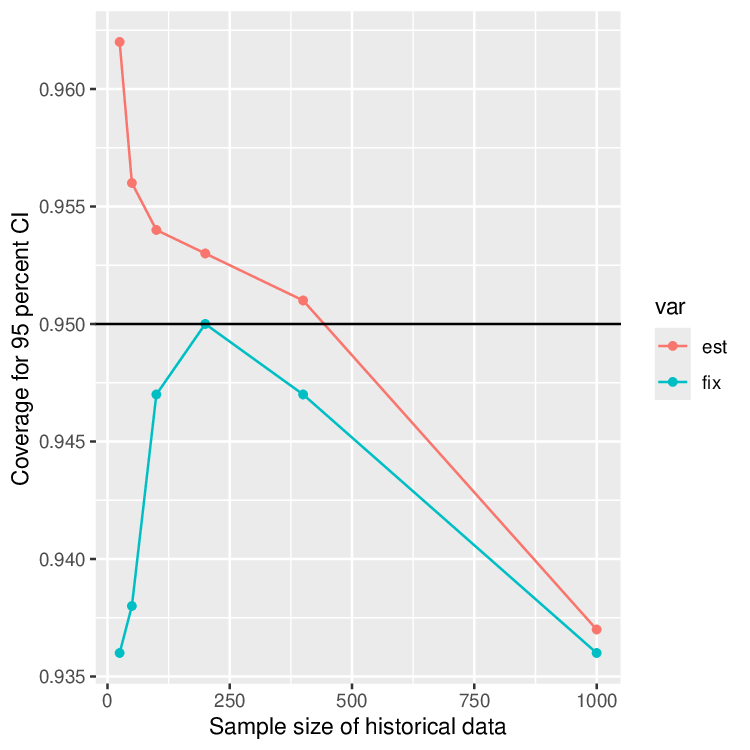}
    \caption{Plots of the coverage probability of the 95\% CI for $\beta_A$ in the PROCOVA model \eqref{eq: PROCOVA linear} over 1000 simulations for Scenario A-9.
     ``fix'' represents $e^\top\hat V_{\text{fix}}e$ with $e=(1,0,0)^\top$ for $\beta_0$, with $e=(0,1,0)^\top$ for $\beta_A$ and $e=(0,0,1)^\top$ for $\beta_1$.
    ``est'' represents $e^\top\hat V_{\text{est}}e$ with $e=(1,0,0)^\top$ for $\beta_0$, with $e=(0,1,0)^\top$ for $\beta_A$ and $e=(0,0,1)^\top$ for $\beta_1$.
    The sample size of trial data is $n=100$. 
    The x-axis represents the sample size of historical data $\tilde n$.
    The y-axis represents the coverage probability which is the proportion of 1000 simulations in which the 95\% CI using each variance estimator includes the true value.}
\end{figure}
\clearpage
\begin{figure}[h]
    \centering
    \includegraphics[width=\linewidth]{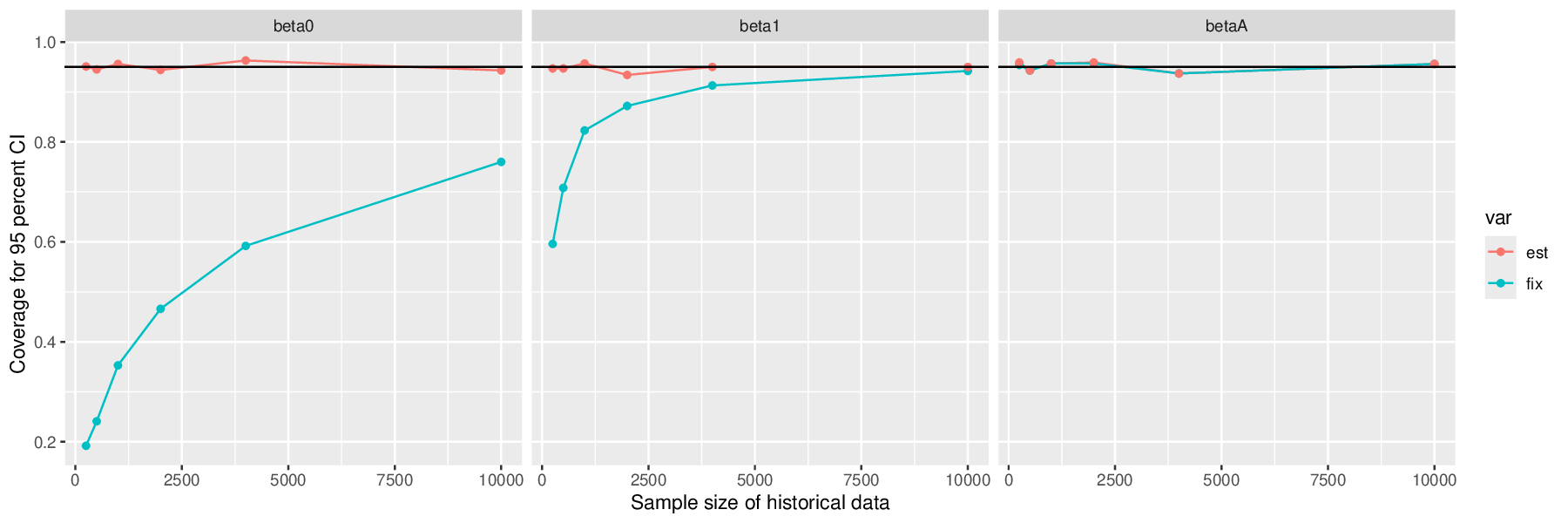}
    \caption{Plots of the coverage probability of 95\% CI over 1000 simulations for Scenario A-9.
    ``beta0'', ``betaA'', ``beta1'' represent the intercept $\beta_0$, the coefficient for the treatment assignment $\beta_A$, the coefficient for the prognostic score $\beta_1$ in the PROCOVA model \eqref{eq: PROCOVA linear}, respectively.
     ``fix'' represents $e^\top\hat V_{\text{fix}}e$ with $e=(1,0,0)^\top$ for $\beta_0$, with $e=(0,1,0)^\top$ for $\beta_A$ and $e=(0,0,1)^\top$ for $\beta_1$.
    ``est'' represents $e^\top\hat V_{\text{est}}e$ with $e=(1,0,0)^\top$ for $\beta_0$, with $e=(0,1,0)^\top$ for $\beta_A$ and $e=(0,0,1)^\top$ for $\beta_1$.
    The sample size of trial data is $n=1000$. 
    The x-axis represents the sample size of historical data $\tilde n$.
    The y-axis represents the coverage probability which is the proportion of 1000 simulations in which the 95\% CI using each variance estimator includes the true value.}
\end{figure}

\clearpage
\subsection{Scenario B-1}

\begin{figure}[h]
    \centering
    \includegraphics[width=\linewidth]{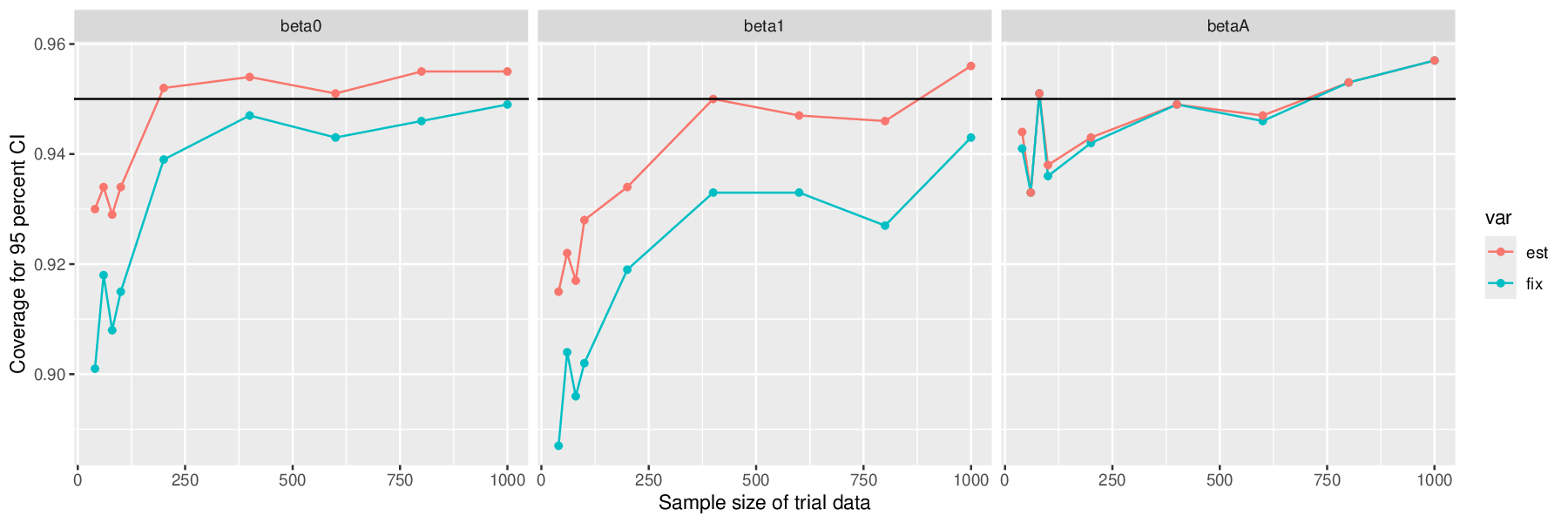}
    \caption{Plots of the coverage probability of 95\% CI over 1000 simulations for Scenario B-1.
    ``beta0'', ``betaA'', ``beta1'' represent the intercept $\beta_0$, the coefficient for the treatment assignment $\beta_A$, the coefficient for the prognostic score $\beta_1$ in the PROCOVA model \eqref{eq: PROCOVA linear}, respectively.
     ``fix'' represents $e^\top\hat V_{\text{fix}}e$ with $e=(1,0,0)^\top$ for $\beta_0$, with $e=(0,1,0)^\top$ for $\beta_A$ and $e=(0,0,1)^\top$ for $\beta_1$.
    ``est'' represents $e^\top\hat V_{\text{est}}e$ with $e=(1,0,0)^\top$ for $\beta_0$, with $e=(0,1,0)^\top$ for $\beta_A$ and $e=(0,0,1)^\top$ for $\beta_1$.
    The x-axis represents the sample size of trial data $n$.
    The sample size of historical data is $\tilde n = 10n$. 
    The y-axis represents the coverage probability which is the proportion of 1000 simulations in which the 95\% CI using each variance estimator includes the true value.
    }
\end{figure}

\clearpage
\begin{figure}[h]
    \centering
    \includegraphics[width=0.4\linewidth]{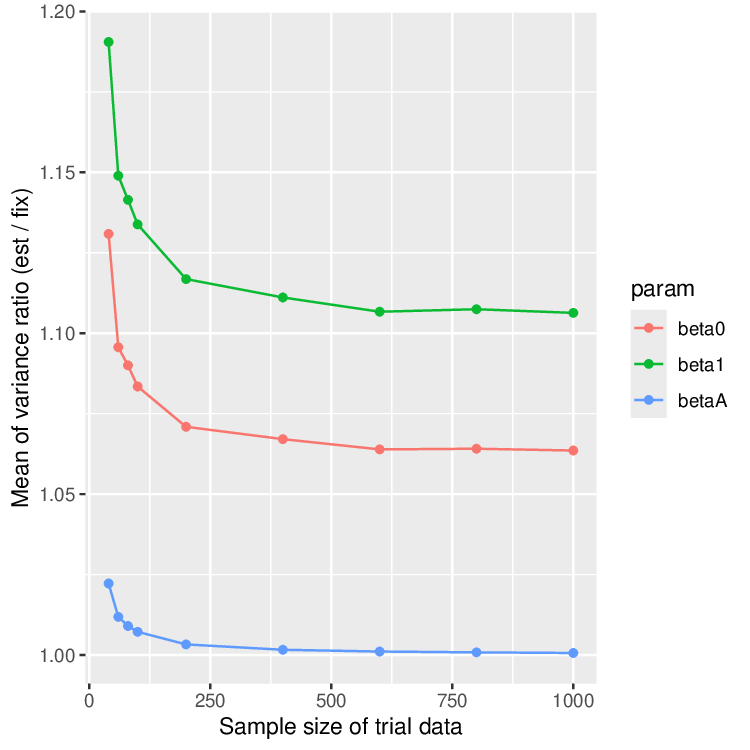}
    \caption{Plots of the mean of the ratio of two variance estimators over 1000 simulations for Scenario B-1.
    ``beta0'', ``betaA'', ``beta1'' represent the intercept $\beta_0$, the coefficient for the treatment assignment $\beta_A$, the coefficient for the prognostic score $\beta_1$ in the PROCOVA model \eqref{eq: PROCOVA linear}, respectively.
    The x-axis represents the sample size of trial data $n$.
    The sample size of historical data is $\tilde n = 10n$. 
    The y-axis represents the mean of the ratio of two variance estimators, i.e., $e^\top\hat V_{\text{est}}e/e^\top\hat V_{\text{fix}}e$ with $e=(1,0,0)^\top$ for $\beta_0$, with $e=(0,1,0)^\top$ for $\beta_A$ and $e=(0,0,1)^\top$ for $\beta_1$, over 1000 simulations.}
\end{figure}
\clearpage
\begin{figure}[h]
    \centering
    \includegraphics[width=0.4\linewidth]{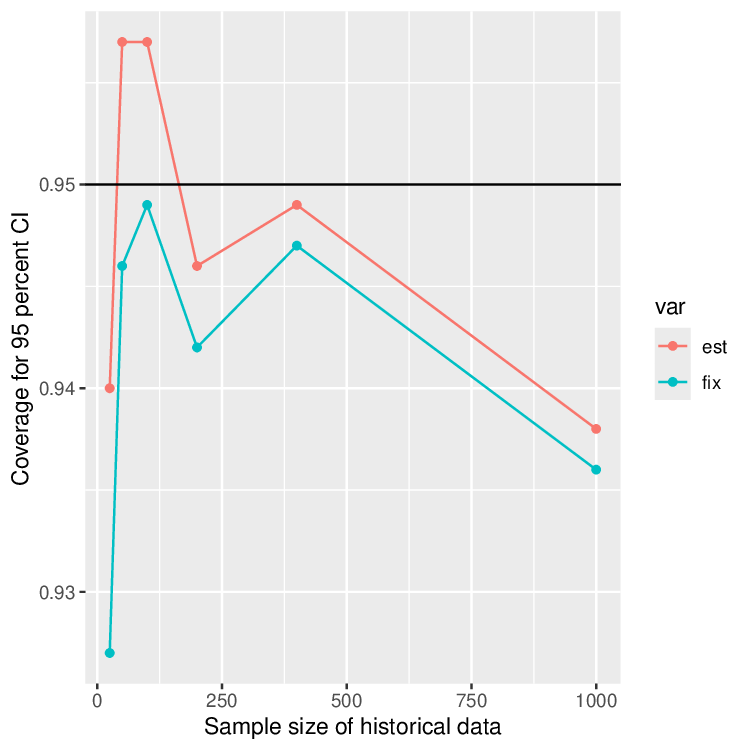}
    \caption{Plots of the coverage probability of the 95\% CI for $\beta_A$ in the PROCOVA model \eqref{eq: PROCOVA linear} over 1000 simulations for Scenario B-1.
     ``fix'' represents $e^\top\hat V_{\text{fix}}e$ with $e=(1,0,0)^\top$ for $\beta_0$, with $e=(0,1,0)^\top$ for $\beta_A$ and $e=(0,0,1)^\top$ for $\beta_1$.
    ``est'' represents $e^\top\hat V_{\text{est}}e$ with $e=(1,0,0)^\top$ for $\beta_0$, with $e=(0,1,0)^\top$ for $\beta_A$ and $e=(0,0,1)^\top$ for $\beta_1$.
    The sample size of trial data is $n=100$. 
    The x-axis represents the sample size of historical data $\tilde n$.
    The y-axis represents the coverage probability which is the proportion of 1000 simulations in which the 95\% CI using each variance estimator includes the true value.}
\end{figure}
\clearpage
\begin{figure}[h]
    \centering
    \includegraphics[width=\linewidth]{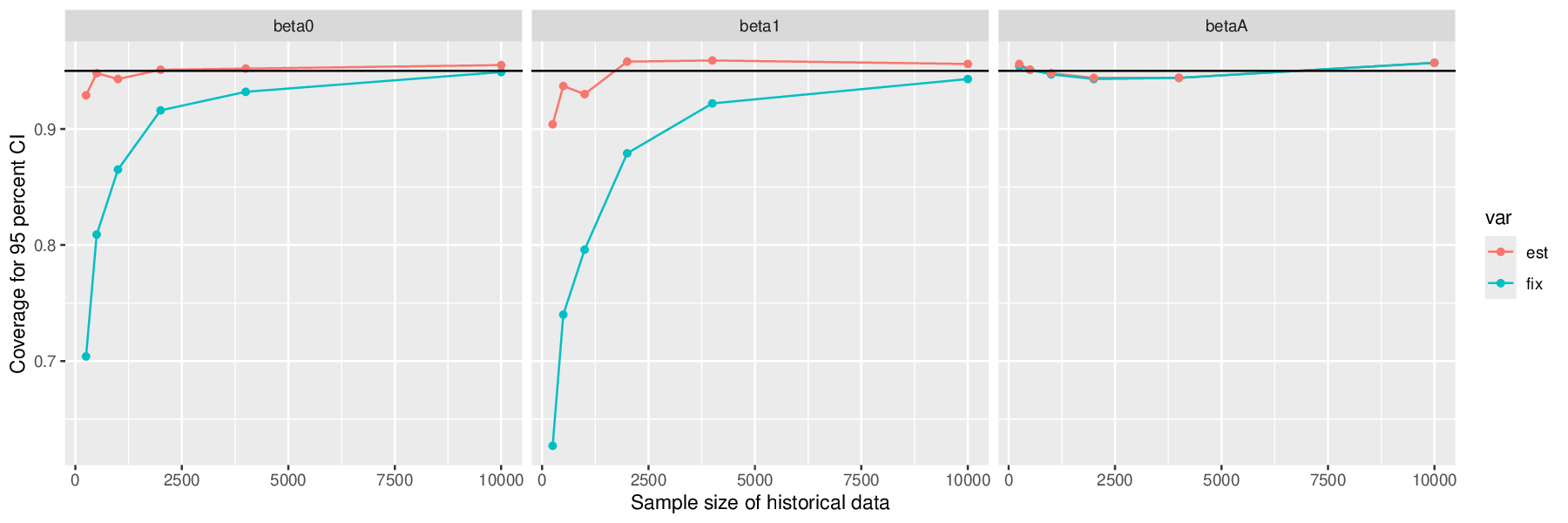}
    \caption{Plots of the coverage probability of 95\% CI over 1000 simulations for Scenario B-1.
    ``beta0'', ``betaA'', ``beta1'' represent the intercept $\beta_0$, the coefficient for the treatment assignment $\beta_A$, the coefficient for the prognostic score $\beta_1$ in the PROCOVA model \eqref{eq: PROCOVA linear}, respectively.
     ``fix'' represents $e^\top\hat V_{\text{fix}}e$ with $e=(1,0,0)^\top$ for $\beta_0$, with $e=(0,1,0)^\top$ for $\beta_A$ and $e=(0,0,1)^\top$ for $\beta_1$.
    ``est'' represents $e^\top\hat V_{\text{est}}e$ with $e=(1,0,0)^\top$ for $\beta_0$, with $e=(0,1,0)^\top$ for $\beta_A$ and $e=(0,0,1)^\top$ for $\beta_1$.
    The sample size of trial data is $n=1000$. 
    The x-axis represents the sample size of historical data $\tilde n$.
    The y-axis represents the coverage probability which is the proportion of 1000 simulations in which the 95\% CI using each variance estimator includes the true value.}
\end{figure}

\clearpage
\subsection{Scenario B-2}

\begin{figure}[h]
    \centering
    \includegraphics[width=\linewidth]{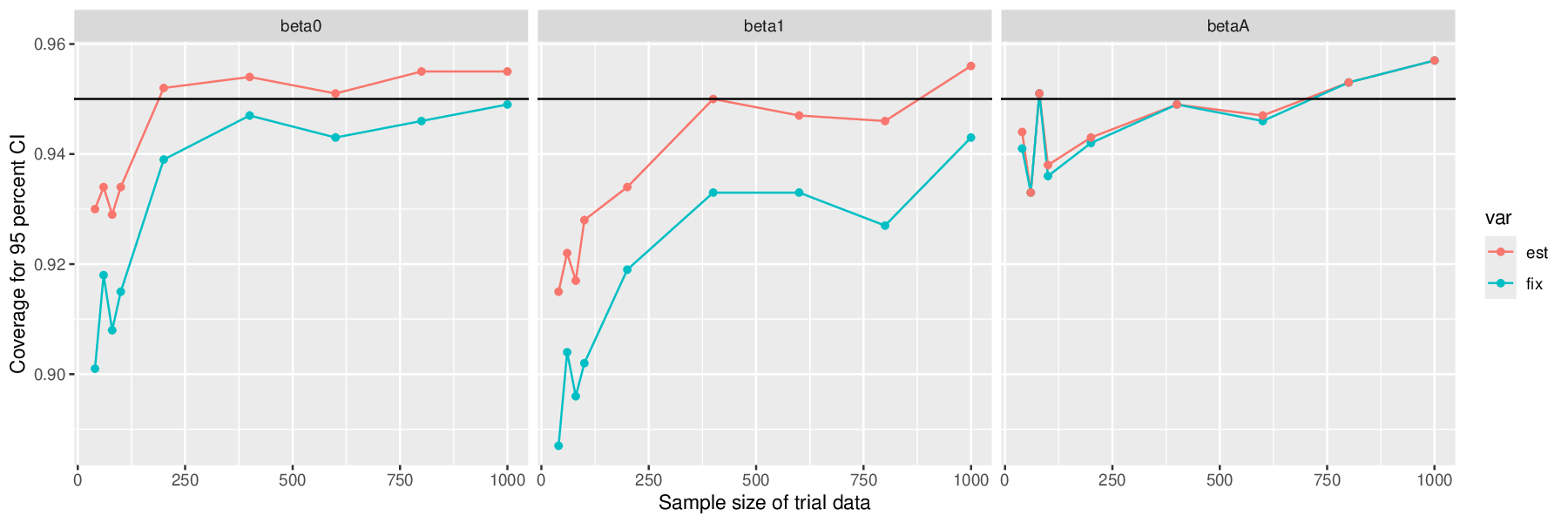}
    \caption{Plots of the coverage probability of 95\% CI over 1000 simulations for Scenario B-2.
    ``beta0'', ``betaA'', ``beta1'' represent the intercept $\beta_0$, the coefficient for the treatment assignment $\beta_A$, the coefficient for the prognostic score $\beta_1$ in the PROCOVA model \eqref{eq: PROCOVA linear}, respectively.
     ``fix'' represents $e^\top\hat V_{\text{fix}}e$ with $e=(1,0,0)^\top$ for $\beta_0$, with $e=(0,1,0)^\top$ for $\beta_A$ and $e=(0,0,1)^\top$ for $\beta_1$.
    ``est'' represents $e^\top\hat V_{\text{est}}e$ with $e=(1,0,0)^\top$ for $\beta_0$, with $e=(0,1,0)^\top$ for $\beta_A$ and $e=(0,0,1)^\top$ for $\beta_1$.
    The x-axis represents the sample size of trial data $n$.
    The sample size of historical data is $\tilde n = 10n$. 
    The y-axis represents the coverage probability which is the proportion of 1000 simulations in which the 95\% CI using each variance estimator includes the true value.
    }
\end{figure}

\clearpage
\begin{figure}[h]
    \centering
    \includegraphics[width=0.4\linewidth]{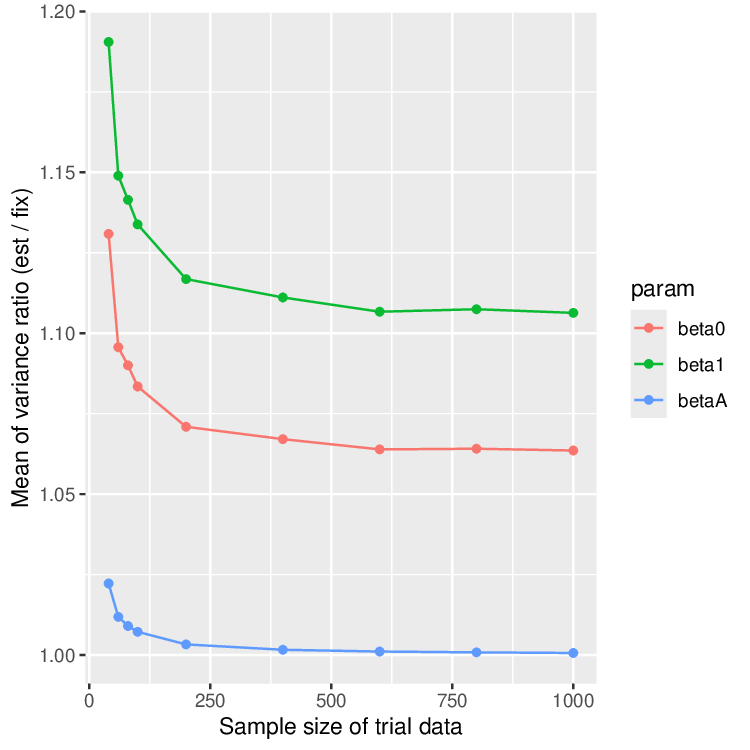}
    \caption{Plots of the mean of the ratio of two variance estimators over 1000 simulations for Scenario B-2.
    ``beta0'', ``betaA'', ``beta1'' represent the intercept $\beta_0$, the coefficient for the treatment assignment $\beta_A$, the coefficient for the prognostic score $\beta_1$ in the PROCOVA model \eqref{eq: PROCOVA linear}, respectively.
    The x-axis represents the sample size of trial data $n$.
    The sample size of historical data is $\tilde n = 10n$. 
    The y-axis represents the mean of the ratio of two variance estimators, i.e., $e^\top\hat V_{\text{est}}e/e^\top\hat V_{\text{fix}}e$ with $e=(1,0,0)^\top$ for $\beta_0$, with $e=(0,1,0)^\top$ for $\beta_A$ and $e=(0,0,1)^\top$ for $\beta_1$, over 1000 simulations.}
\end{figure}
\clearpage
\begin{figure}[h]
    \centering
    \includegraphics[width=0.4\linewidth]{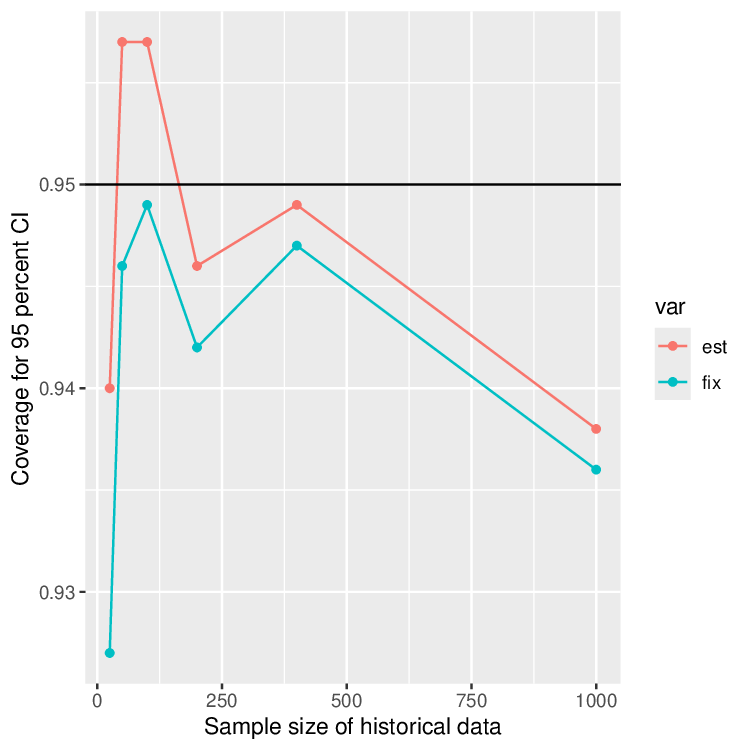}
    \caption{Plots of the coverage probability of the 95\% CI for $\beta_A$ in the PROCOVA model \eqref{eq: PROCOVA linear} over 1000 simulations for Scenario B-2.
     ``fix'' represents $e^\top\hat V_{\text{fix}}e$ with $e=(1,0,0)^\top$ for $\beta_0$, with $e=(0,1,0)^\top$ for $\beta_A$ and $e=(0,0,1)^\top$ for $\beta_1$.
    ``est'' represents $e^\top\hat V_{\text{est}}e$ with $e=(1,0,0)^\top$ for $\beta_0$, with $e=(0,1,0)^\top$ for $\beta_A$ and $e=(0,0,1)^\top$ for $\beta_1$.
    The sample size of trial data is $n=100$. 
    The x-axis represents the sample size of historical data $\tilde n$.
    The y-axis represents the coverage probability which is the proportion of 1000 simulations in which the 95\% CI using each variance estimator includes the true value.}
\end{figure}
\clearpage
\begin{figure}[h]
    \centering
    \includegraphics[width=\linewidth]{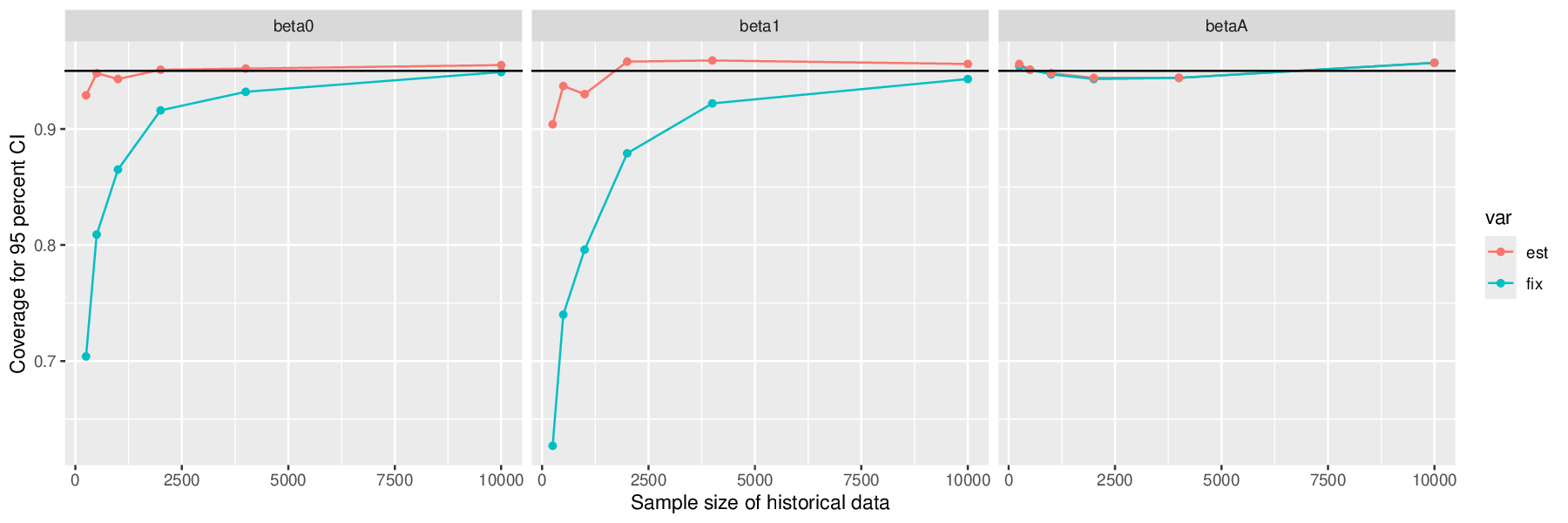}
    \caption{Plots of the coverage probability of 95\% CI over 1000 simulations for Scenario B-2.
    ``beta0'', ``betaA'', ``beta1'' represent the intercept $\beta_0$, the coefficient for the treatment assignment $\beta_A$, the coefficient for the prognostic score $\beta_1$ in the PROCOVA model \eqref{eq: PROCOVA linear}, respectively.
     ``fix'' represents $e^\top\hat V_{\text{fix}}e$ with $e=(1,0,0)^\top$ for $\beta_0$, with $e=(0,1,0)^\top$ for $\beta_A$ and $e=(0,0,1)^\top$ for $\beta_1$.
    ``est'' represents $e^\top\hat V_{\text{est}}e$ with $e=(1,0,0)^\top$ for $\beta_0$, with $e=(0,1,0)^\top$ for $\beta_A$ and $e=(0,0,1)^\top$ for $\beta_1$.
    The sample size of trial data is $n=1000$. 
    The x-axis represents the sample size of historical data $\tilde n$.
    The y-axis represents the coverage probability which is the proportion of 1000 simulations in which the 95\% CI using each variance estimator includes the true value.}
\end{figure}

\clearpage
\subsection{Scenario B-3}

\begin{figure}[h]
    \centering
    \includegraphics[width=\linewidth]{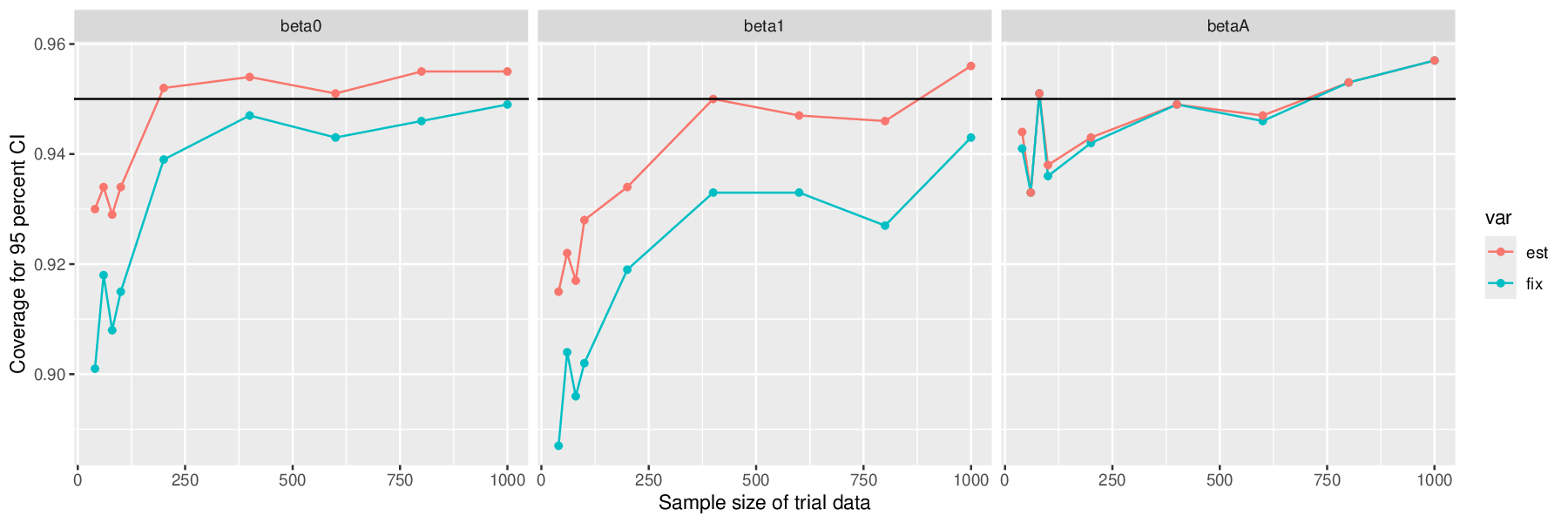}
    \caption{Plots of the coverage probability of 95\% CI over 1000 simulations for Scenario B-3.
    ``beta0'', ``betaA'', ``beta1'' represent the intercept $\beta_0$, the coefficient for the treatment assignment $\beta_A$, the coefficient for the prognostic score $\beta_1$ in the PROCOVA model \eqref{eq: PROCOVA linear}, respectively.
     ``fix'' represents $e^\top\hat V_{\text{fix}}e$ with $e=(1,0,0)^\top$ for $\beta_0$, with $e=(0,1,0)^\top$ for $\beta_A$ and $e=(0,0,1)^\top$ for $\beta_1$.
    ``est'' represents $e^\top\hat V_{\text{est}}e$ with $e=(1,0,0)^\top$ for $\beta_0$, with $e=(0,1,0)^\top$ for $\beta_A$ and $e=(0,0,1)^\top$ for $\beta_1$.
    The x-axis represents the sample size of trial data $n$.
    The sample size of historical data is $\tilde n = 10n$. 
    The y-axis represents the coverage probability which is the proportion of 1000 simulations in which the 95\% CI using each variance estimator includes the true value.
    }
\end{figure}

\clearpage
\begin{figure}[h]
    \centering
    \includegraphics[width=0.4\linewidth]{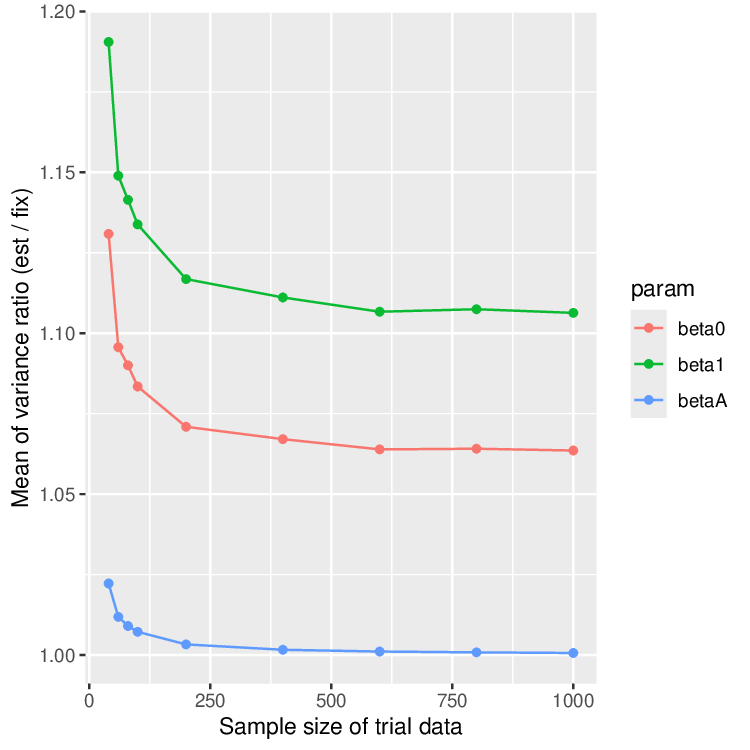}
    \caption{Plots of the mean of the ratio of two variance estimators over 1000 simulations for Scenario B-3.
    ``beta0'', ``betaA'', ``beta1'' represent the intercept $\beta_0$, the coefficient for the treatment assignment $\beta_A$, the coefficient for the prognostic score $\beta_1$ in the PROCOVA model \eqref{eq: PROCOVA linear}, respectively.
    The x-axis represents the sample size of trial data $n$.
    The sample size of historical data is $\tilde n = 10n$. 
    The y-axis represents the mean of the ratio of two variance estimators, i.e., $e^\top\hat V_{\text{est}}e/e^\top\hat V_{\text{fix}}e$ with $e=(1,0,0)^\top$ for $\beta_0$, with $e=(0,1,0)^\top$ for $\beta_A$ and $e=(0,0,1)^\top$ for $\beta_1$, over 1000 simulations.}
\end{figure}
\clearpage
\begin{figure}[h]
    \centering
    \includegraphics[width=0.4\linewidth]{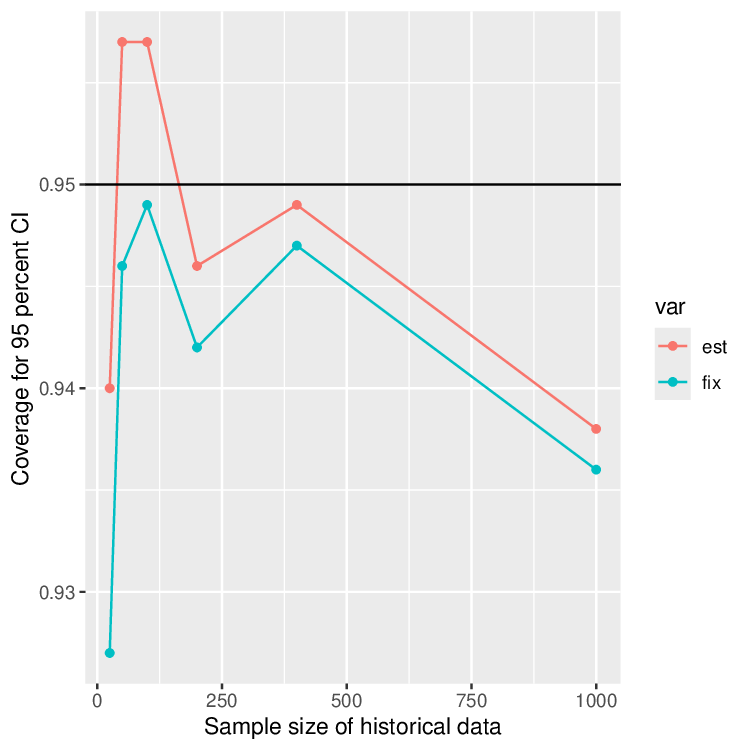}
    \caption{Plots of the coverage probability of the 95\% CI for $\beta_A$ in the PROCOVA model \eqref{eq: PROCOVA linear} over 1000 simulations for Scenario B-3.
     ``fix'' represents $e^\top\hat V_{\text{fix}}e$ with $e=(1,0,0)^\top$ for $\beta_0$, with $e=(0,1,0)^\top$ for $\beta_A$ and $e=(0,0,1)^\top$ for $\beta_1$.
    ``est'' represents $e^\top\hat V_{\text{est}}e$ with $e=(1,0,0)^\top$ for $\beta_0$, with $e=(0,1,0)^\top$ for $\beta_A$ and $e=(0,0,1)^\top$ for $\beta_1$.
    The sample size of trial data is $n=100$. 
    The x-axis represents the sample size of historical data $\tilde n$.
    The y-axis represents the coverage probability which is the proportion of 1000 simulations in which the 95\% CI using each variance estimator includes the true value.}
\end{figure}
\clearpage
\begin{figure}[h]
    \centering
    \includegraphics[width=\linewidth]{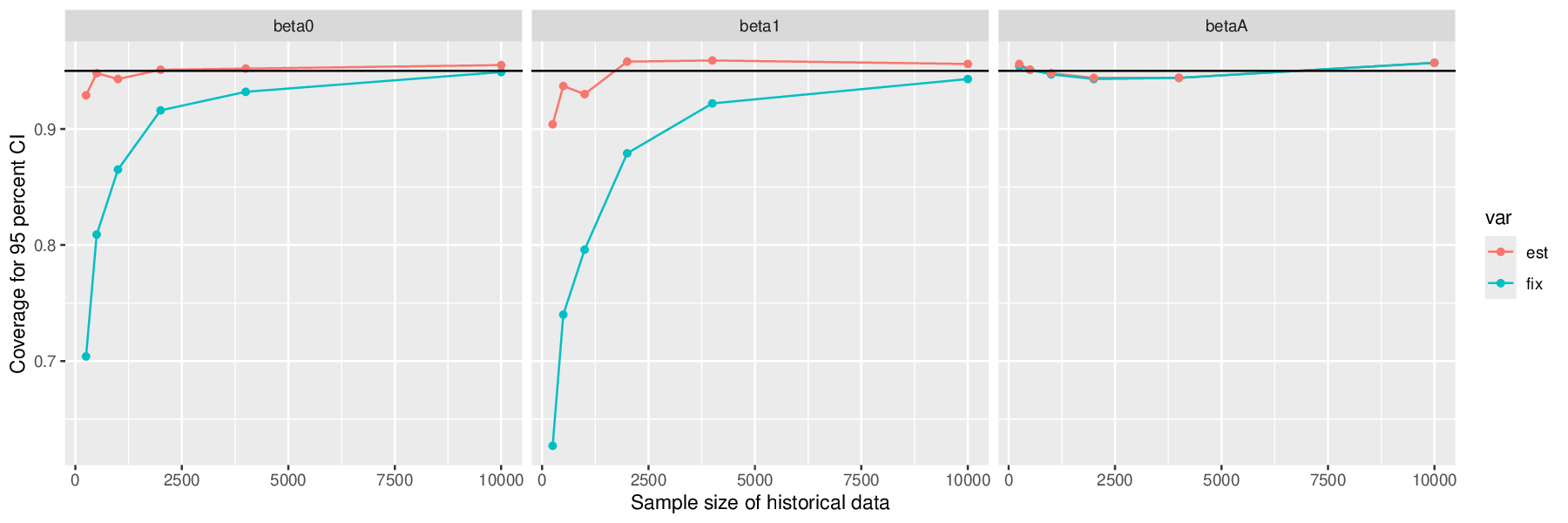}
    \caption{Plots of the coverage probability of 95\% CI over 1000 simulations for Scenario B-3.
    ``beta0'', ``betaA'', ``beta1'' represent the intercept $\beta_0$, the coefficient for the treatment assignment $\beta_A$, the coefficient for the prognostic score $\beta_1$ in the PROCOVA model \eqref{eq: PROCOVA linear}, respectively.
     ``fix'' represents $e^\top\hat V_{\text{fix}}e$ with $e=(1,0,0)^\top$ for $\beta_0$, with $e=(0,1,0)^\top$ for $\beta_A$ and $e=(0,0,1)^\top$ for $\beta_1$.
    ``est'' represents $e^\top\hat V_{\text{est}}e$ with $e=(1,0,0)^\top$ for $\beta_0$, with $e=(0,1,0)^\top$ for $\beta_A$ and $e=(0,0,1)^\top$ for $\beta_1$.
    The sample size of trial data is $n=1000$. 
    The x-axis represents the sample size of historical data $\tilde n$.
    The y-axis represents the coverage probability which is the proportion of 1000 simulations in which the 95\% CI using each variance estimator includes the true value.}
\end{figure}

\clearpage
\subsection{Scenario B-4}

\begin{figure}[h]
    \centering
    \includegraphics[width=\linewidth]{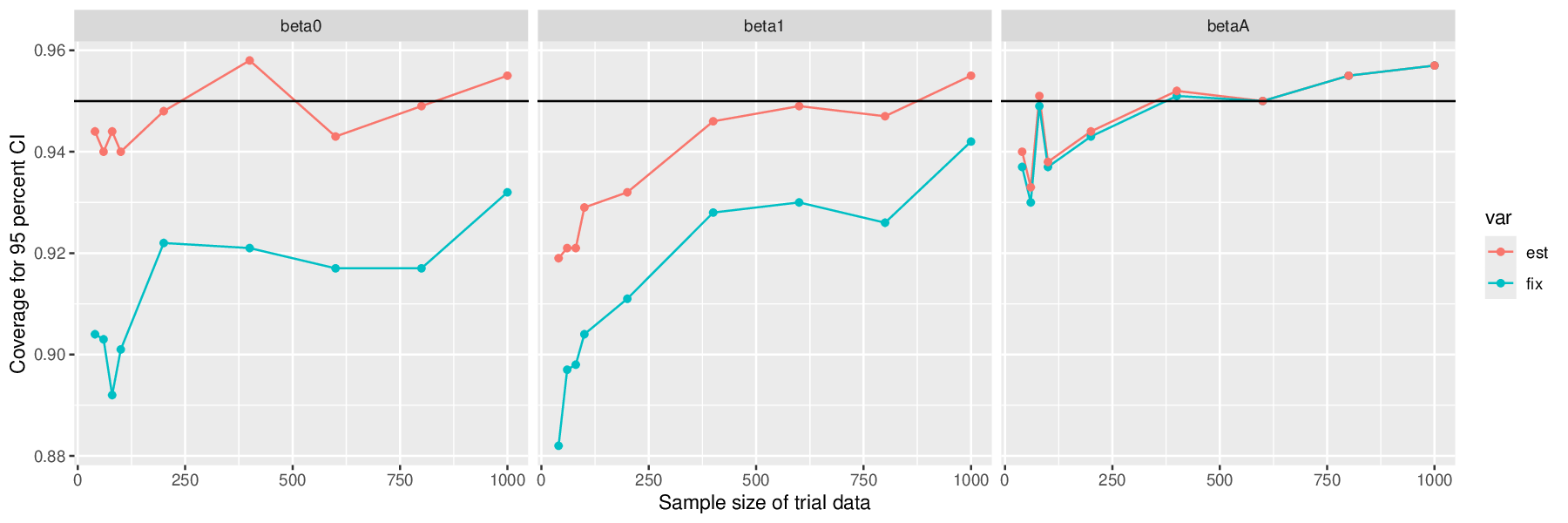}
    \caption{Plots of the coverage probability of 95\% CI over 1000 simulations for Scenario B-4.
    ``beta0'', ``betaA'', ``beta1'' represent the intercept $\beta_0$, the coefficient for the treatment assignment $\beta_A$, the coefficient for the prognostic score $\beta_1$ in the PROCOVA model \eqref{eq: PROCOVA linear}, respectively.
     ``fix'' represents $e^\top\hat V_{\text{fix}}e$ with $e=(1,0,0)^\top$ for $\beta_0$, with $e=(0,1,0)^\top$ for $\beta_A$ and $e=(0,0,1)^\top$ for $\beta_1$.
    ``est'' represents $e^\top\hat V_{\text{est}}e$ with $e=(1,0,0)^\top$ for $\beta_0$, with $e=(0,1,0)^\top$ for $\beta_A$ and $e=(0,0,1)^\top$ for $\beta_1$.
    The x-axis represents the sample size of trial data $n$.
    The sample size of historical data is $\tilde n = 10n$. 
    The y-axis represents the coverage probability which is the proportion of 1000 simulations in which the 95\% CI using each variance estimator includes the true value.
    }
\end{figure}

\clearpage
\begin{figure}[h]
    \centering
    \includegraphics[width=0.4\linewidth]{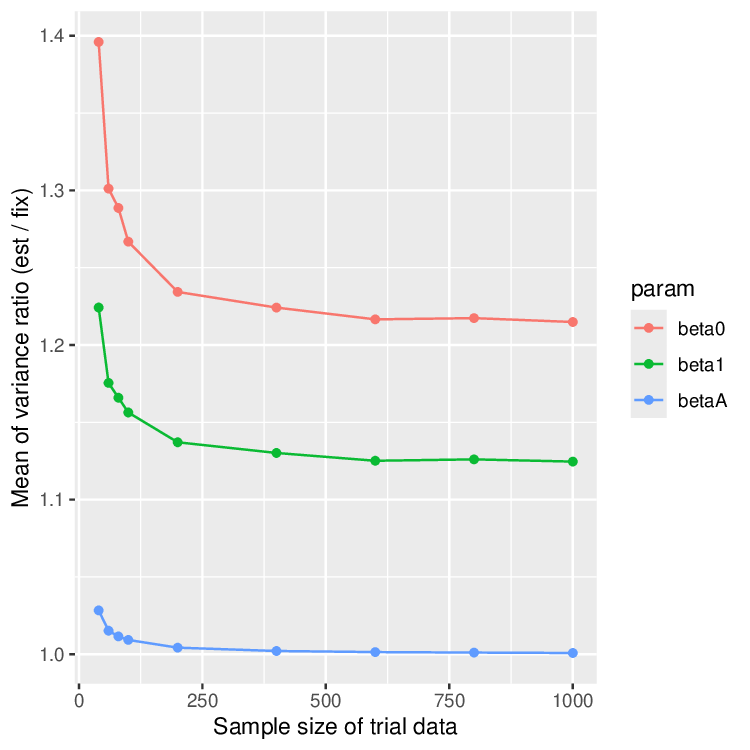}
    \caption{Plots of the mean of the ratio of two variance estimators over 1000 simulations for Scenario B-4.
    ``beta0'', ``betaA'', ``beta1'' represent the intercept $\beta_0$, the coefficient for the treatment assignment $\beta_A$, the coefficient for the prognostic score $\beta_1$ in the PROCOVA model \eqref{eq: PROCOVA linear}, respectively.
    The x-axis represents the sample size of trial data $n$.
    The sample size of historical data is $\tilde n = 10n$. 
    The y-axis represents the mean of the ratio of two variance estimators, i.e., $e^\top\hat V_{\text{est}}e/e^\top\hat V_{\text{fix}}e$ with $e=(1,0,0)^\top$ for $\beta_0$, with $e=(0,1,0)^\top$ for $\beta_A$ and $e=(0,0,1)^\top$ for $\beta_1$, over 1000 simulations.}
\end{figure}

\clearpage
\begin{figure}[h]
    \centering
    \includegraphics[width=0.4\linewidth]{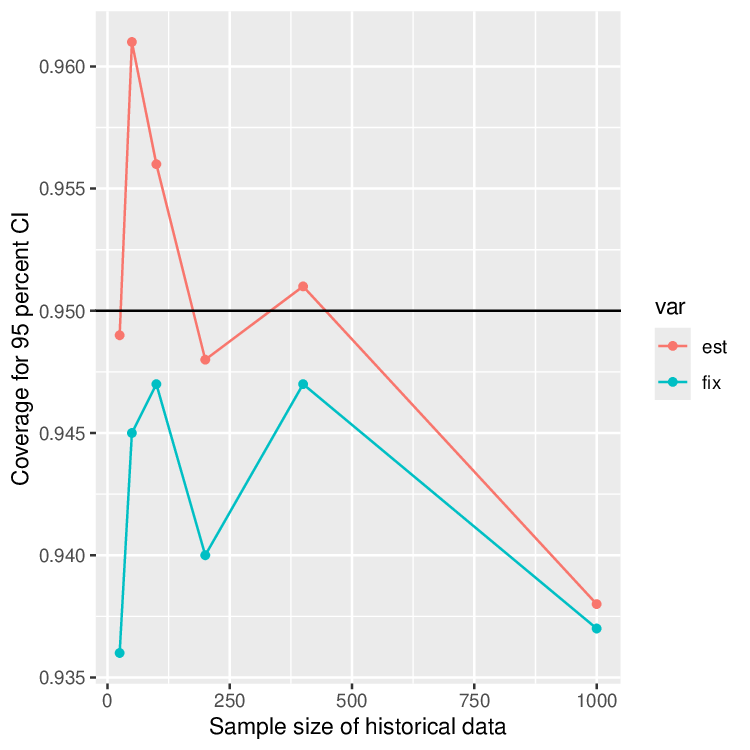}
    \caption{Plots of the coverage probability of the 95\% CI for $\beta_A$ in the PROCOVA model \eqref{eq: PROCOVA linear} over 1000 simulations for Scenario B-4.
     ``fix'' represents $e^\top\hat V_{\text{fix}}e$ with $e=(1,0,0)^\top$ for $\beta_0$, with $e=(0,1,0)^\top$ for $\beta_A$ and $e=(0,0,1)^\top$ for $\beta_1$.
    ``est'' represents $e^\top\hat V_{\text{est}}e$ with $e=(1,0,0)^\top$ for $\beta_0$, with $e=(0,1,0)^\top$ for $\beta_A$ and $e=(0,0,1)^\top$ for $\beta_1$.
    The sample size of trial data is $n=100$. 
    The x-axis represents the sample size of historical data $\tilde n$.
    The y-axis represents the coverage probability which is the proportion of 1000 simulations in which the 95\% CI using each variance estimator includes the true value.}
\end{figure}

\clearpage
\begin{figure}[h]
    \centering
    \includegraphics[width=\linewidth]{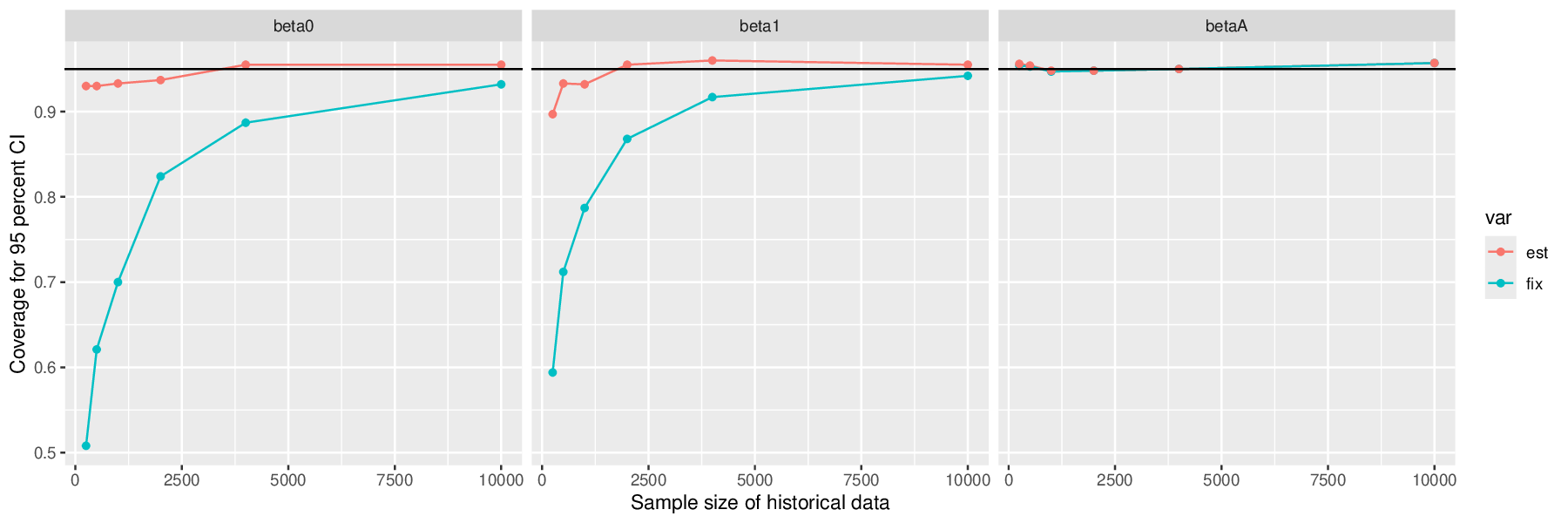}
    \caption{Plots of the coverage probability of 95\% CI over 1000 simulations for Scenario B-4.
    ``beta0'', ``betaA'', ``beta1'' represent the intercept $\beta_0$, the coefficient for the treatment assignment $\beta_A$, the coefficient for the prognostic score $\beta_1$ in the PROCOVA model \eqref{eq: PROCOVA linear}, respectively.
     ``fix'' represents $e^\top\hat V_{\text{fix}}e$ with $e=(1,0,0)^\top$ for $\beta_0$, with $e=(0,1,0)^\top$ for $\beta_A$ and $e=(0,0,1)^\top$ for $\beta_1$.
    ``est'' represents $e^\top\hat V_{\text{est}}e$ with $e=(1,0,0)^\top$ for $\beta_0$, with $e=(0,1,0)^\top$ for $\beta_A$ and $e=(0,0,1)^\top$ for $\beta_1$.
    The sample size of trial data is $n=1000$. 
    The x-axis represents the sample size of historical data $\tilde n$.
    The y-axis represents the coverage probability which is the proportion of 1000 simulations in which the 95\% CI using each variance estimator includes the true value.}
\end{figure}

\clearpage
\subsection{Scenario B-5}

\begin{figure}[h]
    \centering
    \includegraphics[width=\linewidth]{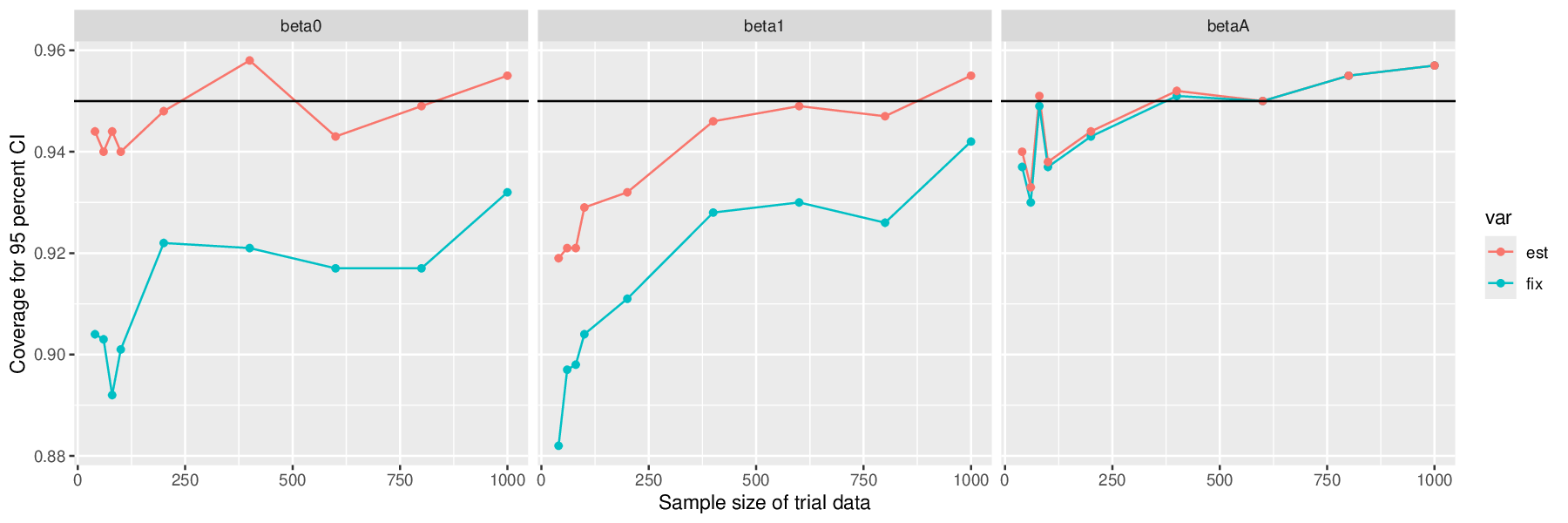}
    \caption{Plots of the coverage probability of 95\% CI over 1000 simulations for Scenario B-5.
    ``beta0'', ``betaA'', ``beta1'' represent the intercept $\beta_0$, the coefficient for the treatment assignment $\beta_A$, the coefficient for the prognostic score $\beta_1$ in the PROCOVA model \eqref{eq: PROCOVA linear}, respectively.
     ``fix'' represents $e^\top\hat V_{\text{fix}}e$ with $e=(1,0,0)^\top$ for $\beta_0$, with $e=(0,1,0)^\top$ for $\beta_A$ and $e=(0,0,1)^\top$ for $\beta_1$.
    ``est'' represents $e^\top\hat V_{\text{est}}e$ with $e=(1,0,0)^\top$ for $\beta_0$, with $e=(0,1,0)^\top$ for $\beta_A$ and $e=(0,0,1)^\top$ for $\beta_1$.
    The x-axis represents the sample size of trial data $n$.
    The sample size of historical data is $\tilde n = 10n$. 
    The y-axis represents the coverage probability which is the proportion of 1000 simulations in which the 95\% CI using each variance estimator includes the true value.
    }
\end{figure}

\clearpage
\begin{figure}[h]
    \centering
    \includegraphics[width=0.4\linewidth]{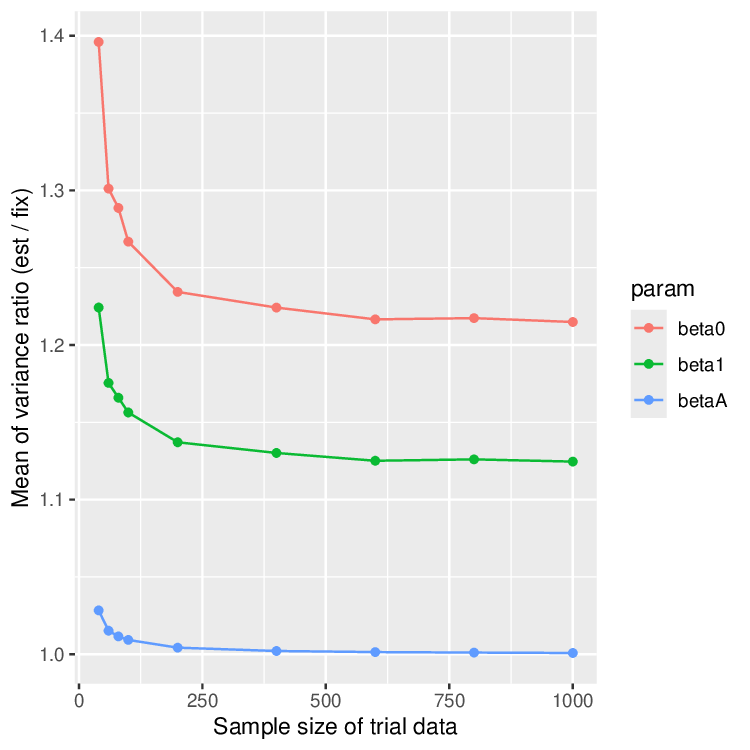}
    \caption{Plots of the mean of the ratio of two variance estimators over 1000 simulations for Scenario B-5.
    ``beta0'', ``betaA'', ``beta1'' represent the intercept $\beta_0$, the coefficient for the treatment assignment $\beta_A$, the coefficient for the prognostic score $\beta_1$ in the PROCOVA model \eqref{eq: PROCOVA linear}, respectively.
    The x-axis represents the sample size of trial data $n$.
    The sample size of historical data is $\tilde n = 10n$. 
    The y-axis represents the mean of the ratio of two variance estimators, i.e., $e^\top\hat V_{\text{est}}e/e^\top\hat V_{\text{fix}}e$ with $e=(1,0,0)^\top$ for $\beta_0$, with $e=(0,1,0)^\top$ for $\beta_A$ and $e=(0,0,1)^\top$ for $\beta_1$, over 1000 simulations.}
\end{figure}

\clearpage
\begin{figure}[h]
    \centering
    \includegraphics[width=0.4\linewidth]{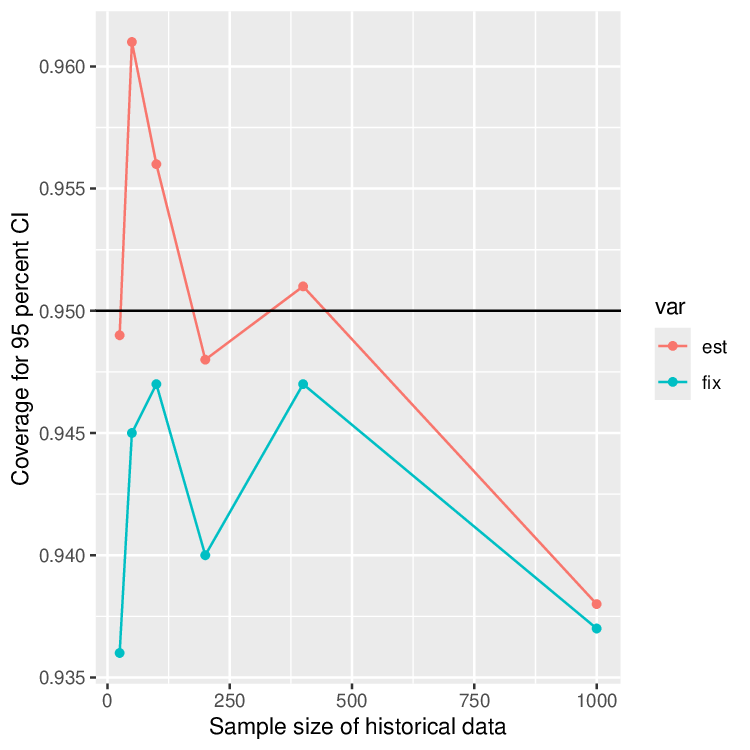}
    \caption{Plots of the coverage probability of the 95\% CI for $\beta_A$ in the PROCOVA model \eqref{eq: PROCOVA linear} over 1000 simulations for Scenario B-5.
     ``fix'' represents $e^\top\hat V_{\text{fix}}e$ with $e=(1,0,0)^\top$ for $\beta_0$, with $e=(0,1,0)^\top$ for $\beta_A$ and $e=(0,0,1)^\top$ for $\beta_1$.
    ``est'' represents $e^\top\hat V_{\text{est}}e$ with $e=(1,0,0)^\top$ for $\beta_0$, with $e=(0,1,0)^\top$ for $\beta_A$ and $e=(0,0,1)^\top$ for $\beta_1$.
    The sample size of trial data is $n=100$. 
    The x-axis represents the sample size of historical data $\tilde n$.
    The y-axis represents the coverage probability which is the proportion of 1000 simulations in which the 95\% CI using each variance estimator includes the true value.}
\end{figure}

\clearpage
\begin{figure}[h]
    \centering
    \includegraphics[width=\linewidth]{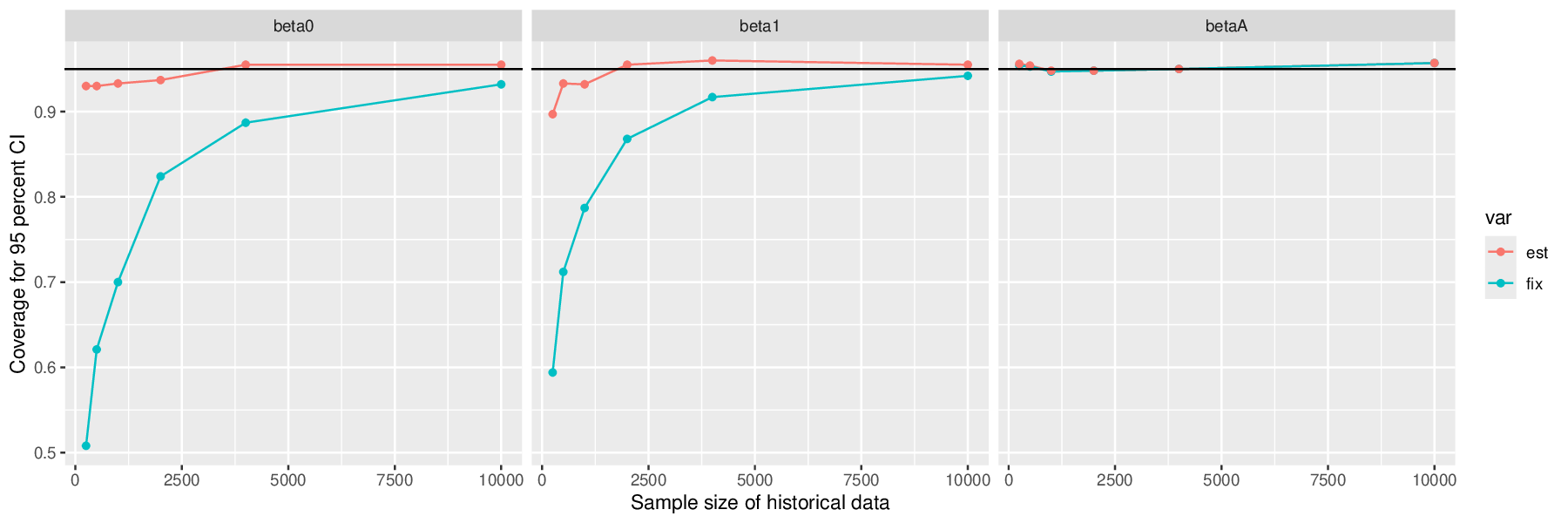}
    \caption{Plots of the coverage probability of 95\% CI over 1000 simulations for Scenario B-5.
    ``beta0'', ``betaA'', ``beta1'' represent the intercept $\beta_0$, the coefficient for the treatment assignment $\beta_A$, the coefficient for the prognostic score $\beta_1$ in the PROCOVA model \eqref{eq: PROCOVA linear}, respectively.
     ``fix'' represents $e^\top\hat V_{\text{fix}}e$ with $e=(1,0,0)^\top$ for $\beta_0$, with $e=(0,1,0)^\top$ for $\beta_A$ and $e=(0,0,1)^\top$ for $\beta_1$.
    ``est'' represents $e^\top\hat V_{\text{est}}e$ with $e=(1,0,0)^\top$ for $\beta_0$, with $e=(0,1,0)^\top$ for $\beta_A$ and $e=(0,0,1)^\top$ for $\beta_1$.
    The sample size of trial data is $n=1000$. 
    The x-axis represents the sample size of historical data $\tilde n$.
    The y-axis represents the coverage probability which is the proportion of 1000 simulations in which the 95\% CI using each variance estimator includes the true value.}
\end{figure}

\clearpage
\subsection{Scenario B-6}

\begin{figure}[h]
    \centering
    \includegraphics[width=\linewidth]{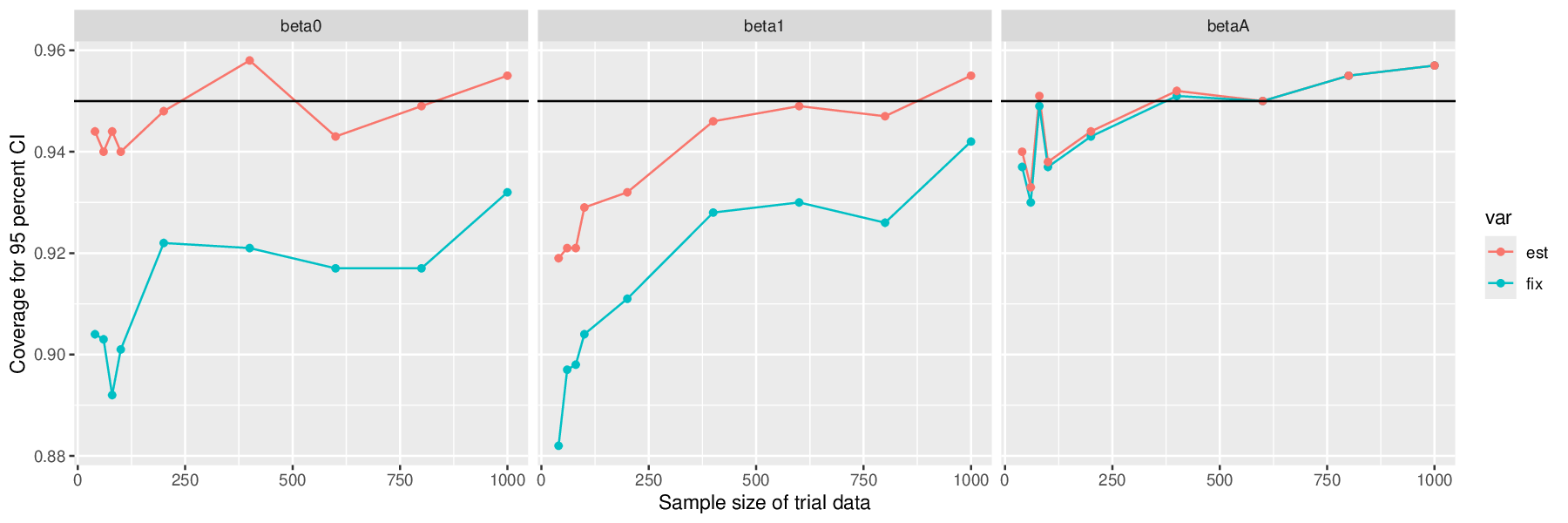}
    \caption{Plots of the coverage probability of 95\% CI over 1000 simulations for Scenario B-6.
    ``beta0'', ``betaA'', ``beta1'' represent the intercept $\beta_0$, the coefficient for the treatment assignment $\beta_A$, the coefficient for the prognostic score $\beta_1$ in the PROCOVA model \eqref{eq: PROCOVA linear}, respectively.
     ``fix'' represents $e^\top\hat V_{\text{fix}}e$ with $e=(1,0,0)^\top$ for $\beta_0$, with $e=(0,1,0)^\top$ for $\beta_A$ and $e=(0,0,1)^\top$ for $\beta_1$.
    ``est'' represents $e^\top\hat V_{\text{est}}e$ with $e=(1,0,0)^\top$ for $\beta_0$, with $e=(0,1,0)^\top$ for $\beta_A$ and $e=(0,0,1)^\top$ for $\beta_1$.
    The x-axis represents the sample size of trial data $n$.
    The sample size of historical data is $\tilde n = 10n$. 
    The y-axis represents the coverage probability which is the proportion of 1000 simulations in which the 95\% CI using each variance estimator includes the true value.
    }
\end{figure}

\clearpage
\begin{figure}[h]
    \centering
    \includegraphics[width=0.4\linewidth]{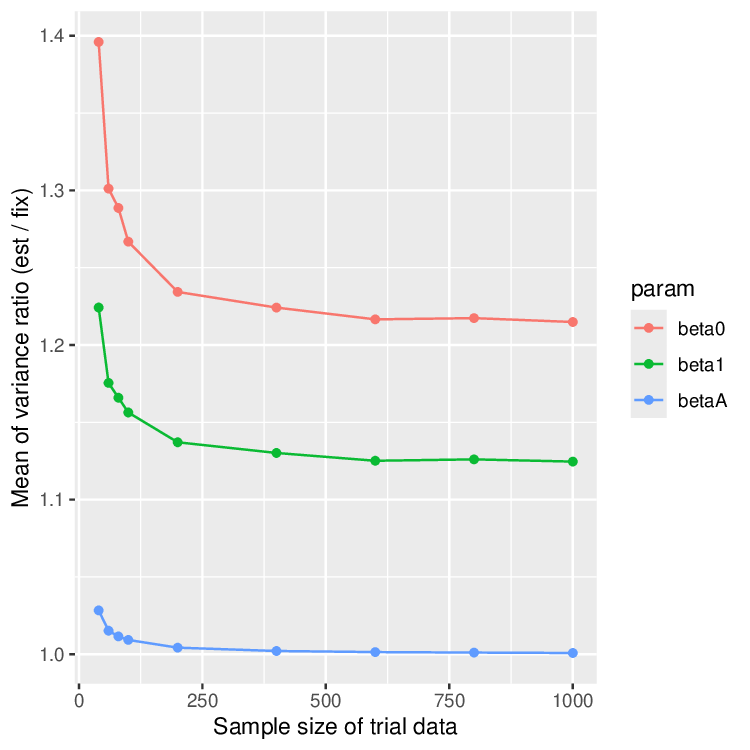}
    \caption{Plots of the mean of the ratio of two variance estimators over 1000 simulations for Scenario B-6.
    ``beta0'', ``betaA'', ``beta1'' represent the intercept $\beta_0$, the coefficient for the treatment assignment $\beta_A$, the coefficient for the prognostic score $\beta_1$ in the PROCOVA model \eqref{eq: PROCOVA linear}, respectively.
    The x-axis represents the sample size of trial data $n$.
    The sample size of historical data is $\tilde n = 10n$. 
    The y-axis represents the mean of the ratio of two variance estimators, i.e., $e^\top\hat V_{\text{est}}e/e^\top\hat V_{\text{fix}}e$ with $e=(1,0,0)^\top$ for $\beta_0$, with $e=(0,1,0)^\top$ for $\beta_A$ and $e=(0,0,1)^\top$ for $\beta_1$, over 1000 simulations.}
\end{figure}

\clearpage
\begin{figure}[h]
    \centering
    \includegraphics[width=0.4\linewidth]{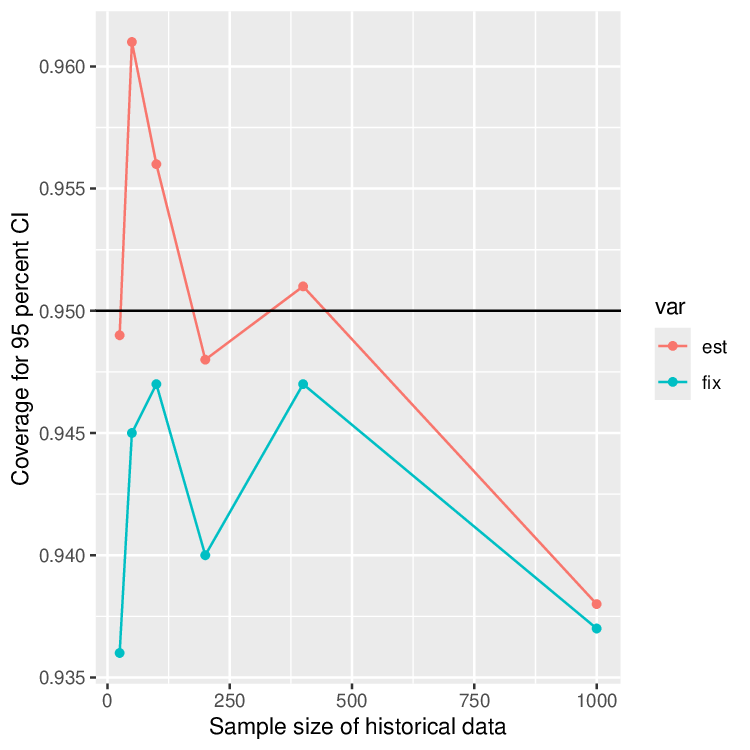}
    \caption{Plots of the coverage probability of the 95\% CI for $\beta_A$ in the PROCOVA model \eqref{eq: PROCOVA linear} over 1000 simulations for Scenario B-6.
     ``fix'' represents $e^\top\hat V_{\text{fix}}e$ with $e=(1,0,0)^\top$ for $\beta_0$, with $e=(0,1,0)^\top$ for $\beta_A$ and $e=(0,0,1)^\top$ for $\beta_1$.
    ``est'' represents $e^\top\hat V_{\text{est}}e$ with $e=(1,0,0)^\top$ for $\beta_0$, with $e=(0,1,0)^\top$ for $\beta_A$ and $e=(0,0,1)^\top$ for $\beta_1$.
    The sample size of trial data is $n=100$. 
    The x-axis represents the sample size of historical data $\tilde n$.
    The y-axis represents the coverage probability which is the proportion of 1000 simulations in which the 95\% CI using each variance estimator includes the true value.}
\end{figure}

\clearpage
\begin{figure}[h]
    \centering
    \includegraphics[width=\linewidth]{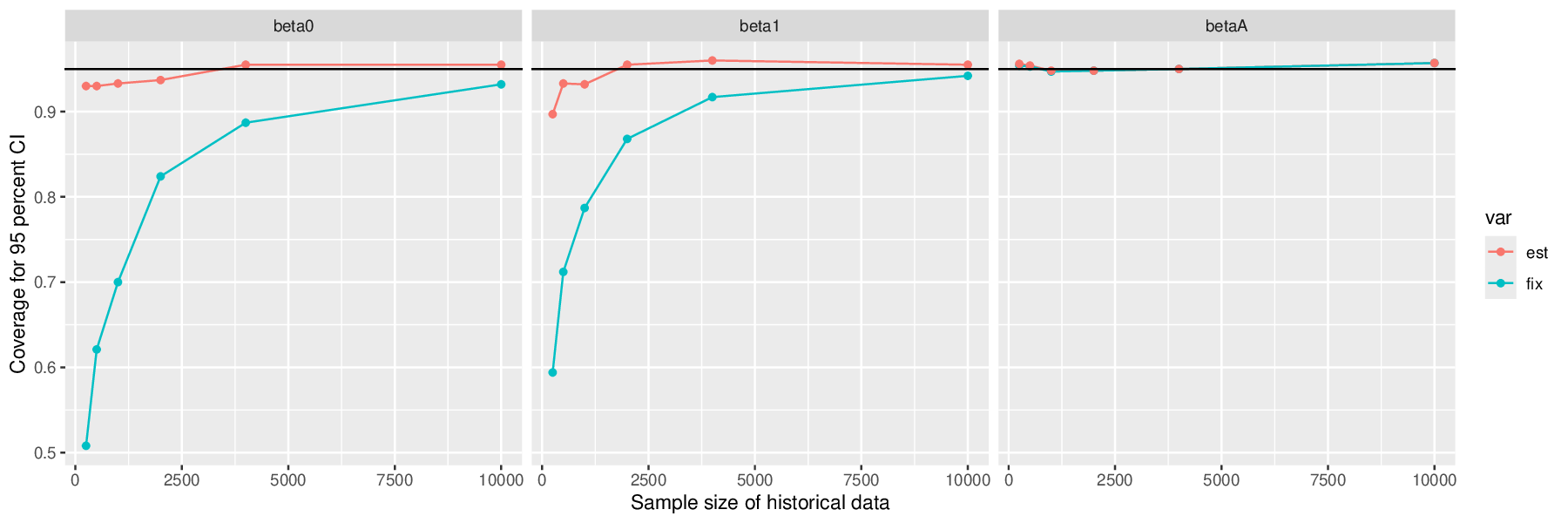}
    \caption{Plots of the coverage probability of 95\% CI over 1000 simulations for Scenario B-6.
    ``beta0'', ``betaA'', ``beta1'' represent the intercept $\beta_0$, the coefficient for the treatment assignment $\beta_A$, the coefficient for the prognostic score $\beta_1$ in the PROCOVA model \eqref{eq: PROCOVA linear}, respectively.
     ``fix'' represents $e^\top\hat V_{\text{fix}}e$ with $e=(1,0,0)^\top$ for $\beta_0$, with $e=(0,1,0)^\top$ for $\beta_A$ and $e=(0,0,1)^\top$ for $\beta_1$.
    ``est'' represents $e^\top\hat V_{\text{est}}e$ with $e=(1,0,0)^\top$ for $\beta_0$, with $e=(0,1,0)^\top$ for $\beta_A$ and $e=(0,0,1)^\top$ for $\beta_1$.
    The sample size of trial data is $n=1000$. 
    The x-axis represents the sample size of historical data $\tilde n$.
    The y-axis represents the coverage probability which is the proportion of 1000 simulations in which the 95\% CI using each variance estimator includes the true value.}
\end{figure}

\clearpage
\subsection{Scenario B-7}

\begin{figure}[h]
    \centering
    \includegraphics[width=\linewidth]{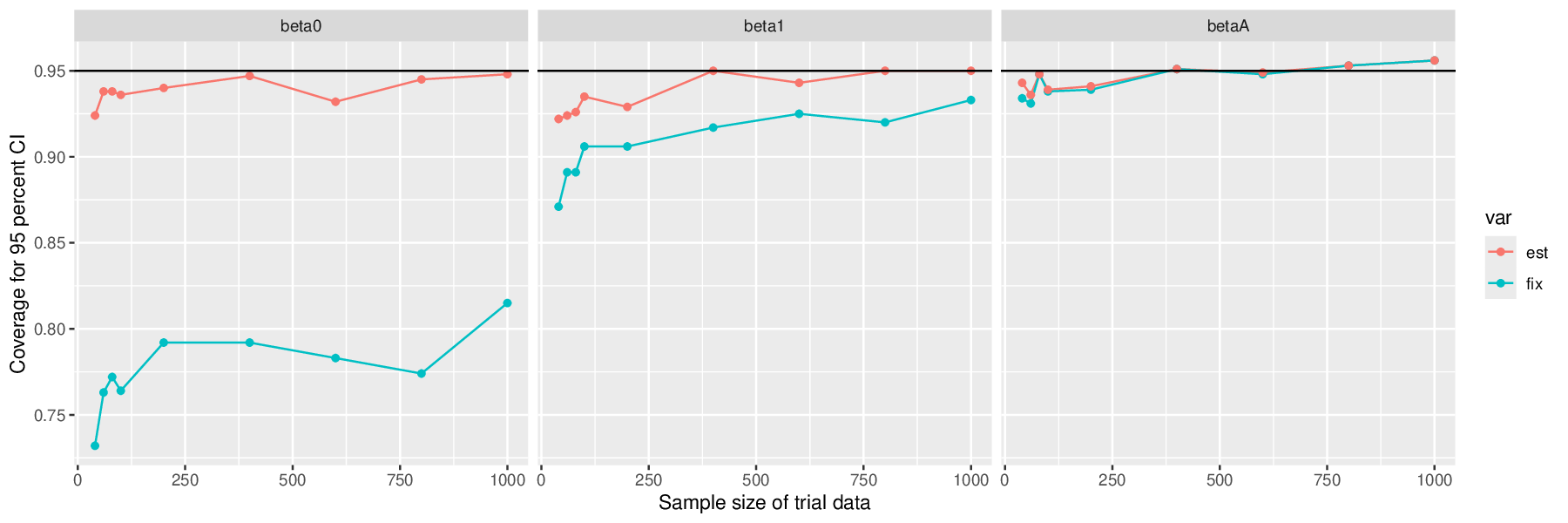}
    \caption{Plots of the coverage probability of 95\% CI over 1000 simulations for Scenario B-7.
    ``beta0'', ``betaA'', ``beta1'' represent the intercept $\beta_0$, the coefficient for the treatment assignment $\beta_A$, the coefficient for the prognostic score $\beta_1$ in the PROCOVA model \eqref{eq: PROCOVA linear}, respectively.
     ``fix'' represents $e^\top\hat V_{\text{fix}}e$ with $e=(1,0,0)^\top$ for $\beta_0$, with $e=(0,1,0)^\top$ for $\beta_A$ and $e=(0,0,1)^\top$ for $\beta_1$.
    ``est'' represents $e^\top\hat V_{\text{est}}e$ with $e=(1,0,0)^\top$ for $\beta_0$, with $e=(0,1,0)^\top$ for $\beta_A$ and $e=(0,0,1)^\top$ for $\beta_1$.
    The x-axis represents the sample size of trial data $n$.
    The sample size of historical data is $\tilde n = 10n$. 
    The y-axis represents the coverage probability which is the proportion of 1000 simulations in which the 95\% CI using each variance estimator includes the true value.
    }
\end{figure}

\clearpage
\begin{figure}[h]
    \centering
    \includegraphics[width=0.4\linewidth]{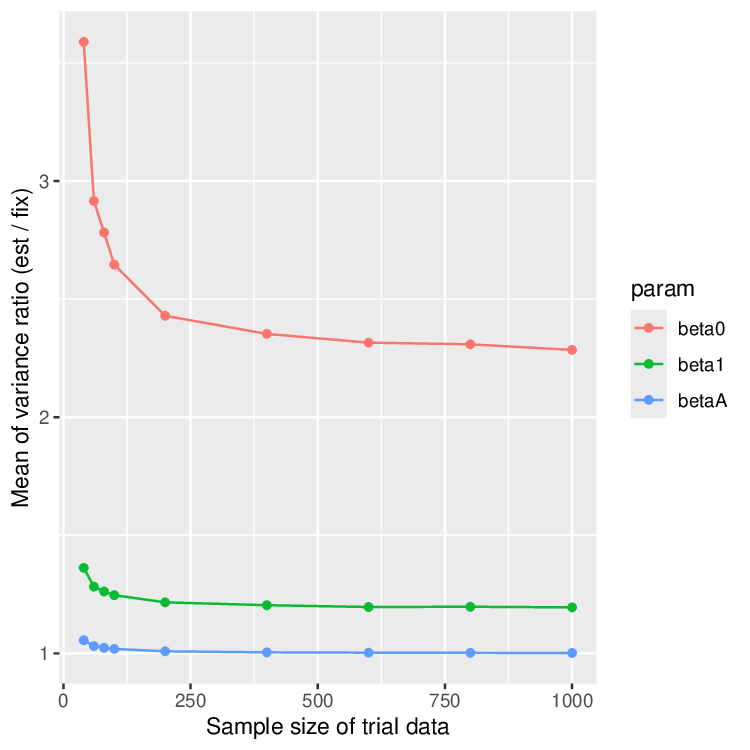}
    \caption{Plots of the mean of the ratio of two variance estimators over 1000 simulations for Scenario B-7.
    ``beta0'', ``betaA'', ``beta1'' represent the intercept $\beta_0$, the coefficient for the treatment assignment $\beta_A$, the coefficient for the prognostic score $\beta_1$ in the PROCOVA model \eqref{eq: PROCOVA linear}, respectively.
    The x-axis represents the sample size of trial data $n$.
    The sample size of historical data is $\tilde n = 10n$. 
    The y-axis represents the mean of the ratio of two variance estimators, i.e., $e^\top\hat V_{\text{est}}e/e^\top\hat V_{\text{fix}}e$ with $e=(1,0,0)^\top$ for $\beta_0$, with $e=(0,1,0)^\top$ for $\beta_A$ and $e=(0,0,1)^\top$ for $\beta_1$, over 1000 simulations.}
\end{figure}
\clearpage
\begin{figure}[h]
    \centering
    \includegraphics[width=0.4\linewidth]{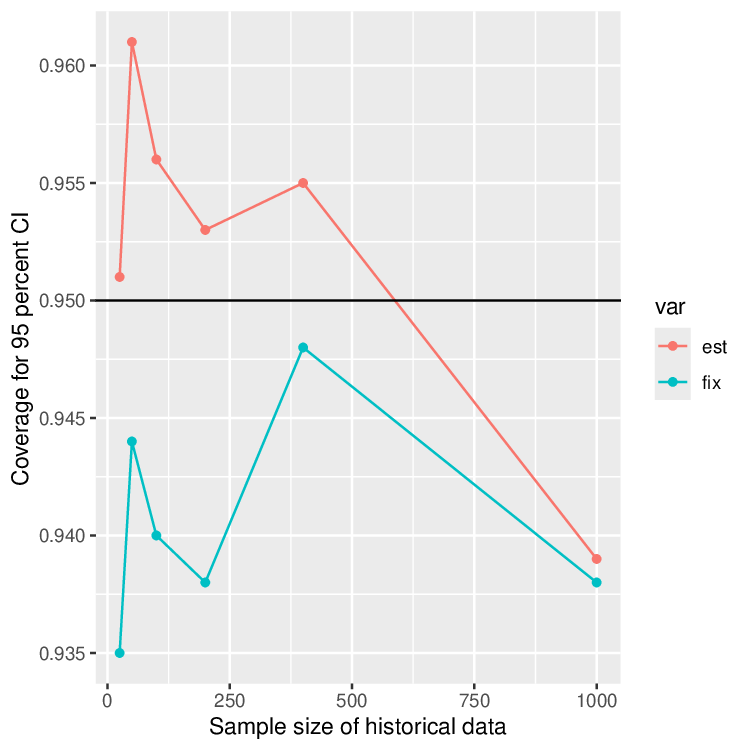}
    \caption{Plots of the coverage probability of the 95\% CI for $\beta_A$ in the PROCOVA model \eqref{eq: PROCOVA linear} over 1000 simulations for Scenario B-7.
     ``fix'' represents $e^\top\hat V_{\text{fix}}e$ with $e=(1,0,0)^\top$ for $\beta_0$, with $e=(0,1,0)^\top$ for $\beta_A$ and $e=(0,0,1)^\top$ for $\beta_1$.
    ``est'' represents $e^\top\hat V_{\text{est}}e$ with $e=(1,0,0)^\top$ for $\beta_0$, with $e=(0,1,0)^\top$ for $\beta_A$ and $e=(0,0,1)^\top$ for $\beta_1$.
    The sample size of trial data is $n=100$. 
    The x-axis represents the sample size of historical data $\tilde n$.
    The y-axis represents the coverage probability which is the proportion of 1000 simulations in which the 95\% CI using each variance estimator includes the true value.}
\end{figure}
\clearpage
\begin{figure}[h]
    \centering
    \includegraphics[width=\linewidth]{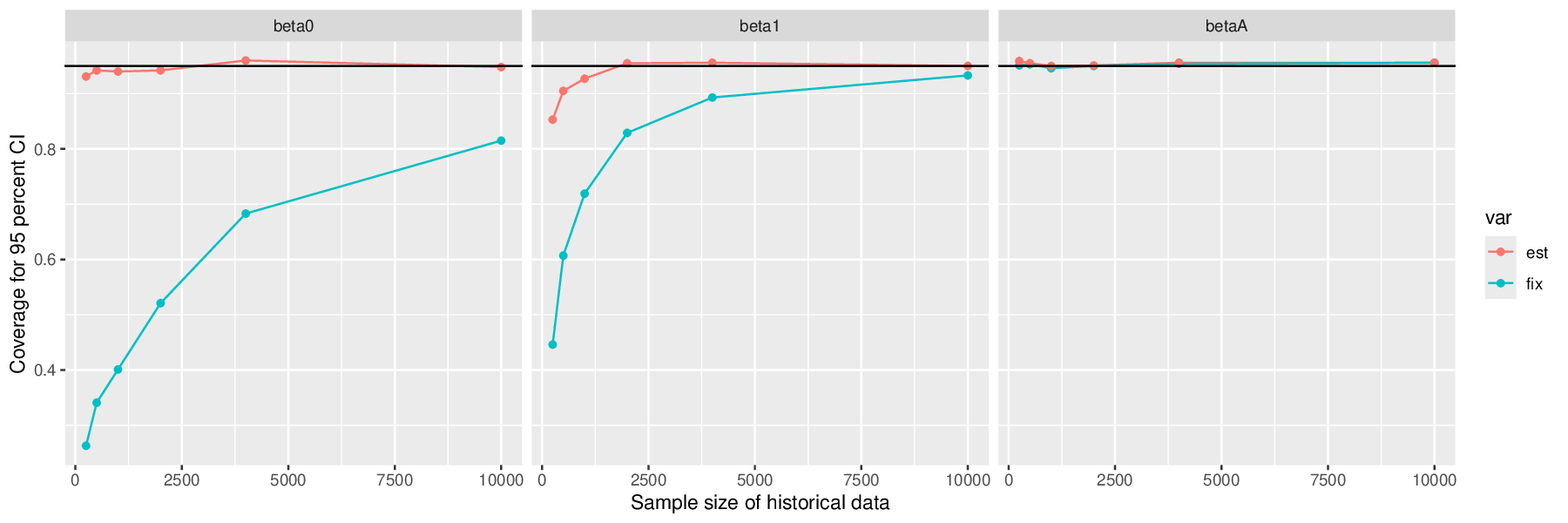}
    \caption{Plots of the coverage probability of 95\% CI over 1000 simulations for Scenario B-7.
    ``beta0'', ``betaA'', ``beta1'' represent the intercept $\beta_0$, the coefficient for the treatment assignment $\beta_A$, the coefficient for the prognostic score $\beta_1$ in the PROCOVA model \eqref{eq: PROCOVA linear}, respectively.
     ``fix'' represents $e^\top\hat V_{\text{fix}}e$ with $e=(1,0,0)^\top$ for $\beta_0$, with $e=(0,1,0)^\top$ for $\beta_A$ and $e=(0,0,1)^\top$ for $\beta_1$.
    ``est'' represents $e^\top\hat V_{\text{est}}e$ with $e=(1,0,0)^\top$ for $\beta_0$, with $e=(0,1,0)^\top$ for $\beta_A$ and $e=(0,0,1)^\top$ for $\beta_1$.
    The sample size of trial data is $n=1000$. 
    The x-axis represents the sample size of historical data $\tilde n$.
    The y-axis represents the coverage probability which is the proportion of 1000 simulations in which the 95\% CI using each variance estimator includes the true value.}
\end{figure}

\clearpage
\subsection{Scenario B-8}

\begin{figure}[h]
    \centering
    \includegraphics[width=\linewidth]{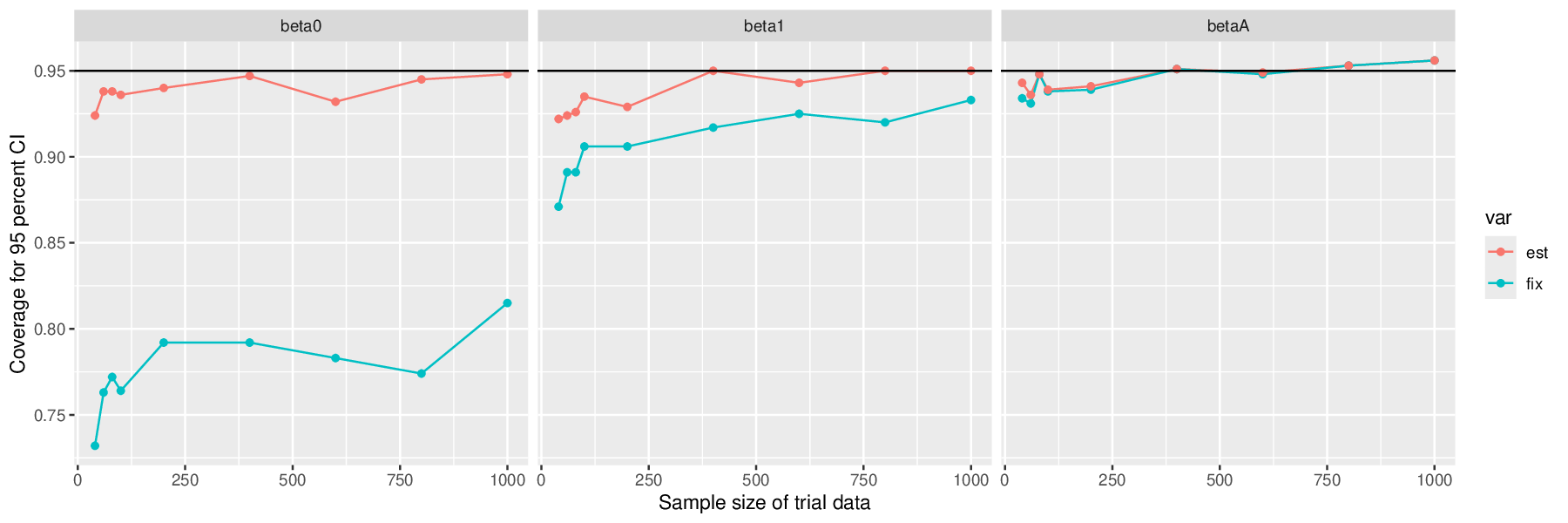}
    \caption{Plots of the coverage probability of 95\% CI over 1000 simulations for Scenario B-8.
    ``beta0'', ``betaA'', ``beta1'' represent the intercept $\beta_0$, the coefficient for the treatment assignment $\beta_A$, the coefficient for the prognostic score $\beta_1$ in the PROCOVA model \eqref{eq: PROCOVA linear}, respectively.
     ``fix'' represents $e^\top\hat V_{\text{fix}}e$ with $e=(1,0,0)^\top$ for $\beta_0$, with $e=(0,1,0)^\top$ for $\beta_A$ and $e=(0,0,1)^\top$ for $\beta_1$.
    ``est'' represents $e^\top\hat V_{\text{est}}e$ with $e=(1,0,0)^\top$ for $\beta_0$, with $e=(0,1,0)^\top$ for $\beta_A$ and $e=(0,0,1)^\top$ for $\beta_1$.
    The x-axis represents the sample size of trial data $n$.
    The sample size of historical data is $\tilde n = 10n$. 
    The y-axis represents the coverage probability which is the proportion of 1000 simulations in which the 95\% CI using each variance estimator includes the true value.
    }
\end{figure}

\clearpage
\begin{figure}[h]
    \centering
    \includegraphics[width=0.4\linewidth]{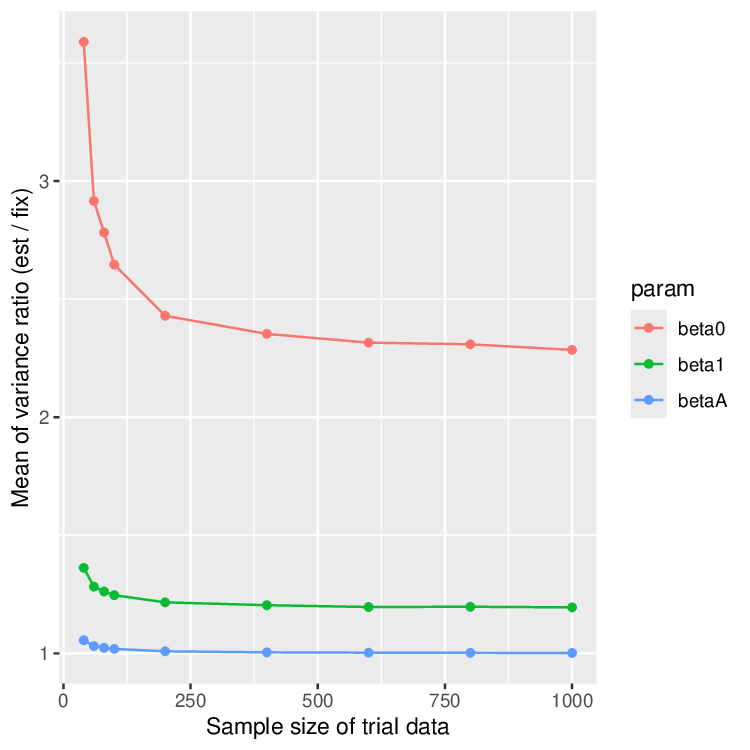}
    \caption{Plots of the mean of the ratio of two variance estimators over 1000 simulations for Scenario B-8.
    ``beta0'', ``betaA'', ``beta1'' represent the intercept $\beta_0$, the coefficient for the treatment assignment $\beta_A$, the coefficient for the prognostic score $\beta_1$ in the PROCOVA model \eqref{eq: PROCOVA linear}, respectively.
    The x-axis represents the sample size of trial data $n$.
    The sample size of historical data is $\tilde n = 10n$. 
    The y-axis represents the mean of the ratio of two variance estimators, i.e., $e^\top\hat V_{\text{est}}e/e^\top\hat V_{\text{fix}}e$ with $e=(1,0,0)^\top$ for $\beta_0$, with $e=(0,1,0)^\top$ for $\beta_A$ and $e=(0,0,1)^\top$ for $\beta_1$, over 1000 simulations.}
\end{figure}
\clearpage
\begin{figure}[h]
    \centering
    \includegraphics[width=0.4\linewidth]{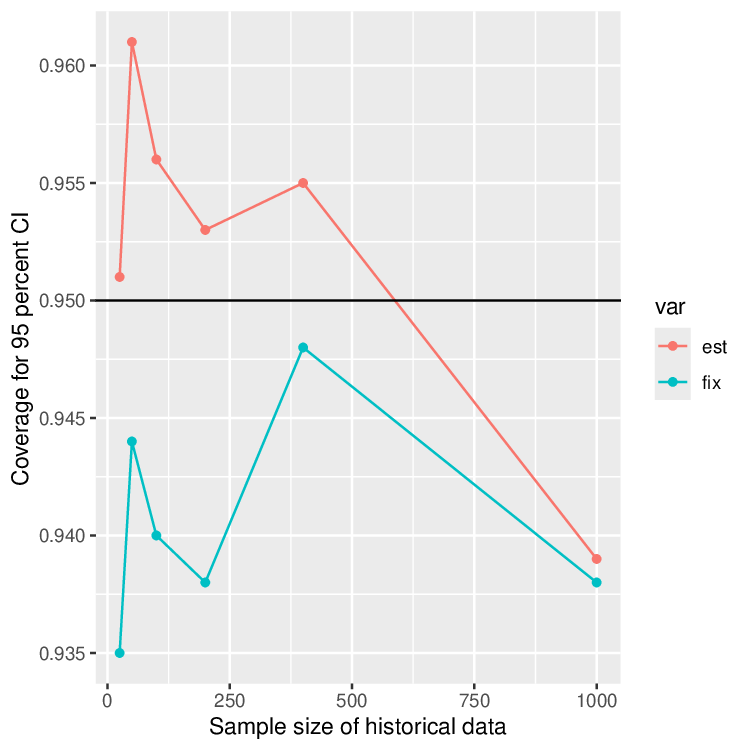}
    \caption{Plots of the coverage probability of the 95\% CI for $\beta_A$ in the PROCOVA model \eqref{eq: PROCOVA linear} over 1000 simulations for Scenario B-8.
     ``fix'' represents $e^\top\hat V_{\text{fix}}e$ with $e=(1,0,0)^\top$ for $\beta_0$, with $e=(0,1,0)^\top$ for $\beta_A$ and $e=(0,0,1)^\top$ for $\beta_1$.
    ``est'' represents $e^\top\hat V_{\text{est}}e$ with $e=(1,0,0)^\top$ for $\beta_0$, with $e=(0,1,0)^\top$ for $\beta_A$ and $e=(0,0,1)^\top$ for $\beta_1$.
    The sample size of trial data is $n=100$. 
    The x-axis represents the sample size of historical data $\tilde n$.
    The y-axis represents the coverage probability which is the proportion of 1000 simulations in which the 95\% CI using each variance estimator includes the true value.}
\end{figure}
\clearpage
\begin{figure}[h]
    \centering
    \includegraphics[width=\linewidth]{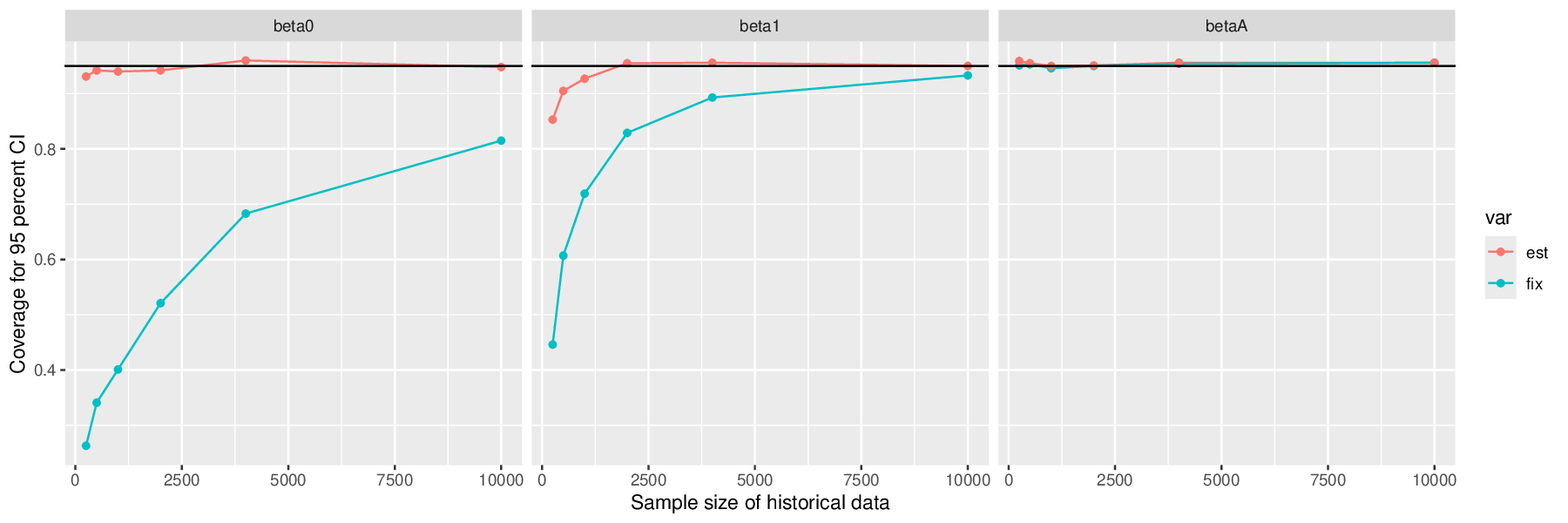}
    \caption{Plots of the coverage probability of 95\% CI over 1000 simulations for Scenario B-8.
    ``beta0'', ``betaA'', ``beta1'' represent the intercept $\beta_0$, the coefficient for the treatment assignment $\beta_A$, the coefficient for the prognostic score $\beta_1$ in the PROCOVA model \eqref{eq: PROCOVA linear}, respectively.
     ``fix'' represents $e^\top\hat V_{\text{fix}}e$ with $e=(1,0,0)^\top$ for $\beta_0$, with $e=(0,1,0)^\top$ for $\beta_A$ and $e=(0,0,1)^\top$ for $\beta_1$.
    ``est'' represents $e^\top\hat V_{\text{est}}e$ with $e=(1,0,0)^\top$ for $\beta_0$, with $e=(0,1,0)^\top$ for $\beta_A$ and $e=(0,0,1)^\top$ for $\beta_1$.
    The sample size of trial data is $n=1000$. 
    The x-axis represents the sample size of historical data $\tilde n$.
    The y-axis represents the coverage probability which is the proportion of 1000 simulations in which the 95\% CI using each variance estimator includes the true value.}
\end{figure}

\clearpage
\subsection{Scenario B-9}

\begin{figure}[h]
    \centering
    \includegraphics[width=\linewidth]{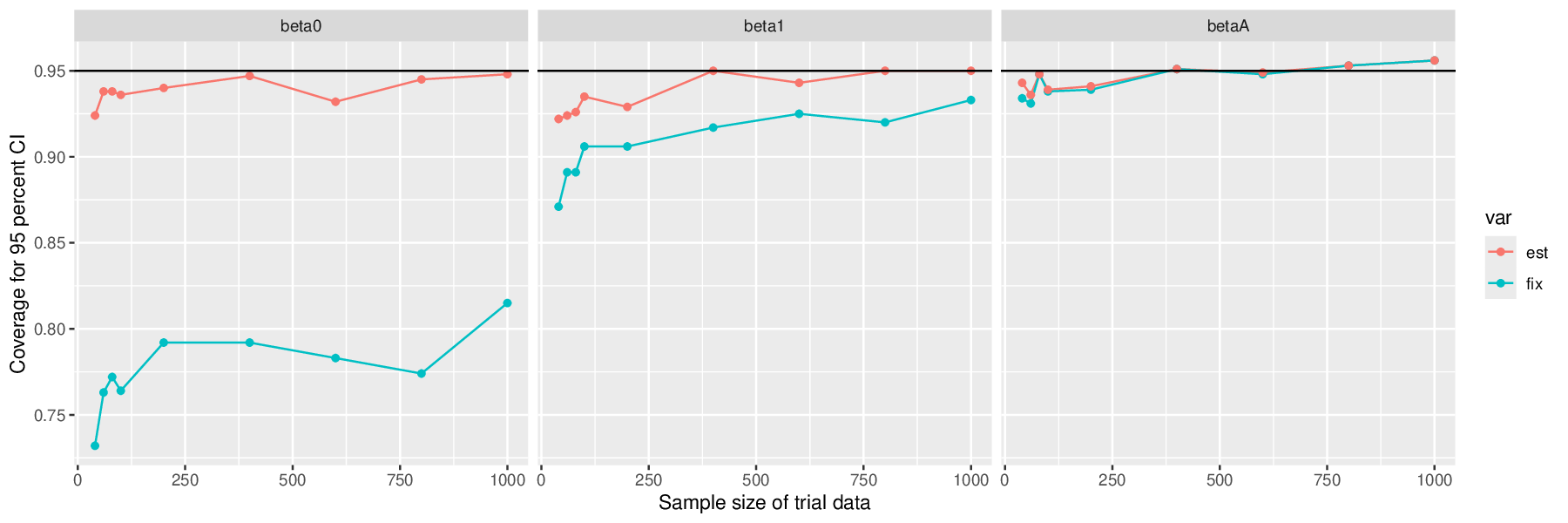}
    \caption{Plots of the coverage probability of 95\% CI over 1000 simulations for Scenario B-9.
    ``beta0'', ``betaA'', ``beta1'' represent the intercept $\beta_0$, the coefficient for the treatment assignment $\beta_A$, the coefficient for the prognostic score $\beta_1$ in the PROCOVA model \eqref{eq: PROCOVA linear}, respectively.
     ``fix'' represents $e^\top\hat V_{\text{fix}}e$ with $e=(1,0,0)^\top$ for $\beta_0$, with $e=(0,1,0)^\top$ for $\beta_A$ and $e=(0,0,1)^\top$ for $\beta_1$.
    ``est'' represents $e^\top\hat V_{\text{est}}e$ with $e=(1,0,0)^\top$ for $\beta_0$, with $e=(0,1,0)^\top$ for $\beta_A$ and $e=(0,0,1)^\top$ for $\beta_1$.
    The x-axis represents the sample size of trial data $n$.
    The sample size of historical data is $\tilde n = 10n$. 
    The y-axis represents the coverage probability which is the proportion of 1000 simulations in which the 95\% CI using each variance estimator includes the true value.
    }
\end{figure}

\clearpage
\begin{figure}[h]
    \centering
    \includegraphics[width=0.4\linewidth]{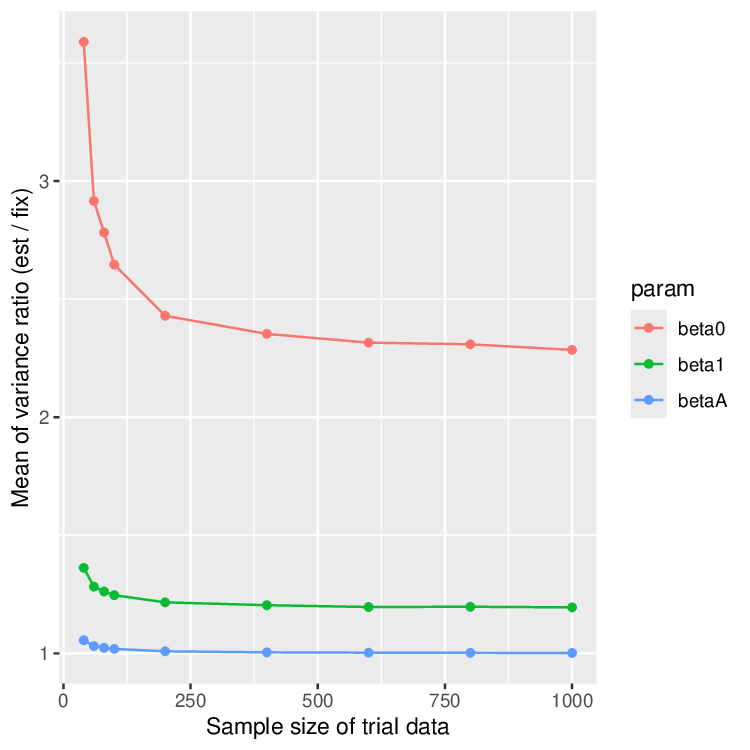}
    \caption{Plots of the mean of the ratio of two variance estimators over 1000 simulations for Scenario B-9.
    ``beta0'', ``betaA'', ``beta1'' represent the intercept $\beta_0$, the coefficient for the treatment assignment $\beta_A$, the coefficient for the prognostic score $\beta_1$ in the PROCOVA model \eqref{eq: PROCOVA linear}, respectively.
    The x-axis represents the sample size of trial data $n$.
    The sample size of historical data is $\tilde n = 10n$. 
    The y-axis represents the mean of the ratio of two variance estimators, i.e., $e^\top\hat V_{\text{est}}e/e^\top\hat V_{\text{fix}}e$ with $e=(1,0,0)^\top$ for $\beta_0$, with $e=(0,1,0)^\top$ for $\beta_A$ and $e=(0,0,1)^\top$ for $\beta_1$, over 1000 simulations.}
\end{figure}
\clearpage
\begin{figure}[h]
    \centering
    \includegraphics[width=0.4\linewidth]{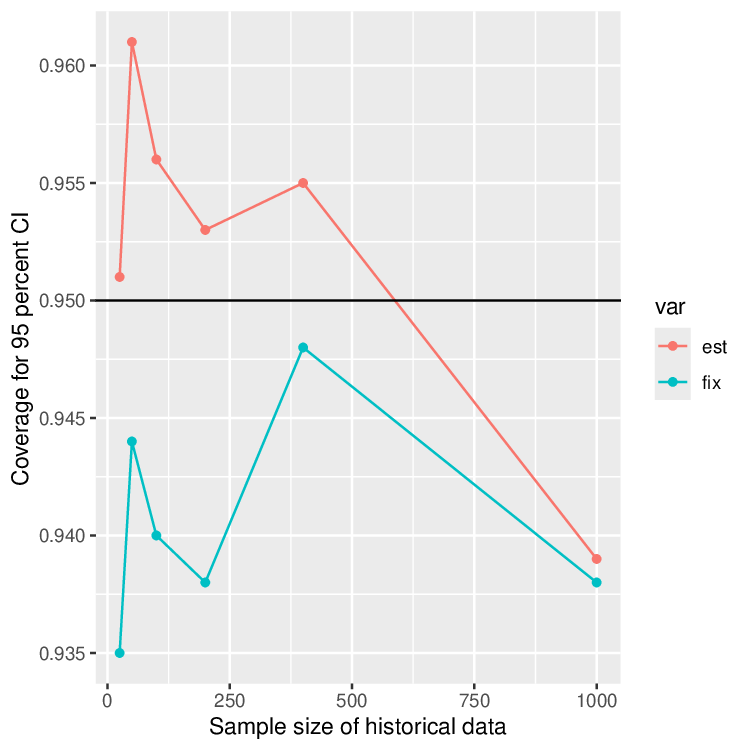}
    \caption{Plots of the coverage probability of the 95\% CI for $\beta_A$ in the PROCOVA model \eqref{eq: PROCOVA linear} over 1000 simulations for Scenario B-9.
     ``fix'' represents $e^\top\hat V_{\text{fix}}e$ with $e=(1,0,0)^\top$ for $\beta_0$, with $e=(0,1,0)^\top$ for $\beta_A$ and $e=(0,0,1)^\top$ for $\beta_1$.
    ``est'' represents $e^\top\hat V_{\text{est}}e$ with $e=(1,0,0)^\top$ for $\beta_0$, with $e=(0,1,0)^\top$ for $\beta_A$ and $e=(0,0,1)^\top$ for $\beta_1$.
    The sample size of trial data is $n=100$. 
    The x-axis represents the sample size of historical data $\tilde n$.
    The y-axis represents the coverage probability which is the proportion of 1000 simulations in which the 95\% CI using each variance estimator includes the true value.}
\end{figure}
\clearpage
\begin{figure}[h]
    \centering
    \includegraphics[width=\linewidth]{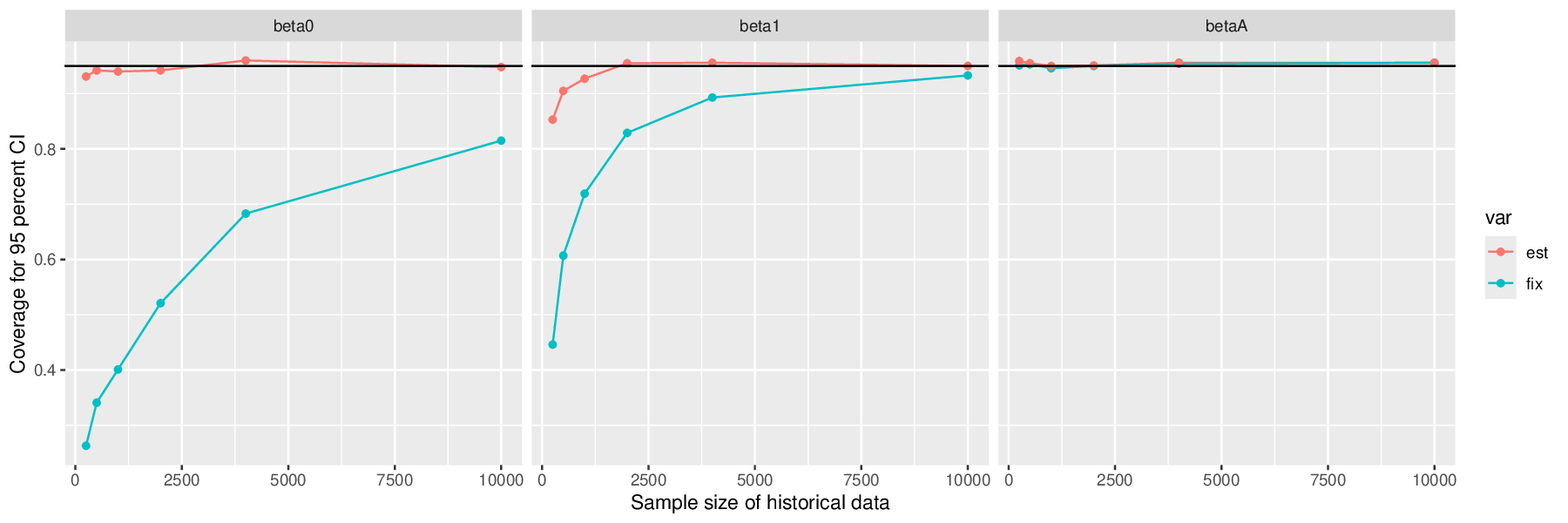}
    \caption{Plots of the coverage probability of 95\% CI over 1000 simulations for Scenario B-9.
    ``beta0'', ``betaA'', ``beta1'' represent the intercept $\beta_0$, the coefficient for the treatment assignment $\beta_A$, the coefficient for the prognostic score $\beta_1$ in the PROCOVA model \eqref{eq: PROCOVA linear}, respectively.
     ``fix'' represents $e^\top\hat V_{\text{fix}}e$ with $e=(1,0,0)^\top$ for $\beta_0$, with $e=(0,1,0)^\top$ for $\beta_A$ and $e=(0,0,1)^\top$ for $\beta_1$.
    ``est'' represents $e^\top\hat V_{\text{est}}e$ with $e=(1,0,0)^\top$ for $\beta_0$, with $e=(0,1,0)^\top$ for $\beta_A$ and $e=(0,0,1)^\top$ for $\beta_1$.
    The sample size of trial data is $n=1000$. 
    The x-axis represents the sample size of historical data $\tilde n$.
    The y-axis represents the coverage probability which is the proportion of 1000 simulations in which the 95\% CI using each variance estimator includes the true value.}
\end{figure}

\clearpage
\subsection{Scenario C-1}

\begin{figure}[h]
    \centering
    \includegraphics[width=\linewidth]{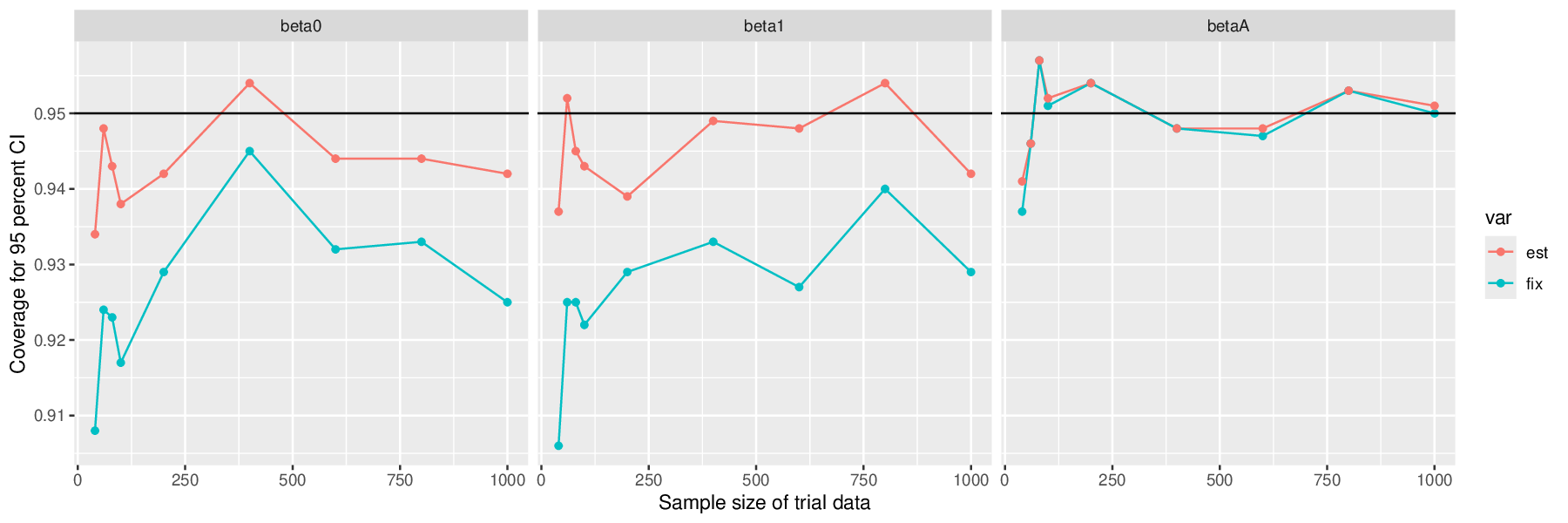}
    \caption{Plots of the coverage probability of 95\% CI over 1000 simulations for Scenario C-1.
    ``beta0'', ``betaA'', ``beta1'' represent the intercept $\beta_0$, the coefficient for the treatment assignment $\beta_A$, the coefficient for the prognostic score $\beta_1$ in the PROCOVA model \eqref{eq: PROCOVA linear}, respectively.
     ``fix'' represents $e^\top\hat V_{\text{fix}}e$ with $e=(1,0,0)^\top$ for $\beta_0$, with $e=(0,1,0)^\top$ for $\beta_A$ and $e=(0,0,1)^\top$ for $\beta_1$.
    ``est'' represents $e^\top\hat V_{\text{est}}e$ with $e=(1,0,0)^\top$ for $\beta_0$, with $e=(0,1,0)^\top$ for $\beta_A$ and $e=(0,0,1)^\top$ for $\beta_1$.
    The x-axis represents the sample size of trial data $n$.
    The sample size of historical data is $\tilde n = 10n$. 
    The y-axis represents the coverage probability which is the proportion of 1000 simulations in which the 95\% CI using each variance estimator includes the true value.
    }
\end{figure}

\clearpage
\begin{figure}[h]
    \centering
    \includegraphics[width=0.4\linewidth]{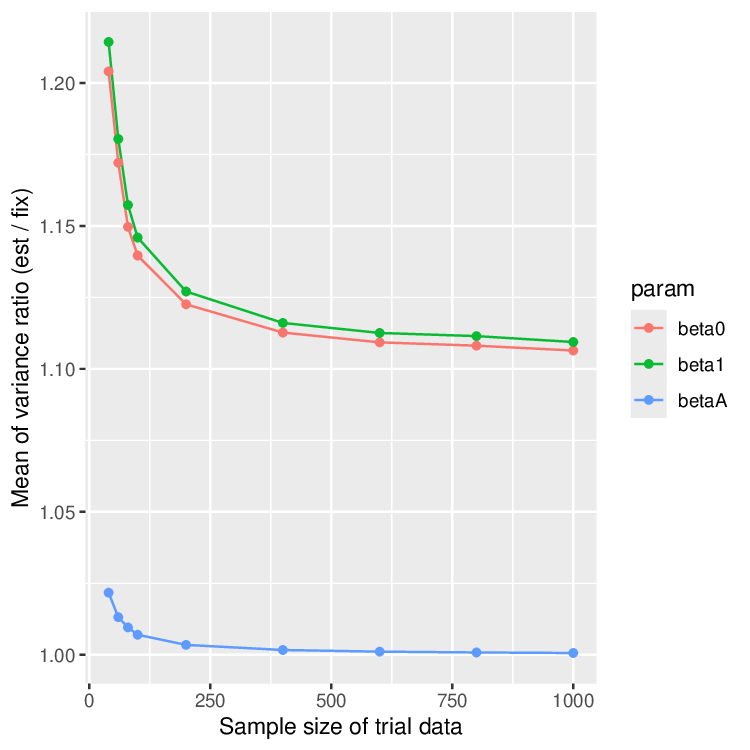}
    \caption{Plots of the mean of the ratio of two variance estimators over 1000 simulations for Scenario C-1.
    ``beta0'', ``betaA'', ``beta1'' represent the intercept $\beta_0$, the coefficient for the treatment assignment $\beta_A$, the coefficient for the prognostic score $\beta_1$ in the PROCOVA model \eqref{eq: PROCOVA linear}, respectively.
    The x-axis represents the sample size of trial data $n$.
    The sample size of historical data is $\tilde n = 10n$. 
    The y-axis represents the mean of the ratio of two variance estimators, i.e., $e^\top\hat V_{\text{est}}e/e^\top\hat V_{\text{fix}}e$ with $e=(1,0,0)^\top$ for $\beta_0$, with $e=(0,1,0)^\top$ for $\beta_A$ and $e=(0,0,1)^\top$ for $\beta_1$, over 1000 simulations.}
\end{figure}
\clearpage
\begin{figure}[h]
    \centering
    \includegraphics[width=0.4\linewidth]{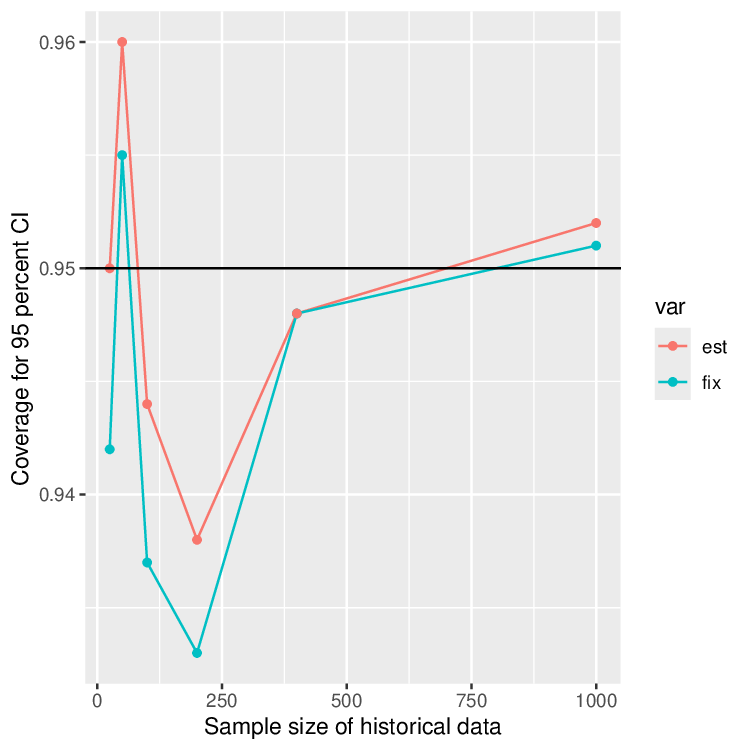}
    \caption{Plots of the coverage probability of the 95\% CI for $\beta_A$ in the PROCOVA model \eqref{eq: PROCOVA linear} over 1000 simulations for Scenario C-1.
     ``fix'' represents $e^\top\hat V_{\text{fix}}e$ with $e=(1,0,0)^\top$ for $\beta_0$, with $e=(0,1,0)^\top$ for $\beta_A$ and $e=(0,0,1)^\top$ for $\beta_1$.
    ``est'' represents $e^\top\hat V_{\text{est}}e$ with $e=(1,0,0)^\top$ for $\beta_0$, with $e=(0,1,0)^\top$ for $\beta_A$ and $e=(0,0,1)^\top$ for $\beta_1$.
    The sample size of trial data is $n=100$. 
    The x-axis represents the sample size of historical data $\tilde n$.
    The y-axis represents the coverage probability which is the proportion of 1000 simulations in which the 95\% CI using each variance estimator includes the true value.}
\end{figure}
\clearpage
\begin{figure}[h]
    \centering
    \includegraphics[width=\linewidth]{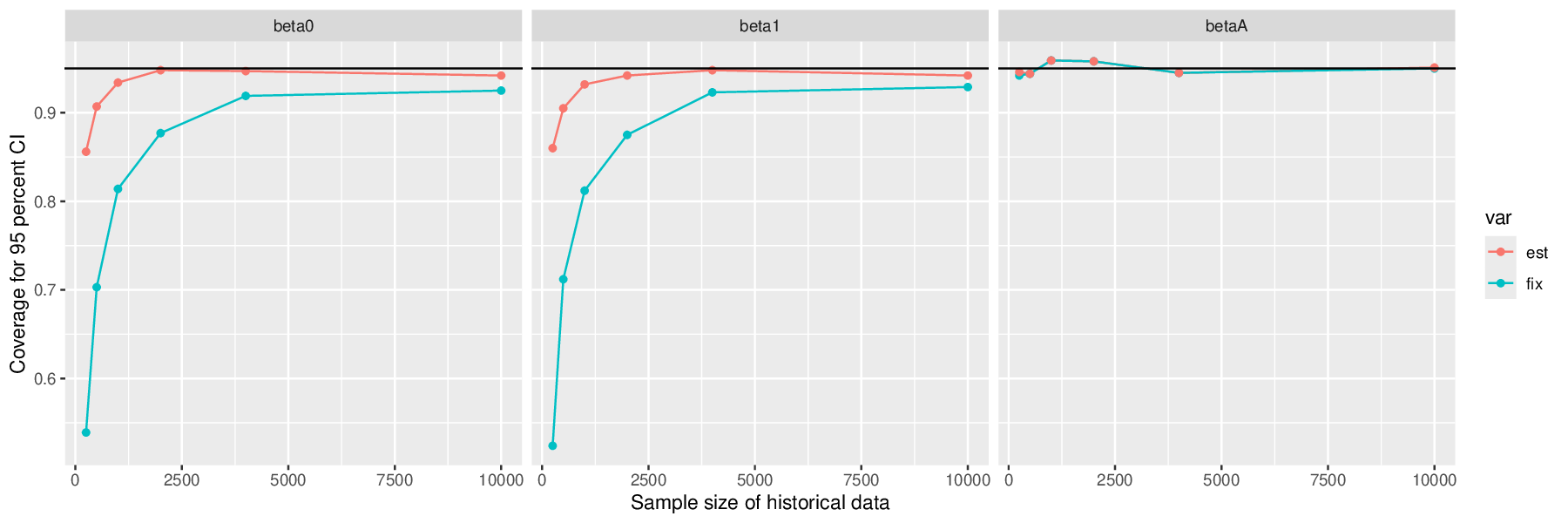}
    \caption{Plots of the coverage probability of 95\% CI over 1000 simulations for Scenario C-1.
    ``beta0'', ``betaA'', ``beta1'' represent the intercept $\beta_0$, the coefficient for the treatment assignment $\beta_A$, the coefficient for the prognostic score $\beta_1$ in the PROCOVA model \eqref{eq: PROCOVA linear}, respectively.
     ``fix'' represents $e^\top\hat V_{\text{fix}}e$ with $e=(1,0,0)^\top$ for $\beta_0$, with $e=(0,1,0)^\top$ for $\beta_A$ and $e=(0,0,1)^\top$ for $\beta_1$.
    ``est'' represents $e^\top\hat V_{\text{est}}e$ with $e=(1,0,0)^\top$ for $\beta_0$, with $e=(0,1,0)^\top$ for $\beta_A$ and $e=(0,0,1)^\top$ for $\beta_1$.
    The sample size of trial data is $n=1000$. 
    The x-axis represents the sample size of historical data $\tilde n$.
    The y-axis represents the coverage probability which is the proportion of 1000 simulations in which the 95\% CI using each variance estimator includes the true value.}
\end{figure}

\clearpage
\subsection{Scenario C-2}

\begin{figure}[h]
    \centering
    \includegraphics[width=\linewidth]{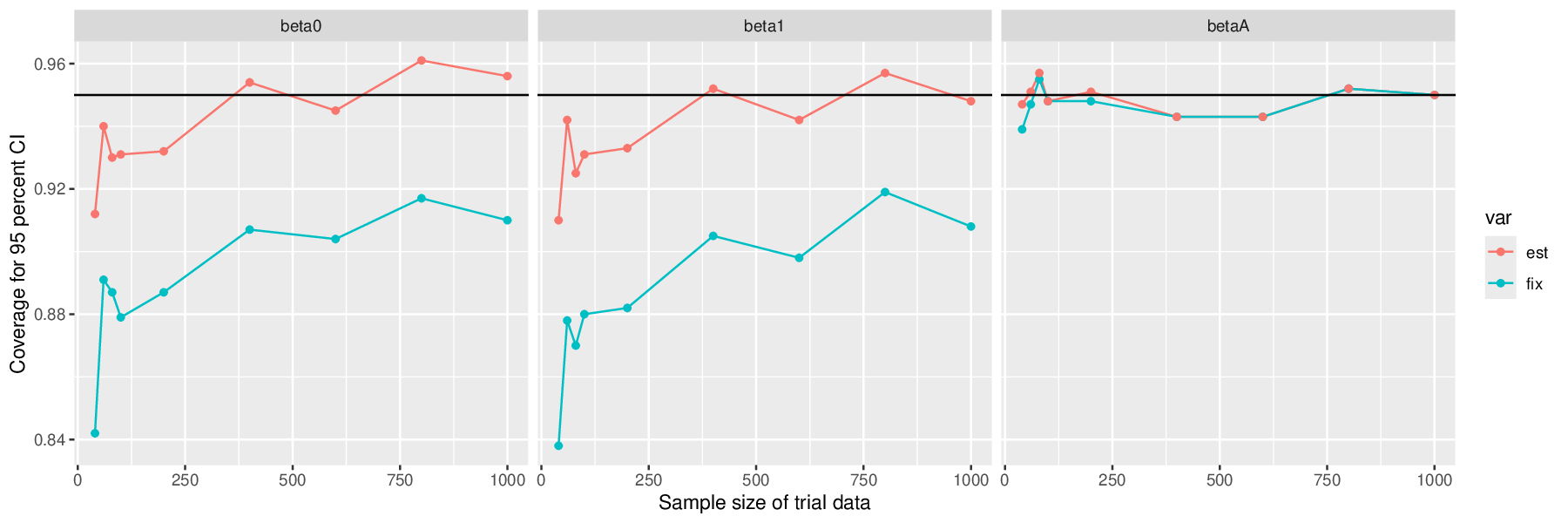}
    \caption{Plots of the coverage probability of 95\% CI over 1000 simulations for Scenario C-2.
    ``beta0'', ``betaA'', ``beta1'' represent the intercept $\beta_0$, the coefficient for the treatment assignment $\beta_A$, the coefficient for the prognostic score $\beta_1$ in the PROCOVA model \eqref{eq: PROCOVA linear}, respectively.
     ``fix'' represents $e^\top\hat V_{\text{fix}}e$ with $e=(1,0,0)^\top$ for $\beta_0$, with $e=(0,1,0)^\top$ for $\beta_A$ and $e=(0,0,1)^\top$ for $\beta_1$.
    ``est'' represents $e^\top\hat V_{\text{est}}e$ with $e=(1,0,0)^\top$ for $\beta_0$, with $e=(0,1,0)^\top$ for $\beta_A$ and $e=(0,0,1)^\top$ for $\beta_1$.
    The x-axis represents the sample size of trial data $n$.
    The sample size of historical data is $\tilde n = 10n$. 
    The y-axis represents the coverage probability which is the proportion of 1000 simulations in which the 95\% CI using each variance estimator includes the true value.
    }
\end{figure}

\clearpage
\begin{figure}[h]
    \centering
    \includegraphics[width=0.4\linewidth]{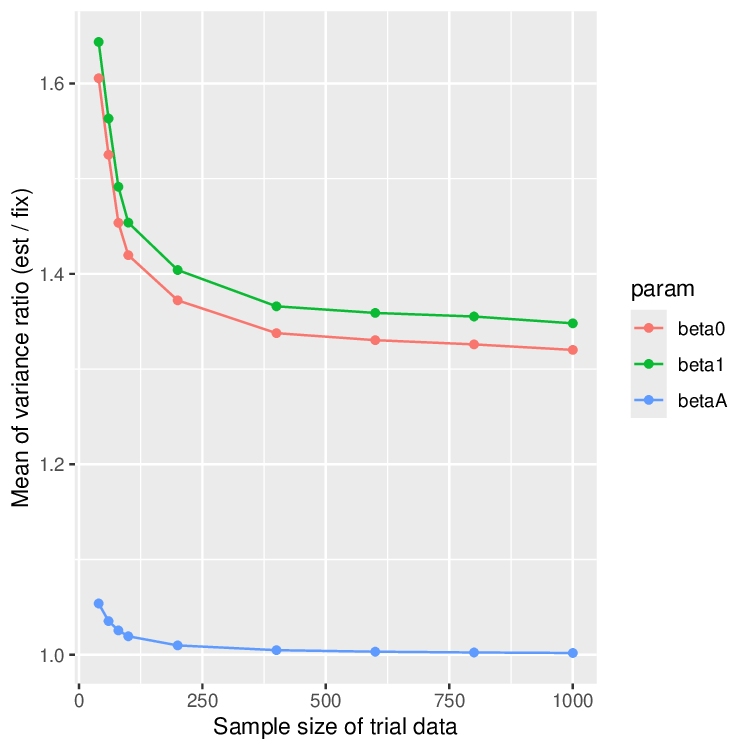}
    \caption{Plots of the mean of the ratio of two variance estimators over 1000 simulations for Scenario C-2.
    ``beta0'', ``betaA'', ``beta1'' represent the intercept $\beta_0$, the coefficient for the treatment assignment $\beta_A$, the coefficient for the prognostic score $\beta_1$ in the PROCOVA model \eqref{eq: PROCOVA linear}, respectively.
    The x-axis represents the sample size of trial data $n$.
    The sample size of historical data is $\tilde n = 10n$. 
    The y-axis represents the mean of the ratio of two variance estimators, i.e., $e^\top\hat V_{\text{est}}e/e^\top\hat V_{\text{fix}}e$ with $e=(1,0,0)^\top$ for $\beta_0$, with $e=(0,1,0)^\top$ for $\beta_A$ and $e=(0,0,1)^\top$ for $\beta_1$, over 1000 simulations.}
\end{figure}
\clearpage
\begin{figure}[h]
    \centering
    \includegraphics[width=0.4\linewidth]{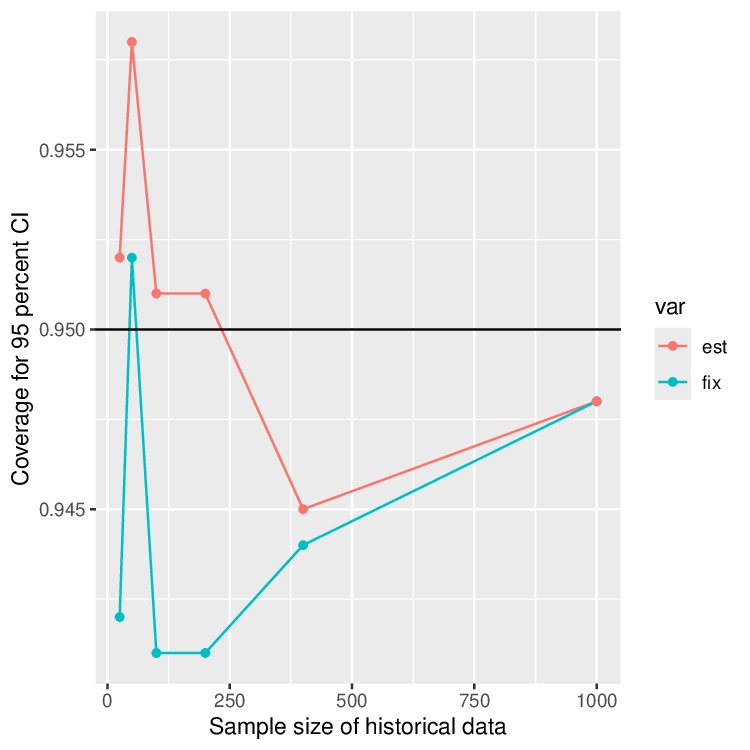}
    \caption{Plots of the coverage probability of the 95\% CI for $\beta_A$ in the PROCOVA model \eqref{eq: PROCOVA linear} over 1000 simulations for Scenario C-2.
     ``fix'' represents $e^\top\hat V_{\text{fix}}e$ with $e=(1,0,0)^\top$ for $\beta_0$, with $e=(0,1,0)^\top$ for $\beta_A$ and $e=(0,0,1)^\top$ for $\beta_1$.
    ``est'' represents $e^\top\hat V_{\text{est}}e$ with $e=(1,0,0)^\top$ for $\beta_0$, with $e=(0,1,0)^\top$ for $\beta_A$ and $e=(0,0,1)^\top$ for $\beta_1$.
    The sample size of trial data is $n=100$. 
    The x-axis represents the sample size of historical data $\tilde n$.
    The y-axis represents the coverage probability which is the proportion of 1000 simulations in which the 95\% CI using each variance estimator includes the true value.}
\end{figure}
\clearpage
\begin{figure}[h]
    \centering
    \includegraphics[width=\linewidth]{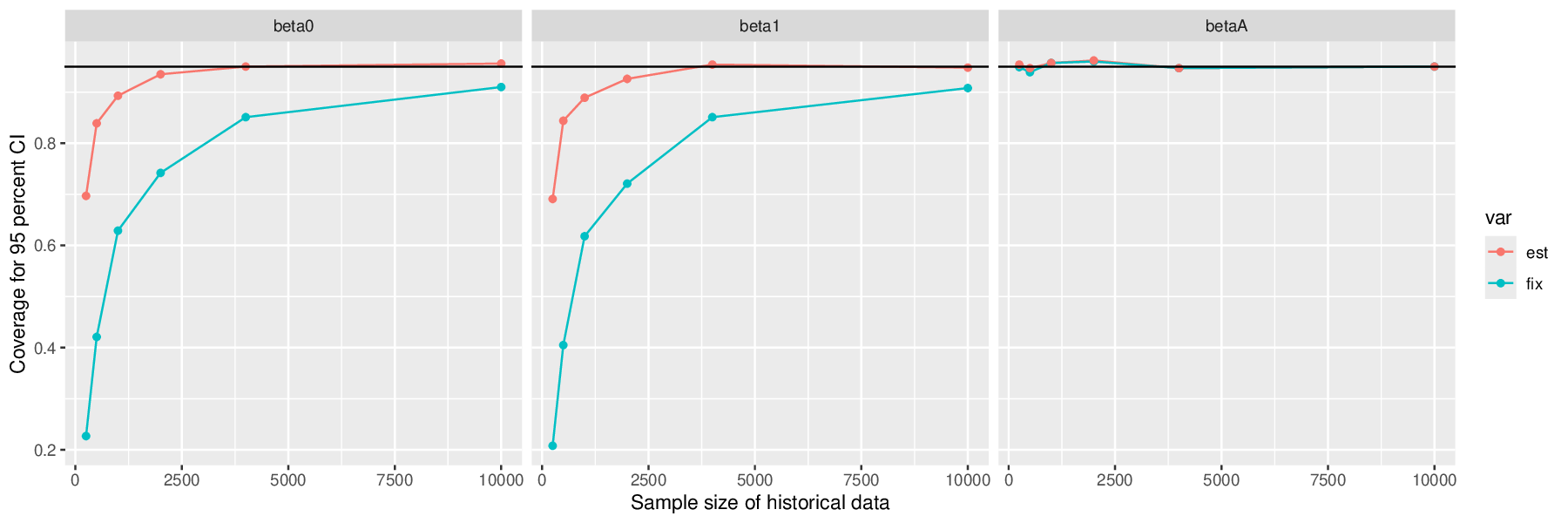}
    \caption{Plots of the coverage probability of 95\% CI over 1000 simulations for Scenario C-2.
    ``beta0'', ``betaA'', ``beta1'' represent the intercept $\beta_0$, the coefficient for the treatment assignment $\beta_A$, the coefficient for the prognostic score $\beta_1$ in the PROCOVA model \eqref{eq: PROCOVA linear}, respectively.
     ``fix'' represents $e^\top\hat V_{\text{fix}}e$ with $e=(1,0,0)^\top$ for $\beta_0$, with $e=(0,1,0)^\top$ for $\beta_A$ and $e=(0,0,1)^\top$ for $\beta_1$.
    ``est'' represents $e^\top\hat V_{\text{est}}e$ with $e=(1,0,0)^\top$ for $\beta_0$, with $e=(0,1,0)^\top$ for $\beta_A$ and $e=(0,0,1)^\top$ for $\beta_1$.
    The sample size of trial data is $n=1000$. 
    The x-axis represents the sample size of historical data $\tilde n$.
    The y-axis represents the coverage probability which is the proportion of 1000 simulations in which the 95\% CI using each variance estimator includes the true value.}
\end{figure}

\clearpage
\subsection{Scenario C-3}

\begin{figure}[h]
    \centering
    \includegraphics[width=\linewidth]{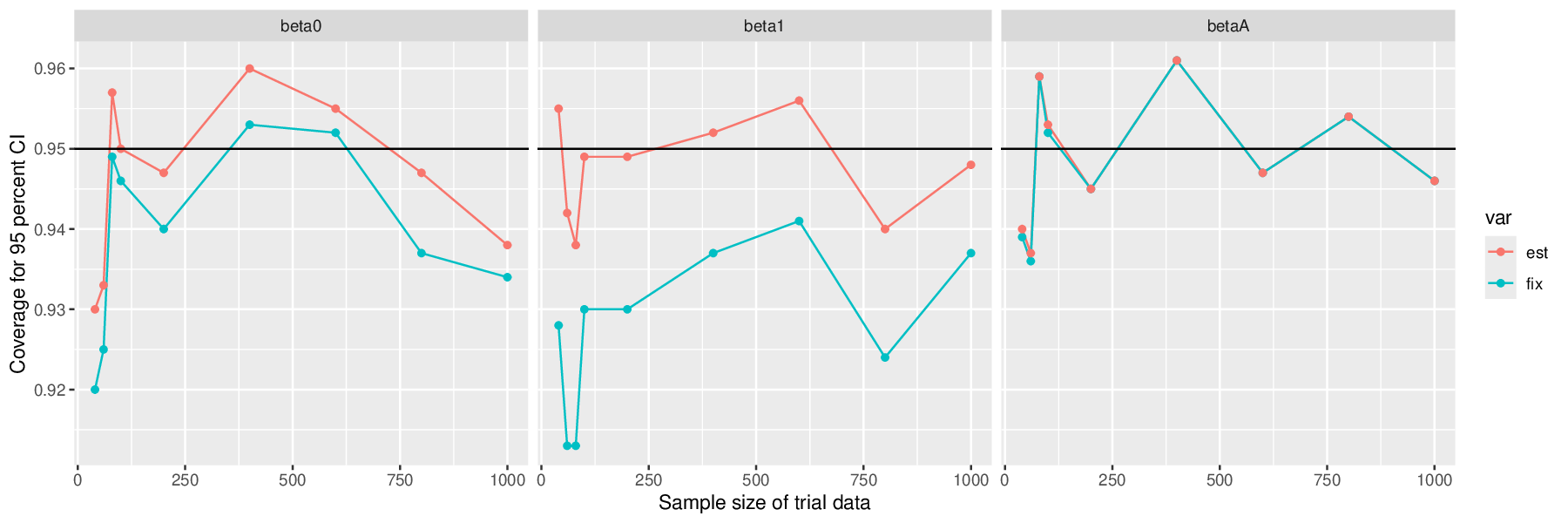}
    \caption{Plots of the coverage probability of 95\% CI over 1000 simulations for Scenario C-3.
    ``beta0'', ``betaA'', ``beta1'' represent the intercept $\beta_0$, the coefficient for the treatment assignment $\beta_A$, the coefficient for the prognostic score $\beta_1$ in the PROCOVA model \eqref{eq: PROCOVA linear}, respectively.
     ``fix'' represents $e^\top\hat V_{\text{fix}}e$ with $e=(1,0,0)^\top$ for $\beta_0$, with $e=(0,1,0)^\top$ for $\beta_A$ and $e=(0,0,1)^\top$ for $\beta_1$.
    ``est'' represents $e^\top\hat V_{\text{est}}e$ with $e=(1,0,0)^\top$ for $\beta_0$, with $e=(0,1,0)^\top$ for $\beta_A$ and $e=(0,0,1)^\top$ for $\beta_1$.
    The x-axis represents the sample size of trial data $n$.
    The sample size of historical data is $\tilde n = 10n$. 
    The y-axis represents the coverage probability which is the proportion of 1000 simulations in which the 95\% CI using each variance estimator includes the true value.
    }
\end{figure}

\clearpage
\begin{figure}[h]
    \centering
    \includegraphics[width=0.4\linewidth]{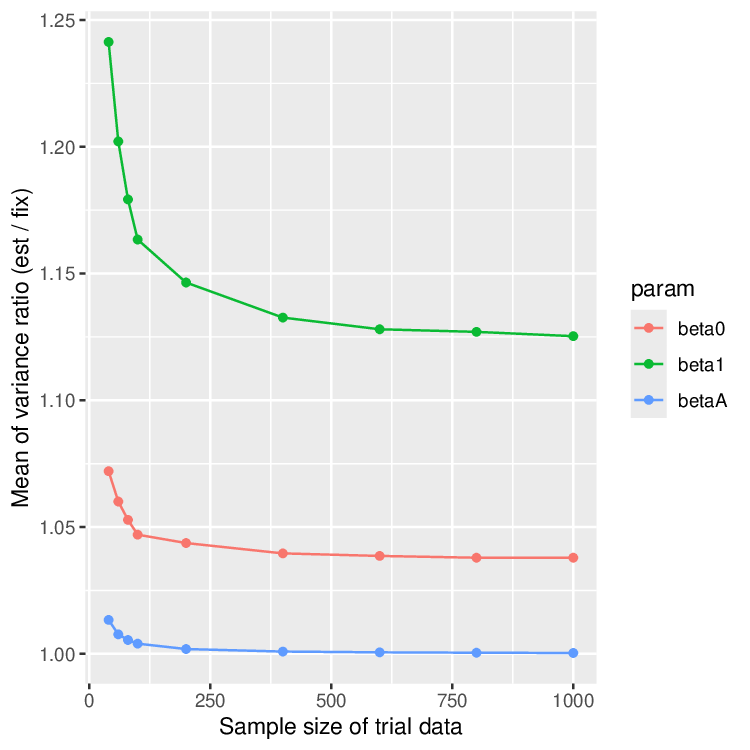}
    \caption{Plots of the mean of the ratio of two variance estimators over 1000 simulations for Scenario C-3.
    ``beta0'', ``betaA'', ``beta1'' represent the intercept $\beta_0$, the coefficient for the treatment assignment $\beta_A$, the coefficient for the prognostic score $\beta_1$ in the PROCOVA model \eqref{eq: PROCOVA linear}, respectively.
    The x-axis represents the sample size of trial data $n$.
    The sample size of historical data is $\tilde n = 10n$. 
    The y-axis represents the mean of the ratio of two variance estimators, i.e., $e^\top\hat V_{\text{est}}e/e^\top\hat V_{\text{fix}}e$ with $e=(1,0,0)^\top$ for $\beta_0$, with $e=(0,1,0)^\top$ for $\beta_A$ and $e=(0,0,1)^\top$ for $\beta_1$, over 1000 simulations.}
\end{figure}
\clearpage
\begin{figure}[h]
    \centering
    \includegraphics[width=0.4\linewidth]{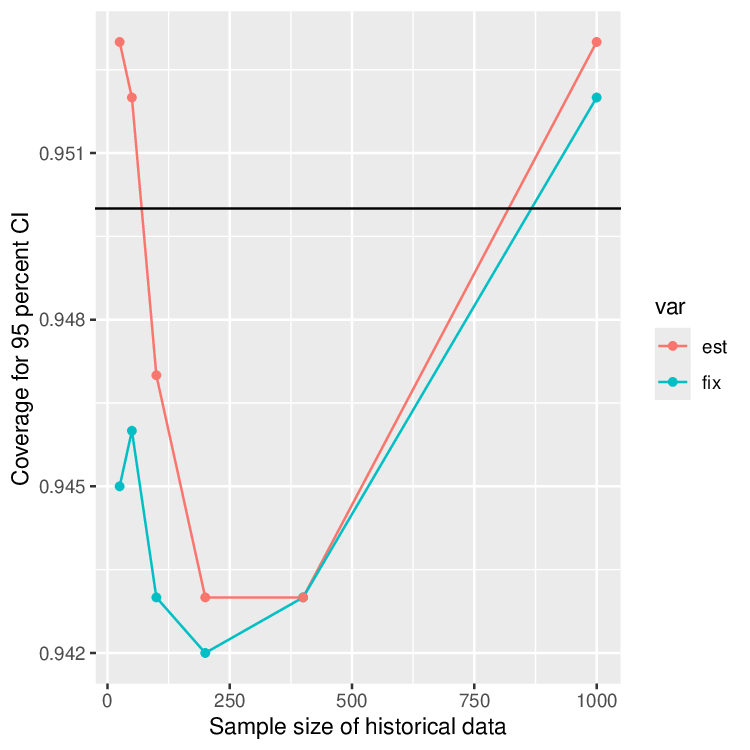}
    \caption{Plots of the coverage probability of the 95\% CI for $\beta_A$ in the PROCOVA model \eqref{eq: PROCOVA linear} over 1000 simulations for Scenario C-3.
     ``fix'' represents $e^\top\hat V_{\text{fix}}e$ with $e=(1,0,0)^\top$ for $\beta_0$, with $e=(0,1,0)^\top$ for $\beta_A$ and $e=(0,0,1)^\top$ for $\beta_1$.
    ``est'' represents $e^\top\hat V_{\text{est}}e$ with $e=(1,0,0)^\top$ for $\beta_0$, with $e=(0,1,0)^\top$ for $\beta_A$ and $e=(0,0,1)^\top$ for $\beta_1$.
    The sample size of trial data is $n=100$. 
    The x-axis represents the sample size of historical data $\tilde n$.
    The y-axis represents the coverage probability which is the proportion of 1000 simulations in which the 95\% CI using each variance estimator includes the true value.}
\end{figure}
\clearpage
\begin{figure}[h]
    \centering
    \includegraphics[width=\linewidth]{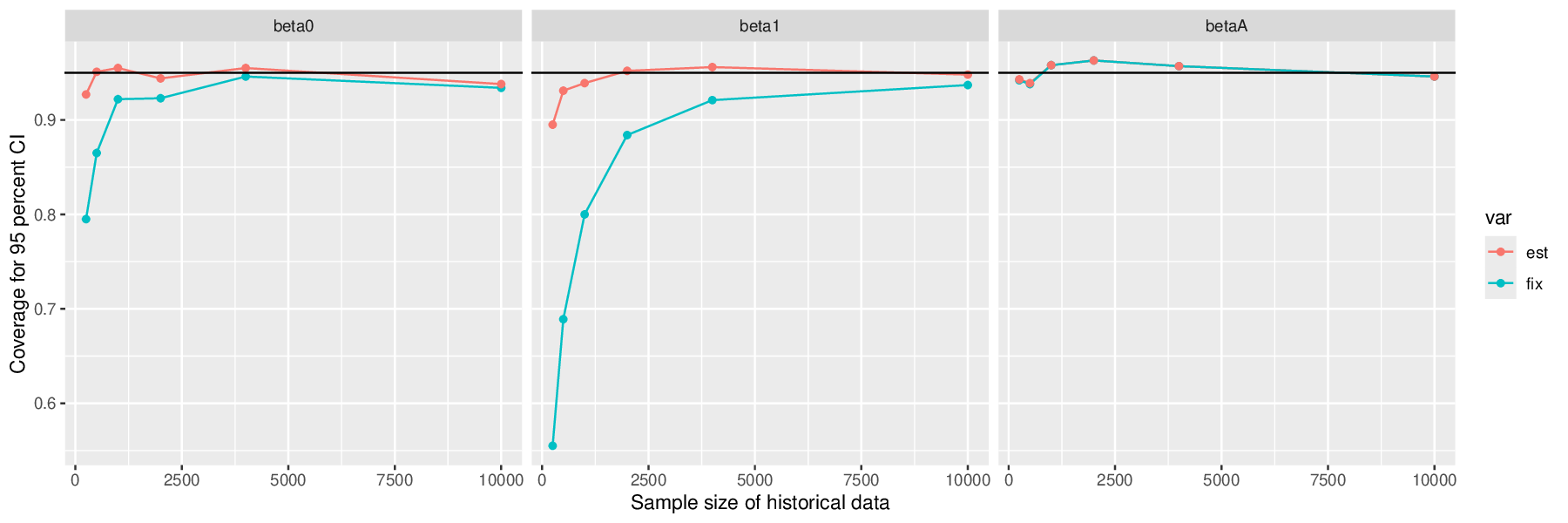}
    \caption{Plots of the coverage probability of 95\% CI over 1000 simulations for Scenario C-3.
    ``beta0'', ``betaA'', ``beta1'' represent the intercept $\beta_0$, the coefficient for the treatment assignment $\beta_A$, the coefficient for the prognostic score $\beta_1$ in the PROCOVA model \eqref{eq: PROCOVA linear}, respectively.
     ``fix'' represents $e^\top\hat V_{\text{fix}}e$ with $e=(1,0,0)^\top$ for $\beta_0$, with $e=(0,1,0)^\top$ for $\beta_A$ and $e=(0,0,1)^\top$ for $\beta_1$.
    ``est'' represents $e^\top\hat V_{\text{est}}e$ with $e=(1,0,0)^\top$ for $\beta_0$, with $e=(0,1,0)^\top$ for $\beta_A$ and $e=(0,0,1)^\top$ for $\beta_1$.
    The sample size of trial data is $n=1000$. 
    The x-axis represents the sample size of historical data $\tilde n$.
    The y-axis represents the coverage probability which is the proportion of 1000 simulations in which the 95\% CI using each variance estimator includes the true value.}
\end{figure}

\clearpage
\subsection{Scenario C-4}

\begin{figure}[h]
    \centering
    \includegraphics[width=\linewidth]{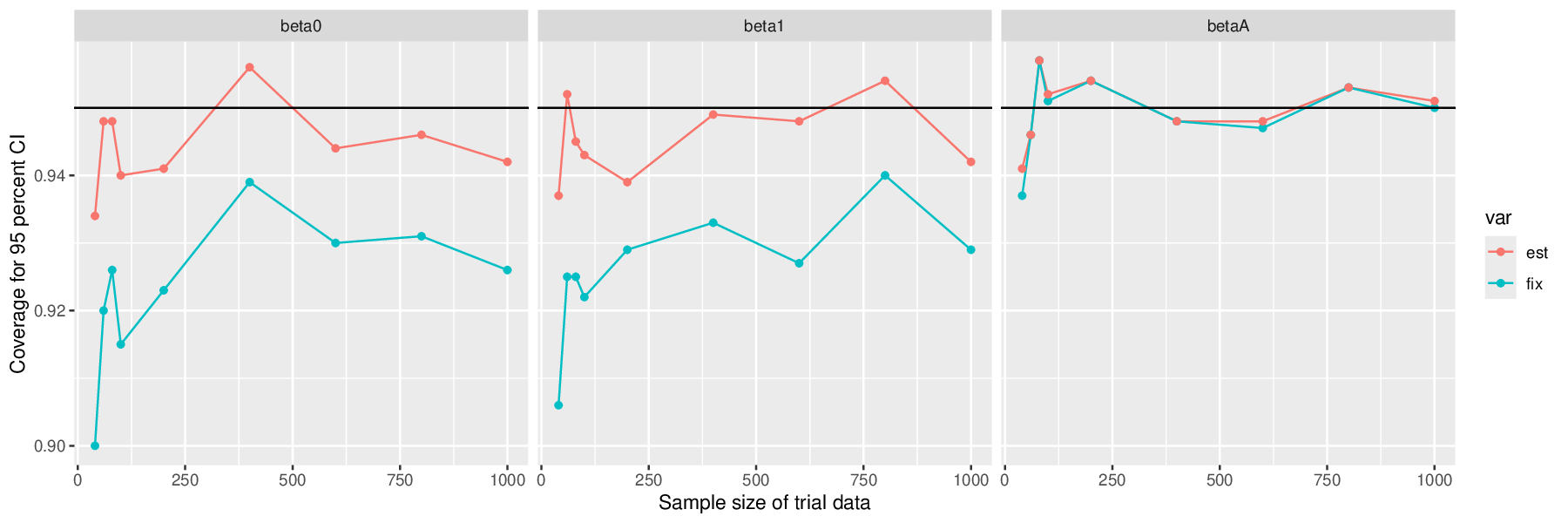}
    \caption{Plots of the coverage probability of 95\% CI over 1000 simulations for Scenario C-4.
    ``beta0'', ``betaA'', ``beta1'' represent the intercept $\beta_0$, the coefficient for the treatment assignment $\beta_A$, the coefficient for the prognostic score $\beta_1$ in the PROCOVA model \eqref{eq: PROCOVA linear}, respectively.
     ``fix'' represents $e^\top\hat V_{\text{fix}}e$ with $e=(1,0,0)^\top$ for $\beta_0$, with $e=(0,1,0)^\top$ for $\beta_A$ and $e=(0,0,1)^\top$ for $\beta_1$.
    ``est'' represents $e^\top\hat V_{\text{est}}e$ with $e=(1,0,0)^\top$ for $\beta_0$, with $e=(0,1,0)^\top$ for $\beta_A$ and $e=(0,0,1)^\top$ for $\beta_1$.
    The x-axis represents the sample size of trial data $n$.
    The sample size of historical data is $\tilde n = 10n$. 
    The y-axis represents the coverage probability which is the proportion of 1000 simulations in which the 95\% CI using each variance estimator includes the true value.
    }
\end{figure}

\clearpage
\begin{figure}[h]
    \centering
    \includegraphics[width=0.4\linewidth]{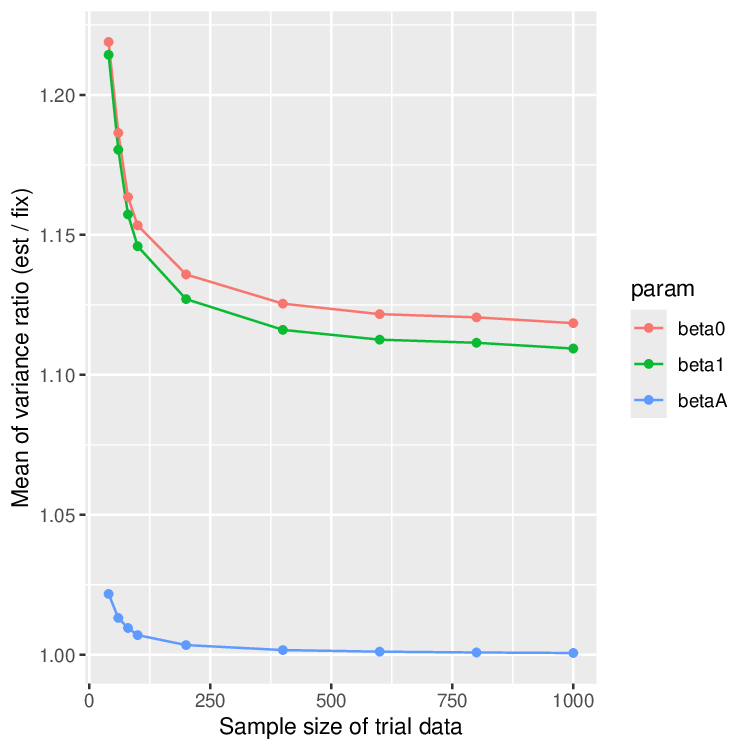}
    \caption{Plots of the mean of the ratio of two variance estimators over 1000 simulations for Scenario C-4.
    ``beta0'', ``betaA'', ``beta1'' represent the intercept $\beta_0$, the coefficient for the treatment assignment $\beta_A$, the coefficient for the prognostic score $\beta_1$ in the PROCOVA model \eqref{eq: PROCOVA linear}, respectively.
    The x-axis represents the sample size of trial data $n$.
    The sample size of historical data is $\tilde n = 10n$. 
    The y-axis represents the mean of the ratio of two variance estimators, i.e., $e^\top\hat V_{\text{est}}e/e^\top\hat V_{\text{fix}}e$ with $e=(1,0,0)^\top$ for $\beta_0$, with $e=(0,1,0)^\top$ for $\beta_A$ and $e=(0,0,1)^\top$ for $\beta_1$, over 1000 simulations.}
\end{figure}

\clearpage
\begin{figure}[h]
    \centering
    \includegraphics[width=0.4\linewidth]{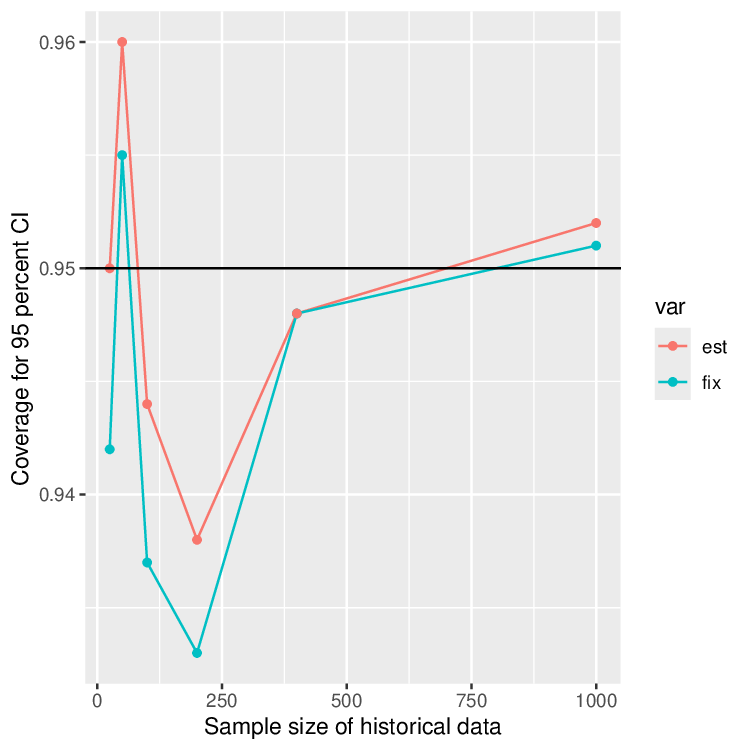}
    \caption{Plots of the coverage probability of the 95\% CI for $\beta_A$ in the PROCOVA model \eqref{eq: PROCOVA linear} over 1000 simulations for Scenario C-4.
     ``fix'' represents $e^\top\hat V_{\text{fix}}e$ with $e=(1,0,0)^\top$ for $\beta_0$, with $e=(0,1,0)^\top$ for $\beta_A$ and $e=(0,0,1)^\top$ for $\beta_1$.
    ``est'' represents $e^\top\hat V_{\text{est}}e$ with $e=(1,0,0)^\top$ for $\beta_0$, with $e=(0,1,0)^\top$ for $\beta_A$ and $e=(0,0,1)^\top$ for $\beta_1$.
    The sample size of trial data is $n=100$. 
    The x-axis represents the sample size of historical data $\tilde n$.
    The y-axis represents the coverage probability which is the proportion of 1000 simulations in which the 95\% CI using each variance estimator includes the true value.}
\end{figure}

\clearpage
\begin{figure}[h]
    \centering
    \includegraphics[width=\linewidth]{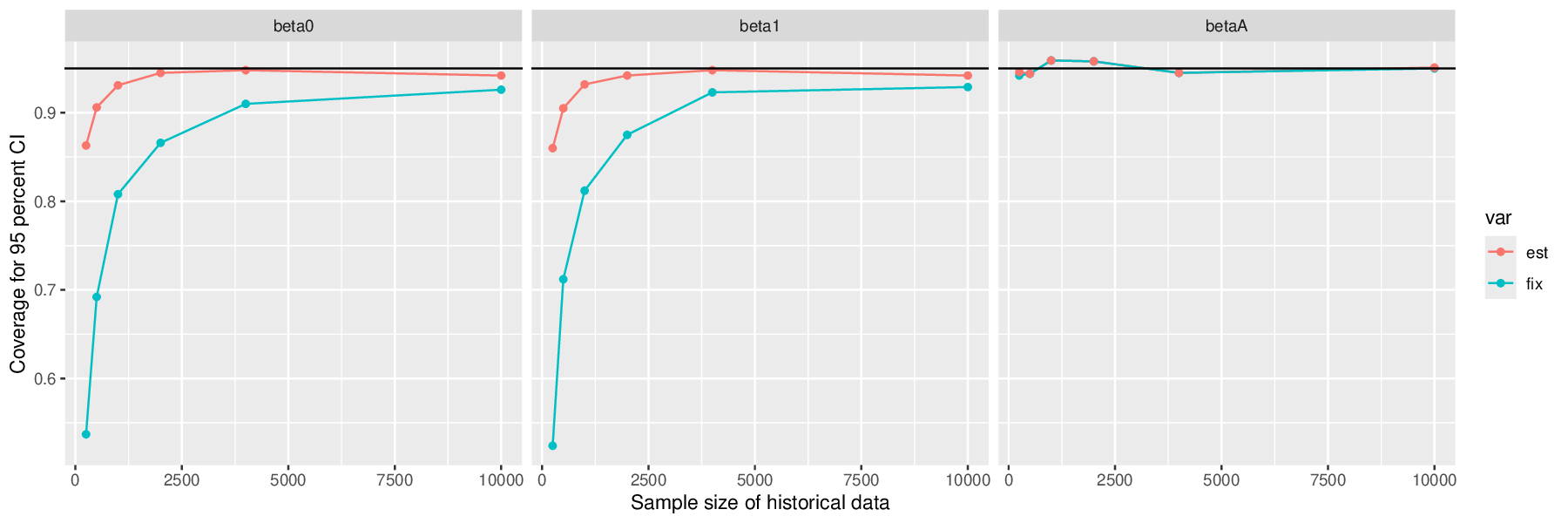}
    \caption{Plots of the coverage probability of 95\% CI over 1000 simulations for Scenario C-4.
    ``beta0'', ``betaA'', ``beta1'' represent the intercept $\beta_0$, the coefficient for the treatment assignment $\beta_A$, the coefficient for the prognostic score $\beta_1$ in the PROCOVA model \eqref{eq: PROCOVA linear}, respectively.
     ``fix'' represents $e^\top\hat V_{\text{fix}}e$ with $e=(1,0,0)^\top$ for $\beta_0$, with $e=(0,1,0)^\top$ for $\beta_A$ and $e=(0,0,1)^\top$ for $\beta_1$.
    ``est'' represents $e^\top\hat V_{\text{est}}e$ with $e=(1,0,0)^\top$ for $\beta_0$, with $e=(0,1,0)^\top$ for $\beta_A$ and $e=(0,0,1)^\top$ for $\beta_1$.
    The sample size of trial data is $n=1000$. 
    The x-axis represents the sample size of historical data $\tilde n$.
    The y-axis represents the coverage probability which is the proportion of 1000 simulations in which the 95\% CI using each variance estimator includes the true value.}
\end{figure}

\clearpage
\subsection{Scenario C-5}

\begin{figure}[h]
    \centering
    \includegraphics[width=\linewidth]{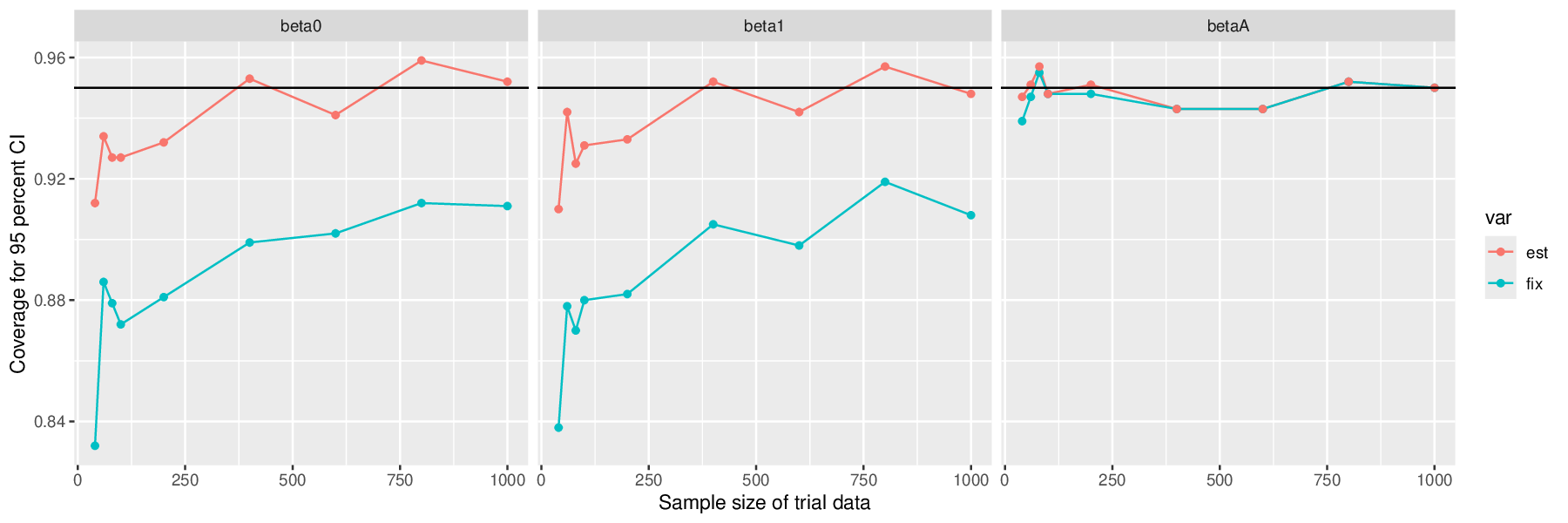}
    \caption{Plots of the coverage probability of 95\% CI over 1000 simulations for Scenario C-5.
    ``beta0'', ``betaA'', ``beta1'' represent the intercept $\beta_0$, the coefficient for the treatment assignment $\beta_A$, the coefficient for the prognostic score $\beta_1$ in the PROCOVA model \eqref{eq: PROCOVA linear}, respectively.
     ``fix'' represents $e^\top\hat V_{\text{fix}}e$ with $e=(1,0,0)^\top$ for $\beta_0$, with $e=(0,1,0)^\top$ for $\beta_A$ and $e=(0,0,1)^\top$ for $\beta_1$.
    ``est'' represents $e^\top\hat V_{\text{est}}e$ with $e=(1,0,0)^\top$ for $\beta_0$, with $e=(0,1,0)^\top$ for $\beta_A$ and $e=(0,0,1)^\top$ for $\beta_1$.
    The x-axis represents the sample size of trial data $n$.
    The sample size of historical data is $\tilde n = 10n$. 
    The y-axis represents the coverage probability which is the proportion of 1000 simulations in which the 95\% CI using each variance estimator includes the true value.
    }
\end{figure}

\clearpage
\begin{figure}[h]
    \centering
    \includegraphics[width=0.4\linewidth]{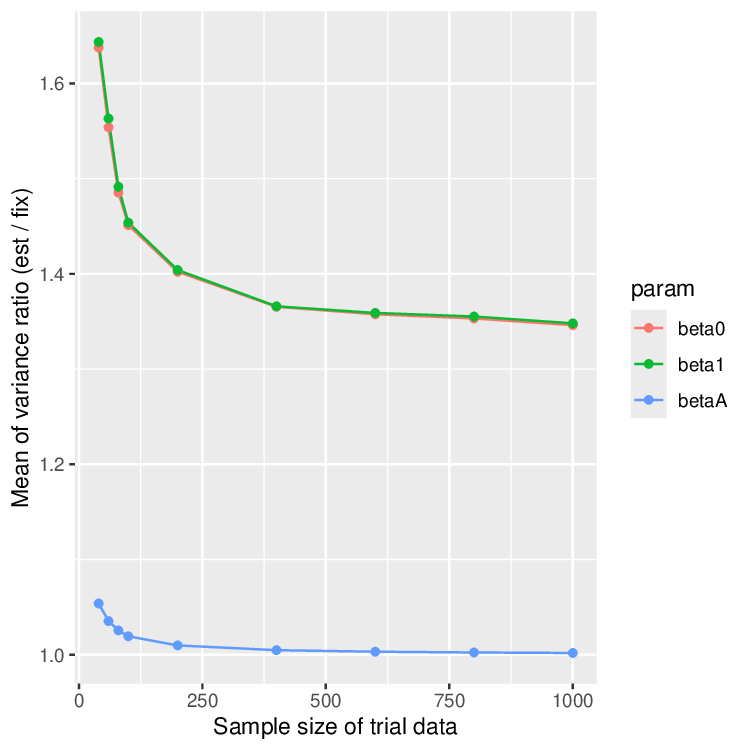}
    \caption{Plots of the mean of the ratio of two variance estimators over 1000 simulations for Scenario C-5.
    ``beta0'', ``betaA'', ``beta1'' represent the intercept $\beta_0$, the coefficient for the treatment assignment $\beta_A$, the coefficient for the prognostic score $\beta_1$ in the PROCOVA model \eqref{eq: PROCOVA linear}, respectively.
    The x-axis represents the sample size of trial data $n$.
    The sample size of historical data is $\tilde n = 10n$. 
    The y-axis represents the mean of the ratio of two variance estimators, i.e., $e^\top\hat V_{\text{est}}e/e^\top\hat V_{\text{fix}}e$ with $e=(1,0,0)^\top$ for $\beta_0$, with $e=(0,1,0)^\top$ for $\beta_A$ and $e=(0,0,1)^\top$ for $\beta_1$, over 1000 simulations.}
\end{figure}

\clearpage
\begin{figure}[h]
    \centering
    \includegraphics[width=0.4\linewidth]{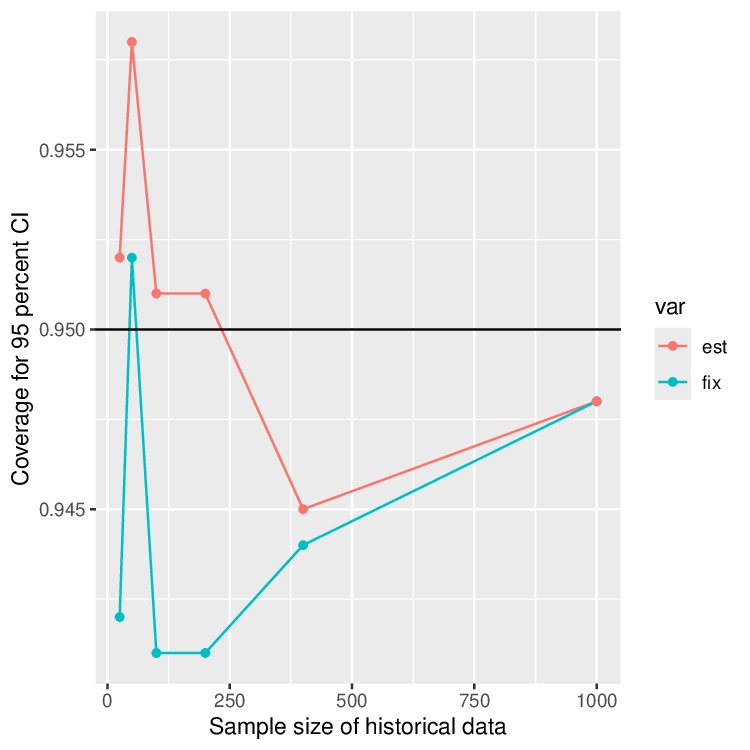}
    \caption{Plots of the coverage probability of the 95\% CI for $\beta_A$ in the PROCOVA model \eqref{eq: PROCOVA linear} over 1000 simulations for Scenario C-5.
     ``fix'' represents $e^\top\hat V_{\text{fix}}e$ with $e=(1,0,0)^\top$ for $\beta_0$, with $e=(0,1,0)^\top$ for $\beta_A$ and $e=(0,0,1)^\top$ for $\beta_1$.
    ``est'' represents $e^\top\hat V_{\text{est}}e$ with $e=(1,0,0)^\top$ for $\beta_0$, with $e=(0,1,0)^\top$ for $\beta_A$ and $e=(0,0,1)^\top$ for $\beta_1$.
    The sample size of trial data is $n=100$. 
    The x-axis represents the sample size of historical data $\tilde n$.
    The y-axis represents the coverage probability which is the proportion of 1000 simulations in which the 95\% CI using each variance estimator includes the true value.}
\end{figure}

\clearpage
\begin{figure}[h]
    \centering
    \includegraphics[width=\linewidth]{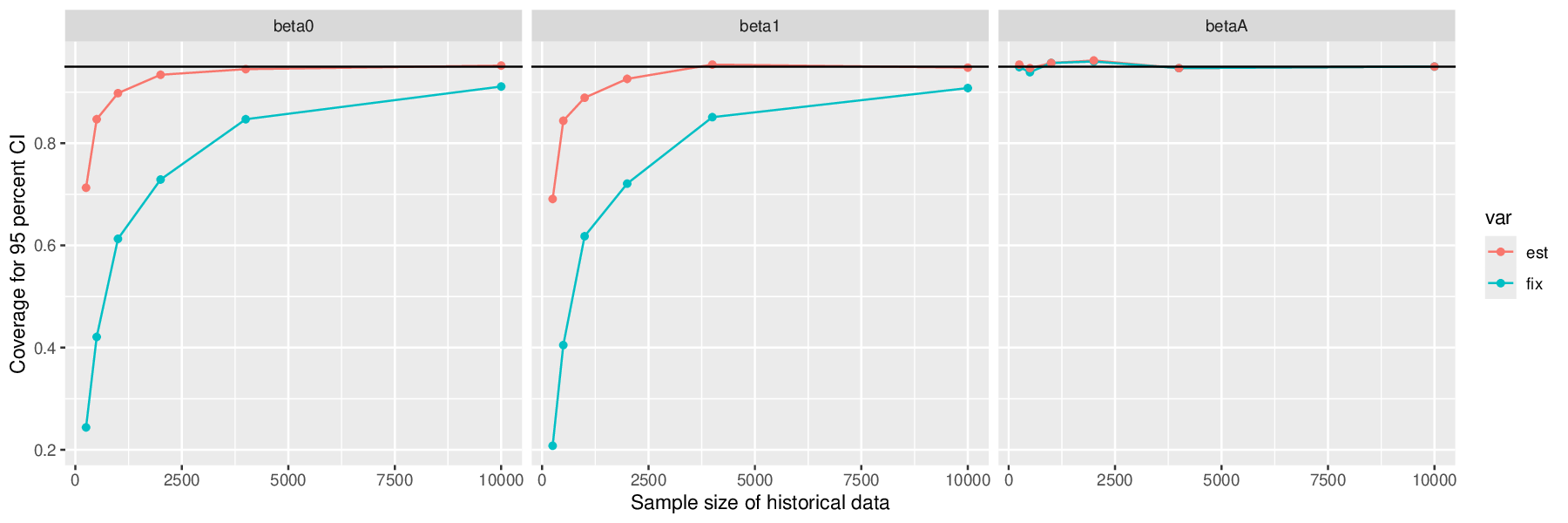}
    \caption{Plots of the coverage probability of 95\% CI over 1000 simulations for Scenario C-5.
    ``beta0'', ``betaA'', ``beta1'' represent the intercept $\beta_0$, the coefficient for the treatment assignment $\beta_A$, the coefficient for the prognostic score $\beta_1$ in the PROCOVA model \eqref{eq: PROCOVA linear}, respectively.
     ``fix'' represents $e^\top\hat V_{\text{fix}}e$ with $e=(1,0,0)^\top$ for $\beta_0$, with $e=(0,1,0)^\top$ for $\beta_A$ and $e=(0,0,1)^\top$ for $\beta_1$.
    ``est'' represents $e^\top\hat V_{\text{est}}e$ with $e=(1,0,0)^\top$ for $\beta_0$, with $e=(0,1,0)^\top$ for $\beta_A$ and $e=(0,0,1)^\top$ for $\beta_1$.
    The sample size of trial data is $n=1000$. 
    The x-axis represents the sample size of historical data $\tilde n$.
    The y-axis represents the coverage probability which is the proportion of 1000 simulations in which the 95\% CI using each variance estimator includes the true value.}
\end{figure}

\clearpage
\subsection{Scenario C-6}

\begin{figure}[h]
    \centering
    \includegraphics[width=\linewidth]{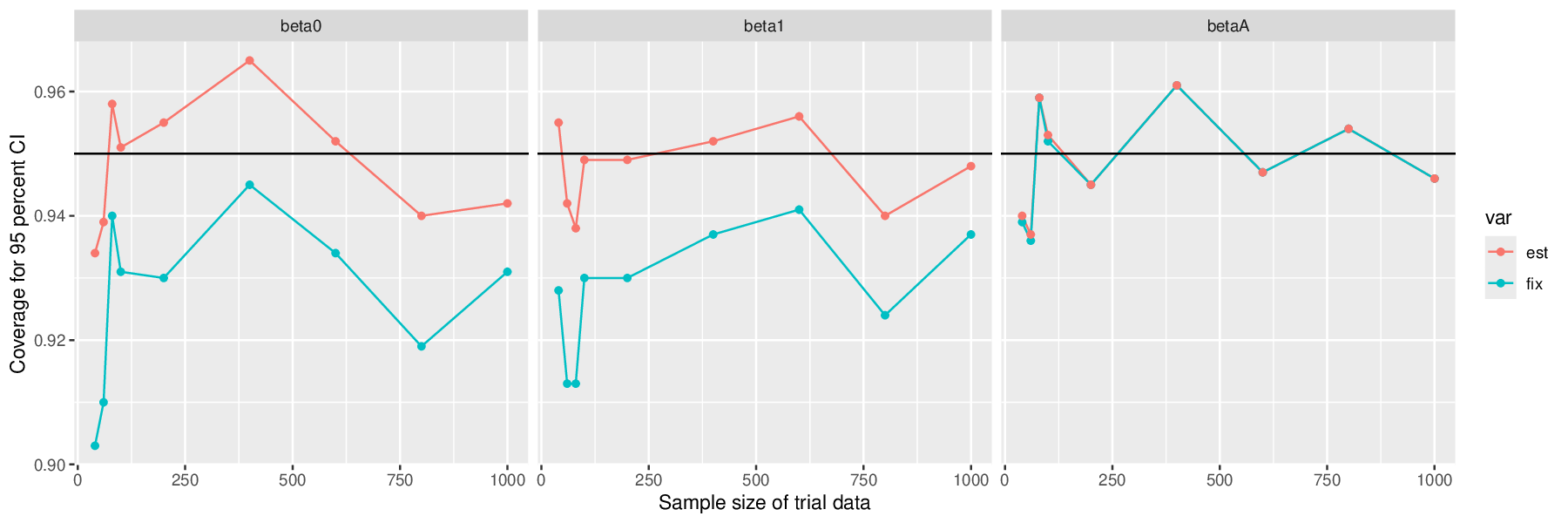}
    \caption{Plots of the coverage probability of 95\% CI over 1000 simulations for Scenario C-6.
    ``beta0'', ``betaA'', ``beta1'' represent the intercept $\beta_0$, the coefficient for the treatment assignment $\beta_A$, the coefficient for the prognostic score $\beta_1$ in the PROCOVA model \eqref{eq: PROCOVA linear}, respectively.
     ``fix'' represents $e^\top\hat V_{\text{fix}}e$ with $e=(1,0,0)^\top$ for $\beta_0$, with $e=(0,1,0)^\top$ for $\beta_A$ and $e=(0,0,1)^\top$ for $\beta_1$.
    ``est'' represents $e^\top\hat V_{\text{est}}e$ with $e=(1,0,0)^\top$ for $\beta_0$, with $e=(0,1,0)^\top$ for $\beta_A$ and $e=(0,0,1)^\top$ for $\beta_1$.
    The x-axis represents the sample size of trial data $n$.
    The sample size of historical data is $\tilde n = 10n$. 
    The y-axis represents the coverage probability which is the proportion of 1000 simulations in which the 95\% CI using each variance estimator includes the true value.
    }
\end{figure}

\clearpage
\begin{figure}[h]
    \centering
    \includegraphics[width=0.4\linewidth]{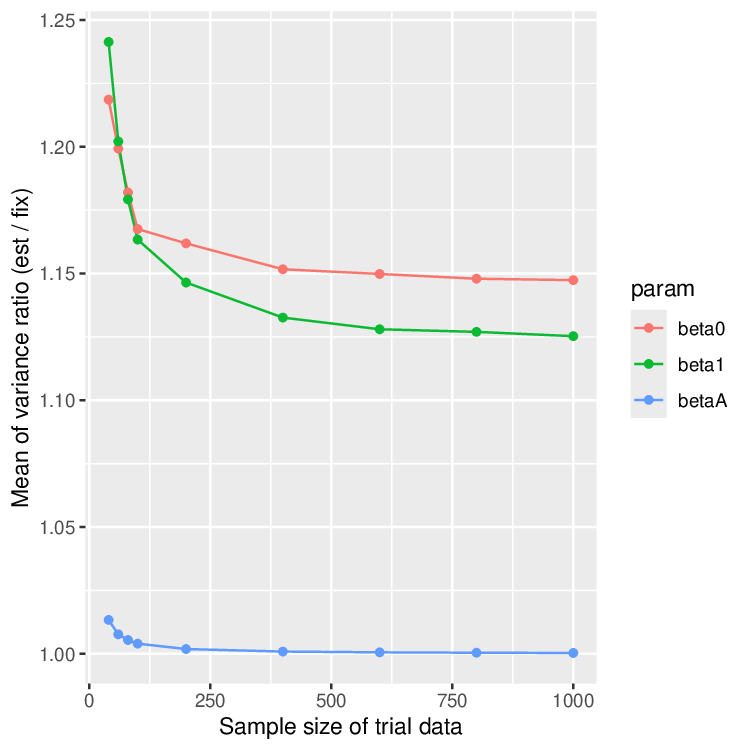}
    \caption{Plots of the mean of the ratio of two variance estimators over 1000 simulations for Scenario C-6.
    ``beta0'', ``betaA'', ``beta1'' represent the intercept $\beta_0$, the coefficient for the treatment assignment $\beta_A$, the coefficient for the prognostic score $\beta_1$ in the PROCOVA model \eqref{eq: PROCOVA linear}, respectively.
    The x-axis represents the sample size of trial data $n$.
    The sample size of historical data is $\tilde n = 10n$. 
    The y-axis represents the mean of the ratio of two variance estimators, i.e., $e^\top\hat V_{\text{est}}e/e^\top\hat V_{\text{fix}}e$ with $e=(1,0,0)^\top$ for $\beta_0$, with $e=(0,1,0)^\top$ for $\beta_A$ and $e=(0,0,1)^\top$ for $\beta_1$, over 1000 simulations.}
\end{figure}

\clearpage
\begin{figure}[h]
    \centering
    \includegraphics[width=0.4\linewidth]{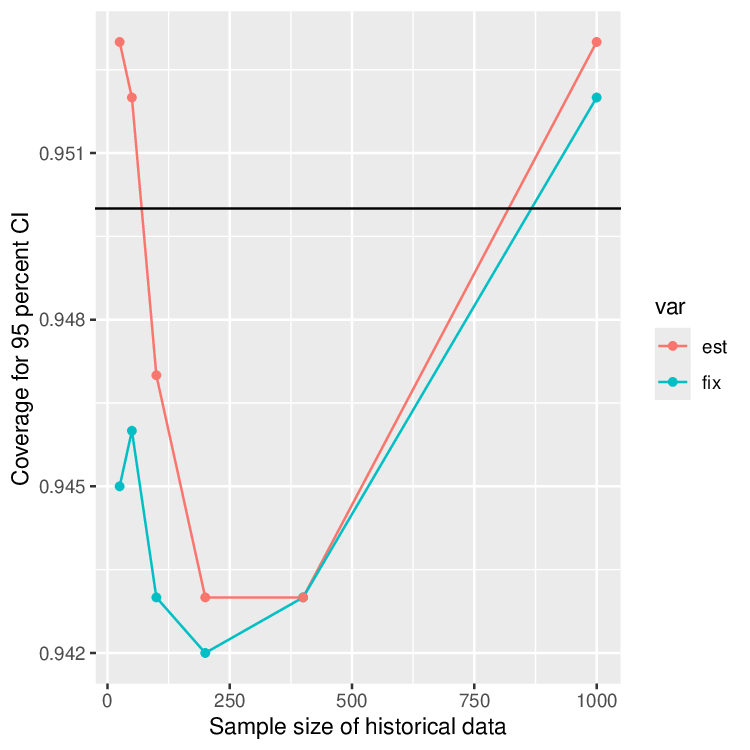}
    \caption{Plots of the coverage probability of the 95\% CI for $\beta_A$ in the PROCOVA model \eqref{eq: PROCOVA linear} over 1000 simulations for Scenario C-6.
     ``fix'' represents $e^\top\hat V_{\text{fix}}e$ with $e=(1,0,0)^\top$ for $\beta_0$, with $e=(0,1,0)^\top$ for $\beta_A$ and $e=(0,0,1)^\top$ for $\beta_1$.
    ``est'' represents $e^\top\hat V_{\text{est}}e$ with $e=(1,0,0)^\top$ for $\beta_0$, with $e=(0,1,0)^\top$ for $\beta_A$ and $e=(0,0,1)^\top$ for $\beta_1$.
    The sample size of trial data is $n=100$. 
    The x-axis represents the sample size of historical data $\tilde n$.
    The y-axis represents the coverage probability which is the proportion of 1000 simulations in which the 95\% CI using each variance estimator includes the true value.}
\end{figure}

\clearpage
\begin{figure}[h]
    \centering
    \includegraphics[width=\linewidth]{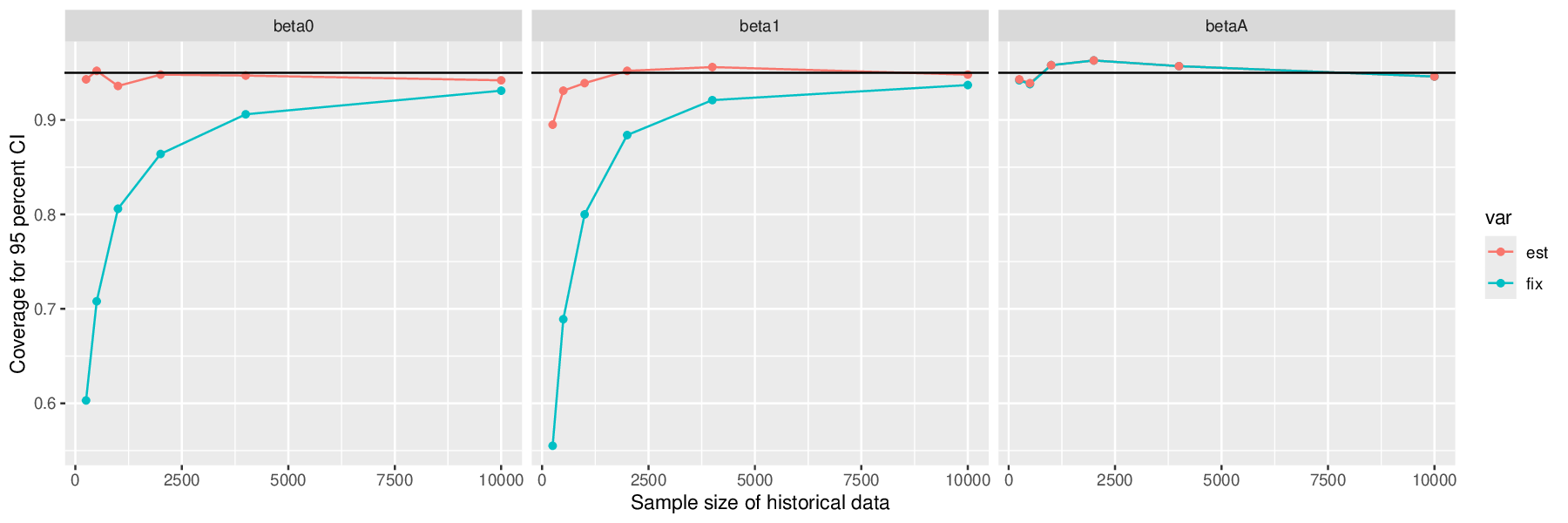}
    \caption{Plots of the coverage probability of 95\% CI over 1000 simulations for Scenario C-6.
    ``beta0'', ``betaA'', ``beta1'' represent the intercept $\beta_0$, the coefficient for the treatment assignment $\beta_A$, the coefficient for the prognostic score $\beta_1$ in the PROCOVA model \eqref{eq: PROCOVA linear}, respectively.
     ``fix'' represents $e^\top\hat V_{\text{fix}}e$ with $e=(1,0,0)^\top$ for $\beta_0$, with $e=(0,1,0)^\top$ for $\beta_A$ and $e=(0,0,1)^\top$ for $\beta_1$.
    ``est'' represents $e^\top\hat V_{\text{est}}e$ with $e=(1,0,0)^\top$ for $\beta_0$, with $e=(0,1,0)^\top$ for $\beta_A$ and $e=(0,0,1)^\top$ for $\beta_1$.
    The sample size of trial data is $n=1000$. 
    The x-axis represents the sample size of historical data $\tilde n$.
    The y-axis represents the coverage probability which is the proportion of 1000 simulations in which the 95\% CI using each variance estimator includes the true value.}
\end{figure}

\clearpage
\subsection{Scenario C-7}

\begin{figure}[h]
    \centering
    \includegraphics[width=\linewidth]{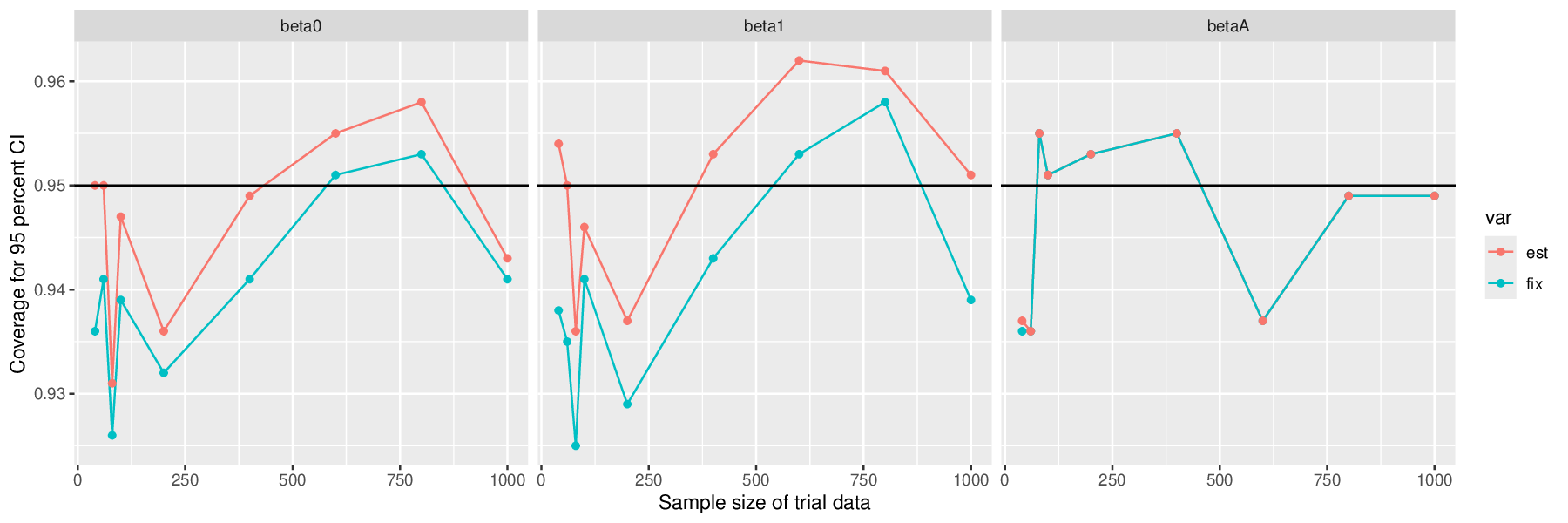}
    \caption{Plots of the coverage probability of 95\% CI over 1000 simulations for Scenario C-7.
    ``beta0'', ``betaA'', ``beta1'' represent the intercept $\beta_0$, the coefficient for the treatment assignment $\beta_A$, the coefficient for the prognostic score $\beta_1$ in the PROCOVA model \eqref{eq: PROCOVA linear}, respectively.
     ``fix'' represents $e^\top\hat V_{\text{fix}}e$ with $e=(1,0,0)^\top$ for $\beta_0$, with $e=(0,1,0)^\top$ for $\beta_A$ and $e=(0,0,1)^\top$ for $\beta_1$.
    ``est'' represents $e^\top\hat V_{\text{est}}e$ with $e=(1,0,0)^\top$ for $\beta_0$, with $e=(0,1,0)^\top$ for $\beta_A$ and $e=(0,0,1)^\top$ for $\beta_1$.
    The x-axis represents the sample size of trial data $n$.
    The sample size of historical data is $\tilde n = 10n$. 
    The y-axis represents the coverage probability which is the proportion of 1000 simulations in which the 95\% CI using each variance estimator includes the true value.
    }
\end{figure}

\clearpage
\begin{figure}[h]
    \centering
    \includegraphics[width=0.4\linewidth]{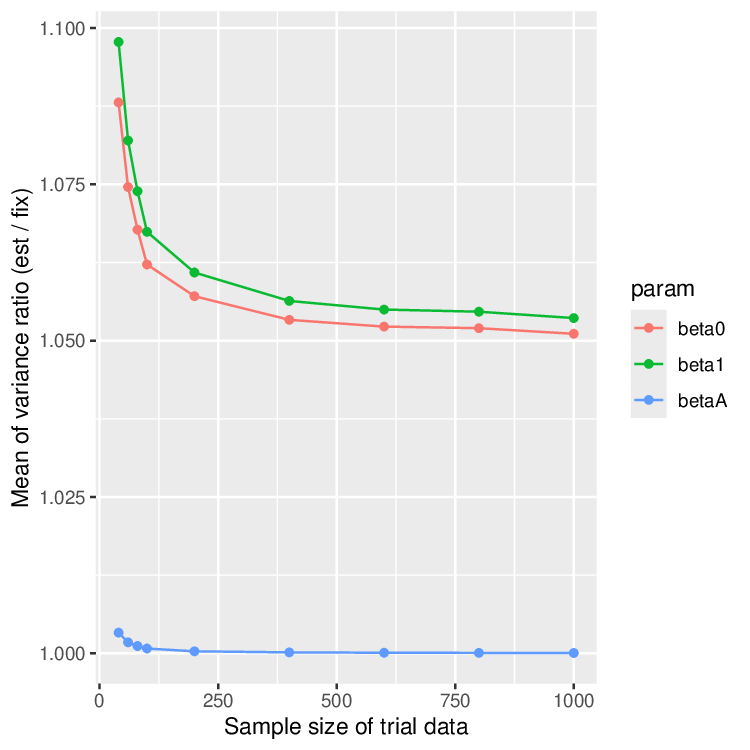}
    \caption{Plots of the mean of the ratio of two variance estimators over 1000 simulations for Scenario C-7.
    ``beta0'', ``betaA'', ``beta1'' represent the intercept $\beta_0$, the coefficient for the treatment assignment $\beta_A$, the coefficient for the prognostic score $\beta_1$ in the PROCOVA model \eqref{eq: PROCOVA linear}, respectively.
    The x-axis represents the sample size of trial data $n$.
    The sample size of historical data is $\tilde n = 10n$. 
    The y-axis represents the mean of the ratio of two variance estimators, i.e., $e^\top\hat V_{\text{est}}e/e^\top\hat V_{\text{fix}}e$ with $e=(1,0,0)^\top$ for $\beta_0$, with $e=(0,1,0)^\top$ for $\beta_A$ and $e=(0,0,1)^\top$ for $\beta_1$, over 1000 simulations.}
\end{figure}
\clearpage
\begin{figure}[h]
    \centering
    \includegraphics[width=0.4\linewidth]{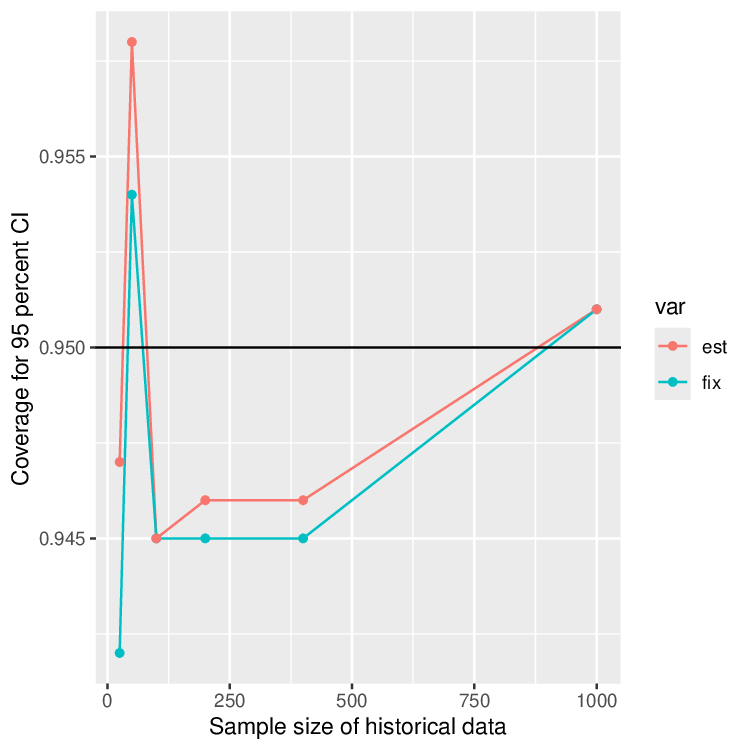}
    \caption{Plots of the coverage probability of the 95\% CI for $\beta_A$ in the PROCOVA model \eqref{eq: PROCOVA linear} over 1000 simulations for Scenario C-7.
     ``fix'' represents $e^\top\hat V_{\text{fix}}e$ with $e=(1,0,0)^\top$ for $\beta_0$, with $e=(0,1,0)^\top$ for $\beta_A$ and $e=(0,0,1)^\top$ for $\beta_1$.
    ``est'' represents $e^\top\hat V_{\text{est}}e$ with $e=(1,0,0)^\top$ for $\beta_0$, with $e=(0,1,0)^\top$ for $\beta_A$ and $e=(0,0,1)^\top$ for $\beta_1$.
    The sample size of trial data is $n=100$. 
    The x-axis represents the sample size of historical data $\tilde n$.
    The y-axis represents the coverage probability which is the proportion of 1000 simulations in which the 95\% CI using each variance estimator includes the true value.}
\end{figure}
\clearpage
\begin{figure}[h]
    \centering
    \includegraphics[width=\linewidth]{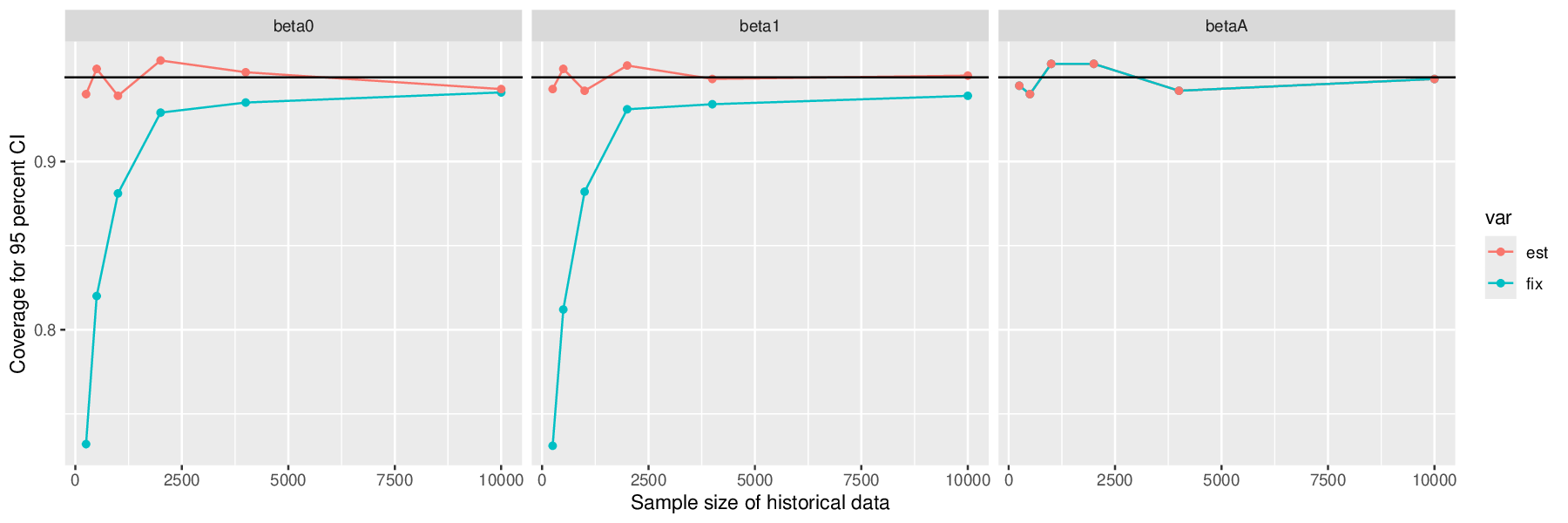}
    \caption{Plots of the coverage probability of 95\% CI over 1000 simulations for Scenario C-7.
    ``beta0'', ``betaA'', ``beta1'' represent the intercept $\beta_0$, the coefficient for the treatment assignment $\beta_A$, the coefficient for the prognostic score $\beta_1$ in the PROCOVA model \eqref{eq: PROCOVA linear}, respectively.
     ``fix'' represents $e^\top\hat V_{\text{fix}}e$ with $e=(1,0,0)^\top$ for $\beta_0$, with $e=(0,1,0)^\top$ for $\beta_A$ and $e=(0,0,1)^\top$ for $\beta_1$.
    ``est'' represents $e^\top\hat V_{\text{est}}e$ with $e=(1,0,0)^\top$ for $\beta_0$, with $e=(0,1,0)^\top$ for $\beta_A$ and $e=(0,0,1)^\top$ for $\beta_1$.
    The sample size of trial data is $n=1000$. 
    The x-axis represents the sample size of historical data $\tilde n$.
    The y-axis represents the coverage probability which is the proportion of 1000 simulations in which the 95\% CI using each variance estimator includes the true value.}
\end{figure}

\clearpage
\subsection{Scenario C-8}

\begin{figure}[h]
    \centering
    \includegraphics[width=\linewidth]{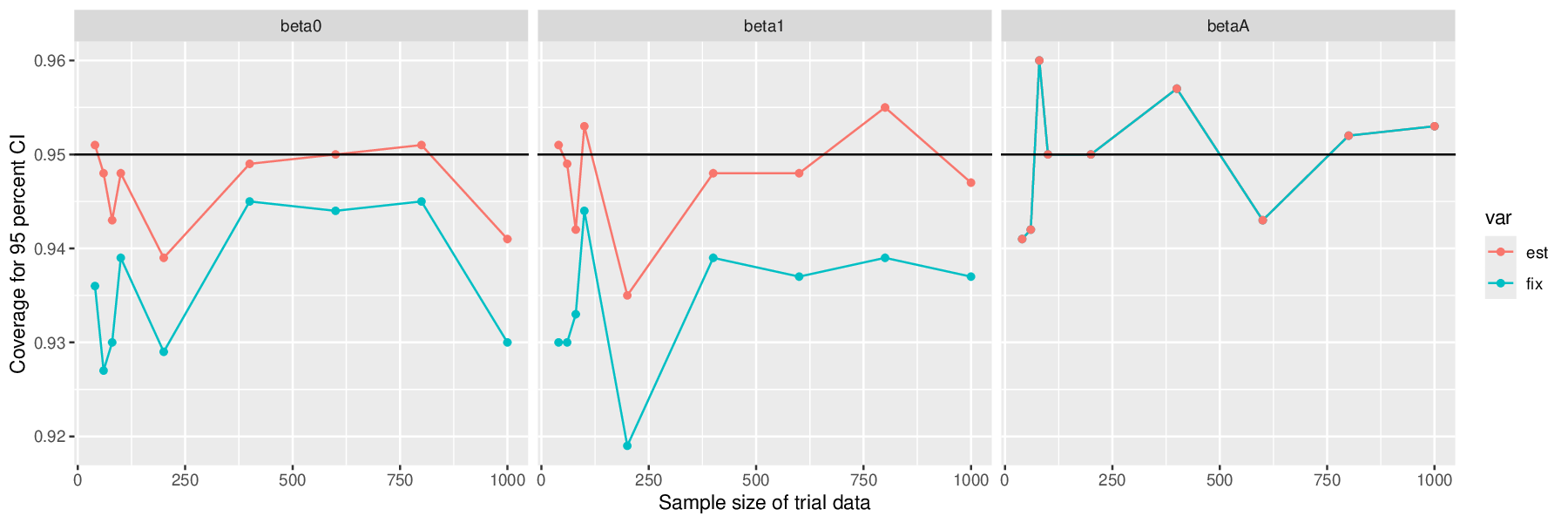}
    \caption{Plots of the coverage probability of 95\% CI over 1000 simulations for Scenario C-8.
    ``beta0'', ``betaA'', ``beta1'' represent the intercept $\beta_0$, the coefficient for the treatment assignment $\beta_A$, the coefficient for the prognostic score $\beta_1$ in the PROCOVA model \eqref{eq: PROCOVA linear}, respectively.
     ``fix'' represents $e^\top\hat V_{\text{fix}}e$ with $e=(1,0,0)^\top$ for $\beta_0$, with $e=(0,1,0)^\top$ for $\beta_A$ and $e=(0,0,1)^\top$ for $\beta_1$.
    ``est'' represents $e^\top\hat V_{\text{est}}e$ with $e=(1,0,0)^\top$ for $\beta_0$, with $e=(0,1,0)^\top$ for $\beta_A$ and $e=(0,0,1)^\top$ for $\beta_1$.
    The x-axis represents the sample size of trial data $n$.
    The sample size of historical data is $\tilde n = 10n$. 
    The y-axis represents the coverage probability which is the proportion of 1000 simulations in which the 95\% CI using each variance estimator includes the true value.
    }
\end{figure}

\clearpage
\begin{figure}[h]
    \centering
    \includegraphics[width=0.4\linewidth]{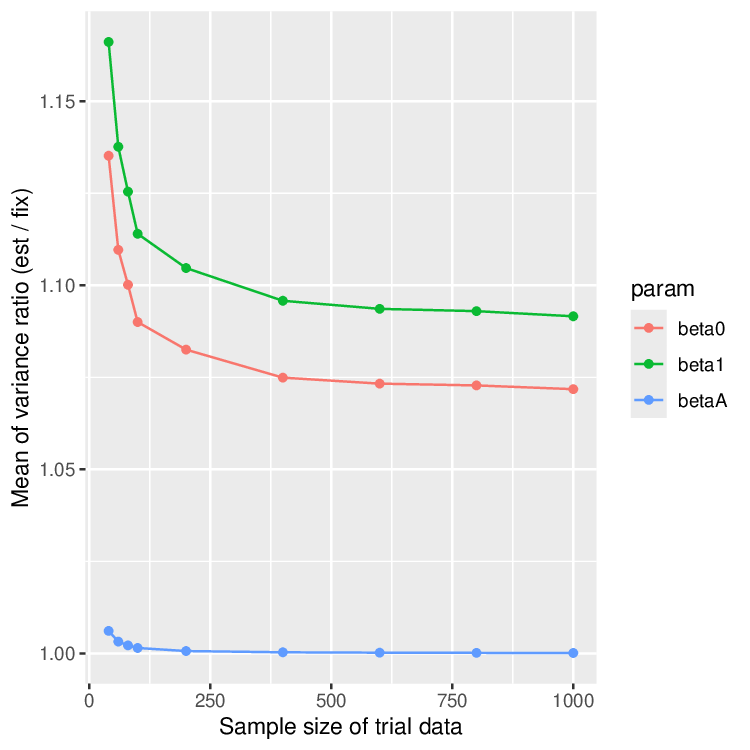}
    \caption{Plots of the mean of the ratio of two variance estimators over 1000 simulations for Scenario C-8.
    ``beta0'', ``betaA'', ``beta1'' represent the intercept $\beta_0$, the coefficient for the treatment assignment $\beta_A$, the coefficient for the prognostic score $\beta_1$ in the PROCOVA model \eqref{eq: PROCOVA linear}, respectively.
    The x-axis represents the sample size of trial data $n$.
    The sample size of historical data is $\tilde n = 10n$. 
    The y-axis represents the mean of the ratio of two variance estimators, i.e., $e^\top\hat V_{\text{est}}e/e^\top\hat V_{\text{fix}}e$ with $e=(1,0,0)^\top$ for $\beta_0$, with $e=(0,1,0)^\top$ for $\beta_A$ and $e=(0,0,1)^\top$ for $\beta_1$, over 1000 simulations.}
\end{figure}
\clearpage
\begin{figure}[h]
    \centering
    \includegraphics[width=0.4\linewidth]{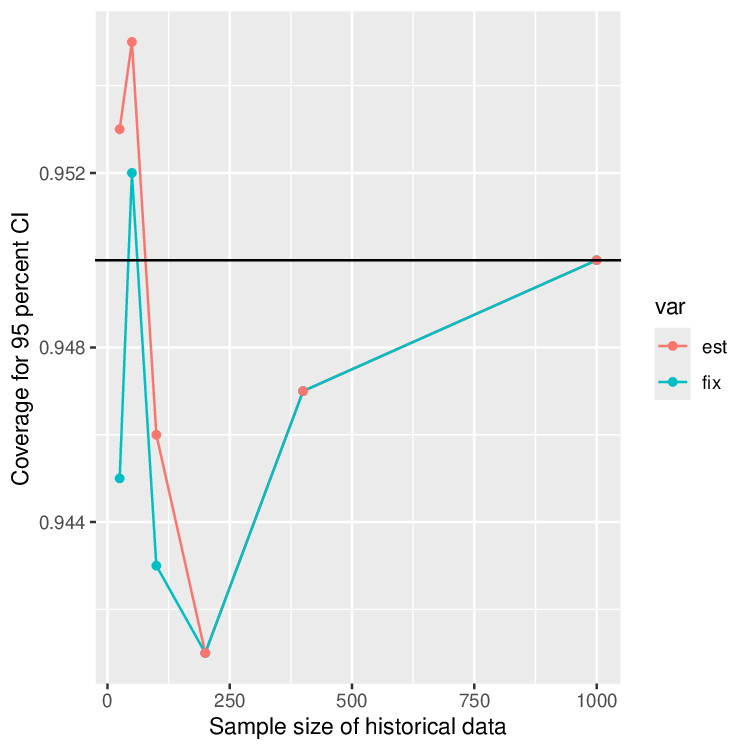}
    \caption{Plots of the coverage probability of the 95\% CI for $\beta_A$ in the PROCOVA model \eqref{eq: PROCOVA linear} over 1000 simulations for Scenario C-8.
     ``fix'' represents $e^\top\hat V_{\text{fix}}e$ with $e=(1,0,0)^\top$ for $\beta_0$, with $e=(0,1,0)^\top$ for $\beta_A$ and $e=(0,0,1)^\top$ for $\beta_1$.
    ``est'' represents $e^\top\hat V_{\text{est}}e$ with $e=(1,0,0)^\top$ for $\beta_0$, with $e=(0,1,0)^\top$ for $\beta_A$ and $e=(0,0,1)^\top$ for $\beta_1$.
    The sample size of trial data is $n=100$. 
    The x-axis represents the sample size of historical data $\tilde n$.
    The y-axis represents the coverage probability which is the proportion of 1000 simulations in which the 95\% CI using each variance estimator includes the true value.}
\end{figure}
\clearpage
\begin{figure}[h]
    \centering
    \includegraphics[width=\linewidth]{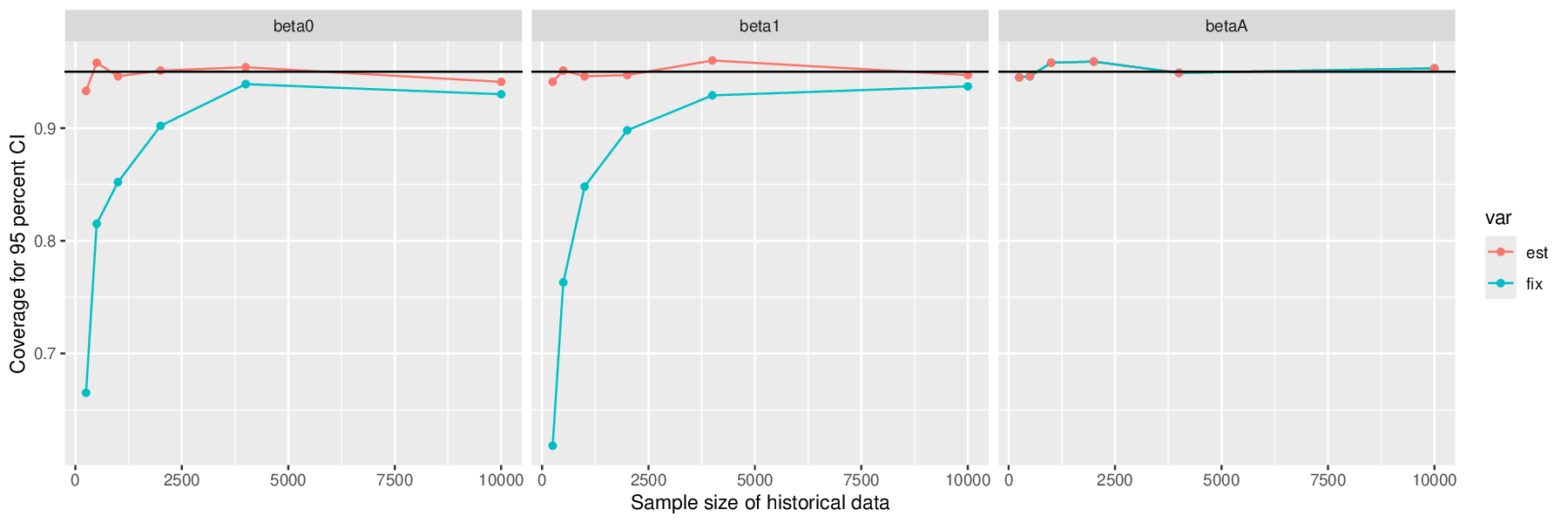}
    \caption{Plots of the coverage probability of 95\% CI over 1000 simulations for Scenario C-8.
    ``beta0'', ``betaA'', ``beta1'' represent the intercept $\beta_0$, the coefficient for the treatment assignment $\beta_A$, the coefficient for the prognostic score $\beta_1$ in the PROCOVA model \eqref{eq: PROCOVA linear}, respectively.
     ``fix'' represents $e^\top\hat V_{\text{fix}}e$ with $e=(1,0,0)^\top$ for $\beta_0$, with $e=(0,1,0)^\top$ for $\beta_A$ and $e=(0,0,1)^\top$ for $\beta_1$.
    ``est'' represents $e^\top\hat V_{\text{est}}e$ with $e=(1,0,0)^\top$ for $\beta_0$, with $e=(0,1,0)^\top$ for $\beta_A$ and $e=(0,0,1)^\top$ for $\beta_1$.
    The sample size of trial data is $n=1000$. 
    The x-axis represents the sample size of historical data $\tilde n$.
    The y-axis represents the coverage probability which is the proportion of 1000 simulations in which the 95\% CI using each variance estimator includes the true value.}
\end{figure}

\clearpage
\subsection{Scenario C-9}

\begin{figure}[h]
    \centering
    \includegraphics[width=\linewidth]{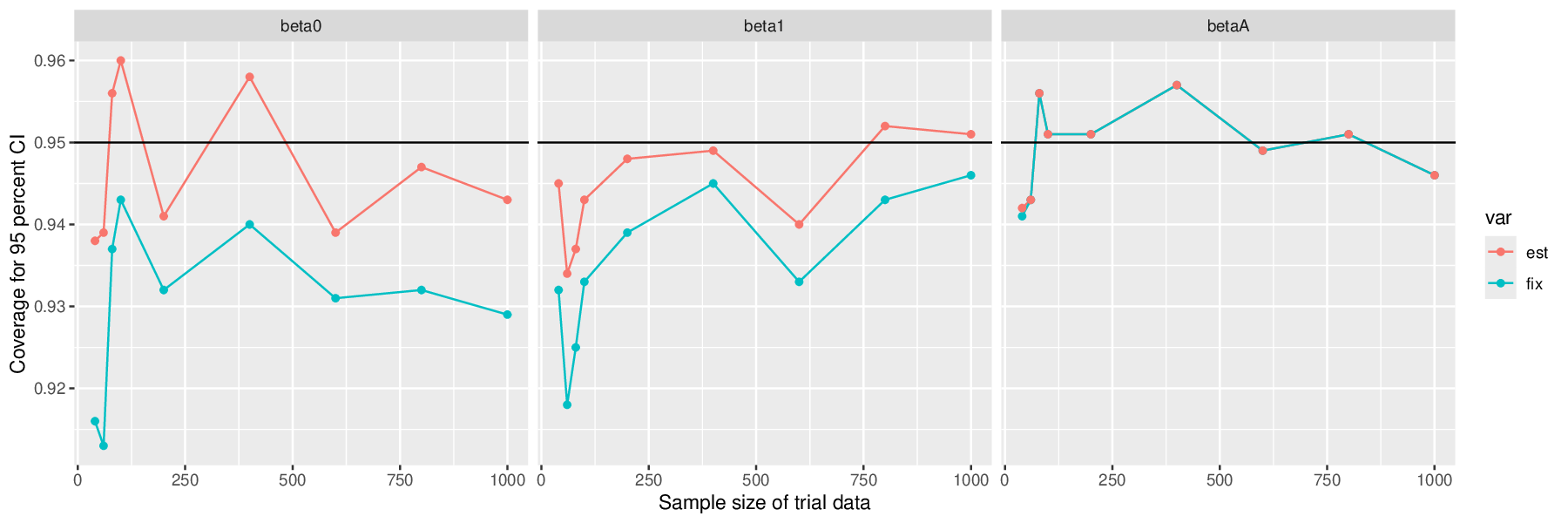}
    \caption{Plots of the coverage probability of 95\% CI over 1000 simulations for Scenario C-9.
    ``beta0'', ``betaA'', ``beta1'' represent the intercept $\beta_0$, the coefficient for the treatment assignment $\beta_A$, the coefficient for the prognostic score $\beta_1$ in the PROCOVA model \eqref{eq: PROCOVA linear}, respectively.
     ``fix'' represents $e^\top\hat V_{\text{fix}}e$ with $e=(1,0,0)^\top$ for $\beta_0$, with $e=(0,1,0)^\top$ for $\beta_A$ and $e=(0,0,1)^\top$ for $\beta_1$.
    ``est'' represents $e^\top\hat V_{\text{est}}e$ with $e=(1,0,0)^\top$ for $\beta_0$, with $e=(0,1,0)^\top$ for $\beta_A$ and $e=(0,0,1)^\top$ for $\beta_1$.
    The x-axis represents the sample size of trial data $n$.
    The sample size of historical data is $\tilde n = 10n$. 
    The y-axis represents the coverage probability which is the proportion of 1000 simulations in which the 95\% CI using each variance estimator includes the true value.
    }
\end{figure}

\clearpage
\begin{figure}[h]
    \centering
    \includegraphics[width=0.4\linewidth]{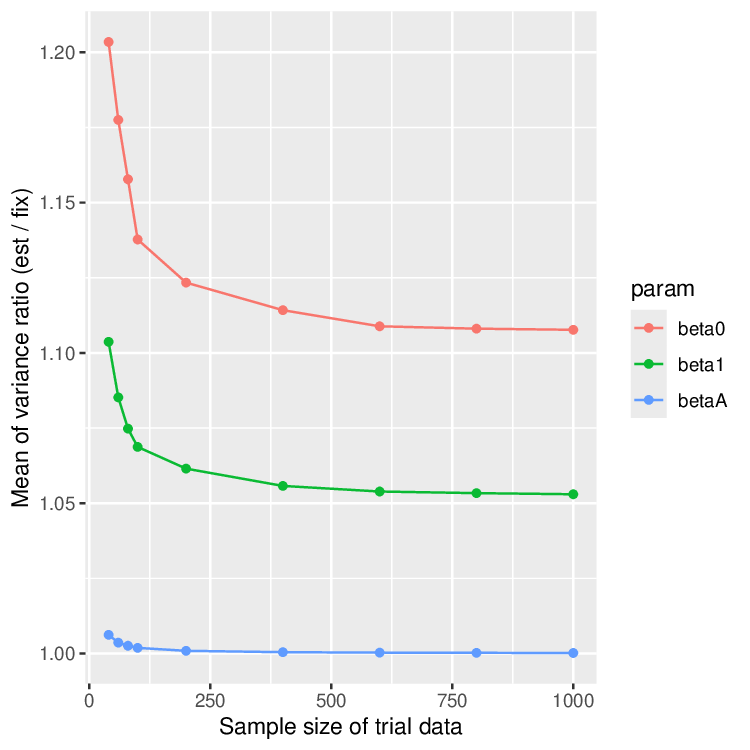}
    \caption{Plots of the mean of the ratio of two variance estimators over 1000 simulations for Scenario C-9.
    ``beta0'', ``betaA'', ``beta1'' represent the intercept $\beta_0$, the coefficient for the treatment assignment $\beta_A$, the coefficient for the prognostic score $\beta_1$ in the PROCOVA model \eqref{eq: PROCOVA linear}, respectively.
    The x-axis represents the sample size of trial data $n$.
    The sample size of historical data is $\tilde n = 10n$. 
    The y-axis represents the mean of the ratio of two variance estimators, i.e., $e^\top\hat V_{\text{est}}e/e^\top\hat V_{\text{fix}}e$ with $e=(1,0,0)^\top$ for $\beta_0$, with $e=(0,1,0)^\top$ for $\beta_A$ and $e=(0,0,1)^\top$ for $\beta_1$, over 1000 simulations.}
\end{figure}
\clearpage
\begin{figure}[h]
    \centering
    \includegraphics[width=0.4\linewidth]{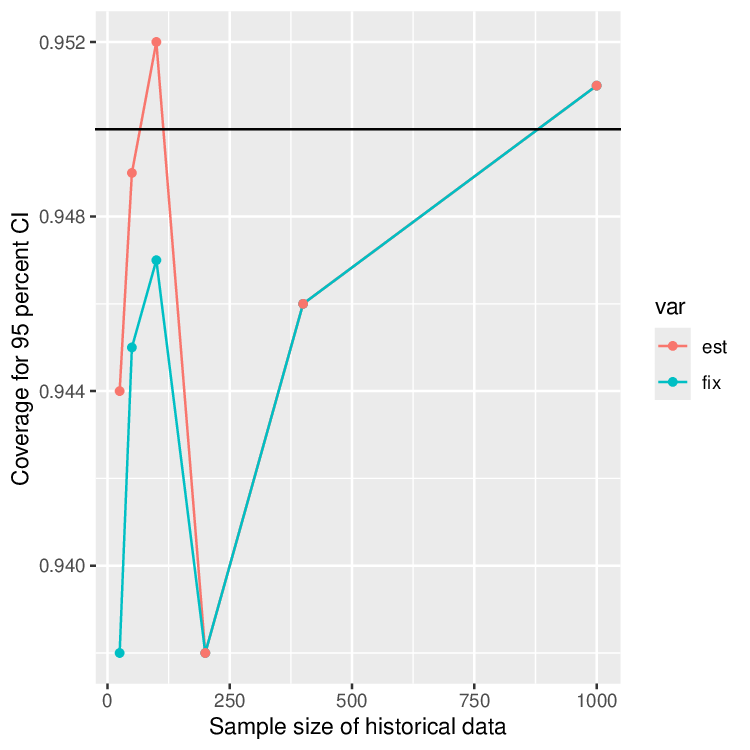}
    \caption{Plots of the coverage probability of the 95\% CI for $\beta_A$ in the PROCOVA model \eqref{eq: PROCOVA linear} over 1000 simulations for Scenario C-9.
     ``fix'' represents $e^\top\hat V_{\text{fix}}e$ with $e=(1,0,0)^\top$ for $\beta_0$, with $e=(0,1,0)^\top$ for $\beta_A$ and $e=(0,0,1)^\top$ for $\beta_1$.
    ``est'' represents $e^\top\hat V_{\text{est}}e$ with $e=(1,0,0)^\top$ for $\beta_0$, with $e=(0,1,0)^\top$ for $\beta_A$ and $e=(0,0,1)^\top$ for $\beta_1$.
    The sample size of trial data is $n=100$. 
    The x-axis represents the sample size of historical data $\tilde n$.
    The y-axis represents the coverage probability which is the proportion of 1000 simulations in which the 95\% CI using each variance estimator includes the true value.}
\end{figure}
\clearpage
\begin{figure}[h]
    \centering
    \includegraphics[width=\linewidth]{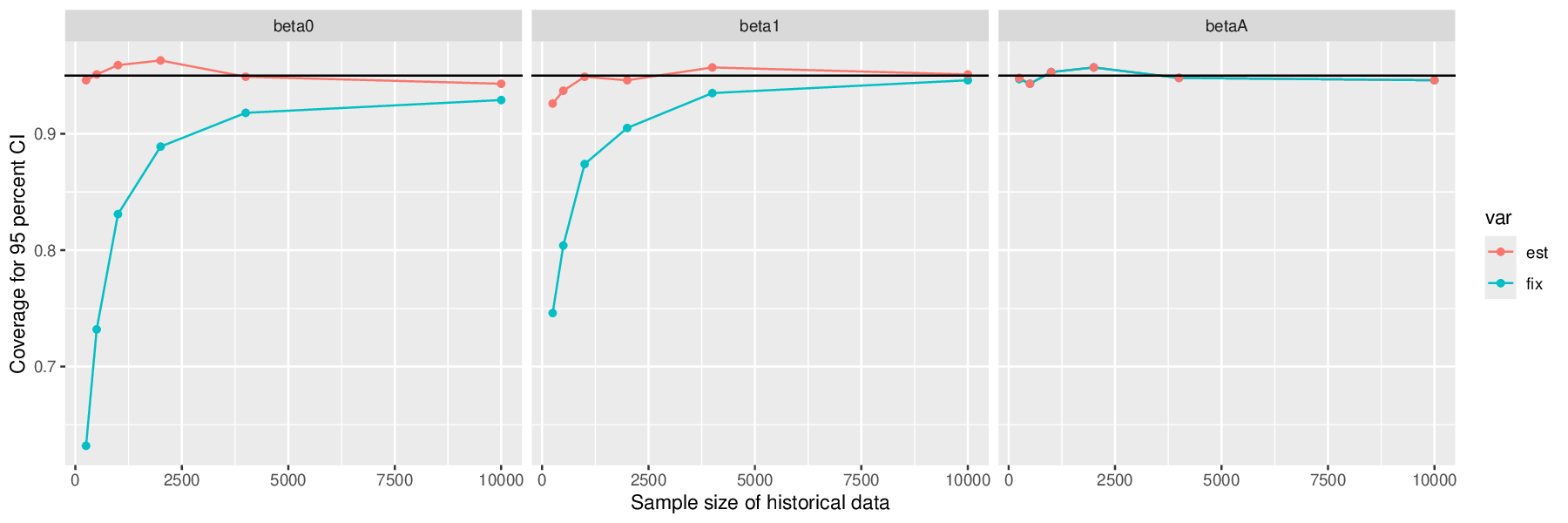}
    \caption{Plots of the coverage probability of 95\% CI over 1000 simulations for Scenario C-9.
    ``beta0'', ``betaA'', ``beta1'' represent the intercept $\beta_0$, the coefficient for the treatment assignment $\beta_A$, the coefficient for the prognostic score $\beta_1$ in the PROCOVA model \eqref{eq: PROCOVA linear}, respectively.
     ``fix'' represents $e^\top\hat V_{\text{fix}}e$ with $e=(1,0,0)^\top$ for $\beta_0$, with $e=(0,1,0)^\top$ for $\beta_A$ and $e=(0,0,1)^\top$ for $\beta_1$.
    ``est'' represents $e^\top\hat V_{\text{est}}e$ with $e=(1,0,0)^\top$ for $\beta_0$, with $e=(0,1,0)^\top$ for $\beta_A$ and $e=(0,0,1)^\top$ for $\beta_1$.
    The sample size of trial data is $n=1000$. 
    The x-axis represents the sample size of historical data $\tilde n$.
    The y-axis represents the coverage probability which is the proportion of 1000 simulations in which the 95\% CI using each variance estimator includes the true value.}
\end{figure}

\clearpage
\subsection{Scenario D-1}

\begin{figure}[h]
    \centering
    \includegraphics[width=\linewidth]{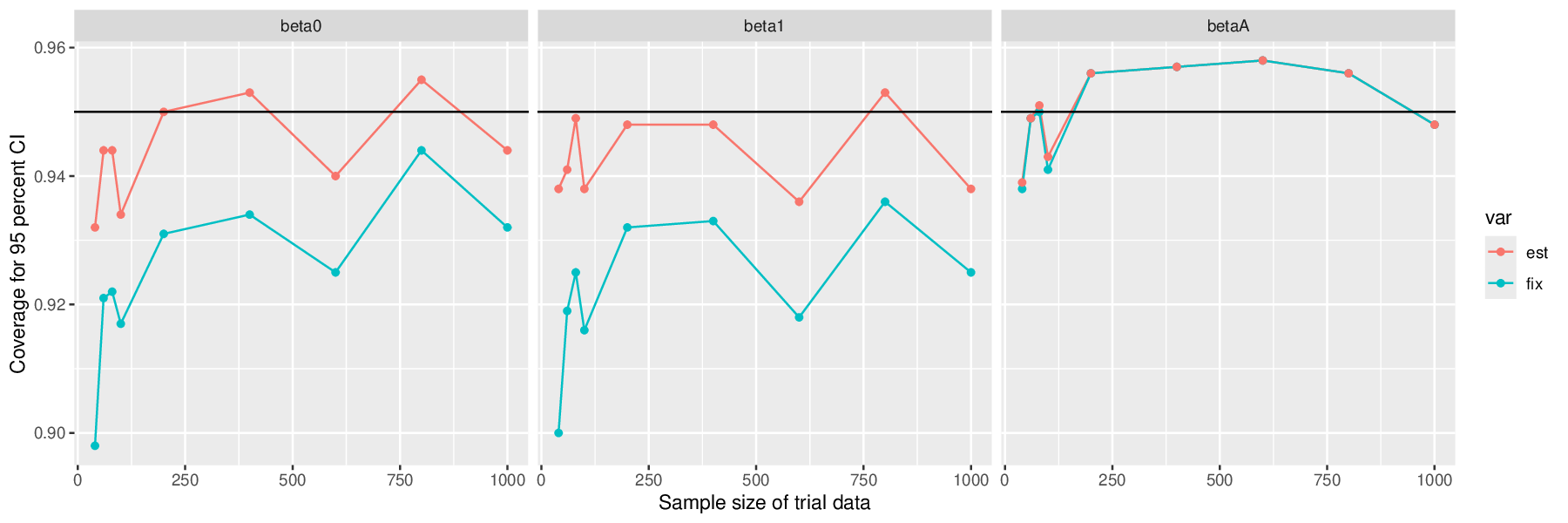}
    \caption{Plots of the coverage probability of 95\% CI over 1000 simulations for Scenario D-1.
    ``beta0'', ``betaA'', ``beta1'' represent the intercept $\beta_0$, the coefficient for the treatment assignment $\beta_A$, the coefficient for the prognostic score $\beta_1$ in the PROCOVA model \eqref{eq: PROCOVA linear}, respectively.
     ``fix'' represents $e^\top\hat V_{\text{fix}}e$ with $e=(1,0,0)^\top$ for $\beta_0$, with $e=(0,1,0)^\top$ for $\beta_A$ and $e=(0,0,1)^\top$ for $\beta_1$.
    ``est'' represents $e^\top\hat V_{\text{est}}e$ with $e=(1,0,0)^\top$ for $\beta_0$, with $e=(0,1,0)^\top$ for $\beta_A$ and $e=(0,0,1)^\top$ for $\beta_1$.
    The x-axis represents the sample size of trial data $n$.
    The sample size of historical data is $\tilde n = 10n$. 
    The y-axis represents the coverage probability which is the proportion of 1000 simulations in which the 95\% CI using each variance estimator includes the true value.
    }
\end{figure}

\clearpage
\begin{figure}[h]
    \centering
    \includegraphics[width=0.4\linewidth]{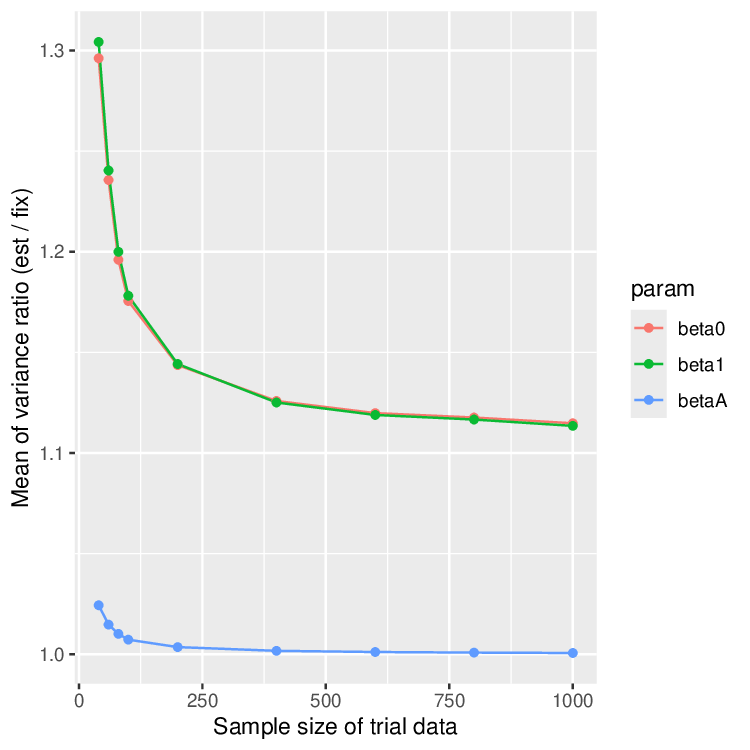}
    \caption{Plots of the mean of the ratio of two variance estimators over 1000 simulations for Scenario D-1.
    ``beta0'', ``betaA'', ``beta1'' represent the intercept $\beta_0$, the coefficient for the treatment assignment $\beta_A$, the coefficient for the prognostic score $\beta_1$ in the PROCOVA model \eqref{eq: PROCOVA linear}, respectively.
    The x-axis represents the sample size of trial data $n$.
    The sample size of historical data is $\tilde n = 10n$. 
    The y-axis represents the mean of the ratio of two variance estimators, i.e., $e^\top\hat V_{\text{est}}e/e^\top\hat V_{\text{fix}}e$ with $e=(1,0,0)^\top$ for $\beta_0$, with $e=(0,1,0)^\top$ for $\beta_A$ and $e=(0,0,1)^\top$ for $\beta_1$, over 1000 simulations.}
\end{figure}
\clearpage
\begin{figure}[h]
    \centering
    \includegraphics[width=0.4\linewidth]{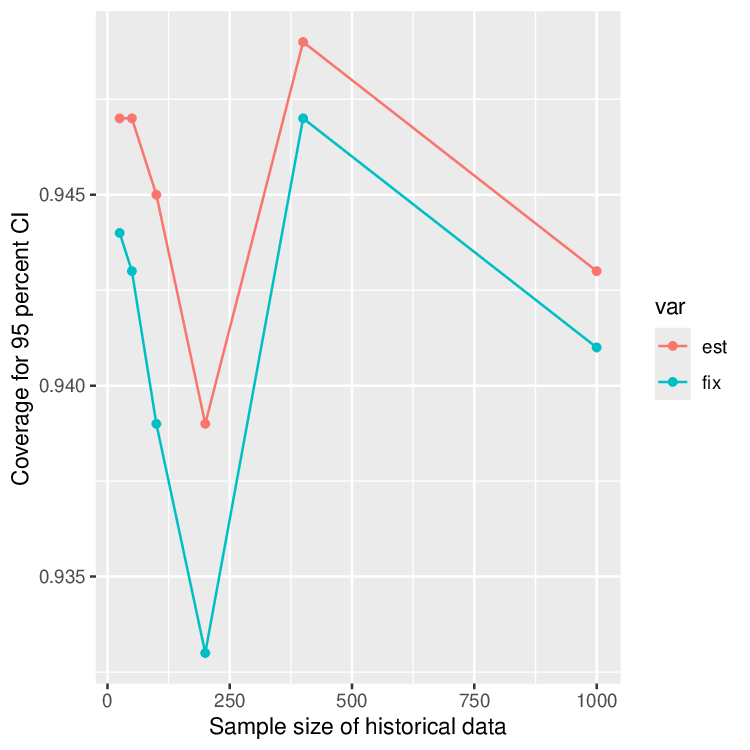}
    \caption{Plots of the coverage probability of the 95\% CI for $\beta_A$ in the PROCOVA model \eqref{eq: PROCOVA linear} over 1000 simulations for Scenario D-1.
     ``fix'' represents $e^\top\hat V_{\text{fix}}e$ with $e=(1,0,0)^\top$ for $\beta_0$, with $e=(0,1,0)^\top$ for $\beta_A$ and $e=(0,0,1)^\top$ for $\beta_1$.
    ``est'' represents $e^\top\hat V_{\text{est}}e$ with $e=(1,0,0)^\top$ for $\beta_0$, with $e=(0,1,0)^\top$ for $\beta_A$ and $e=(0,0,1)^\top$ for $\beta_1$.
    The sample size of trial data is $n=100$. 
    The x-axis represents the sample size of historical data $\tilde n$.
    The y-axis represents the coverage probability which is the proportion of 1000 simulations in which the 95\% CI using each variance estimator includes the true value.}
\end{figure}
\clearpage
\begin{figure}[h]
    \centering
    \includegraphics[width=\linewidth]{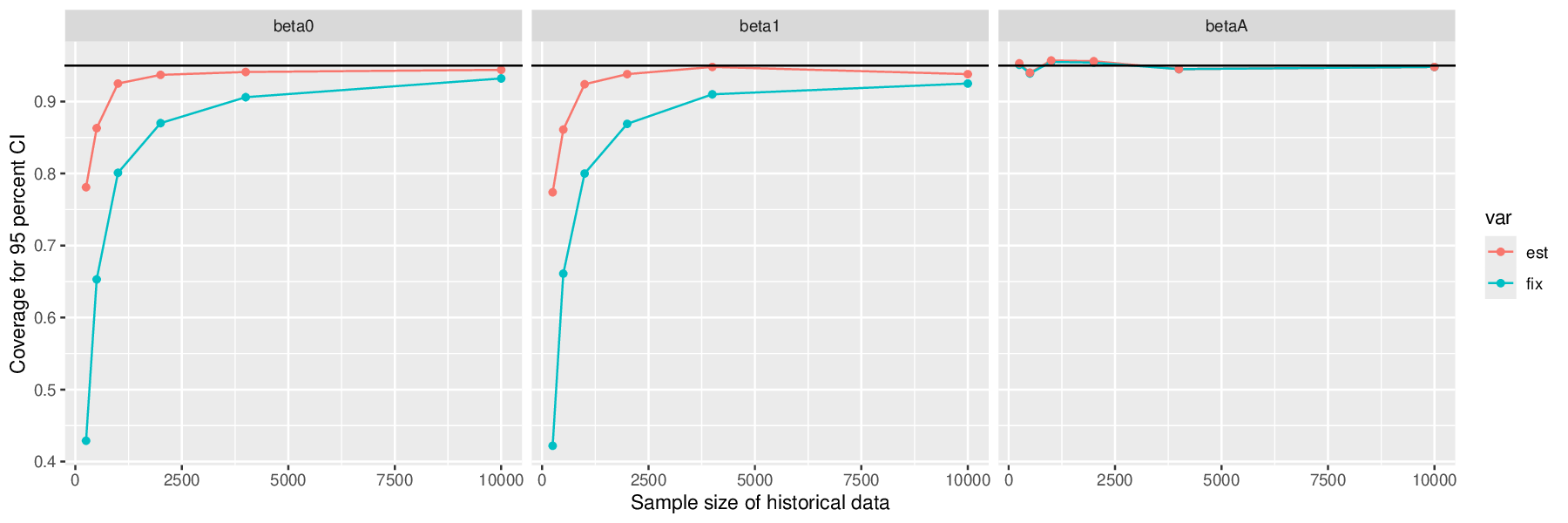}
    \caption{Plots of the coverage probability of 95\% CI over 1000 simulations for Scenario D-1.
    ``beta0'', ``betaA'', ``beta1'' represent the intercept $\beta_0$, the coefficient for the treatment assignment $\beta_A$, the coefficient for the prognostic score $\beta_1$ in the PROCOVA model \eqref{eq: PROCOVA linear}, respectively.
     ``fix'' represents $e^\top\hat V_{\text{fix}}e$ with $e=(1,0,0)^\top$ for $\beta_0$, with $e=(0,1,0)^\top$ for $\beta_A$ and $e=(0,0,1)^\top$ for $\beta_1$.
    ``est'' represents $e^\top\hat V_{\text{est}}e$ with $e=(1,0,0)^\top$ for $\beta_0$, with $e=(0,1,0)^\top$ for $\beta_A$ and $e=(0,0,1)^\top$ for $\beta_1$.
    The sample size of trial data is $n=1000$. 
    The x-axis represents the sample size of historical data $\tilde n$.
    The y-axis represents the coverage probability which is the proportion of 1000 simulations in which the 95\% CI using each variance estimator includes the true value.}
\end{figure}

\clearpage
\subsection{Scenario D-2}

\begin{figure}[h]
    \centering
    \includegraphics[width=\linewidth]{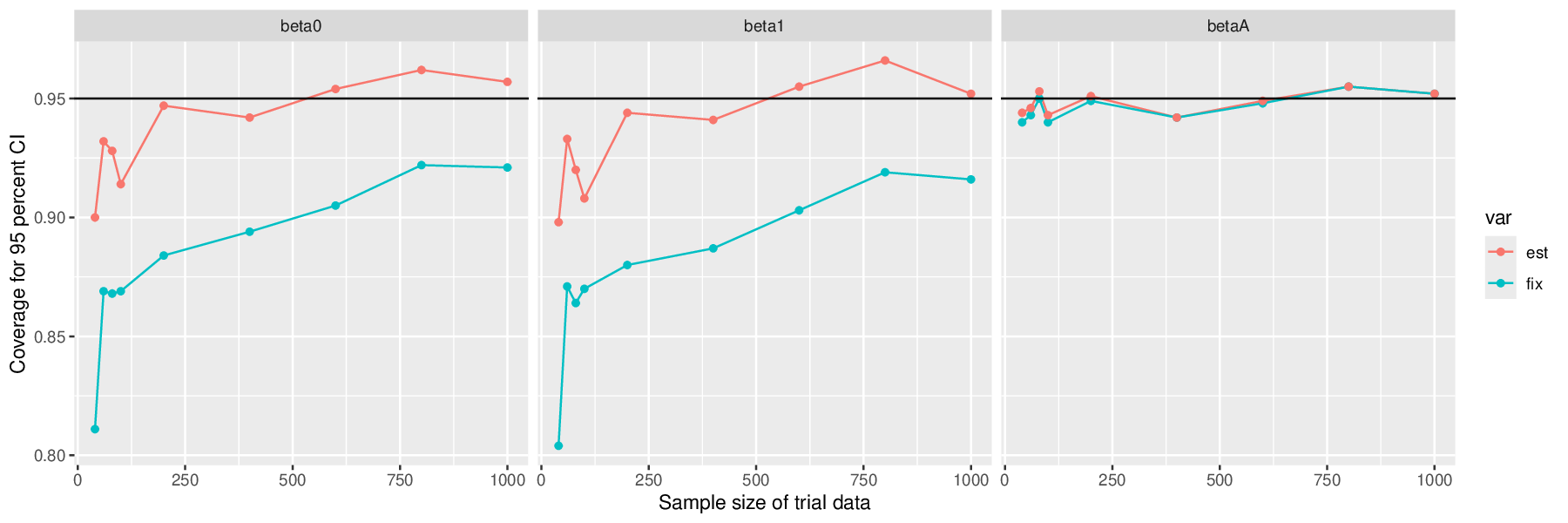}
    \caption{Plots of the coverage probability of 95\% CI over 1000 simulations for Scenario D-2.
    ``beta0'', ``betaA'', ``beta1'' represent the intercept $\beta_0$, the coefficient for the treatment assignment $\beta_A$, the coefficient for the prognostic score $\beta_1$ in the PROCOVA model \eqref{eq: PROCOVA linear}, respectively.
     ``fix'' represents $e^\top\hat V_{\text{fix}}e$ with $e=(1,0,0)^\top$ for $\beta_0$, with $e=(0,1,0)^\top$ for $\beta_A$ and $e=(0,0,1)^\top$ for $\beta_1$.
    ``est'' represents $e^\top\hat V_{\text{est}}e$ with $e=(1,0,0)^\top$ for $\beta_0$, with $e=(0,1,0)^\top$ for $\beta_A$ and $e=(0,0,1)^\top$ for $\beta_1$.
    The x-axis represents the sample size of trial data $n$.
    The sample size of historical data is $\tilde n = 10n$. 
    The y-axis represents the coverage probability which is the proportion of 1000 simulations in which the 95\% CI using each variance estimator includes the true value.
    }
\end{figure}

\clearpage
\begin{figure}[h]
    \centering
    \includegraphics[width=0.4\linewidth]{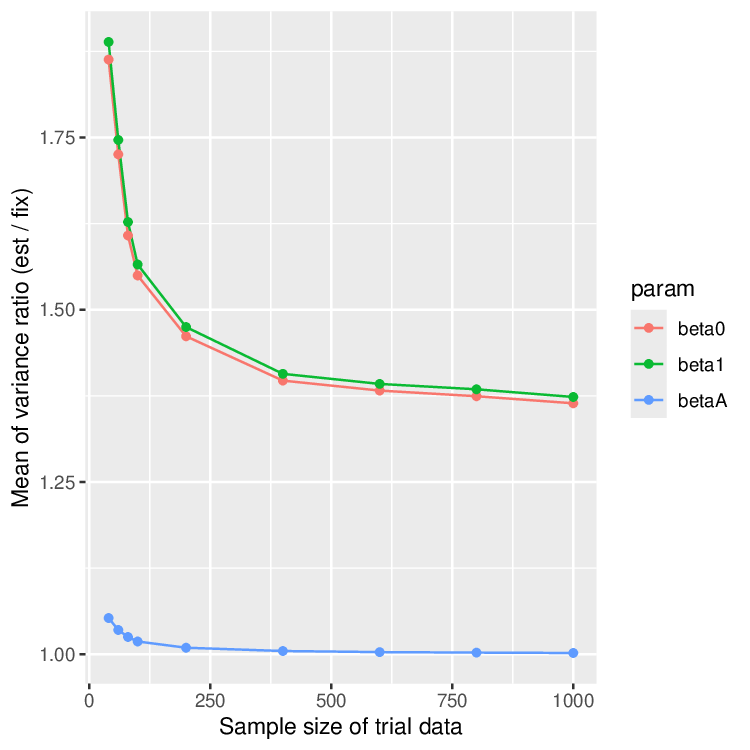}
    \caption{Plots of the mean of the ratio of two variance estimators over 1000 simulations for Scenario D-2.
    ``beta0'', ``betaA'', ``beta1'' represent the intercept $\beta_0$, the coefficient for the treatment assignment $\beta_A$, the coefficient for the prognostic score $\beta_1$ in the PROCOVA model \eqref{eq: PROCOVA linear}, respectively.
    The x-axis represents the sample size of trial data $n$.
    The sample size of historical data is $\tilde n = 10n$. 
    The y-axis represents the mean of the ratio of two variance estimators, i.e., $e^\top\hat V_{\text{est}}e/e^\top\hat V_{\text{fix}}e$ with $e=(1,0,0)^\top$ for $\beta_0$, with $e=(0,1,0)^\top$ for $\beta_A$ and $e=(0,0,1)^\top$ for $\beta_1$, over 1000 simulations.}
\end{figure}
\clearpage
\begin{figure}[h]
    \centering
    \includegraphics[width=0.4\linewidth]{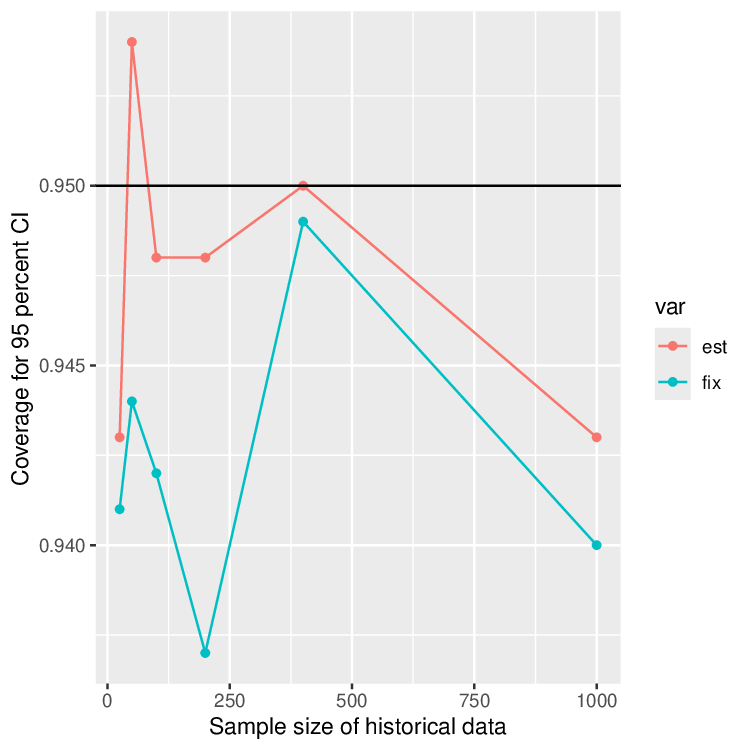}
    \caption{Plots of the coverage probability of the 95\% CI for $\beta_A$ in the PROCOVA model \eqref{eq: PROCOVA linear} over 1000 simulations for Scenario D-2.
     ``fix'' represents $e^\top\hat V_{\text{fix}}e$ with $e=(1,0,0)^\top$ for $\beta_0$, with $e=(0,1,0)^\top$ for $\beta_A$ and $e=(0,0,1)^\top$ for $\beta_1$.
    ``est'' represents $e^\top\hat V_{\text{est}}e$ with $e=(1,0,0)^\top$ for $\beta_0$, with $e=(0,1,0)^\top$ for $\beta_A$ and $e=(0,0,1)^\top$ for $\beta_1$.
    The sample size of trial data is $n=100$. 
    The x-axis represents the sample size of historical data $\tilde n$.
    The y-axis represents the coverage probability which is the proportion of 1000 simulations in which the 95\% CI using each variance estimator includes the true value.}
\end{figure}
\clearpage
\begin{figure}[h]
    \centering
    \includegraphics[width=\linewidth]{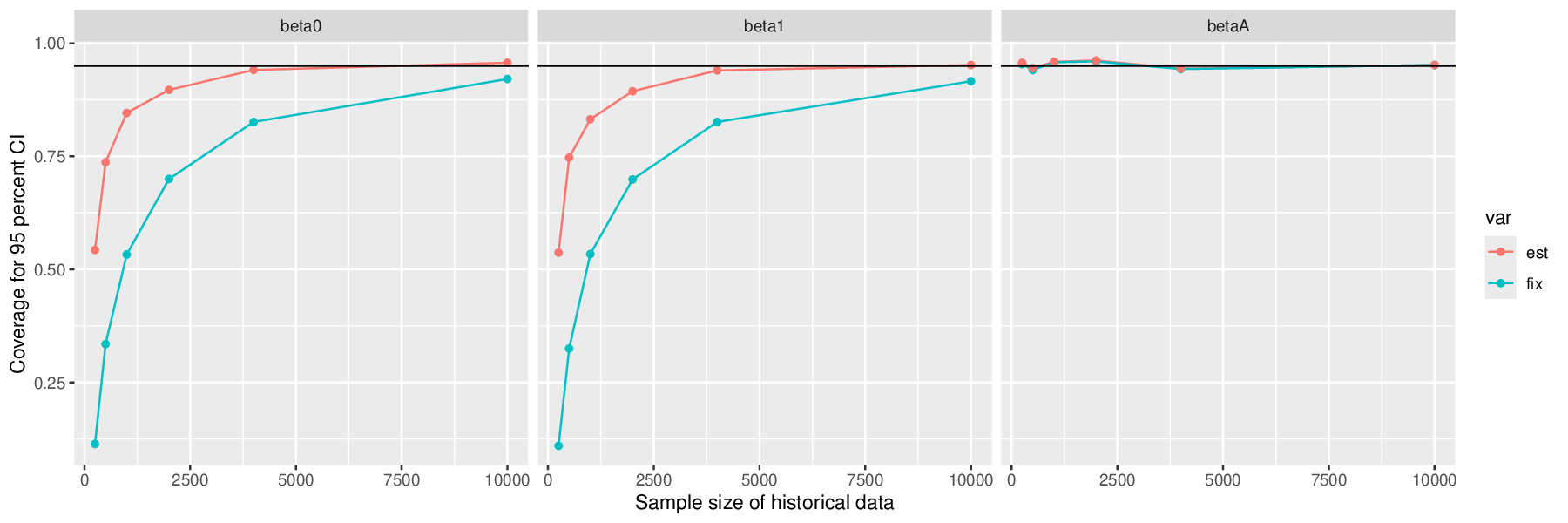}
    \caption{Plots of the coverage probability of 95\% CI over 1000 simulations for Scenario D-2.
    ``beta0'', ``betaA'', ``beta1'' represent the intercept $\beta_0$, the coefficient for the treatment assignment $\beta_A$, the coefficient for the prognostic score $\beta_1$ in the PROCOVA model \eqref{eq: PROCOVA linear}, respectively.
     ``fix'' represents $e^\top\hat V_{\text{fix}}e$ with $e=(1,0,0)^\top$ for $\beta_0$, with $e=(0,1,0)^\top$ for $\beta_A$ and $e=(0,0,1)^\top$ for $\beta_1$.
    ``est'' represents $e^\top\hat V_{\text{est}}e$ with $e=(1,0,0)^\top$ for $\beta_0$, with $e=(0,1,0)^\top$ for $\beta_A$ and $e=(0,0,1)^\top$ for $\beta_1$.
    The sample size of trial data is $n=1000$. 
    The x-axis represents the sample size of historical data $\tilde n$.
    The y-axis represents the coverage probability which is the proportion of 1000 simulations in which the 95\% CI using each variance estimator includes the true value.}
\end{figure}

\clearpage
\subsection{Scenario D-3}

\begin{figure}[h]
    \centering
    \includegraphics[width=\linewidth]{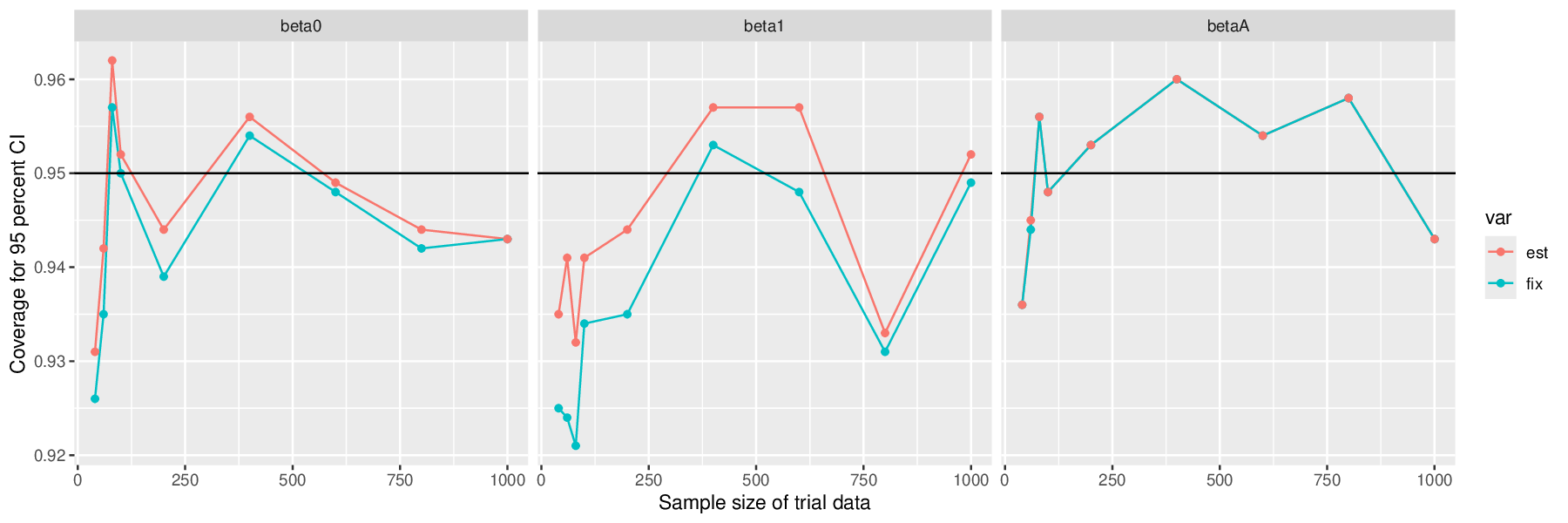}
    \caption{Plots of the coverage probability of 95\% CI over 1000 simulations for Scenario D-3.
    ``beta0'', ``betaA'', ``beta1'' represent the intercept $\beta_0$, the coefficient for the treatment assignment $\beta_A$, the coefficient for the prognostic score $\beta_1$ in the PROCOVA model \eqref{eq: PROCOVA linear}, respectively.
     ``fix'' represents $e^\top\hat V_{\text{fix}}e$ with $e=(1,0,0)^\top$ for $\beta_0$, with $e=(0,1,0)^\top$ for $\beta_A$ and $e=(0,0,1)^\top$ for $\beta_1$.
    ``est'' represents $e^\top\hat V_{\text{est}}e$ with $e=(1,0,0)^\top$ for $\beta_0$, with $e=(0,1,0)^\top$ for $\beta_A$ and $e=(0,0,1)^\top$ for $\beta_1$.
    The x-axis represents the sample size of trial data $n$.
    The sample size of historical data is $\tilde n = 10n$. 
    The y-axis represents the coverage probability which is the proportion of 1000 simulations in which the 95\% CI using each variance estimator includes the true value.
    }
\end{figure}

\clearpage
\begin{figure}[h]
    \centering
    \includegraphics[width=0.4\linewidth]{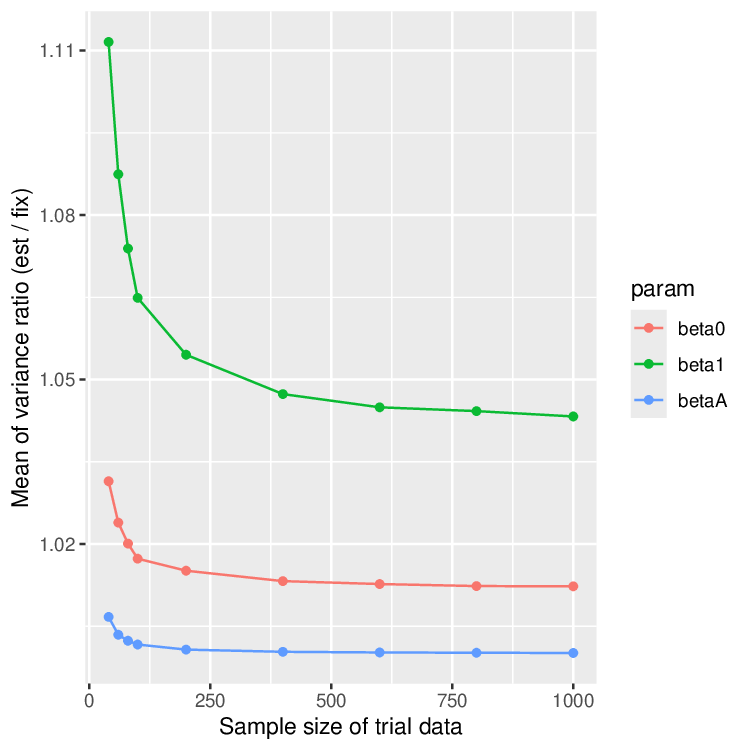}
    \caption{Plots of the mean of the ratio of two variance estimators over 1000 simulations for Scenario D-3.
    ``beta0'', ``betaA'', ``beta1'' represent the intercept $\beta_0$, the coefficient for the treatment assignment $\beta_A$, the coefficient for the prognostic score $\beta_1$ in the PROCOVA model \eqref{eq: PROCOVA linear}, respectively.
    The x-axis represents the sample size of trial data $n$.
    The sample size of historical data is $\tilde n = 10n$. 
    The y-axis represents the mean of the ratio of two variance estimators, i.e., $e^\top\hat V_{\text{est}}e/e^\top\hat V_{\text{fix}}e$ with $e=(1,0,0)^\top$ for $\beta_0$, with $e=(0,1,0)^\top$ for $\beta_A$ and $e=(0,0,1)^\top$ for $\beta_1$, over 1000 simulations.}
\end{figure}
\clearpage
\begin{figure}[h]
    \centering
    \includegraphics[width=0.4\linewidth]{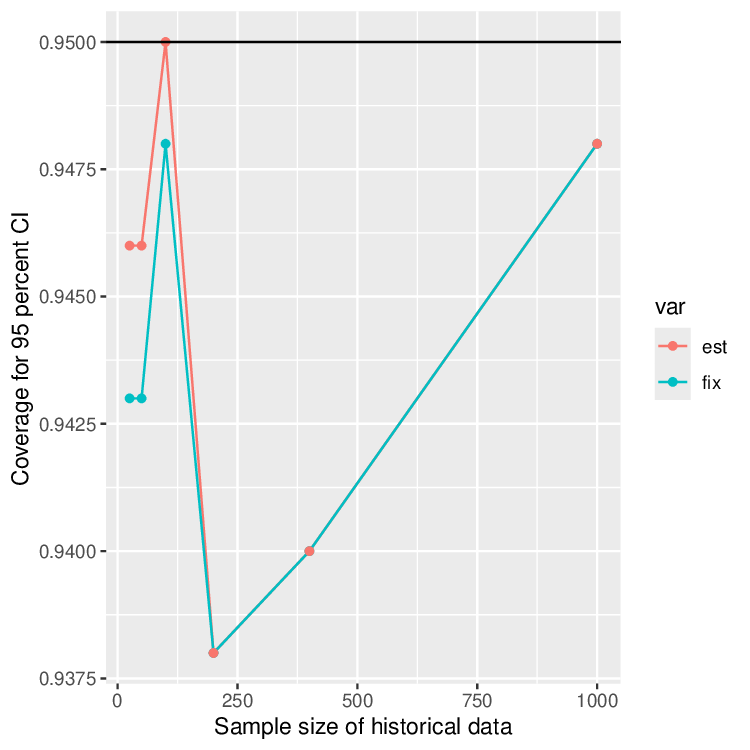}
    \caption{Plots of the coverage probability of the 95\% CI for $\beta_A$ in the PROCOVA model \eqref{eq: PROCOVA linear} over 1000 simulations for Scenario D-3.
     ``fix'' represents $e^\top\hat V_{\text{fix}}e$ with $e=(1,0,0)^\top$ for $\beta_0$, with $e=(0,1,0)^\top$ for $\beta_A$ and $e=(0,0,1)^\top$ for $\beta_1$.
    ``est'' represents $e^\top\hat V_{\text{est}}e$ with $e=(1,0,0)^\top$ for $\beta_0$, with $e=(0,1,0)^\top$ for $\beta_A$ and $e=(0,0,1)^\top$ for $\beta_1$.
    The sample size of trial data is $n=100$. 
    The x-axis represents the sample size of historical data $\tilde n$.
    The y-axis represents the coverage probability which is the proportion of 1000 simulations in which the 95\% CI using each variance estimator includes the true value.}
\end{figure}
\clearpage
\begin{figure}[h]
    \centering
    \includegraphics[width=\linewidth]{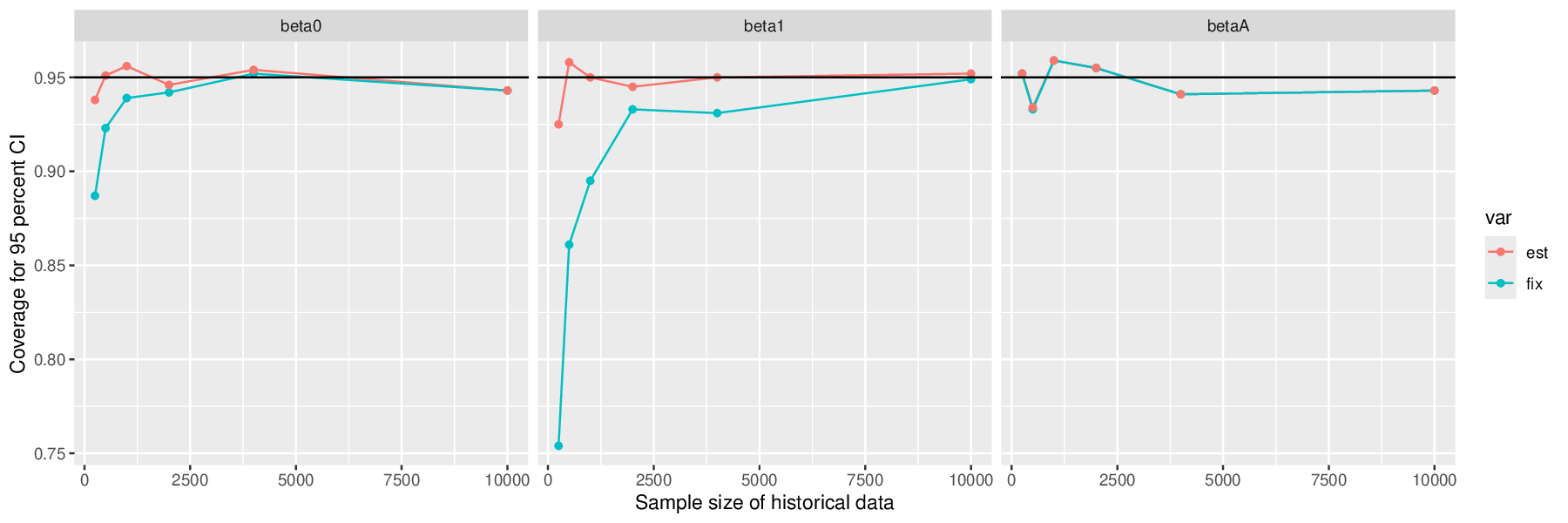}
    \caption{Plots of the coverage probability of 95\% CI over 1000 simulations for Scenario D-3.
    ``beta0'', ``betaA'', ``beta1'' represent the intercept $\beta_0$, the coefficient for the treatment assignment $\beta_A$, the coefficient for the prognostic score $\beta_1$ in the PROCOVA model \eqref{eq: PROCOVA linear}, respectively.
     ``fix'' represents $e^\top\hat V_{\text{fix}}e$ with $e=(1,0,0)^\top$ for $\beta_0$, with $e=(0,1,0)^\top$ for $\beta_A$ and $e=(0,0,1)^\top$ for $\beta_1$.
    ``est'' represents $e^\top\hat V_{\text{est}}e$ with $e=(1,0,0)^\top$ for $\beta_0$, with $e=(0,1,0)^\top$ for $\beta_A$ and $e=(0,0,1)^\top$ for $\beta_1$.
    The sample size of trial data is $n=1000$. 
    The x-axis represents the sample size of historical data $\tilde n$.
    The y-axis represents the coverage probability which is the proportion of 1000 simulations in which the 95\% CI using each variance estimator includes the true value.}
\end{figure}

\clearpage
\subsection{Scenario D-4}

\begin{figure}[h]
    \centering
    \includegraphics[width=\linewidth]{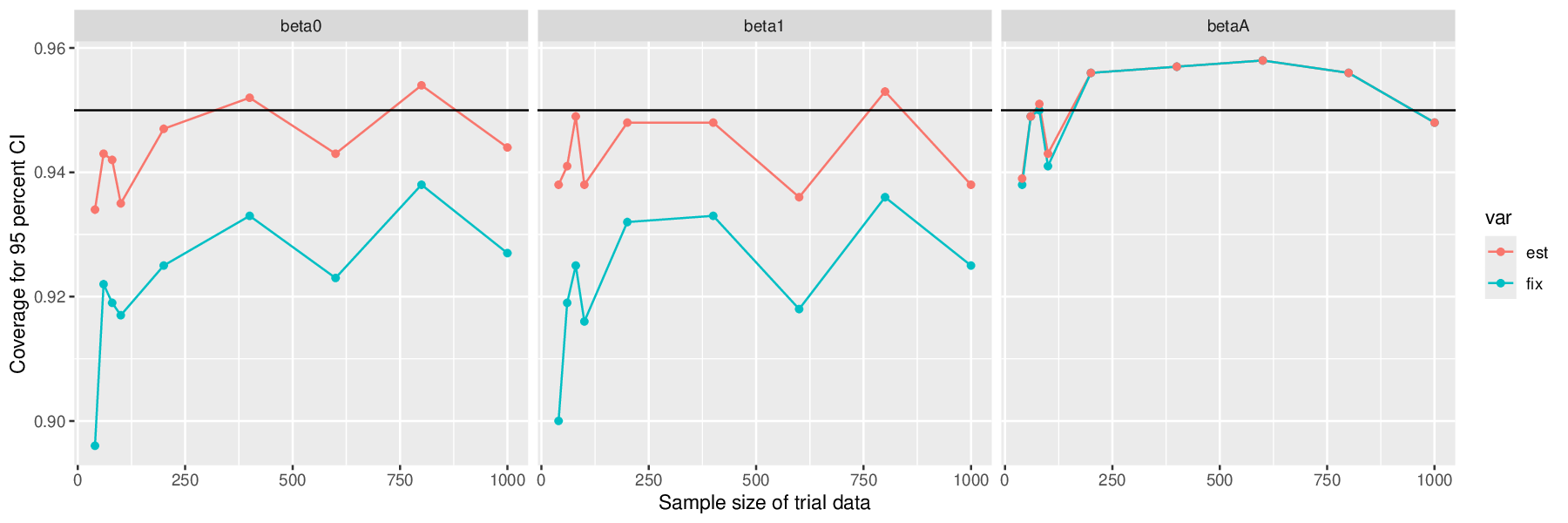}
    \caption{Plots of the coverage probability of 95\% CI over 1000 simulations for Scenario D-4.
    ``beta0'', ``betaA'', ``beta1'' represent the intercept $\beta_0$, the coefficient for the treatment assignment $\beta_A$, the coefficient for the prognostic score $\beta_1$ in the PROCOVA model \eqref{eq: PROCOVA linear}, respectively.
     ``fix'' represents $e^\top\hat V_{\text{fix}}e$ with $e=(1,0,0)^\top$ for $\beta_0$, with $e=(0,1,0)^\top$ for $\beta_A$ and $e=(0,0,1)^\top$ for $\beta_1$.
    ``est'' represents $e^\top\hat V_{\text{est}}e$ with $e=(1,0,0)^\top$ for $\beta_0$, with $e=(0,1,0)^\top$ for $\beta_A$ and $e=(0,0,1)^\top$ for $\beta_1$.
    The x-axis represents the sample size of trial data $n$.
    The sample size of historical data is $\tilde n = 10n$. 
    The y-axis represents the coverage probability which is the proportion of 1000 simulations in which the 95\% CI using each variance estimator includes the true value.
    }
\end{figure}

\clearpage
\begin{figure}[h]
    \centering
    \includegraphics[width=0.4\linewidth]{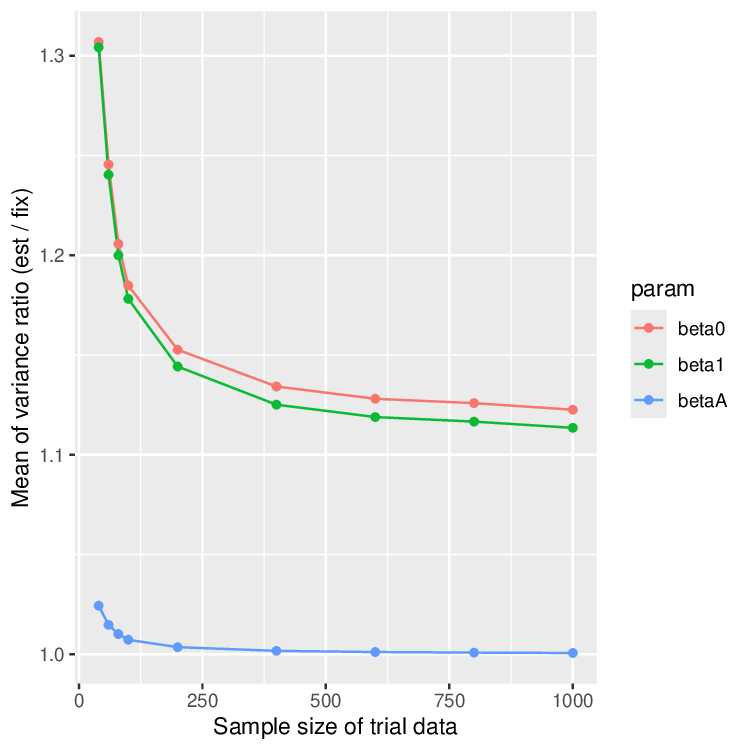}
    \caption{Plots of the mean of the ratio of two variance estimators over 1000 simulations for Scenario D-4.
    ``beta0'', ``betaA'', ``beta1'' represent the intercept $\beta_0$, the coefficient for the treatment assignment $\beta_A$, the coefficient for the prognostic score $\beta_1$ in the PROCOVA model \eqref{eq: PROCOVA linear}, respectively.
    The x-axis represents the sample size of trial data $n$.
    The sample size of historical data is $\tilde n = 10n$. 
    The y-axis represents the mean of the ratio of two variance estimators, i.e., $e^\top\hat V_{\text{est}}e/e^\top\hat V_{\text{fix}}e$ with $e=(1,0,0)^\top$ for $\beta_0$, with $e=(0,1,0)^\top$ for $\beta_A$ and $e=(0,0,1)^\top$ for $\beta_1$, over 1000 simulations.}
\end{figure}

\clearpage
\begin{figure}[h]
    \centering
    \includegraphics[width=0.4\linewidth]{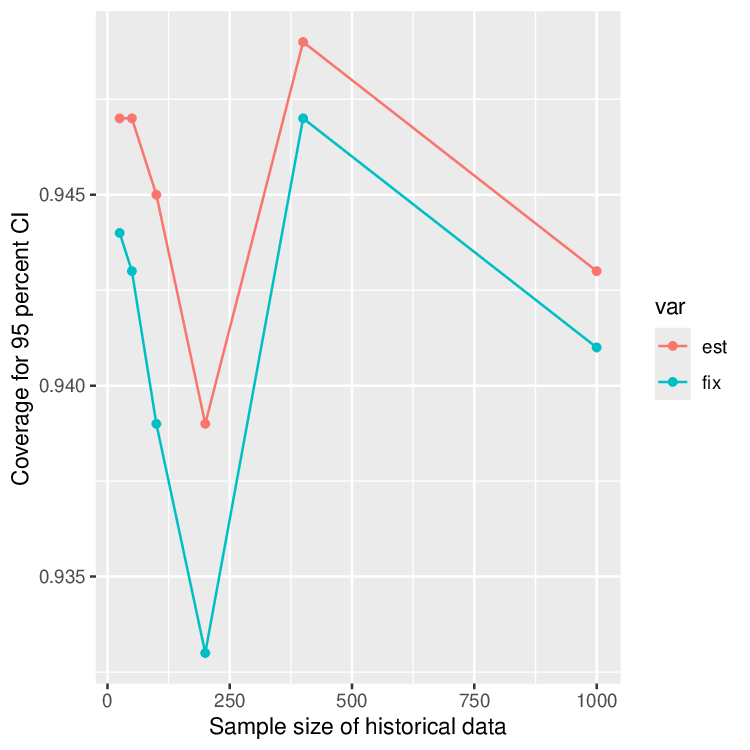}
    \caption{Plots of the coverage probability of the 95\% CI for $\beta_A$ in the PROCOVA model \eqref{eq: PROCOVA linear} over 1000 simulations for Scenario D-4.
     ``fix'' represents $e^\top\hat V_{\text{fix}}e$ with $e=(1,0,0)^\top$ for $\beta_0$, with $e=(0,1,0)^\top$ for $\beta_A$ and $e=(0,0,1)^\top$ for $\beta_1$.
    ``est'' represents $e^\top\hat V_{\text{est}}e$ with $e=(1,0,0)^\top$ for $\beta_0$, with $e=(0,1,0)^\top$ for $\beta_A$ and $e=(0,0,1)^\top$ for $\beta_1$.
    The sample size of trial data is $n=100$. 
    The x-axis represents the sample size of historical data $\tilde n$.
    The y-axis represents the coverage probability which is the proportion of 1000 simulations in which the 95\% CI using each variance estimator includes the true value.}
\end{figure}

\clearpage
\begin{figure}[h]
    \centering
    \includegraphics[width=\linewidth]{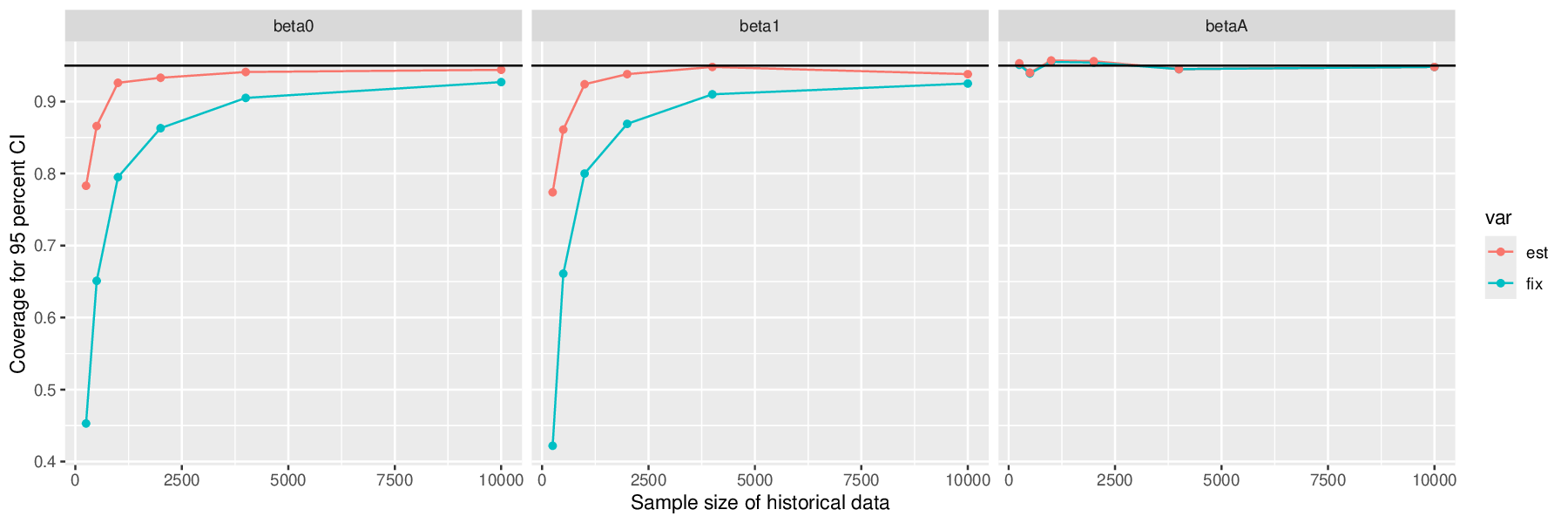}
    \caption{Plots of the coverage probability of 95\% CI over 1000 simulations for Scenario D-4.
    ``beta0'', ``betaA'', ``beta1'' represent the intercept $\beta_0$, the coefficient for the treatment assignment $\beta_A$, the coefficient for the prognostic score $\beta_1$ in the PROCOVA model \eqref{eq: PROCOVA linear}, respectively.
     ``fix'' represents $e^\top\hat V_{\text{fix}}e$ with $e=(1,0,0)^\top$ for $\beta_0$, with $e=(0,1,0)^\top$ for $\beta_A$ and $e=(0,0,1)^\top$ for $\beta_1$.
    ``est'' represents $e^\top\hat V_{\text{est}}e$ with $e=(1,0,0)^\top$ for $\beta_0$, with $e=(0,1,0)^\top$ for $\beta_A$ and $e=(0,0,1)^\top$ for $\beta_1$.
    The sample size of trial data is $n=1000$. 
    The x-axis represents the sample size of historical data $\tilde n$.
    The y-axis represents the coverage probability which is the proportion of 1000 simulations in which the 95\% CI using each variance estimator includes the true value.}
\end{figure}

\clearpage
\subsection{Scenario D-6}

\begin{figure}[h]
    \centering
    \includegraphics[width=\linewidth]{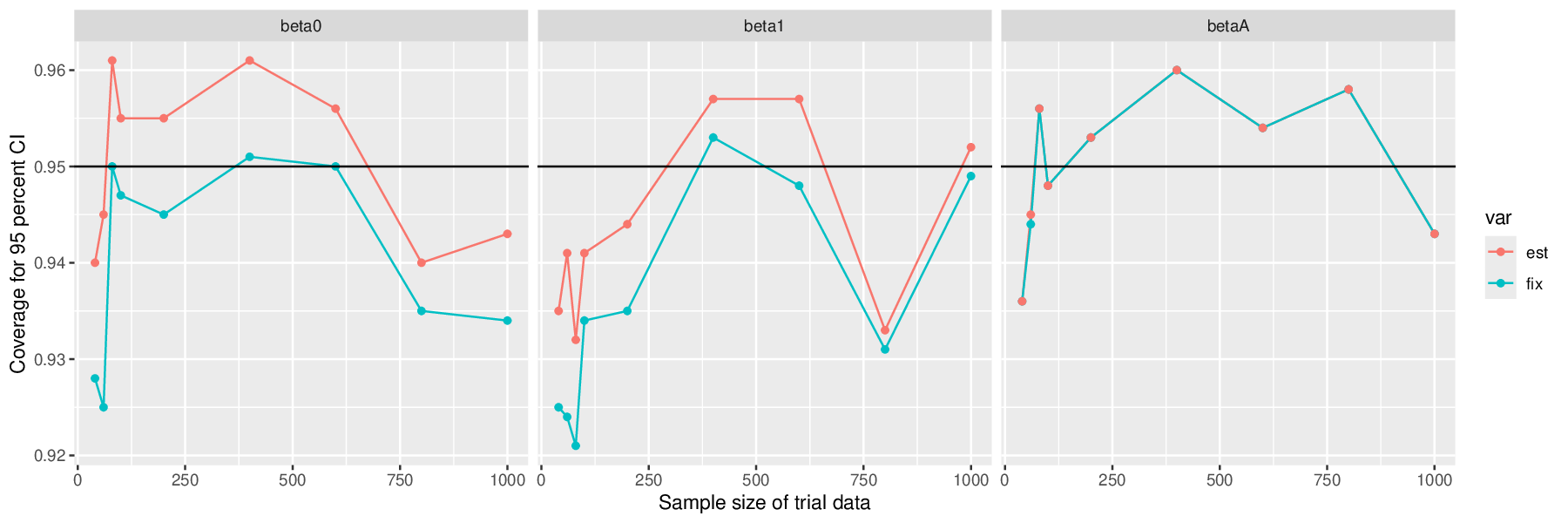}
    \caption{Plots of the coverage probability of 95\% CI over 1000 simulations for Scenario D-6.
    ``beta0'', ``betaA'', ``beta1'' represent the intercept $\beta_0$, the coefficient for the treatment assignment $\beta_A$, the coefficient for the prognostic score $\beta_1$ in the PROCOVA model \eqref{eq: PROCOVA linear}, respectively.
     ``fix'' represents $e^\top\hat V_{\text{fix}}e$ with $e=(1,0,0)^\top$ for $\beta_0$, with $e=(0,1,0)^\top$ for $\beta_A$ and $e=(0,0,1)^\top$ for $\beta_1$.
    ``est'' represents $e^\top\hat V_{\text{est}}e$ with $e=(1,0,0)^\top$ for $\beta_0$, with $e=(0,1,0)^\top$ for $\beta_A$ and $e=(0,0,1)^\top$ for $\beta_1$.
    The x-axis represents the sample size of trial data $n$.
    The sample size of historical data is $\tilde n = 10n$. 
    The y-axis represents the coverage probability which is the proportion of 1000 simulations in which the 95\% CI using each variance estimator includes the true value.
    }
\end{figure}

\clearpage
\begin{figure}[h]
    \centering
    \includegraphics[width=0.4\linewidth]{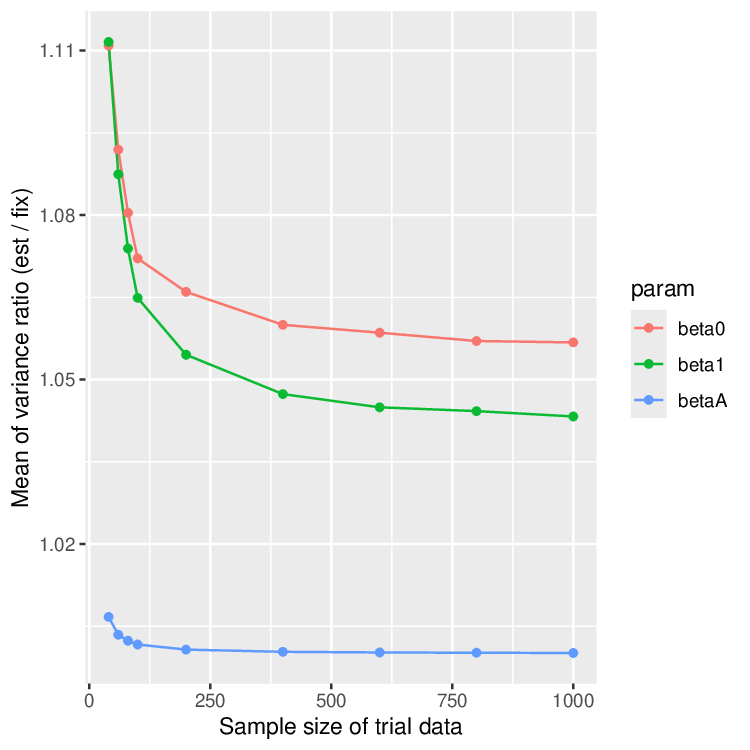}
    \caption{Plots of the mean of the ratio of two variance estimators over 1000 simulations for Scenario D-6.
    ``beta0'', ``betaA'', ``beta1'' represent the intercept $\beta_0$, the coefficient for the treatment assignment $\beta_A$, the coefficient for the prognostic score $\beta_1$ in the PROCOVA model \eqref{eq: PROCOVA linear}, respectively.
    The x-axis represents the sample size of trial data $n$.
    The sample size of historical data is $\tilde n = 10n$. 
    The y-axis represents the mean of the ratio of two variance estimators, i.e., $e^\top\hat V_{\text{est}}e/e^\top\hat V_{\text{fix}}e$ with $e=(1,0,0)^\top$ for $\beta_0$, with $e=(0,1,0)^\top$ for $\beta_A$ and $e=(0,0,1)^\top$ for $\beta_1$, over 1000 simulations.}
\end{figure}

\clearpage
\begin{figure}[h]
    \centering
    \includegraphics[width=0.4\linewidth]{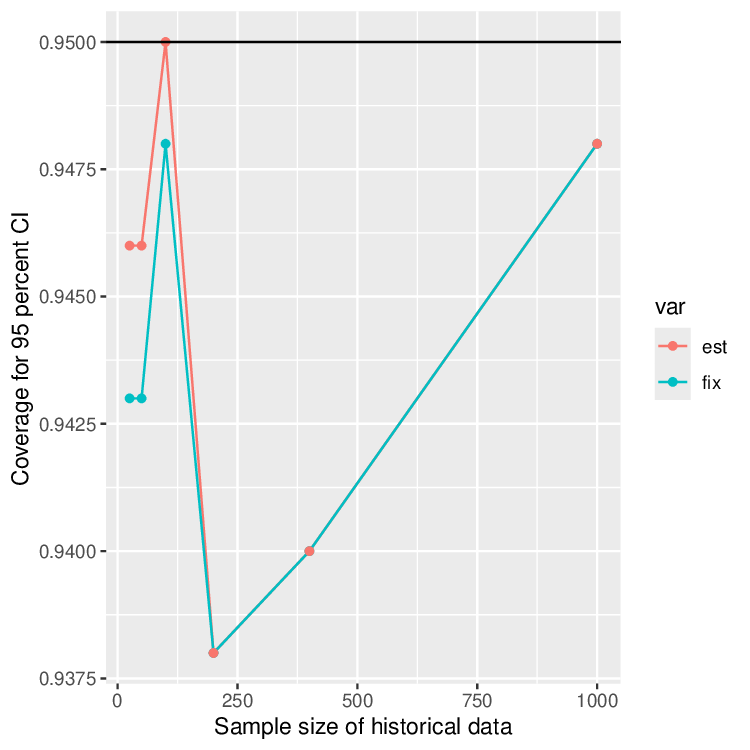}
    \caption{Plots of the coverage probability of the 95\% CI for $\beta_A$ in the PROCOVA model \eqref{eq: PROCOVA linear} over 1000 simulations for Scenario D-6.
     ``fix'' represents $e^\top\hat V_{\text{fix}}e$ with $e=(1,0,0)^\top$ for $\beta_0$, with $e=(0,1,0)^\top$ for $\beta_A$ and $e=(0,0,1)^\top$ for $\beta_1$.
    ``est'' represents $e^\top\hat V_{\text{est}}e$ with $e=(1,0,0)^\top$ for $\beta_0$, with $e=(0,1,0)^\top$ for $\beta_A$ and $e=(0,0,1)^\top$ for $\beta_1$.
    The sample size of trial data is $n=100$. 
    The x-axis represents the sample size of historical data $\tilde n$.
    The y-axis represents the coverage probability which is the proportion of 1000 simulations in which the 95\% CI using each variance estimator includes the true value.}
\end{figure}

\clearpage
\begin{figure}[h]
    \centering
    \includegraphics[width=\linewidth]{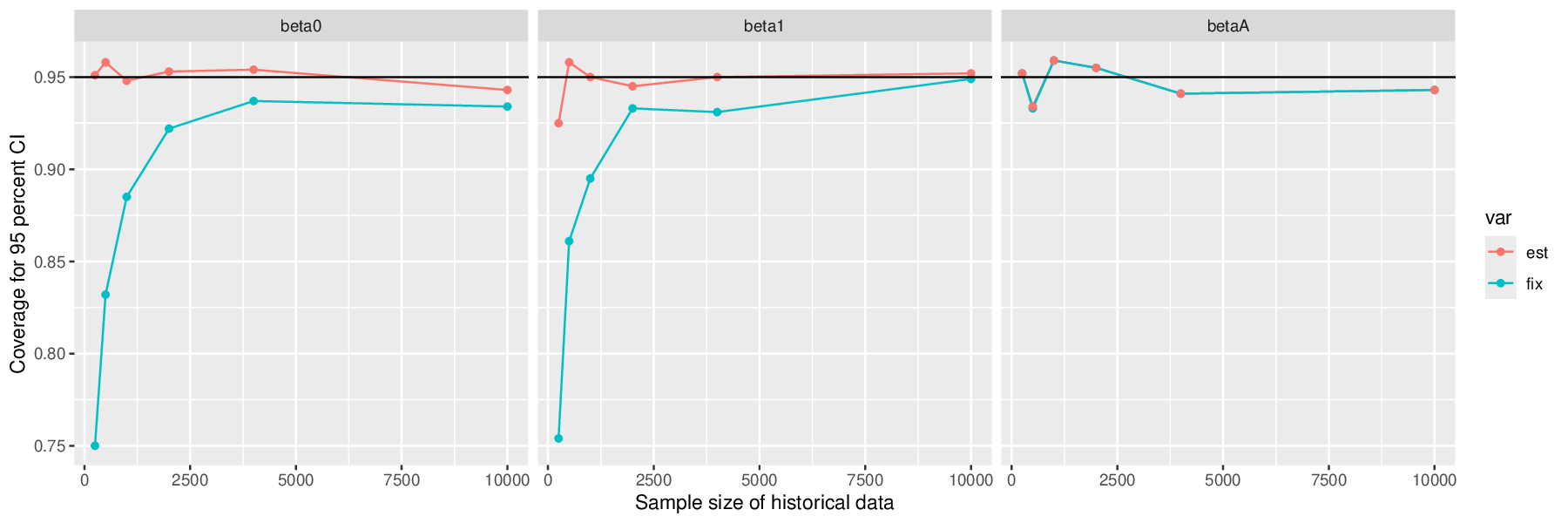}
    \caption{Plots of the coverage probability of 95\% CI over 1000 simulations for Scenario D-6.
    ``beta0'', ``betaA'', ``beta1'' represent the intercept $\beta_0$, the coefficient for the treatment assignment $\beta_A$, the coefficient for the prognostic score $\beta_1$ in the PROCOVA model \eqref{eq: PROCOVA linear}, respectively.
     ``fix'' represents $e^\top\hat V_{\text{fix}}e$ with $e=(1,0,0)^\top$ for $\beta_0$, with $e=(0,1,0)^\top$ for $\beta_A$ and $e=(0,0,1)^\top$ for $\beta_1$.
    ``est'' represents $e^\top\hat V_{\text{est}}e$ with $e=(1,0,0)^\top$ for $\beta_0$, with $e=(0,1,0)^\top$ for $\beta_A$ and $e=(0,0,1)^\top$ for $\beta_1$.
    The sample size of trial data is $n=1000$. 
    The x-axis represents the sample size of historical data $\tilde n$.
    The y-axis represents the coverage probability which is the proportion of 1000 simulations in which the 95\% CI using each variance estimator includes the true value.}
\end{figure}

\clearpage
\subsection{Scenario D-7}

\begin{figure}[h]
    \centering
    \includegraphics[width=\linewidth]{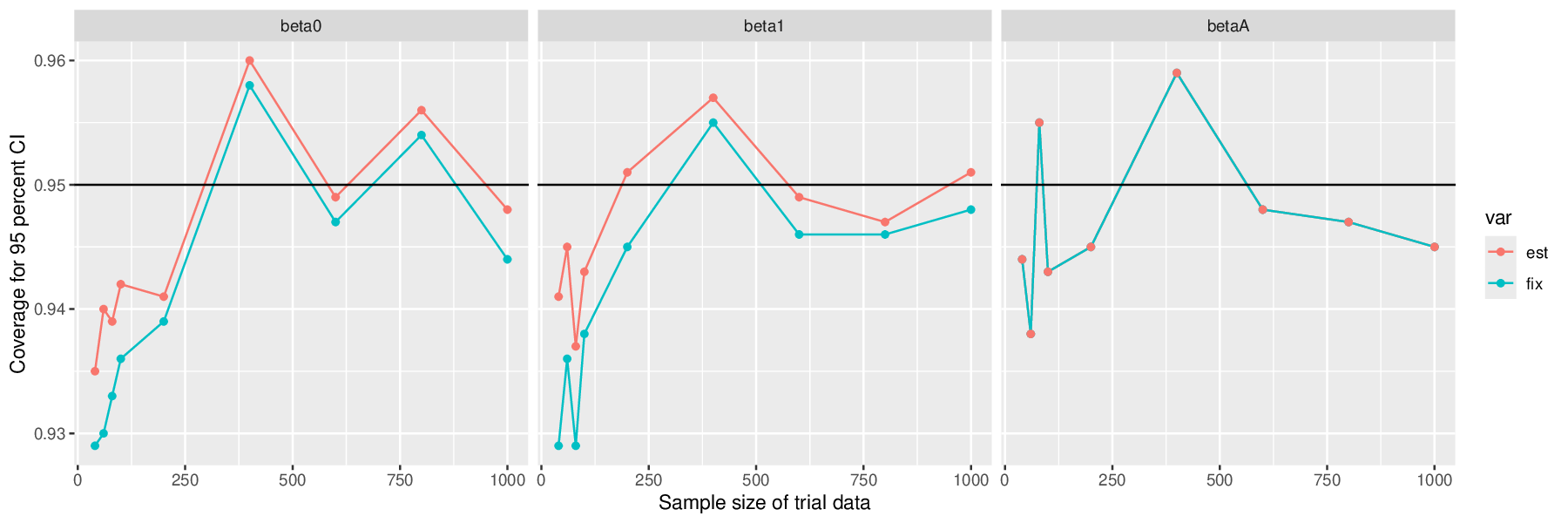}
    \caption{Plots of the coverage probability of 95\% CI over 1000 simulations for Scenario D-7.
    ``beta0'', ``betaA'', ``beta1'' represent the intercept $\beta_0$, the coefficient for the treatment assignment $\beta_A$, the coefficient for the prognostic score $\beta_1$ in the PROCOVA model \eqref{eq: PROCOVA linear}, respectively.
     ``fix'' represents $e^\top\hat V_{\text{fix}}e$ with $e=(1,0,0)^\top$ for $\beta_0$, with $e=(0,1,0)^\top$ for $\beta_A$ and $e=(0,0,1)^\top$ for $\beta_1$.
    ``est'' represents $e^\top\hat V_{\text{est}}e$ with $e=(1,0,0)^\top$ for $\beta_0$, with $e=(0,1,0)^\top$ for $\beta_A$ and $e=(0,0,1)^\top$ for $\beta_1$.
    The x-axis represents the sample size of trial data $n$.
    The sample size of historical data is $\tilde n = 10n$. 
    The y-axis represents the coverage probability which is the proportion of 1000 simulations in which the 95\% CI using each variance estimator includes the true value.
    }
\end{figure}

\clearpage
\begin{figure}[h]
    \centering
    \includegraphics[width=0.4\linewidth]{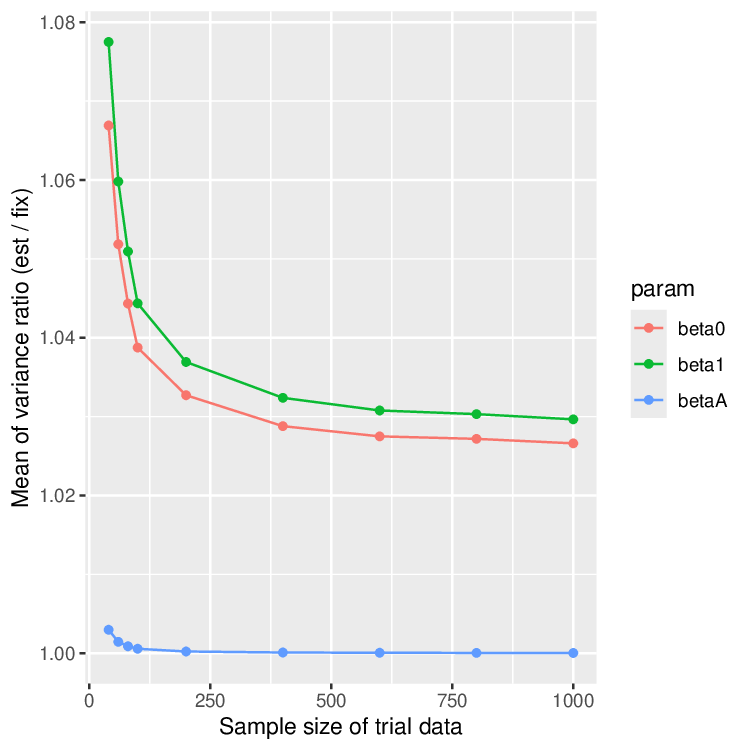}
    \caption{Plots of the mean of the ratio of two variance estimators over 1000 simulations for Scenario D-7.
    ``beta0'', ``betaA'', ``beta1'' represent the intercept $\beta_0$, the coefficient for the treatment assignment $\beta_A$, the coefficient for the prognostic score $\beta_1$ in the PROCOVA model \eqref{eq: PROCOVA linear}, respectively.
    The x-axis represents the sample size of trial data $n$.
    The sample size of historical data is $\tilde n = 10n$. 
    The y-axis represents the mean of the ratio of two variance estimators, i.e., $e^\top\hat V_{\text{est}}e/e^\top\hat V_{\text{fix}}e$ with $e=(1,0,0)^\top$ for $\beta_0$, with $e=(0,1,0)^\top$ for $\beta_A$ and $e=(0,0,1)^\top$ for $\beta_1$, over 1000 simulations.}
\end{figure}
\clearpage
\begin{figure}[h]
    \centering
    \includegraphics[width=0.4\linewidth]{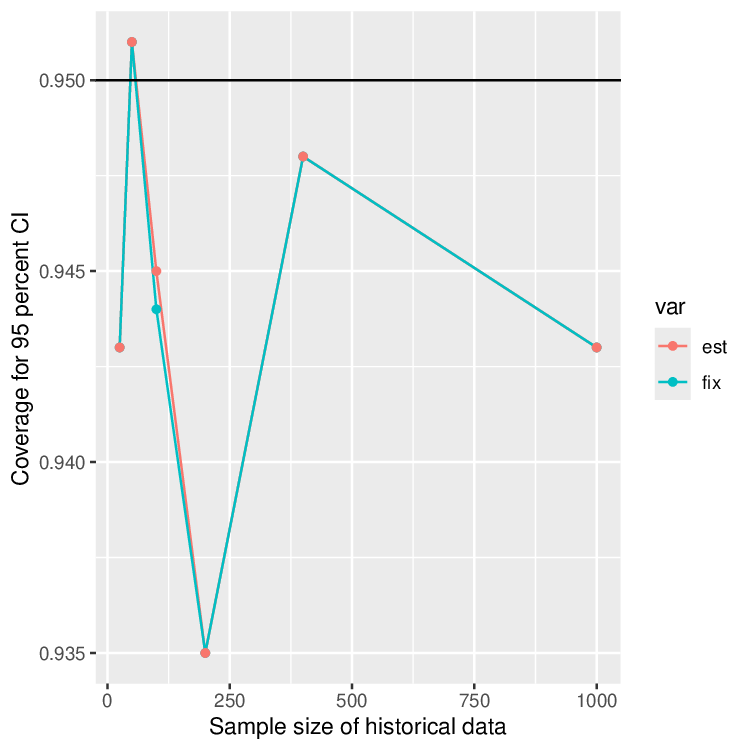}
    \caption{Plots of the coverage probability of the 95\% CI for $\beta_A$ in the PROCOVA model \eqref{eq: PROCOVA linear} over 1000 simulations for Scenario D-7.
     ``fix'' represents $e^\top\hat V_{\text{fix}}e$ with $e=(1,0,0)^\top$ for $\beta_0$, with $e=(0,1,0)^\top$ for $\beta_A$ and $e=(0,0,1)^\top$ for $\beta_1$.
    ``est'' represents $e^\top\hat V_{\text{est}}e$ with $e=(1,0,0)^\top$ for $\beta_0$, with $e=(0,1,0)^\top$ for $\beta_A$ and $e=(0,0,1)^\top$ for $\beta_1$.
    The sample size of trial data is $n=100$. 
    The x-axis represents the sample size of historical data $\tilde n$.
    The y-axis represents the coverage probability which is the proportion of 1000 simulations in which the 95\% CI using each variance estimator includes the true value.}
\end{figure}
\clearpage
\begin{figure}[h]
    \centering
    \includegraphics[width=\linewidth]{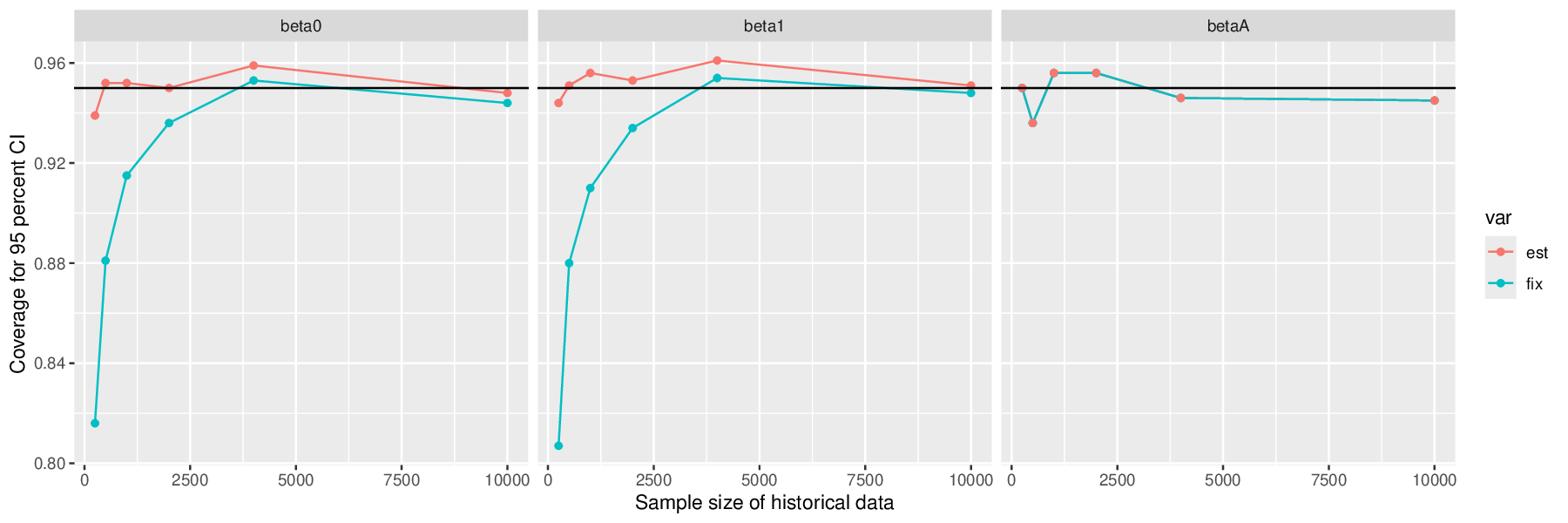}
    \caption{Plots of the coverage probability of 95\% CI over 1000 simulations for Scenario D-7.
    ``beta0'', ``betaA'', ``beta1'' represent the intercept $\beta_0$, the coefficient for the treatment assignment $\beta_A$, the coefficient for the prognostic score $\beta_1$ in the PROCOVA model \eqref{eq: PROCOVA linear}, respectively.
     ``fix'' represents $e^\top\hat V_{\text{fix}}e$ with $e=(1,0,0)^\top$ for $\beta_0$, with $e=(0,1,0)^\top$ for $\beta_A$ and $e=(0,0,1)^\top$ for $\beta_1$.
    ``est'' represents $e^\top\hat V_{\text{est}}e$ with $e=(1,0,0)^\top$ for $\beta_0$, with $e=(0,1,0)^\top$ for $\beta_A$ and $e=(0,0,1)^\top$ for $\beta_1$.
    The sample size of trial data is $n=1000$. 
    The x-axis represents the sample size of historical data $\tilde n$.
    The y-axis represents the coverage probability which is the proportion of 1000 simulations in which the 95\% CI using each variance estimator includes the true value.}
\end{figure}

\clearpage
\subsection{Scenario D-8}

\begin{figure}[h]
    \centering
    \includegraphics[width=\linewidth]{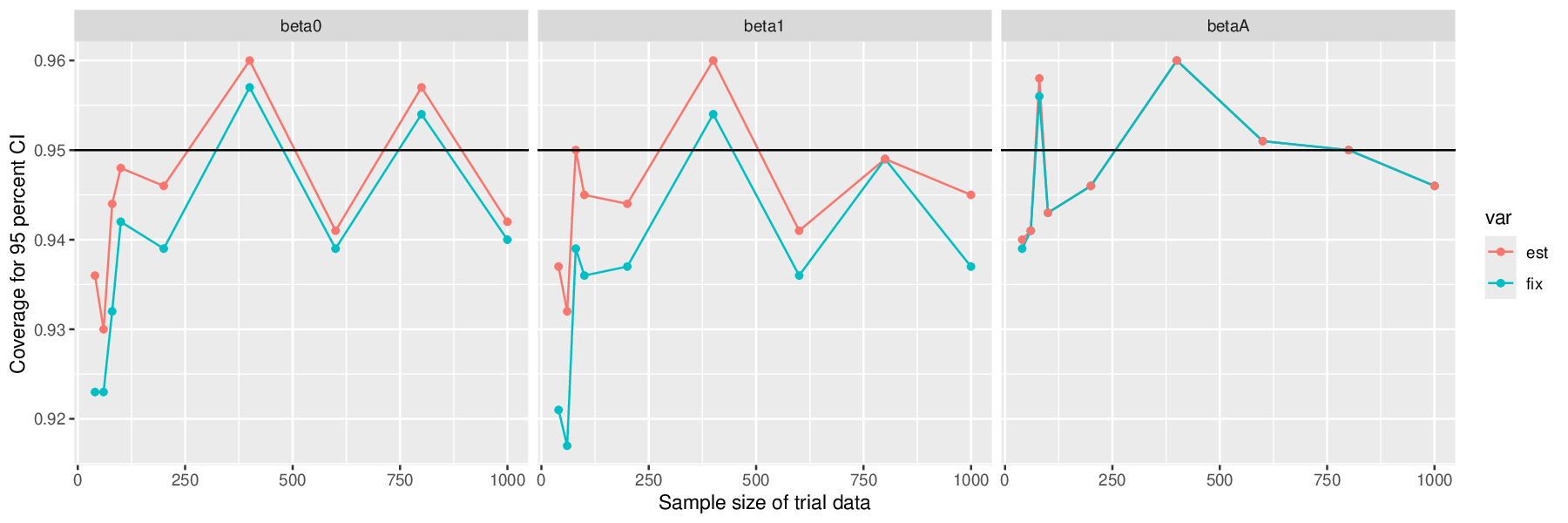}
    \caption{Plots of the coverage probability of 95\% CI over 1000 simulations for Scenario D-8.
    ``beta0'', ``betaA'', ``beta1'' represent the intercept $\beta_0$, the coefficient for the treatment assignment $\beta_A$, the coefficient for the prognostic score $\beta_1$ in the PROCOVA model \eqref{eq: PROCOVA linear}, respectively.
     ``fix'' represents $e^\top\hat V_{\text{fix}}e$ with $e=(1,0,0)^\top$ for $\beta_0$, with $e=(0,1,0)^\top$ for $\beta_A$ and $e=(0,0,1)^\top$ for $\beta_1$.
    ``est'' represents $e^\top\hat V_{\text{est}}e$ with $e=(1,0,0)^\top$ for $\beta_0$, with $e=(0,1,0)^\top$ for $\beta_A$ and $e=(0,0,1)^\top$ for $\beta_1$.
    The x-axis represents the sample size of trial data $n$.
    The sample size of historical data is $\tilde n = 10n$. 
    The y-axis represents the coverage probability which is the proportion of 1000 simulations in which the 95\% CI using each variance estimator includes the true value.
    }
\end{figure}

\clearpage
\begin{figure}[h]
    \centering
    \includegraphics[width=0.4\linewidth]{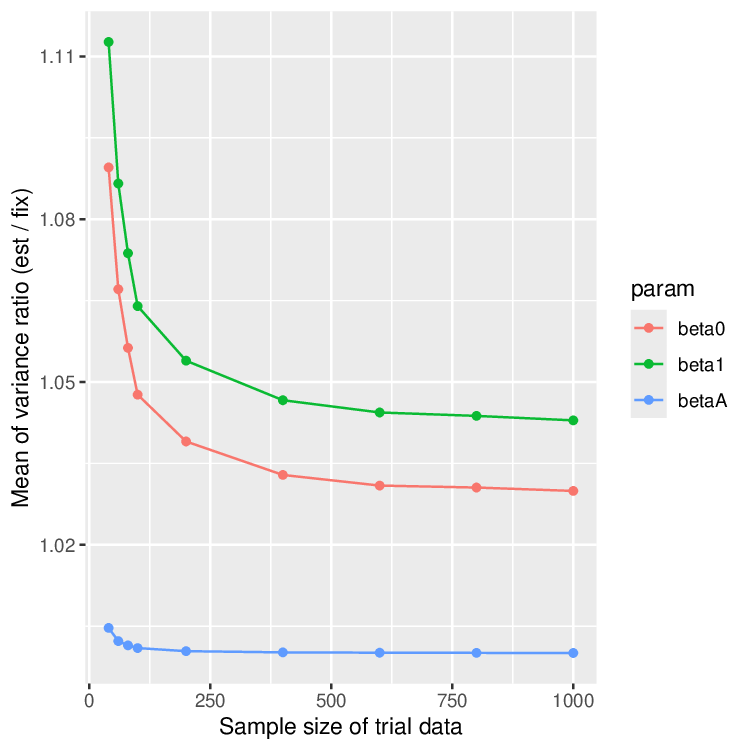}
    \caption{Plots of the mean of the ratio of two variance estimators over 1000 simulations for Scenario D-8.
    ``beta0'', ``betaA'', ``beta1'' represent the intercept $\beta_0$, the coefficient for the treatment assignment $\beta_A$, the coefficient for the prognostic score $\beta_1$ in the PROCOVA model \eqref{eq: PROCOVA linear}, respectively.
    The x-axis represents the sample size of trial data $n$.
    The sample size of historical data is $\tilde n = 10n$. 
    The y-axis represents the mean of the ratio of two variance estimators, i.e., $e^\top\hat V_{\text{est}}e/e^\top\hat V_{\text{fix}}e$ with $e=(1,0,0)^\top$ for $\beta_0$, with $e=(0,1,0)^\top$ for $\beta_A$ and $e=(0,0,1)^\top$ for $\beta_1$, over 1000 simulations.}
\end{figure}
\clearpage
\begin{figure}[h]
    \centering
    \includegraphics[width=0.4\linewidth]{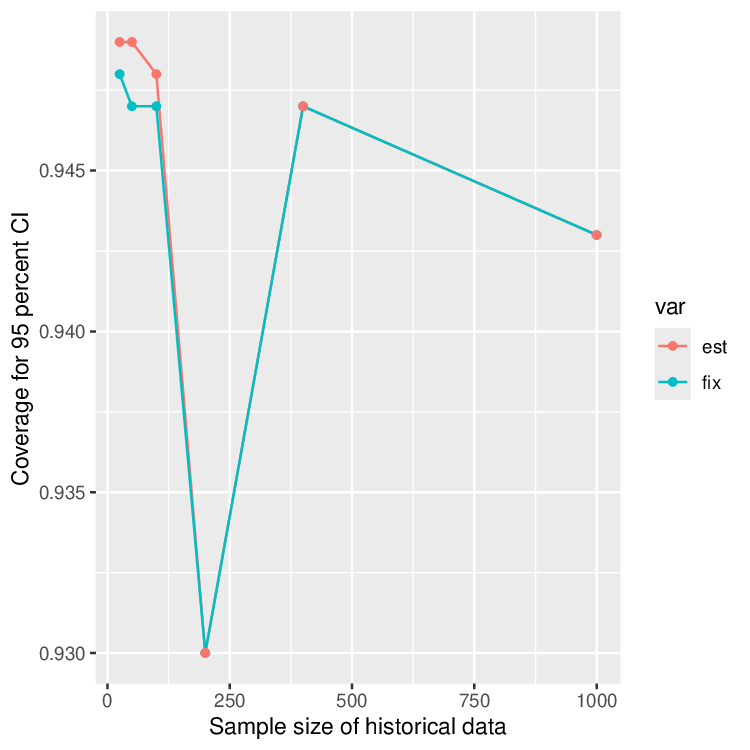}
    \caption{Plots of the coverage probability of the 95\% CI for $\beta_A$ in the PROCOVA model \eqref{eq: PROCOVA linear} over 1000 simulations for Scenario D-8.
     ``fix'' represents $e^\top\hat V_{\text{fix}}e$ with $e=(1,0,0)^\top$ for $\beta_0$, with $e=(0,1,0)^\top$ for $\beta_A$ and $e=(0,0,1)^\top$ for $\beta_1$.
    ``est'' represents $e^\top\hat V_{\text{est}}e$ with $e=(1,0,0)^\top$ for $\beta_0$, with $e=(0,1,0)^\top$ for $\beta_A$ and $e=(0,0,1)^\top$ for $\beta_1$.
    The sample size of trial data is $n=100$. 
    The x-axis represents the sample size of historical data $\tilde n$.
    The y-axis represents the coverage probability which is the proportion of 1000 simulations in which the 95\% CI using each variance estimator includes the true value.}
\end{figure}
\clearpage
\begin{figure}[h]
    \centering
    \includegraphics[width=\linewidth]{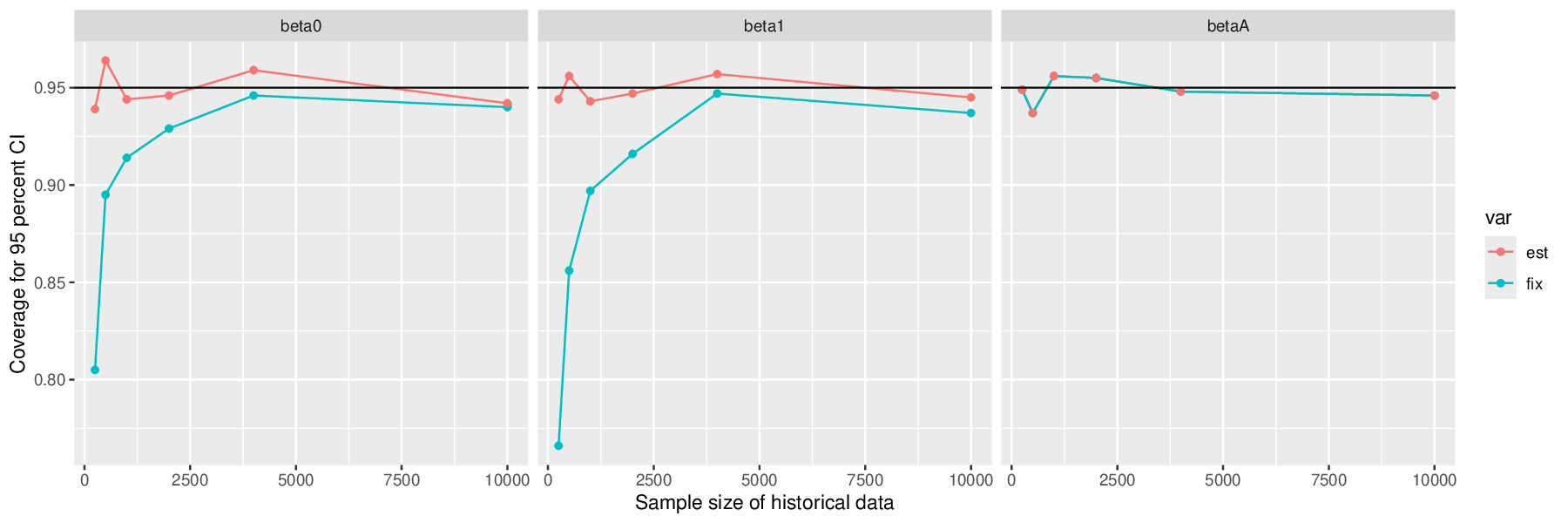}
    \caption{Plots of the coverage probability of 95\% CI over 1000 simulations for Scenario D-8.
    ``beta0'', ``betaA'', ``beta1'' represent the intercept $\beta_0$, the coefficient for the treatment assignment $\beta_A$, the coefficient for the prognostic score $\beta_1$ in the PROCOVA model \eqref{eq: PROCOVA linear}, respectively.
     ``fix'' represents $e^\top\hat V_{\text{fix}}e$ with $e=(1,0,0)^\top$ for $\beta_0$, with $e=(0,1,0)^\top$ for $\beta_A$ and $e=(0,0,1)^\top$ for $\beta_1$.
    ``est'' represents $e^\top\hat V_{\text{est}}e$ with $e=(1,0,0)^\top$ for $\beta_0$, with $e=(0,1,0)^\top$ for $\beta_A$ and $e=(0,0,1)^\top$ for $\beta_1$.
    The sample size of trial data is $n=1000$. 
    The x-axis represents the sample size of historical data $\tilde n$.
    The y-axis represents the coverage probability which is the proportion of 1000 simulations in which the 95\% CI using each variance estimator includes the true value.}
\end{figure}

\clearpage
\subsection{Scenario D-9}

\begin{figure}[h]
    \centering
    \includegraphics[width=\linewidth]{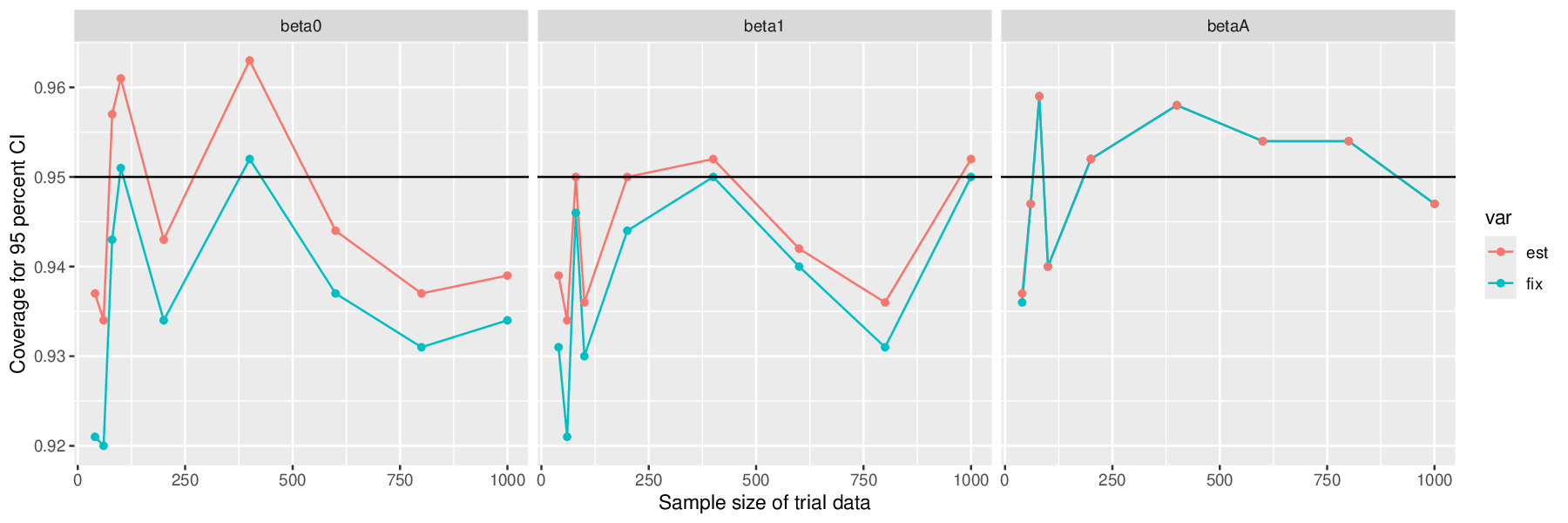}
    \caption{Plots of the coverage probability of 95\% CI over 1000 simulations for Scenario D-9.
    ``beta0'', ``betaA'', ``beta1'' represent the intercept $\beta_0$, the coefficient for the treatment assignment $\beta_A$, the coefficient for the prognostic score $\beta_1$ in the PROCOVA model \eqref{eq: PROCOVA linear}, respectively.
     ``fix'' represents $e^\top\hat V_{\text{fix}}e$ with $e=(1,0,0)^\top$ for $\beta_0$, with $e=(0,1,0)^\top$ for $\beta_A$ and $e=(0,0,1)^\top$ for $\beta_1$.
    ``est'' represents $e^\top\hat V_{\text{est}}e$ with $e=(1,0,0)^\top$ for $\beta_0$, with $e=(0,1,0)^\top$ for $\beta_A$ and $e=(0,0,1)^\top$ for $\beta_1$.
    The x-axis represents the sample size of trial data $n$.
    The sample size of historical data is $\tilde n = 10n$. 
    The y-axis represents the coverage probability which is the proportion of 1000 simulations in which the 95\% CI using each variance estimator includes the true value.
    }
\end{figure}

\clearpage
\begin{figure}[h]
    \centering
    \includegraphics[width=0.4\linewidth]{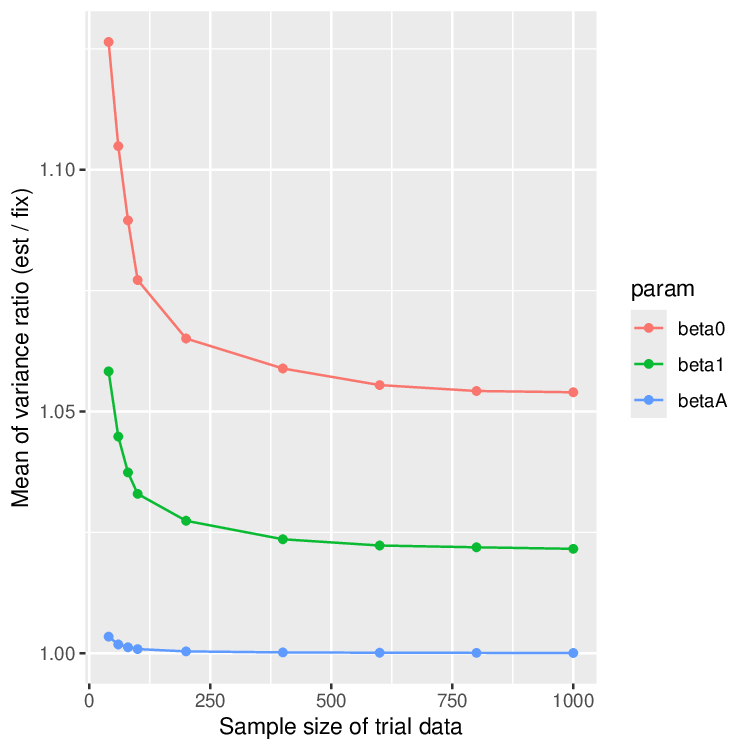}
    \caption{Plots of the mean of the ratio of two variance estimators over 1000 simulations for Scenario D-9.
    ``beta0'', ``betaA'', ``beta1'' represent the intercept $\beta_0$, the coefficient for the treatment assignment $\beta_A$, the coefficient for the prognostic score $\beta_1$ in the PROCOVA model \eqref{eq: PROCOVA linear}, respectively.
    The x-axis represents the sample size of trial data $n$.
    The sample size of historical data is $\tilde n = 10n$. 
    The y-axis represents the mean of the ratio of two variance estimators, i.e., $e^\top\hat V_{\text{est}}e/e^\top\hat V_{\text{fix}}e$ with $e=(1,0,0)^\top$ for $\beta_0$, with $e=(0,1,0)^\top$ for $\beta_A$ and $e=(0,0,1)^\top$ for $\beta_1$, over 1000 simulations.}
\end{figure}
\clearpage
\begin{figure}[h]
    \centering
    \includegraphics[width=0.4\linewidth]{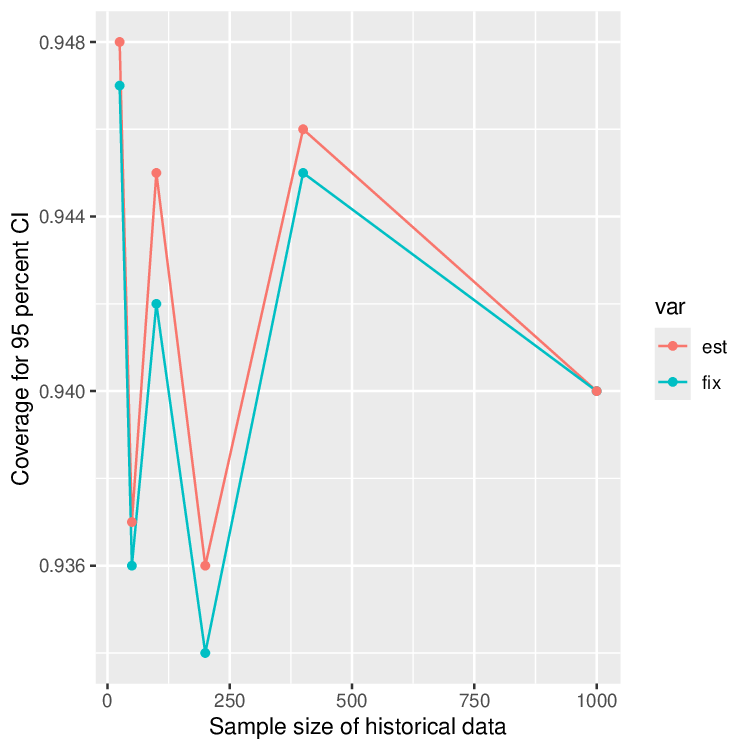}
    \caption{Plots of the coverage probability of the 95\% CI for $\beta_A$ in the PROCOVA model \eqref{eq: PROCOVA linear} over 1000 simulations for Scenario D-9.
     ``fix'' represents $e^\top\hat V_{\text{fix}}e$ with $e=(1,0,0)^\top$ for $\beta_0$, with $e=(0,1,0)^\top$ for $\beta_A$ and $e=(0,0,1)^\top$ for $\beta_1$.
    ``est'' represents $e^\top\hat V_{\text{est}}e$ with $e=(1,0,0)^\top$ for $\beta_0$, with $e=(0,1,0)^\top$ for $\beta_A$ and $e=(0,0,1)^\top$ for $\beta_1$.
    The sample size of trial data is $n=100$. 
    The x-axis represents the sample size of historical data $\tilde n$.
    The y-axis represents the coverage probability which is the proportion of 1000 simulations in which the 95\% CI using each variance estimator includes the true value.}
\end{figure}
\clearpage
\begin{figure}[h]
    \centering
    \includegraphics[width=\linewidth]{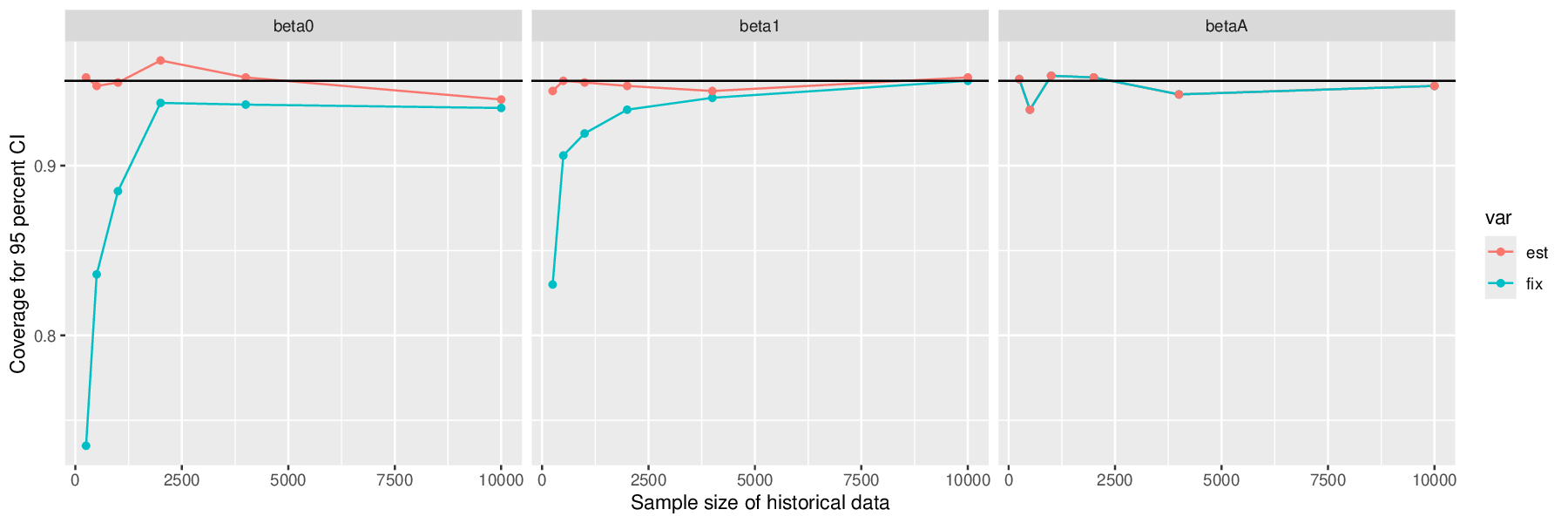}
    \caption{Plots of the coverage probability of 95\% CI over 1000 simulations for Scenario D-9.
    ``beta0'', ``betaA'', ``beta1'' represent the intercept $\beta_0$, the coefficient for the treatment assignment $\beta_A$, the coefficient for the prognostic score $\beta_1$ in the PROCOVA model \eqref{eq: PROCOVA linear}, respectively.
     ``fix'' represents $e^\top\hat V_{\text{fix}}e$ with $e=(1,0,0)^\top$ for $\beta_0$, with $e=(0,1,0)^\top$ for $\beta_A$ and $e=(0,0,1)^\top$ for $\beta_1$.
    ``est'' represents $e^\top\hat V_{\text{est}}e$ with $e=(1,0,0)^\top$ for $\beta_0$, with $e=(0,1,0)^\top$ for $\beta_A$ and $e=(0,0,1)^\top$ for $\beta_1$.
    The sample size of trial data is $n=1000$. 
    The x-axis represents the sample size of historical data $\tilde n$.
    The y-axis represents the coverage probability which is the proportion of 1000 simulations in which the 95\% CI using each variance estimator includes the true value.}
\end{figure}

\end{appendix}

\end{document}